\documentclass[11pt]{article}

\usepackage[text={6.8in,8.5in},centering]{geometry}
\usepackage{times}
\usepackage{amsfonts}
\usepackage{amsmath}
\usepackage{systeme}

\usepackage{amssymb}
\usepackage{amsthm}
\usepackage{graphicx}
\usepackage{url}
\usepackage{hyperref}
\usepackage[table]{xcolor}

\usepackage{color}
\usepackage{tcolorbox}
\usepackage[]{mdframed}
\usepackage{subcaption}
\usepackage{bbm}

\usepackage{mdframed}

\usepackage{float}

\graphicspath{{FigBeyond/}}  % Add both search paths

\usepackage{cancel}
\usepackage{fancyhdr}
\newtheorem{lemma}{Lemma}[section]
\newtheorem{theorem}[lemma]{Theorem}
 \newtheorem{definition}[lemma]{Definition}

\newtheorem{corollary}[lemma]{Corollary }

\newtheorem{remark}[lemma]{Remark}

\usepackage{blkarray}
\usepackage{tikz}
\usetikzlibrary{trees}

\usepackage{hyperref}
\newcommand{\doi}[1]{\href{https://doi.org/#1}{DOI: #1}}

\newcommand{\bra}[1]{\ensuremath{\langle#1|}}
\newcommand{\ket}[1]{\ensuremath{|#1\rangle}}
\newcommand{\mc}[1]{\ensuremath{\mathcal{#1}}}   % For Hilbert spaces and subspaces
\newcommand{\mb}[1]{\ensuremath{\mathbbm{#1}}}   % For Hamiltonians
\newcommand{\mr}[1]{\ensuremath{\mathrm{#1}}}    % For projection and transformation operators
\newcommand{\ms}[1]{\ensuremath{\mathsf{#1}}} % For subscripts or driver labels (already defined)
\newcommand{\mt}[1]{\ensuremath{\mathtt{#1}}}
\newcommand{\abbr}[1]{\ensuremath{\mathtt{#1}}}  % For abbreviations like \LM, \GM, etc.

\renewcommand{\bar}{\overline}

\newcommand{\X}{\ms{X}}
\newcommand{\Z}{\ms{Z}}
\newcommand{\sZ}{\ms{\tilde{Z}}}
\newcommand{\XX}{\ms{XX}}
\newcommand{\ZZ}{\ms{ZZ}}

\newcommand{\Sop}[1]{\mathrm{S}_{#1}}  % Generic spin operator

\numberwithin{equation}{section}

\newcommand{\ver}{{\ms{V}}}
\newcommand{\edge}{{\ms{E}}}

\newcommand{\Jxx}{J_{\ms{xx}}}
\newcommand{\Jzz}{J_{\ms{zz}}}

\newcommand{\eff}{{\ms{eff}}}
\newcommand{\Heff}{\mb{H}_1^{\ms{eff}}}

\newcommand{\Hinter}{\mb{H}_{\ms{inter\!-\!block}}}

\newcommand{\mis}{\ms{mis}}

\newcommand{\cl}{\ms{Clique}}

\usepackage{mathtools}

\usepackage{tabularx}
\usepackage{booktabs}
\usepackage{braket}
\usepackage{color}
\usepackage{tcolorbox}
\usepackage{array, xcolor, colortbl}
\newcolumntype{L}{>{\raggedright\arraybackslash}X}

\usepackage{enumitem}

\newcommand{\shz}[1]{\tilde{\sigma}^z_{#1}}

\newcommand{\Gdis}{\ensuremath{G_{\ms{dis}}}}
\newcommand{\Gshare}{\ensuremath{G_{\ms{share}}}}
\newcommand{\Gcnt}{\ensuremath{G_{\ms{contract}}}}

\newcommand{\x}[1]{\mt{x}(#1)}
\newcommand{\jxx}[1]{\mt{jxx}(#1)}

\newcommand{\bst}[1]{\mathbf{b}_{#1}}  % for binary string basis states

\newcommand{\SLE}{\mathcal{L}^{\text{ind}}}

\newcommand{\ELz}{E^{\ms{\tiny LM}}_0}
\newcommand{\EGz}{E^{\ms{\tiny GM}}_0}

\newcommand{\ESz}{E^{\ms{\tiny AS0}}_0}

\newcommand{\EGtrue}{E^{\ms{\tiny true}}_0}

\newcommand{\Jxxlift}{\Jxx^{\text{\tiny lift}}}
\newcommand{\Jxxsteer}{\Jxx^{\text{\tiny steer}}}
\newcommand{\Jxxsep}{\Jxx^{\text{\tiny sep}}}
\newcommand{\Jxxsink}{\Jxx^{\text{\tiny sink}}}

\newcommand{\Jzzinter}{\Jzz^{\text{\tiny inter}}}
\newcommand{\Jzzsteer}{\Jzz^{\text{\tiny steer}}}

\newcommand{\weff}{w^{\text{\tiny eff}}}

\newcommand{\LM}{\abbr{LM}}
\newcommand{\GM}{\abbr{GM}}
\newcommand{\MLIS}{\abbr{MLIS}}
\newcommand{\DMS}{\abbr{deg}-\abbr{MLIS}}
\newcommand{\MIC}{\abbr{dMIC}}
\newcommand{\MDC}{\abbr{dMDC}}
\newcommand{\GIC}{\abbr{GIC}}

\newcommand{\Ucombine}{U_{\ms{merge}}}

\newcommand{\DDD}{{\sc Dic-Dac-Doa{}}}

\newcommand{\Bc}{{\mathcal{B}}_{a}}
\newcommand{\Ba}{{\mathcal{B}}_{\ms{ang}}}

\newcommand{\PLE}{{\Pi^{\tiny{\ms{ind}}}}}

\newcommand{\Uc}{{\mathbf{U}^{\ms{clique}}}}

\def\final{1} % set this to 0 to get a comment-free version
\ifnum\final=1  %namely if we allow comments in the output
\newcommand{\vnote}[1]{[{\small Vicky: \bf #1}]\marginpar{*}}
\newcommand{\sidecomment}[1]{\marginpar{\tiny #1}}
\else % in this case [final=0] we don't want any comments to show
\newcommand{\vnote}[1]{}
\newcommand{\sidecomment}[1]{}
\fi  %ok, we're done with defining comment macros

\begin{document}

\title{Beyond Stoquasticity: Structural Steering and Interference in Quantum Optimization}

\author{
  Vicky Choi\\
  Gladiolus Veritatis Consulting Co.\footnote{\url{https://www.vc-gladius.com}}
}

\maketitle       

\begin{abstract}
We present a theoretical analysis of the \DDD{} algorithm, a
non-stoquastic quantum algorithm for solving the Maximum Independent
Set (MIS) problem. The algorithm runs in polynomial time and achieves
exponential speedup over both transverse-field quantum annealing (TFQA)
and classical algorithms on a structured family of NP-hard MIS
instances, under assumptions supported by analytical and numerical
evidence.
The core of this speedup lies in the ability of the evolving ground
state to develop both positive and negative amplitudes, enabled by the
non-stoquastic $\XX$-driver. This sign structure permits quantum
interference that produces negative amplitudes in the computational
basis, allowing efficient evolution paths beyond the reach of
stoquastic algorithms, whose ground states remain strictly
non-negative.
In our analysis, the efficiency of the algorithm is measured by the 
presence or absence of an anti-crossing, rather than by spectral gap 
estimation as in traditional approaches. The key idea is to infer it
from the crossing behavior of bare energy levels of relevant 
subsystems associated with the degenerate local minima (\LM{}) and the global minimum (\GM{}).
The cliques of the critical \LM{}, responsible for the anti-crossing in TFQA, 
can be efficiently identified to form the $\XX$-driver graph.
Based on the clique structure of \LM{}, we construct a decomposition
of the Hilbert space into same-sign and opposite-sign sectors,
yielding a corresponding block-diagonal form of the Hamiltonian. 
We then show that the non-stoquastic $\XX$-driver induces a see-saw 
effect that shifts their bare energies, leading to analytically derived 
bounds on $\Jxx$ which support a two-stage annealing schedule that 
prevents anti-crossings throughout the evolution. The resulting speedup 
can be attributed to two mechanisms: in the first stage, energy-guided 
localization within the same-sign block steers the ground state smoothly 
into the $\GM$-supporting region, while in the second stage, 
the opposite-sign blocks are invoked and sign-generating quantum 
interference drives the evolution along an opposite-sign path.
Finally, our analysis produces scalable small-scale models, derived
from our structural reduction, that capture the essential dynamics of
the algorithm. These models provide a concrete opportunity for
verification of the quantum advantage mechanism on currently available
universal quantum computers.
\end{abstract}

%% In essence, DDD first evolves to the GM-supporting region via (1) structural steering, and then completes the evolution to GM via (2) sign-generating quantum interference. The evolution follows a smooth, opposite-sign path. In contrast, TFQA (\(J_{xx}=0\)) first localizes to the LM-supporting region and then transitions to the GM-supporting region via anti-crossing tunneling; its path lies within the same-sign cone (nonnegative amplitudes).

\newpage
\tableofcontents
\newpage
%0.Introduction
\section{Introduction}
\label{sec:Intro}

A distinctive feature of quantum mechanics is that the wavefunction of a quantum state, expressed in a given basis, can have both positive and negative amplitudes---a characteristic with no classical counterpart, as classical probabilities must be non-negative.  
Motivated by the goal of harnessing this intrinsically quantum property,
we proposed the \DDD{} algorithm in~\cite{Choi2021,ChoiPatent} for the NP-hard Maximum Independent Set (MIS) problem.  
The acronym originally stood for \textit{Driver graph from Independent Cliques, Double Anti-Crossing, Diabatic quantum Optimization Annealing}.  
The algorithm modifies standard transverse-field quantum annealing (TFQA)~\cite{Farhi2000,Farhi2001,AQC-eq,AQC-Review} by adding a specially designed $\XX$-driver term, aiming to overcome small-gap anti-crossings. 
For completeness, we include the full algorithmic details and necessary refinements in this paper.

We now recall the system Hamiltonian for \DDD{}:
\[
\mb{H}(t) = \x{t} \mb{H}_{\ms{X}} + \jxx{t} \mb{H}_{\ms{XX}} + \mt{p}(t)\mb{H}_{\ms{problem}},
\]
where
\(
\mb{H}_{\ms{X}} = - \sum_{i} \sigma_i^x, 
\mb{H}_{\ms{XX}} = \sum_{(i,j) \in \edge(G_{\ms{driver}})} \sigma_i^x \sigma_j^x,
\)
and $\mb{H}_{\ms{problem}}$ is the MIS-Ising Hamiltonian defined in Equation~\eqref{eq:problem-Ham}.
The time-dependent parameter schedule $\jxx{t}$ depends on the $\XX$-coupling strength $\Jxx$.  
In particular, $\jxx{t} \equiv 0$ when $\Jxx = 0$,  
so the case $\Jxx = 0$ corresponds to TFQA, without the $\XX$-driver.
The system Hamiltonian is stoquastic (in the computational basis) if $\Jxx \le 0$, and non-stoquastic if $\Jxx >0$. 

The goal of our analysis is to demonstrate how and why an appropriately chosen \emph{non-stoquastic} (i.e., positive) coupling strength~$\Jxx$ enables exponential speedup with \DDD{}.
Specifically, we develop a theoretical framework for understanding the role of $\Jxx$ in shaping the Hamiltonian structure and guiding ground state evolution.  
We show that the algorithm succeeds in polynomial time by successively dissolving small-gap anti-crossings, one by one,  
for a structured class of input instances, called \GIC{} graphs.  
Despite their structure, we argue in Section~\ref{sec:MIS-hardness} that \GIC{} graphs exhibit classical hardness:  
solving MIS on such instances would require exponential time unless \( \mathrm{P} = \mathrm{NP} \).
Our analysis methods combine analytical techniques and numerical verification, with perturbative and effective Hamiltonian arguements
that can in principle be made rigorous with further work.

%%%%% same vs opposte sign, stoquasticity

First, we introduce some necessary terminology.
Throughout this paper, all Hamiltonians considered are real and Hermitian.  
Accordingly, we restrict our attention to quantum states with real-valued amplitudes.  
That is, the phase of each component is either \( 0 \) (corresponding to a positive sign) or \( \pi \) (corresponding to a negative sign).
We now formalize the sign structure of quantum states, which plays a central role in our analysis:
\begin{mdframed}
\begin{definition}
Let \( |\psi\rangle = \sum_{x \in \mathcal{B}} \psi(x) |x\rangle \) be a quantum state with real amplitudes in a basis \( \mathcal{B} \).

\begin{itemize}
    \item \( |\psi\rangle \) is called a \textbf{same-sign state} if \( \psi(x) \geq 0 \) for all \( x \in \mathcal{B} \).  
    That is, all components are in phase (with relative phase \( 0 \)).

    \item \( |\psi\rangle \) is called an \textbf{opposite-sign state} if there exist \( x, x' \in \mathcal{B} \) such that \( \psi(x) > 0 \) and \( \psi(x') < 0 \).  
    In this case, some components are out of phase, differing by a relative phase of \( \pi \).
\end{itemize}
\end{definition}
\end{mdframed}
Unless stated otherwise, we take \( \mathcal{B} \) to be the computational basis when referring to same-sign or opposite-sign states.
More generally, the computational basis is assumed whenever the basis is unspecified.

Accordingly, we define the notions of same-sign and opposite-sign bases, sectors, and blocks:

\begin{definition}
An orthonormal basis consisting entirely of same-sign states is called a same-sign basis.  
A basis that includes at least one opposite-sign state is called an opposite-sign basis.
A subspace spanned by a same-sign basis is called a same-sign sector; otherwise, it is called an opposite-sign sector.
A submatrix (or block) of a Hamiltonian is called a same-sign block if it is expressed in a same-sign basis.  
Otherwise, it is referred to as an opposite-sign block.
\end{definition}

\paragraph{Example.}
The state \( \ket{+} = \tfrac{1}{\sqrt{2}}(\ket{0} + \ket{1}) \) is a same-sign state,  
while \( \ket{-} = \tfrac{1}{\sqrt{2}}(\ket{0} - \ket{1}) \) is an opposite-sign state.  
The computational basis \( \{ \ket{0}, \ket{1} \} \) is a same-sign basis of \( \mb{C}^2 \),  
and the Hadamard basis \( \{ \ket{+}, \ket{-} \} \) is an opposite-sign basis.

These definitions are closely related to the concept of stoquasticity, which constrains the sign structure of the ground state.
It is well-known that, by the Perron--Frobenius theorem, the ground state of a stoquastic Hamiltonian (i.e., one whose off-diagonal elements are non-positive in a given basis) is a same-sign state in that basis~\cite{AQC-eq,general-PF2}. 
In particular, when expressed in the computational basis, the ground state of a stoquastic Hamiltonian is necessarily a same-sign state.  
By contrast, the ground state of a non-stoquastic Hamiltonian may be either a same-sign or an opposite-sign state.
We refer to the former as \emph{Eventually Stoquastic} and the latter as \emph{Proper Non-stoquastic}, as introduced in \cite{Choi2021}.

This distinction is crucial and can be understood in terms of the \emph{admissible subspace}, defined as the subspace of Hilbert space that the instantaneous ground state of the system Hamiltonian can dynamically access.
The shape and sign structure of this subspace differ fundamentally depending on whether the Hamiltonian is stoquastic or properly non-stoquastic.
In the (eventually) stoquastic case, the admissible subspace remains confined to the same-sign sector of the computational basis, where all amplitudes are non-negative.
In the properly non-stoquastic case, this subspace expands to include the opposite-sign sector, allowing superpositions with both positive and negative amplitudes.

It is therefore clear that by ``de-signing'' the non-stoquastic Hamiltonian~\cite{de-sign}, one reduces the admissible subspace.
The larger the admissible subspace, the more possible evolution paths become accessible.
This raises the central question: can some of these new paths---particularly those involving smooth transitions rather than exponentially slow tunneling---be efficiently exploited?

We answer this question affirmatively in this work.
In particular, we show that
\begin{itemize}
\item The \( \XX \)-driver opens access to opposite-sign sectors, enlarging the admissible subspace to include opposite-sign states;
\item In this expanded subspace, \emph{sign-generating quantum interference}---interference that produces negative amplitudes in the computational basis---enables a \emph{smooth evolution path} that bypasses tunneling.
\end{itemize}

This is the essence of going beyond stoquasticity: the ability to exploit a larger sign-structured subspace, which enables new evolution paths that are dynamically inaccessible in the stoquastic regime.

%%% general to bipartite structure
To demonstrate this idea concretely, 
we now turn to the algorithmic framework and its guiding structure.  
First, let $\LM$ and $\GM$ denote a set of degenerate local minima and the global 
minimum of the problem Hamiltonian. We now introduce the notion of the 
associated bare energy, which plays a key role in our analysis. 
For a subset $A \subseteq \ver(G)$, let $\mb{H}_A(t)$ denote the restriction of the Hamiltonian 
to the subspace spanned by the subsets of $A$. We define the ground 
state energy $E_0^{A}(t)$ as the bare energy associated with $A$. 
In particular, $E_0^{\LM}(t)$ and $E_0^{\GM}(t)$ denote the bare energies 
associated with $\LM$ (where $A = \bigcup_{M \in \LM} M$) and $\GM$, respectively.
The full algorithm consists of two phases; see Table~\ref{alg-overview}.
Phase~I uses polynomial-time TFQA to identify an independent-clique (IC)
structure associated with the critical $\LM$, which is responsible for the
(right-most) anti-crossing with $\GM$ in the TFQA ground state evolution.
Informally, such an anti-crossing arises from the crossing of the bare
energies associated with $\LM$ and $\GM$, see Figure~\ref{fig:AC-intro}.
%A precise definition is given in Section~\ref{sec:anti-crossings}.
Phase~II constructs an \( \XX \)-driver---based on the identified IC---to lift this anti-crossing.
This process is repeated iteratively: at each iteration, the algorithm identifies the next critical \LM{}, 
constructs an additional driver term, added to the system Hamiltonian.
%That is, the full driver Hamiltonian is constructed as an accumulated sum.
Provided that certain structural conditions are met, this iterative process can be applied effectively (see Figure~\ref{fig:iterative-dissolve} for a schematic depiction).
%further discussion is provided in Section~\ref{sec:full-system}.
In this way, the algorithm progressively removes small-gap obstructions from the system. 
The main focus of our analysis is therefore on Phase~II applied to a single bipartite substructure---in particular, on how the \( \XX \)-driver, through an appropriate choice of the parameter \( \Jxx \), enables the key mechanisms that dissolve each anti-crossing \emph{without introducing new ones}.

\begin{figure}[!htbp]
  \subcaptionbox{}{\includegraphics[width=0.30\textwidth]{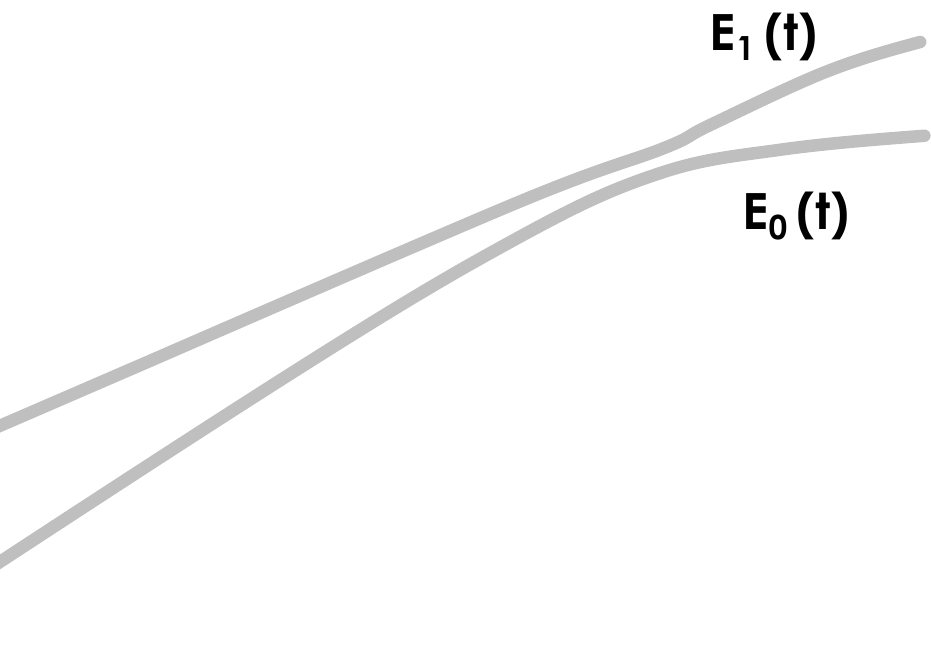}}
  \subcaptionbox{}{\includegraphics[width=0.32\textwidth]{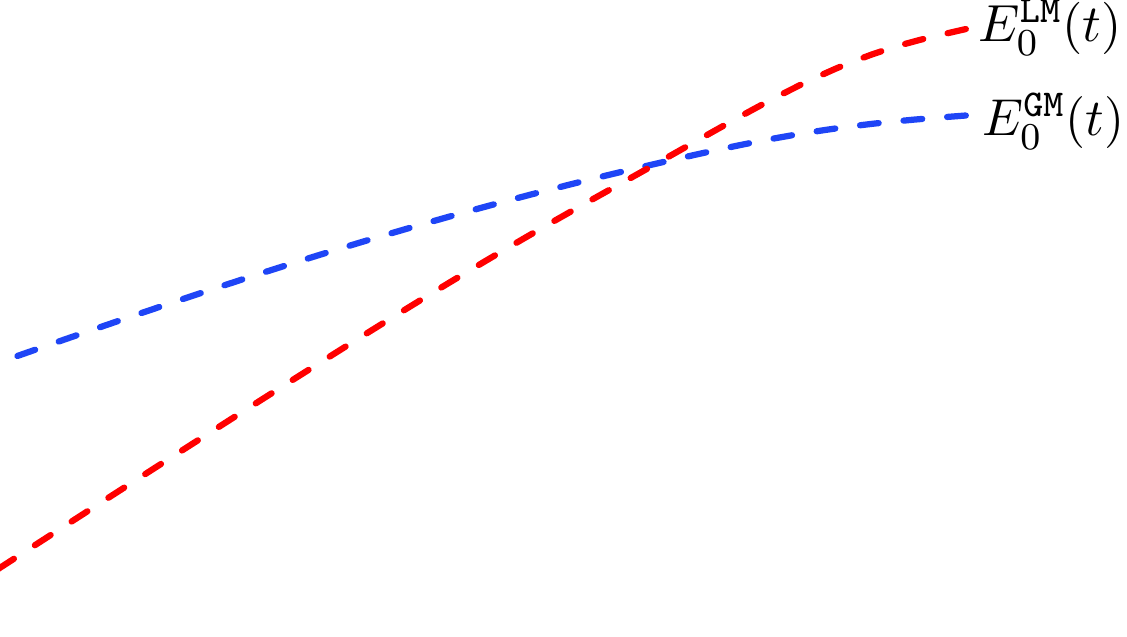}}
  \subcaptionbox{}{\includegraphics[width=0.34\textwidth]{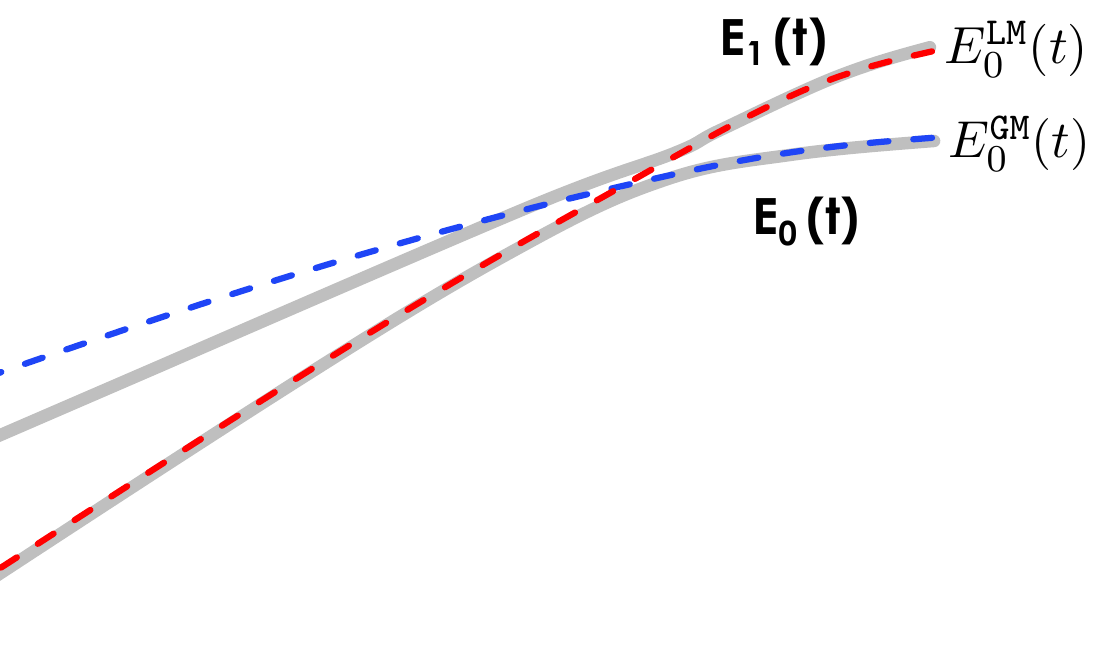}}
  \caption{
Illustration of an $(\LM,\GM)$-anti-crossing.  
(a) An anti-crossing between the lowest two levels $E_0(t)$ and $E_1(t)$.  
(b) Bare energies $E_0^{\LM}(t)$ and $E_0^{\GM}(t)$ cross.  
(c) Overlay showing that the anti-crossing originates from the bare crossing of $E_0^{\LM}(t)$ and $E_0^{\GM}(t)$.
}
\label{fig:AC-intro}
\end{figure}

%%%%%%

\begin{figure}[!htbp]
  \centering
  \subcaptionbox{}{\includegraphics[width=0.32\textwidth]{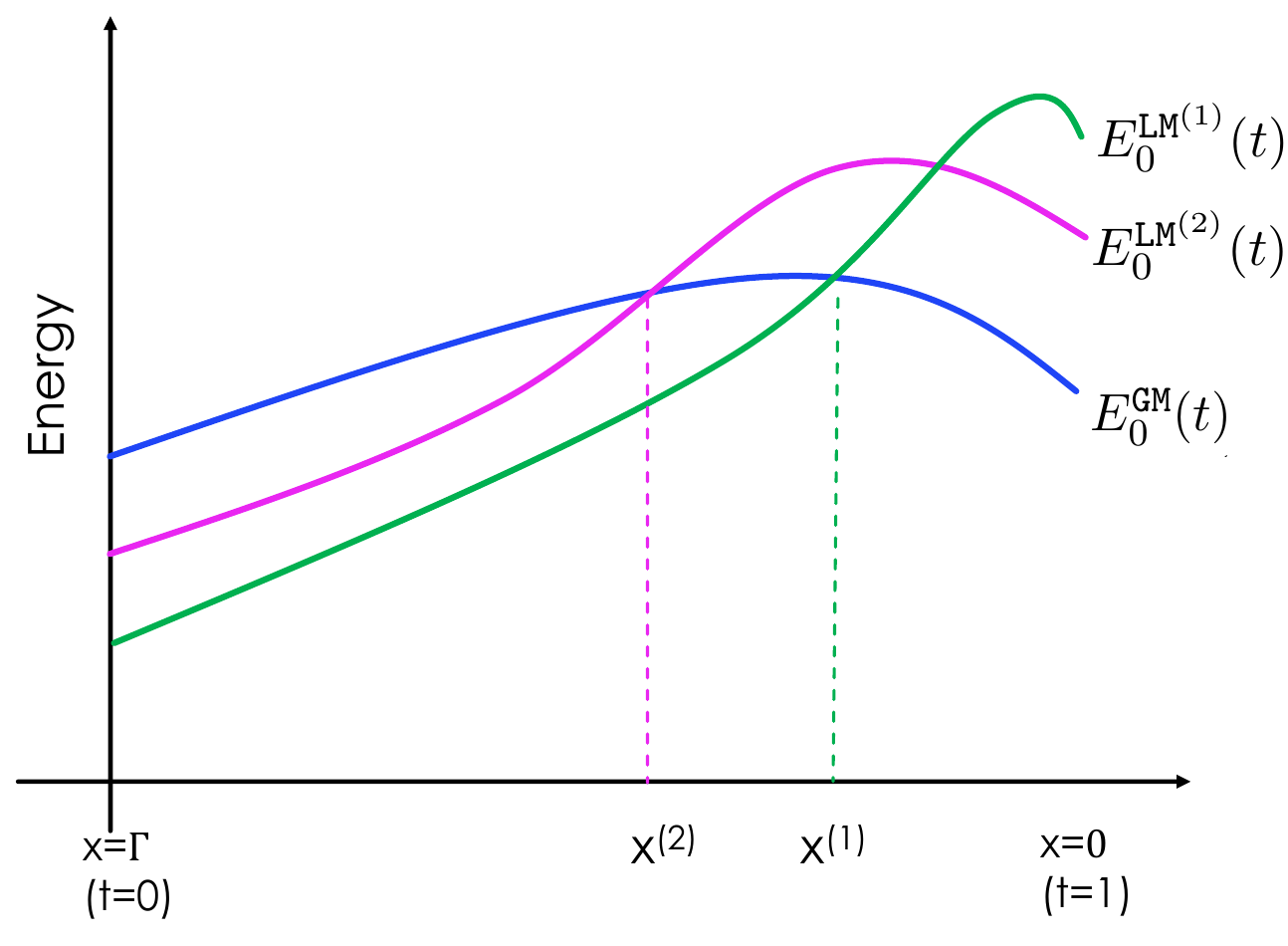}}
  \subcaptionbox{}{\includegraphics[width=0.32\textwidth]{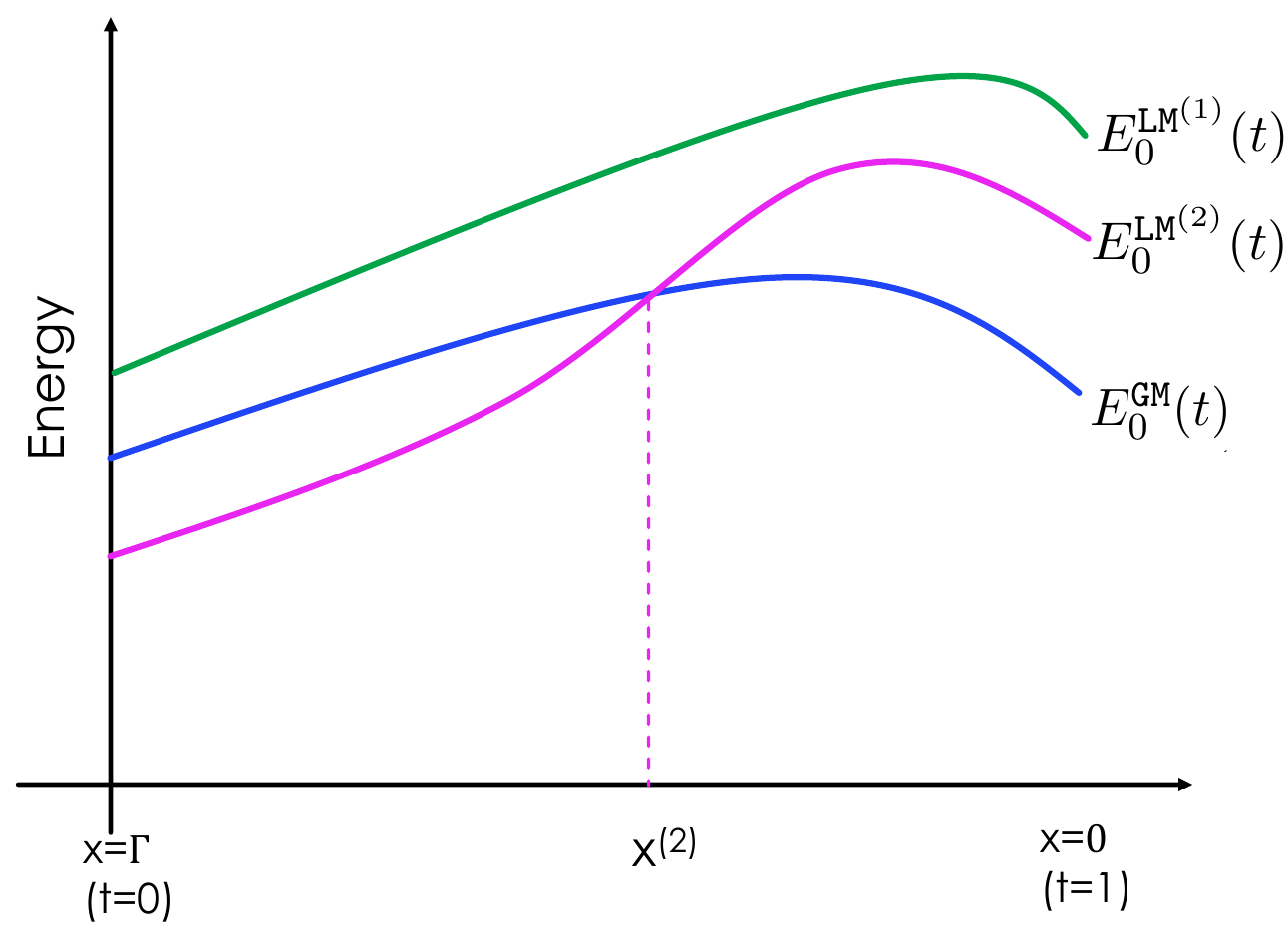}}
  \subcaptionbox{}{\includegraphics[width=0.32\textwidth]{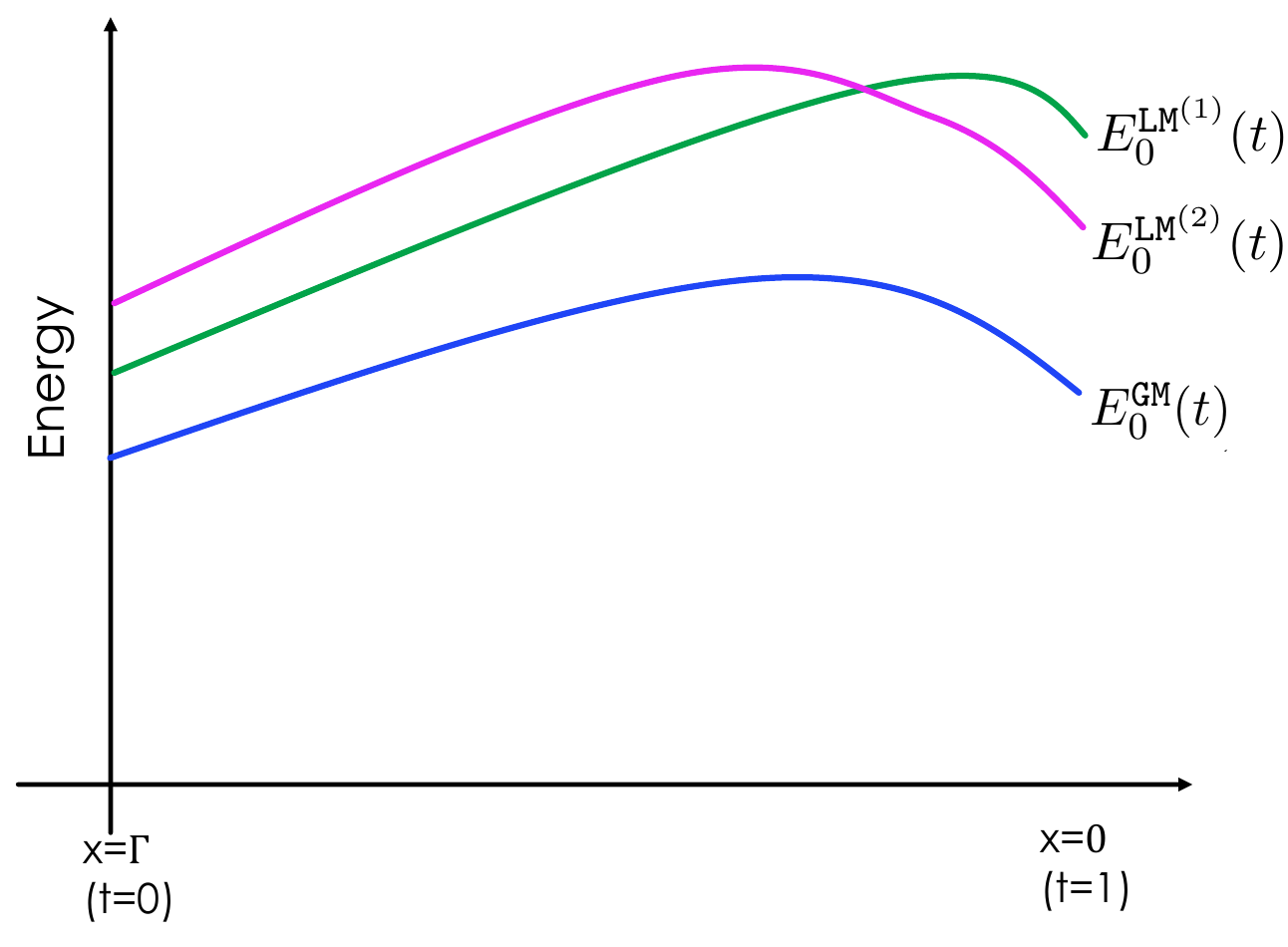}}
  \caption{
    Illustration of the iterative removal of anti-crossings by \DDD{}.  
    (a) Under TFQA, the energy associated with the global minimum $E_0^{\GM}(t)$ (blue) has bare crossings with the bare energies of two local minima, $E_0^{\LM^{(1)}}(t)$ and $E_0^{\LM^{(2)}}(t)$, at positions $x^{(1)}=\mt{x}(t_1)$ and $x^{(2)}=\mt{x}(t_2)$.  
    Each such bare crossing corresponds to an $(\LM^{(i)},\GM)$-anti-crossing in the system spectrum.  
    (b) After Phase~I of \DDD{}, the first crossing at $x^{(1)}$ is removed by lifting $E_0^{\LM^{(1)}}(t)$, while the second crossing at $x^{(2)}$ remains.  
    (c) Phase~II applies the same procedure to lift $E_0^{\LM^{(2)}}(t)$, thereby completing the removal of both anti-crossings.
  }
  \label{fig:iterative-dissolve}
\end{figure}

Informally,  the $\XX$-driver graph, which consists of a set of independent cliques, induces a natural angular momentum structure in the Hilbert space. This structure allows for a block decomposition of the Hamiltonian into disjoint \emph{same-sign} and \emph{opposite-sign} blocks. 
Importantly, in the signed basis induced by this decomposition, the off-diagonal terms involving non-stoquastic coupling strength \( \Jxx \) become diagonal.  
Consequently, \( \Jxx \) induces a \emph{see-saw} energy effect: it raises the energy associated with local minima in the same-sign block while lowering that of the opposite-sign blocks.  
This effect introduces a fundamental trade-off: if $\Jxx$ is too small, 
the local minima in the same-sign block are not lifted high enough 
(to remove the anti-crossing); if it is too large, the opposite-sign 
blocks drop too low, inducing new anti-crossings.
As a result, the success of the algorithmic design depends critically on choosing appropriate \( \Jxx \).
A central contribution of this work is the identification of four analytically derived bounds on \( \Jxx \).
%% These include:  
%% a \emph{Stage-Separation Upper Bound}, which ensures the irrelevance of opposite-sign blocks during early evolution;  
%% a \emph{Lifting Lower Bound}, which ensures the energy associated with \LM{} in the same-sign block is lifted high enough;  
%% a \emph{Steering Lower Bound}, which ensures that the localization to the global-minimum-supporting region occurs smoothly, without tunneling through an anti-crossing; and  
%% a \emph{Sinking Upper Bound},
%% which ensures that the opposite-sign block does not drop too low, avoiding a double anti-crossing whose gaps might not be small enough to allow successful diabatic transitions.
In addition, we identify two upper bounds on $\Jzz$ that constrain the 
problem Hamiltonian; these, together with the $\Jxx$ bounds, define the 
full feasibility window for successful evolution.
We then attribute the quantum speedup to two key mechanisms:
(1) steering the evolving ground state away from the local-minima-supporting region and directly into the global-minimum-supporting region (a mechanism we later formalize as \emph{structural steering});
and  
(2) enabling sign-generating quantum interference between blocks of different signs, which produces negative amplitudes in the computational basis.
%which lowers the energy of the interfered state, allowing it to become the true ground state.
The detailed quantum interference mechanism will be illustrated in Section~\ref{sec:V3}.
\footnote{By quantum interference, we mean the superposition of components in a quantum system that leads to amplitude cancellation or enhancement of the same basis state.
It is inherently quantum in nature, as only quantum systems admit decomposition into opposite-sign bases with complex or signed amplitudes. However, the mere presence of such interference does not, by itself, imply quantum advantage.  
In our system, the observed speedup is attributed to a restricted form of quantum interference---namely, interference that produces negative amplitudes in the computational basis.
This form of sign-generating interference is not accessible to classical simulation or stoquastic dynamics, whose ground states remain confined to the non-negative cone.
}

%Note that sign-generating interference is structurally entailed by successful steering.  
%We elaborate on this structural relationship in the Discussion.

To better understand the role of these mechanisms, we now return to the original motivation for introducing the \( \XX \)-driver---the formation of a double anti-crossing---and reinterpret its role in light of our current structural understanding.

\paragraph{Revisiting the Original Idea.}
We begin by revisiting the original motivation behind the \DDD{} algorithm, which was to convert the narrow anti-crossing in TFQA into a \emph{double anti-crossing} (double-AC) by introducing the $\XX$-driver. The idea was that this transformation would enable exponential speedup through diabatic annealing, via a ground-to-excited-to-ground state transition. Numerical evidence supporting this
idea---including the splitting of degenerate local minima and the formation of a double-AC---was presented in~\cite{Choi2021}.%
\footnote{More concretely, we constructed an $\XX$-driver graph \( G_{\ms{driver}} \) from the set of independent cliques corresponding to the critical degenerate local minima \( L \) identified via polynomial-time TFQA. By applying a sufficiently large $\XX$-coupling strength \( \Jxx \), we induced a splitting of the states in \( L \) into two subsets with opposite amplitudes, denoted \( L^+ \) and \( L^- \), ensuring that \( L \) remained in the instantaneous ground state when its energy was minimized. This resulted in two anti-crossings bridged by \( (L^+, L^-) \), forming what we called a double-AC. This behavior was confirmed numerically through exact diagonalization.}

While this observation guided the early design of the algorithm, our current structural understanding provides a more precise explanation.
%However, our further analysis reveals a more nuanced explanation of this phenomenon.  
The apparent double-AC actually arises from a pair of anti-crossings
between the ground-state energies of two blocks---one same-sign and one opposite-sign.\footnote{This was first
observed by J.~Kerman.}  
When these blocks are decoupled or only weakly coupled, the transition can be reinterpreted as effective adiabatic evolution confined to the same-sign block,
with the double-AC dynamically bypassed (see Section~\ref{sec:anti-crossings}).
%See Figure~\ref{fig:block-ac} in Section~\ref{sec:anti-crossings} for a schematic illustration of how an anti-crossing may be dynamically bypassed.
In general, however, the coupling between the same-sign and opposite-sign blocks may not be weak,  
and transitions between blocks become possible,  
and the double-AC mechanism does not necessarily apply.
\footnote{For example, the first anti-crossing may remain a block-level transition, 
while the second---involving quantum interference---does not correspond to a small gap.}

We conclude that the double-AC---though visually striking and initially suggestive of significance---is not essential for achieving quantum advantage.  
What initially appeared as negative amplitude in the ground state can now be more precisely understood  
as a manifestation of sign-generating quantum interference between same-sign and opposite-sign blocks.
Still, it was the numerical observation of this phenomenon in our original weighted MIS example  
that led to the discovery of the Hamiltonian’s block structure  
and the mechanisms governing the evolution of the quantum ground state.
Even if the double-AC ultimately serves only as a by-product of a deeper structural principle,  
its appearance was a gateway to the framework we now develop.

\noindent\textbf{Remark.} Due to historical reasons, we retain the name \DDD{},  
even though the roles of {\sc Dac} and {\sc Doa} have evolved in light of our current understanding.  
Alternatively, one might consider renaming it to {\sc Dic-Dee-Daw}, where {\sc Dee-Daw} reflects the \emph{see-saw} effect.

%% Our algorithm proceeds in two phases (see Table~\ref{alg-overview}):  
%% Phase~I identifies critical local minima that cause anti-crossings in TFQA,  
%% and constructs a tailored driver graph that targets the removal of the associated anti-crossing.  
%% Phase~II applies a structured two-stage quantum annealing process using this driver  
%% to achieve successful ground state evolution without encountering a (new) anti-crossing.

In summary, \DDD{} is motivated by turning around the obstacle of traditional stoquastic quantum annealing (TFQA)---specifically,
the presence of an anti-crossing caused by the competition between the
energies associated with $\LM{}$ and $\GM$.
This anti-crossing obstacle has long been observed and explained using perturbative arguments in the small transverse-field regime (e.g.,~\cite{Amin-Choi,AKR}).  
In our companion paper~\cite{Choi-Limitation}, we provide an explicit structural explanation and an analytical proof  
for such an anti-crossing beyond the small transverse-field perturbative regime, for the class of structured input graphs.

In our analysis, the efficiency of the algorithm is measured by the 
presence or absence of an anti-crossing, rather than by spectral gap 
estimation as in traditional approaches (see e.g.~\cite{AQC-Review} 
and references therein). An anti-crossing can be rigorously defined 
through reduction to a $2 \times 2$ effective Hamiltonian, with the 
precise definition given in Section~\ref{sec:anti-crossings}. 
The key idea is that the unperturbed energies (the diagonal entries of the 
effective Hamiltonian) can be well-approximated by the bare energies of 
relevant subsystems: the local minima (\LM{}), whose cliques are efficiently 
identified in Phase~I and used to construct the $\XX$-driver graph, and the 
global minimum (\GM{}), which is assumed as a reference in the analysis.
This makes it possible to infer the 
presence or absence of an anti-crossing directly from the crossing 
behavior of bare energy levels, without explicitly constructing 
the effective two-level Hamiltonian.
In particular, our structural analysis consists of three main 
components, supported by both analytical and numerical evidence:

\begin{itemize}
  \item A structural decomposition of the Hilbert space induced by 
  the $\XX$-driver graph, separating same-sign and opposite-sign 
  blocks via angular momentum techniques. 

  \item Identification of a see-saw energy effect induced by the 
  non-stoquastic parameter $\Jxx$, which raises the energy associated with local 
  minima in the same-sign block while lowering that of opposite-sign 
  blocks. 

  \item Derivation of four analytical bounds on $\Jxx$ and design of 
  a two-stage annealing schedule that guarantees correct ground state 
  evolution without encountering an anti-crossing, confirmed by 
  numerical evidence. 
\end{itemize}

%% The resulting speedup can be attributed to (1) structural steering and (2) sign-generating quantum interference effect
%% that results in an opposite-sign path, beyond reach of the stoquastic annealing algorithm.

A by-product contribution of this work is the design of effective small-scale models, derived from our structural reduction, that capture the essential dynamics of the algorithm. These models provide a concrete opportunity for verification of the quantum advantage mechanism on currently available gate-model quantum devices through simulation.

The paper is organized into four Parts.

\textbf{Part~I: Structural Setup and Annealing Framework.}  
In Section~\ref{sec:MIS}, we review the MIS-Ising Hamiltonian and the role of the \( \Jzz \) coupling.
Section~\ref{sec:MIS-hardness} introduces the class of \GIC{} graphs for which we argue classical hardness, and describe the two
fundamental bipartite substructures formed by the critical degenerate local minima and the global minimum.
Section~\ref{sec:2-stage} introduces the revised annealing schedule used in Phase~II of the algorithm, and explains how the components \( \x{t} \), \( \jxx{t} \), and \( \mt{p}(t) \) vary across the evolution.

\textbf{Part~II: Analytical Tools and Block Decomposition.}  
Section~\ref{sec:anti-crossings} defines and clarifies the notion of anti-crossings,  
and presents the basic two-level matrix used throughout the analysis.  
Section~\ref{sec:single-clique} develops the core framework for analyzing a single clique,  
including the block decomposition of its low-energy subspace based on its angular momentum structure.
Section~\ref{sec:bare-sub} presents a closed-form analytical solution for the \emph{bare subsystem},  
which consists of a collection of independent cliques.

\textbf{Part~III: Two-Stage Dynamics and Full-System Application.}  
Section~\ref{sec:stage-0} analyzes the optional Stage~0 and derives the effective low-energy Hamiltonian that initiates Stage~1.  
Section~\ref{sec:stage-main} contains the main technical results of the paper, focusing on a single bipartite substructure. We analyze the decomposition of the Hamiltonian into same-sign and opposite-sign blocks, the role of \( \Jxx \) in steering and interference, and provide analytical bounds along with numerical illustrations.
%% We elaborate the
%%  V3 conceptual model to illustrate the 
%%   sign-generating quantum interference in Section~\ref{sec:V3}.
In Section~\ref{sec:full-system}, we extend the analysis to the full system by iteratively applying the two-phase procedure to successive substructures.

\textbf{Part~IV: Conclusion.}  
We conclude in Section~\ref{sec:discussion} with a discussion of implications, limitations, and future directions.

\begin{table*}[ht]
\centering
\begin{mdframed}[linewidth=0.5pt, roundcorner=4pt]
\paragraph*{Algorithm Overview.}
The full algorithm consists of two phases:

\begin{itemize}
  \item \textbf{Phase I:} Extract an independent-clique (IC) structure using a polynomial-time TFQA process:
  \begin{enumerate}
    \item Use TFQA and polynomial annealing time to return the excited states involved in the anti-crossing.
    \item Extract a set of seeds (local minima states) from the resulting output.
    \item Apply a classical procedure to identify an independent-clique (IC) structure based on these seeds.
  \end{enumerate}

  \item \textbf{Phase II:} Define an $\XX$-driver graph using the IC structure and perform an annealing process with tunable \( \Jxx \):
  \begin{enumerate}
    \item Construct an $\XX$-driver graph \( G_{\ms{driver}} \) using the IC structure identified in Phase~I.
    \item Define a new time-dependent Hamiltonian \( \mb{H}(t) = \x{t} \mb{H}_{\ms{X}} + \jxx{t} \mb{H}_{\ms{XX}} + \mt{p}(t)\mb{H}_{\ms{problem}} \).
    \item For each feasible \( \Jxx \), evolve adiabatically according to a three-stage schedule:
    \begin{itemize}
      \item \textbf{Stage~0} (optional initialization);
      \item \textbf{Stage~1} (energy-guided localization);
      \item \textbf{Stage~2} (interference-driven transition).
    \end{itemize}
    \item Measure the final ground state to extract the optimal solution.
  \end{enumerate}
\end{itemize}
\end{mdframed}
\caption{Two-phase structure of the full \DDD{} algorithm. Phase~I uses stoquastic TFQA to extract seeds and identify an IC structure; Phase~II applies a non-stoquastic driver with tunable \( \Jxx \) to lift anti-crossings.}
\label{alg-overview}
\end{table*}

%% %Part I. MIS and GIC and schedule
\section{MIS Problem and $\Jzz$ Coupling}
\label{sec:MIS}

The NP-hard weighted Maximum Independent Set (MIS) problem is formulated as follows:

\begin{mdframed}
\textbf{Input:} An undirected graph \( G = (\ver(G), \edge(G)) \) with \( N = |\ver(G)| \), where each vertex \( i \in \ver(G) = \{1, \ldots, N \} \) is assigned a positive rational weight \( w_i \).\\
\noindent
\textbf{Output:} A subset \( S \subseteq \ver(G) \) such that \( S \) is independent (i.e., for each \( i,j \in S \), \( (i,j) \not \in \edge(G) \)), and the total weight of \( S \), given by  
\(
\mt{w}(S) = \sum_{i \in S} w_i,
\) 
is maximized. We denote this optimal set as \( \mis(G) \).  
\end{mdframed}

For simplicity, we focus on the unweighted MIS problem, where all \( w_i =w=1 \). However, we retain \( w_i \) in our formulation to allow for generalization to the weighted case and for later analysis purpose. In the unweighted setting, \( \mis(G) \) corresponds to the largest independent set in \( G \).

%% \paragraph{Notation.}
%% We use \( 2^{[N]} \) to denote the set of all bit strings of length \( N \), where each string corresponds to a subset of \( \{1, 2, \ldots, N\} \); positions with value 1 indicate inclusion in the subset. When the meaning is clear, we interchangeably refer to a subset and its corresponding bit string. For any \( k \in 2^{[N]} \), the ket \( |k\rangle \) denotes the quantum state associated with subset \( k \), and the size of \( k \) is given by its Hamming weight---the number of 1s in the bit string.

\subsection{The Role of \( \Jzz \) in the MIS-Ising Hamiltonian}
\label{sec:mis-Ising}

As shown in~\cite{minor1}, the MIS problem can be encoded in the ground state of the MIS-Ising Hamiltonian:
\begin{equation}
  \label{eq:Ising}
\mb{H}_{\ms{MIS\!-Ising}}(G) = \sum_{i \in \ver(G)} (-w_i) \shz{i} + \sum_{(i,j) \in \edge(G)} J_{ij} \shz{i} \shz{j},
\end{equation}
where \( J_{ij} > \max\{w_i, w_j\} \) for all \( (i,j) \in \edge(G) \). This formulation slightly modifies that of~\cite{minor1,Choi2021}, replacing \( \sigma^z \) with the shifted-\( \sigma^z \) operator \( \shz{i} := \tfrac{I + \sigma^z_i}{2} \), whose eigenvalues are \( \{0, 1\} \), making the correspondence with the classical energy function direct:
\begin{equation}
  \label{eq:MIS-eng}
  \mathcal{E}(G) = \sum_{i \in \ver(G)} (-w_i) x_i + \sum_{(i,j) \in \edge(G)} J_{ij} x_i x_j,
\end{equation}
where \( x_i \in \{0,1\} \). The energy is minimized when no two adjacent vertices \( i, j \) satisfy \( x_i = x_j = 1 \), ensuring that the ground state corresponds to \( \mis(G) \).
For convenience, we refer to all \( J_{ij} \) as \( \Jzz \), although these couplings need not be uniform, even in the unweighted case. The only requirement is that \( \Jzz > \max\{w_i, w_j\} \) for each edge \( (i,j) \in \edge(G) \).

In the unweighted case (\( w_i = 1 \)), an independent set of size \( m \) has energy
\(
E_{\ms{ind}} = -m,
\)
while a dependent set has energy
\(
E_{\ms{dep}} = \Jzz \cdot (\# \text{edges}) - (\# \text{vertices}).
\)
Thus, large \( \Jzz \) creates a large energy separation between independent-set states and dependent-set states. 
In principle, \( \Jzz \) can be chosen arbitrarily large; however, excessively large values of \( \Jzz \) can be detrimental. 
With the use of the \( \XX \)-driver in \DDD{}, the dependent-set states can act as ``bridges'' between different groups of independent-set states to facilitate the smooth structural steering (to be elaborated in Section~\ref{sec:stage1}).
The appropriate choice of \( \Jzz \) is crucial for the success of the algorithm.

For \DDD{}, we use two different couplings:
\begin{itemize}
\item \( \Jzz^{\ms{clique}} \): assigned to edges within cliques of the driver graph;
\item \( \Jzz \): assigned to all other edges.
\end{itemize}
That is, the MIS–Ising problem Hamiltonian takes the form
\begin{align}
\mb{H}_{\ms{problem}} 
&= \sum_{i \in \ver(G)} (-w_i) \shz{i} 
+ \Jzz^{\ms{clique}} \sum_{(i,j) \in \edge(G_{\ms{driver}})} \shz{i} \shz{j} 
+ \Jzz \sum_{(i,j) \in \edge(G) \setminus \edge(G_{\ms{driver}})} \shz{i} \shz{j}.
\label{eq:problem-Ham}
\end{align}
The value of \( \Jzz^{\ms{clique}} \) is set sufficiently large to restrict the system to the clique low-energy subspace, 
while \( \Jzz \) must satisfy two upper bounds, \( \Jzzinter \) and \( \Jzzsteer \), 
to ensure the success of the algorithm. 
The precise role of these bounds will be analyzed in Section~\ref{sec:sub4}.

\section{Classical Hard Instances of MIS and the \GIC{} Instances}
\label{sec:MIS-hardness}

Recall that an independent set is \emph{maximal} if no larger independent set contains it.  
Each maximal independent set corresponds to a local minimum of the energy function in Eq.~\eqref{eq:MIS-eng}.  
A collection of maximal independent sets (\MLIS{}) all having the same size \( m \)  
corresponds to a set of \emph{degenerate} local minima with equal energy \( -m \).  
When the meaning is clear (i.e., they all have the same size), we refer to such a collection simply as a \DMS{},  
and its cardinality as the \emph{degeneracy}.  
In this work, we use the terms \emph{degenerate local minima} and \DMS{} interchangeably.

This section explores the relationship between classical hardness and the structure of \DMS{}, 
with a focus on how degeneracy influences computational complexity in both classical and quantum settings. 
In Section~\ref{sec:nece-hardness} we show that classical hardness requires exponentially many \MLIS{}, 
with at least one \DMS{} exhibiting exponential degeneracy. 
In Section~\ref{sec:struct-deg} we introduce a structured form of degeneracy, denoted by \MIC{}. 
Finally, in Section~\ref{sec:GIC} we define the class of \emph{Graphs with Independent Cliques} (\GIC{}), 
together with a reduction to fundamental bipartite structures---\Gdis{} and \Gshare{}---that capture both classical hardness and 
the structural features exploited by our algorithm.

\subsection{Necessary Classical Hardness Conditions}
\label{sec:nece-hardness}
The MIS problem was among the first shown NP-complete~\cite{GJ79}. A simple branching
algorithm~\cite{Tsukiyama1977,Johnson1988} without pruning (i.e., branch-without-bound)
enumerates all \MLIS{} in time \(O(N \cdot \#\MLIS)\), where $\#\MLIS$ denotes the number of sets in \MLIS{}.
Thus, MIS can be solved by identifying the largest among them. In the worst case,
$\#\MLIS = O(3^{N/3})$, though often significantly smaller---yet still exponential.

\begin{mdframed}
\noindent\textbf{Observation 1.}
A necessary condition for classical hardness is that the MIS instance contains exponentially many \MLIS{}.\\
\textbf{Observation 2.}
A necessary consequence is that at least one \DMS{} must exhibit \emph{exponential degeneracy}.
\end{mdframed}

Observation~1 follows directly from the enumeration algorithm above.  
Observation~2 follows from the pigeonhole principle:  
since an independent set can have at most \( N \) distinct sizes,  
at least one size class must contain exponentially many \MLIS{}.  
Otherwise, if all \DMS{} had at most polynomial degeneracy,  
the total number of \MLIS{} would be polynomial---contradicting Observation~1.

As a side note, it is worth emphasizing that classical algorithms for {\sc MIS}
outperform Grover's unstructured quantum search~\cite{Grover1996,Roland2002}, which requires \( O(2^{N/2}) \) time. 
As early as 1977, a simple and implementable algorithm achieved a runtime of \( O(2^{N/3}) \) via careful branching analysis~\cite{Tarjan1977}, and this exponent has been further improved over time using more sophisticated branch-and-bound techniques, see~\cite{MWIS2019} for references.
Since other optimization problems such as {\sc Exact-Cover} can be reduced to {\sc MIS} with the same problem size~\cite{Choi-NP-reduction}, it follows that unstructured adiabatic quantum optimization~\cite{UnstructuredAQO2024}---while perhaps of theoretical interest---offers no practical algorithmic advantage.

\subsection{\MIC{}: Structured Form of a \DMS{}}
\label{sec:struct-deg}
It is widely believed that a \DMS{} with exponential degeneracy causes a tunneling-induced anti-crossing in TFQA.  
We call such a \DMS{} \emph{critical}.
The exponentially many \MLIS{} in a critical \DMS{} must exhibit substantial vertex-sharing.  
This suggests a natural partitioning of the involved vertices into \( k \) cliques,  
where each maximal independent set includes exactly one vertex from each clique.  
However, efficiently identifying such a partition may not always be feasible,  
and multiple partitions may be needed to fully characterize a critical \DMS{}.

In its simplest structured form, a critical \DMS{} can be represented as a \MIC{}, 
namely a collection of mutually independent cliques whose vertices together generate exactly all the \MLIS{} in that \DMS{}.

\begin{definition}
A \MIC{} of size \( k \) consists of \( k \) \emph{independent cliques} (i.e., no edges exist between them),  
denoted as \( \text{Clique}(w_i, n_i) \),  
where \( w_i \) is the vertex weight and \( n_i \) is the clique size, for \( 1 \leq i \leq k \).  
Each maximal independent set in the corresponding \DMS{} is formed by selecting exactly one vertex from each clique.  
In this case, the degeneracy of the \MIC{} is given by
\(
\prod_{i=1}^{k} n_i.
\)
\end{definition}

Moreover, given any single maximal independent set as a seed,  
all \( k \) cliques in the \MIC{} can be identified in linear time due to the independence condition.  
For each vertex in the seed, the corresponding clique is obtained by collecting all vertices  
that are mutually adjacent to it and simultaneously non-adjacent to all other seed vertices.  
This procedure depends only on vertex adjacencies in the graph.

Under the assumption of the \MIC{} structure,  
we analytically establish in the companion paper~\cite{Choi-Limitation}  
that TFQA encounters a tunneling-induced anti-crossing with an exponentially small gap.  
(This generalizes the so-called perturbative crossing,  
which is defined only in the small transverse-field regime.)  
While such an $(\LM{},\GM{})$-anti-crossing is typically a bottleneck for reaching the global minimum $\GM$,  
here we turn it around and use it constructively to reach the local minima $\LM{}$ (the obstacle) instead.  
More specifically, a polynomial-time TFQA algorithm evolves adiabatically through this anti-crossing,  
tracking the energy associated with $\LM{}$ (blue) instead of transitioning to the global minimum $\GM{}$ (red),  
as illustrated in Figure~\ref{fig:AC}.
This outputs a maximal independent set configuration from $\LM$, which can be used to seed the IC-based driver construction in Phase~I
of the \DDD{} algorithm; see Table~\ref{alg-overview}.

\begin{remark}
When the cliques are \emph{dependent}---that is, edges may exist between them---we refer to the structure as an \MDC{},  
a \DMS{} formed by a set of dependent cliques.  
This represents a relaxed structure to which our analysis may be generalized.  
The main difficulty in this relaxation is that the cliques can no longer be unambiguously identified,  
unlike in the \MIC{} case. Further discussion is deferred to future work.
\end{remark}

\paragraph{Terminology and Abbreviations.}
For convenience, we summarize below the key abbreviations and their meanings.
These notions are closely related: a \DMS{} corresponds to degenerate local minima,  
while \MIC{} and \MDC{} describe structured forms of such degeneracy.

\vspace{0.5em}
\begin{center}
\begin{tabular}{ll}
\toprule
\textbf{Abbreviation} & \textbf{Meaning} \\
\midrule
\MLIS{} & Maximal independent sets \\
\DMS{}  & A collection of \MLIS{} of equal size, corresponding to degenerate local minima \\
\MIC{}  & A \DMS{} formed by a set of independent cliques \\
\MDC{}  & A \DMS{} formed by a set if dependent cliques \\
\bottomrule
\end{tabular}
\end{center}
\vspace{0.5em}

\subsection{\GIC{} and Its Hardness}
\label{sec:GIC}
We now define the class of structured problem instances assumed in this paper.  
A \emph{Graph with Independent Cliques} (\GIC{}) is a graph whose structure satisfies the following properties:
\begin{itemize}
\item Each critical \DMS{} in the graph assumes the structure of a \MIC{};
\item It has a unique (global) maximum independent set, denoted as \GM{}.
      which may be viewed as a \MIC{} consisting entirely of singleton cliques;
\item The \GM{} shares at least one vertex with at least one clique in some critical \MIC{}.
\end{itemize}

A critical \MIC{} contains exponentially many \MLIS{},  
and its partial overlap with the \GM{} obscures the global structure,  
making it difficult for classical algorithms to distinguish the true optimum from nearby local maxima.  
The classical hardness of \GIC{} instances stems from these same two features:  
(1) the exponential number of maximal independent sets in the \MIC{}, and  
(2) the partial vertex overlap between the global maximum and these local maxima,  
which creates ambiguity and hinders global optimization.

While we do not prove the classical hardness of \GIC{} in the formal complexity-theoretic sense  
(otherwise we would have proven \( \mathrm{P} = \mathrm{NP} \),
we note that the definition can be further strengthened if needed.  
For example, one may require that the graph contains at least two critical \MIC{}s,  
and that most vertices of the \GM{} belong to cliques in these \MIC{}s.  
In other words, \GIC{} instances are designed to be structured enough for rigorous analysis,  
yet rich enough to present classical difficulty.

\paragraph{Reduction to Fundamental Bipartite Structures.}
Our key argument is that the analysis of any general \GIC{} instance  
can be reduced to a sequence of simpler bipartite substructures,  
each formed by a critical \MIC{} and the global maximum (\GM{}).  
%% Recall that each \MIC{} corresponds to degenerate local minima (\LM{}) in the MIS-Ising energy landscape.  
%% The terms \LM{} and \MIC{} are thus used interchangeably.  
%% We use \GM{} to refer both to the global maximum independent set and to its corresponding global minimum in the energy landscape.  
%% In what follows, we will use \LM{} and \GM{} to denote the \MIC{} and \GM{} in each such bipartite substructure, respectively.

\begin{mdframed}
Recall that each \MIC{} corresponds to a set of degenerate local minima (\LM{}) in the MIS--Ising energy landscape. 
The terms \LM{} and \MIC{} are thus used interchangeably. 
We use \GM{} to refer both to the (global) maximum independent set and to its corresponding global minimum in the energy landscape. 
In what follows, we use \LM{} and \GM{} to denote the \MIC{} and \GM{} in each such bipartite substructure, respectively.
\end{mdframed}

\subsubsection{Bipartite Substructures: \Gdis{} and \Gshare{}}
\label{sec:graphs}

We consider two bipartite substructures:
\begin{itemize}
    \item \textbf{\Gdis}: The local minima (\LM{}) and the global minimum (\GM{}) are vertex-disjoint.
    \item \textbf{\Gshare}: The \GM{} shares exactly one vertex with each clique in the \LM{}.
\end{itemize}

We begin with the \emph{disjoint-structure graph} \( \Gdis = (V, E) \),  
in which the vertex set \( V \) is partitioned into left and right (disjoint) vertex sets,  
with the following structural properties:
\begin{itemize}
  \item The left component is defined by a set \( L = \{ C_1, \dots, C_{m_l} \} \) of \( m_l \) disjoint cliques,  
        each denoted \( C_i = \cl(w_i, n_i) \).  
        We let \( V_L = \bigcup_{i=1}^{m_l} C_i \) denote the full vertex set.
  \item The right component \( R \) consists of \( m_r \) independent vertices, each with weight \( w_r \). 
  \item Every vertex in \( V_L \) is adjacent to every vertex in \( R \).
\end{itemize}

In this paper we mainly focus on the unweighted MIS case,  
and assume uniform weights \( w_i = w_r = w \).  
Under the MIS--Ising mapping with these uniform weights,  
\( V_L \) corresponds to the degenerate local minima (\LM{}) with degeneracy \( \prod_{i=1}^{m_l} n_i \),  
while \( R \) defines the global minimum (\GM{}) with \( m_g = m_r \).
Each local minimum (in \LM{}) corresponds to a maximal independent set of size \( m_l \), and thus has energy \( -m_l \).  
The global minimum (\GM{}) has energy \( -m_g \).

We now define the \emph{shared-structure graph} \( \Gshare \),  
which differs from \( \Gdis \) in that each vertex in \( R \) is adjacent to all but one vertex in each clique of \( L \).  
This modification allows the \GM{} to include shared vertices from the cliques in \( L \),  
thereby introducing overlap with the \LM{}.
Structurally, \( L \) and \( R \) are defined exactly as in \( \Gdis \),  
but with the adjacency rule modified as above.  
Specifically, the global maximum \GM{} consists of one shared vertex from each clique \(C_i \in L\) together with all \(m_r\) independent vertices in \(R\), yielding a total size \(m_g = m_l + m_r\).

Figure~\ref{fig:dis-share-graph} illustrates both cases.

\begin{figure}[h]
  \centering
  \begin{subfigure}[b]{0.36\textwidth}
    \includegraphics[width=\textwidth]{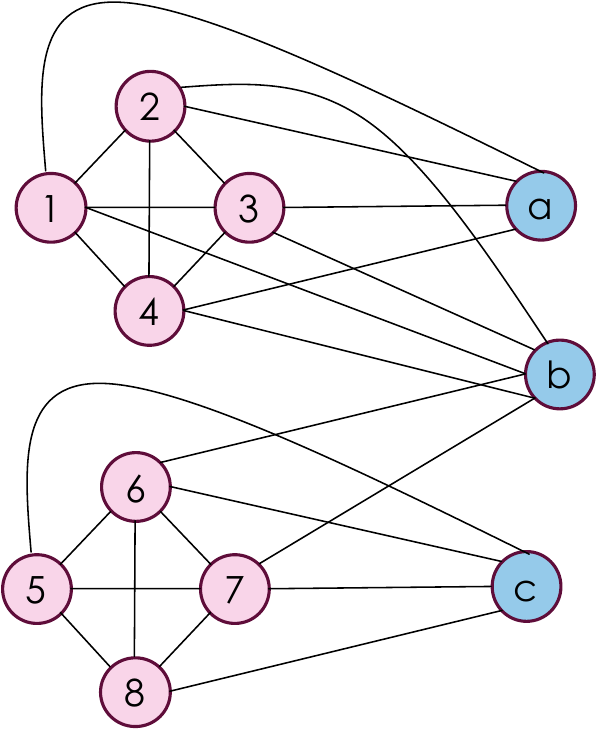}
    \caption{The \LM{} and \GM{} are vertex-disjoint.}
  \end{subfigure}
  \hfill
  \begin{subfigure}[b]{0.36\textwidth}
    \includegraphics[width=\textwidth]{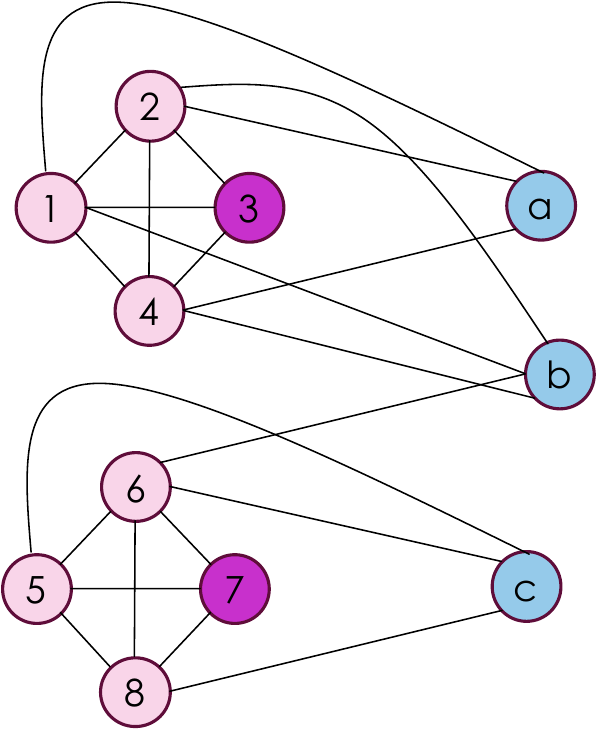}
    \caption{The \GM{} shares exactly one vertex with each clique in the \LM{}.}
  \end{subfigure}
\caption{
Example graphs illustrating the \Gdis{} and \Gshare{} structures.  
Recall that each \LM{} here is structurally a \MIC{}.  
(a) Disjoint-structure graph \Gdis: The set \( L \) consists of \( m_l = 2 \) disjoint cliques,  
each of size \( n_1 = n_2 = 4 \), with their vertices (pink) forming the local minima \LM{}.  
The set \( R \) (blue) consists of \( m_r = 3 \) independent vertices,  
forming the global minimum \GM{}.\\  
(b) Shared-structure graph \Gshare: The set \( L \) again consists of two cliques of size \( n_1 = n_2 = 4 \),  
with pink and purple vertices. The purple vertices (one per clique) are shared between \LM{} and \GM{}.  
The set \( R \) (blue) contains \( m_r = 3 \) independent vertices.  
The global minimum consists of all vertices in \( R \), together with the shared purple vertices in \( L \), giving \( m_g = 5 \).  
In both cases, edges between the pink vertices in \( L \) and all vertices in \( R \) are complete,  
though not all are shown for visual clarity.
}
\label{fig:dis-share-graph}
\end{figure}

For convenience, we write \( m := m_l \), dropping the subscript $l$ when no confusion arises.  
We assume \( \sum_{i=1}^{m} \sqrt{n_i} > m_g \), so that an anti-crossing is induced by the competition between \( \LM \) and \( \GM \).

These bipartite substructures represent two extreme regimes of connectivity and  
overlap between the critical \LM{} and the \GM{}.  
One can view \Gdis{} as a special case of \Gshare{} in which no vertices are shared between \LM{} and \GM{}.  
Note that \Gdis{} is not classically hard on its own  
(for example, one can remove all \LM{} vertices and directly identify the \GM{}).  
However, it may still arise as a substructure during a single iteration of the full algorithm,  
as discussed in Section~\ref{sec:full-system}.  
We consider \Gdis{} separately here for two reasons.  
First, in the \Gdis{} case, the relevant blocks of the Hamiltonian are decoupled,  
allowing the mechanisms of our algorithm to be simply illustrated.  
Second, it highlights the core distinction introduced by shared vertices:  
namely, that a single shared vertex can enable block coupling and necessitate  
an upper bound on \( \Jxx \) to prevent an interference-involved block-level anti-crossing.  
We note that for the TFQA case (\( \Jxx = 0 \)),  
we argue in the companion paper~\cite{Choi-Limitation} that it is sufficient  
to analyze only the disjoint-structure case (for the presence of the anti-crossing and the gap size).

\begin{remark}
Our reduction to bipartite substructures is not presented as a formal derivation,  
but rather as a structural perspective informed by effective Hamiltonian techniques.  
Each individual anti-crossing is primarily governed by the interaction  
between a set of critical local minima and the global minimum.  
This motivates a stepwise reduction to low-energy effective Hamiltonians  
that couple the \GM{} to one critical structure at a time---yielding a bipartite substructure,  
with connectivity and shared vertices bounded by the \Gshare{} worst case.
While a fully rigorous justification of this reduction is beyond the scope of this paper,  
our detailed analysis of the \Gshare{} instance,  
together with numerical examples involving multiple local minima in Section~\ref{sec:full-system},  
provides supporting evidence for the validity of this bipartite reduction perspective.
\end{remark}

%% %% %2. Revised DDD: System Hamiltonian and Parameters Schedule
 
\section{Revised \DDD{}: System Hamiltonian and Annealing Schedule}
\label{sec:2-stage}

In this section, we first recall the system Hamiltonian used in the original \DDD{} algorithm,  
then introduce the revised annealing schedule.  
The full evolution consists of an optional initialization phase (Stage~0),  
followed by the two-stage algorithmic core (Stage~1 and Stage~2).  
%% This two-stage structure is motivated by the \emph{Initialization Upper Bound} on \( \Jxx \) (to be elaborated in Section~\ref{sec:JxxUBone}),  
%% which requires the ground state to start and remain in the same-sign block throughout Stage~1.
%% \vnote{revise this later}

\subsection{Recall: System Hamiltonian of \DDD{}}
The system acts on a Hilbert space of \( N \) spin-\( \tfrac{1}{2} \) particles,  
one for each vertex of the problem graph $G$.  
The Hamiltonian is expressed in terms of the spin operators  
\( \Sop{\ms{x}} = \tfrac{1}{2} \sigma^x \) and \( \Sop{\ms{\tilde{z}}} = \shz{} \),  
where \( \shz{} \) denotes the shifted \( \sigma^z \) operator.
The system Hamiltonian is defined in terms of a time-dependent annealing schedule:
\[
\mb{H}(t) = \x{t} \mb{H}_{\ms{X}} + \jxx{t} \mb{H}_{\ms{XX}} + \mt{p}(t) \mb{H}_{\ms{problem}},
\]
where
\[
\mb{H}_{\ms{X}} = -\sum_{i} \sigma_i^x, \quad
\mb{H}_{\ms{XX}} = \sum_{(i,j) \in \edge(G_{\ms{driver}})} \sigma_i^x \sigma_j^x.
\]

The problem Hamiltonian \( \mb{H}_{\ms{problem}} \), introduced in Section~\ref{sec:MIS}, is recalled here for completeness:
\[
\mb{H}_{\ms{problem}} = \sum_{i \in \ver(G)} (-w_i)\, \shz{i}
+ \Jzz^{\ms{clique}} \sum_{(i,j) \in \edge(G_{\ms{driver}})} \shz{i} \shz{j}
+ \Jzz \sum_{(i,j) \in \edge(G) \setminus \edge(G_{\ms{driver}})} \shz{i} \shz{j}.
\]

The original annealing parameters evolve as:
\(
\x{t} = (1 - t)\Gamma, 
\jxx{t} = t(1 - t)\Jxx,
\mt{p}(t) = t.
\)
We revise \DDD{} by modifying its annealing schedule, specifically the functions  \( \x{t} \), \( \jxx{t} \), and \( \mt{p}(t) \).

\subsection{Revised Annealing Schedule: Stages 0, 1 and 2}

We describe the system using two Hamiltonians: \( \mb{H}_0(t) \), defined for \( t \in [0,1] \), corresponding to the optional Stage~0; and \( \mb{H}_1(t) \), also defined over \( t \in [0,1] \), governing the evolution during Stages~1 and~2. 

There are three distinct values for the transverse field: \( \Gamma_0 > \Gamma_1 > \Gamma_2 \).  
The three stages are distinguished by the values of these \( \Gamma \)'s: the transverse field is reduced from \( \Gamma_0 \) to \( \Gamma_1 \) during Stage~0, from \( \Gamma_1 \) to \( \Gamma_2 \) during Stage~1, and from \( \Gamma_2 \) to \( 0 \) during Stage~2.

\paragraph{Stage~0 (optional).}  
This stage begins in the uniform superposition state  
and evolves adiabatically toward the ground state of \( \mb{H}_0(1) = \mb{H}_1(0) \).

The Hamiltonian for Stage~0 is
\[
\mb{H}_0(t) = \x{t} \mb{H}_{\ms{X}} + \jxx{t} \mb{H}_{\ms{XX}} + \mt{p}(t) \mb{H}_{\ms{problem}}, \quad t \in [0,1],
\]
with
\(
\x{t} = (1 - t)(\Gamma_0 - \Gamma_1) + \Gamma_1,
\jxx{t} = t \Jxx,  
\mt{p}(t) = t.
\)
During this stage, the problem parameters \( w_i \), \( \Jzz^{\ms{clique}} \), and \( \Jzz \)  
are gradually ramped to their final values as \( \mt{p}(t) \) increases from \( 0 \) to \( 1 \),  
while \( \jxx{t} \) increases linearly to its target value \( \Jxx \).
If the ground state of \( \mb{H}_0(1) = \mb{H}_1(0) \) can be prepared directly---e.g., via quantum hardware initialization---  
then Stage~0 may be omitted.

\paragraph{Main Stages 1 and 2 (algorithmic core).}

The system Hamiltonian during the main two-stage evolution is
\[
\mb{H}_1(t) = \x{t} \mb{H}_{\ms{X}} + \jxx{t} \mb{H}_{\ms{XX}} + \mb{H}_{\ms{problem}}, \quad t \in [0,1].
\]
The annealing schedule divides into two stages, determined by a structure-dependent parameter \( \Gamma_2 \).
We define the transition time \( t_{\ms{sep}} := 1 - \frac{\Gamma_2}{\Gamma_1} \),  
so that Stage~1 corresponds to \( t \in [0, t_{\ms{sep}}] \), and Stage~2 to \( t \in [t_{\ms{sep}}, 1] \).
The transverse field is given by
\(
\x{t} = (1 - t)\Gamma_1,
\)
which decreases linearly from \( \Gamma_1 \) at \( t = 0 \) to \( \Gamma_2 \) at \( t = t_{\ms{sep}} \), and then to \( 0 \) at \( t = 1 \).
The XX-driver schedule is
\[
\jxx{t} = 
\begin{cases}
\Jxx & \text{for } t \in [0, t_{\ms{sep}}] \Leftrightarrow \mt{x} \in [\Gamma_2, \Gamma_1], \\
\alpha \x{t} & \text{for } t \in [t_{\ms{sep}}, 1] \Leftrightarrow \mt{x} \in [0, \Gamma_2],
\end{cases}
\]
where \( \Jxx = \alpha \Gamma_2 \). That is, \( \jxx{t} \) remains constant during Stage~1,  
and both \( \x{t} \) and \( \jxx{t} \) ramp linearly to zero during Stage~2.
The parameter \( \alpha \) controls the strength of the XX-coupling relative to the transverse field during Stage~2.  
Its value is critical 
and will be determined analytically in Section~\ref{sec:Jxx-bounds}.

%% %\begin{tcolorbox}[colback=white,colframe=black,boxrule=0.4pt,left=4pt,right=4pt,top=2pt,bottom=2pt]
%% \paragraph{Time-Dependent Quantities and Font Convention.} Note that \( \mt{jxx} \) can be reparametrized as a function of \( \mt{x} \)\footnote{Since \( \mt{x} = (1 - t)\Gamma_1 \) decreases monotonically from \( \Gamma_1 \) to 0,  
%% each value \( t \in [0,1] \) corresponds uniquely to a transverse field value \( \mt{x} \in [0, \Gamma_1] \).  
%% This one-to-one correspondence allows us to equivalently express the annealing schedule  
%% in terms of \( \mt{x} \) rather than \( t \), as indicated by the domain annotations.}.  
%% Accordingly, the Hamiltonian \( \mb{H}_1 \) can be viewed as a function of the transverse field \( \mt{x} \),  
%% which decreases monotonically with \( t \).
%% This reparametrization is particularly useful when describing multiple iterations of the algorithm,  
%% in which each substructure evolves under a Hamiltonian governed by \( \mt{x} \).
%% This dual viewpoint allows us to parametrize the system either explicitly in \( t \) or implicitly in terms of \( \mt{x} \).  
%% In particular, we may describe the Hamiltonian \( \mb{H}_1 \) and its components as functions of the time-dependent transverse field \( \mt{x} \),  
%% without referring directly to the parameter \( t \).  
%% When there is no ambiguity, we often omit the parameter \( t \)  
%% and use text font (e.g., \( \mt{x} \), \( \mt{jxx} \)) to indicate that a quantity is time-dependent.
%% %\end{tcolorbox}

\begin{remark}[Time-Dependent Quantities and Font Convention]
Note that \( \mt{jxx} \) can be reparametrized as a function of \( \mt{x} \)\footnote{Since \( \mt{x} = (1 - t)\Gamma_1 \) decreases monotonically from \( \Gamma_1 \) to 0,  
each value \( t \in [0,1] \) corresponds uniquely to a transverse field value \( \mt{x} \in [0, \Gamma_1] \).  
This one-to-one correspondence allows us to equivalently express the annealing schedule  
in terms of \( \mt{x} \) rather than \( t \), as indicated by the domain annotations.}.  
Accordingly, the Hamiltonian \( \mb{H}_1 \) can be viewed as a function of the transverse field \( \mt{x} \),  
which decreases monotonically with \( t \). This reparametrization is particularly useful when describing multiple iterations of the algorithm,  
in which each substructure evolves under a Hamiltonian governed by \( \mt{x} \).  
This dual viewpoint allows us to parametrize the system either explicitly in \( t \) or implicitly in terms of \( \mt{x} \).  
When there is no ambiguity, we often omit the parameter \( t \) and use text font (e.g., \( \mt{x}, \mt{jxx} \)) to indicate time dependence.
\end{remark}

\begin{figure}[!htbp]
  \centering
  \includegraphics[width=0.65\textwidth]{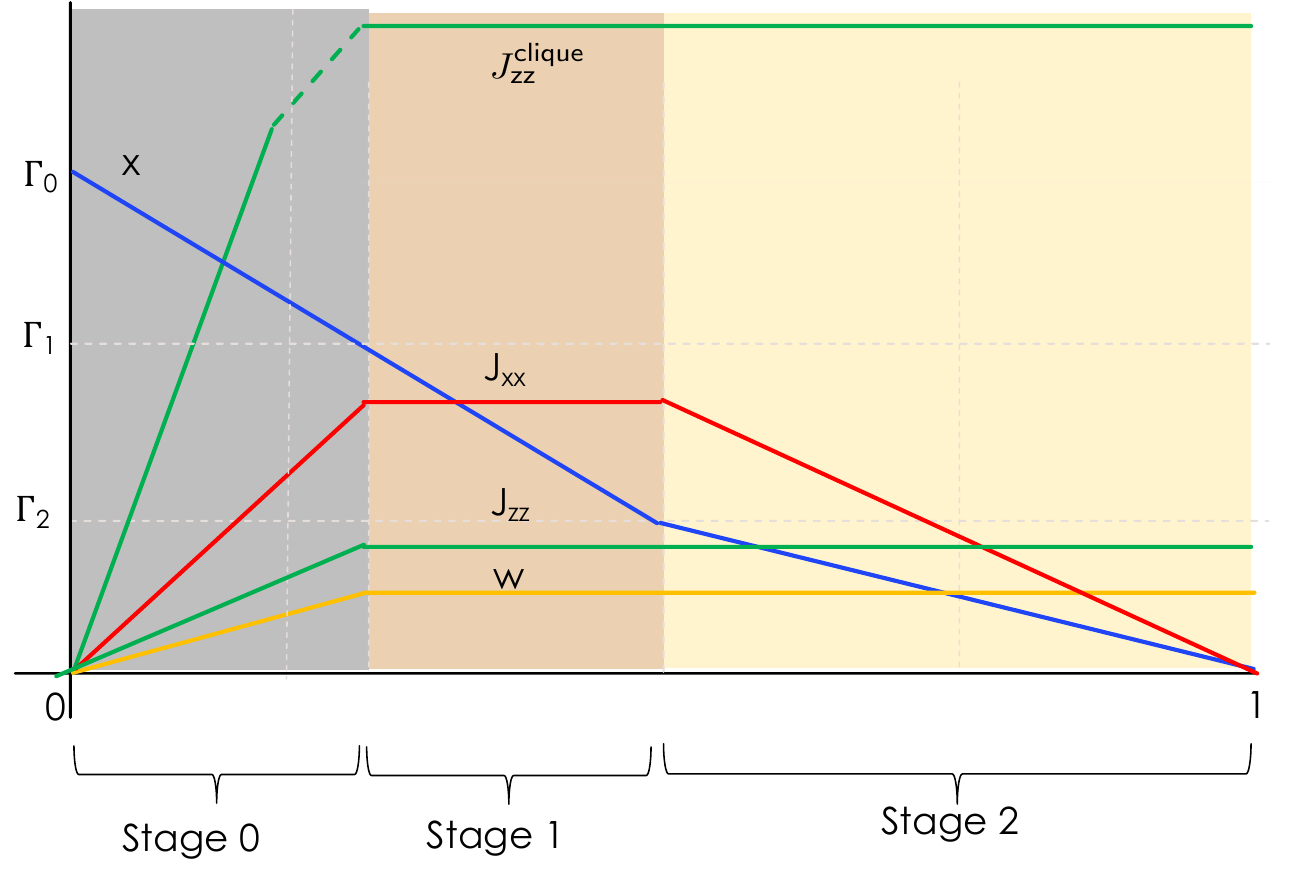}
  \caption{
  Annealing parameter schedule for the system Hamiltonian  
  \( \mb{H}(t) = \x{t} \mb{H}_{\ms{X}} + \jxx{t} \mb{H}_{\ms{XX}} + \mt{p}(t)\mb{H}_{\ms{problem}} \).
  Transverse field \( \x{t} \) (blue) begins at \( \Gamma_0 \),  
  decreases to \( \Gamma_1 \) in Stage~0, then to \( \Gamma_2 \) in Stage~1,  
  and reaches 0 at the end of Stage 2.
  The \( \XX \)-coupling strength \( \jxx{t} \) (red) increases linearly to \( \Jxx \) in Stage~0,  
  remains constant in Stage~1, and decreases to 0 in Stage~2.
  Vertex weights \( \mt{w}_i \) (orange) and $\ZZ$-couplings \( \Jzz, \Jzz^{\ms{clique}} \)  
  (green) ramp to final values in Stage~0 and remain fixed thereafter.
  }
  \label{fig:schedule}
\end{figure}

\subsubsection*{Parameter Values.}
The values of \( \Gamma_0 \), \( \Gamma_1 \), \( \Gamma_2 \), \( \Jxx \) $(\alpha)$, and \( \Jzz \)  
are critical to the success of the algorithm.  
In particular, suitable values of \( \Gamma\)s and \( \Jzz \) are chosen in relation to the analytical bounds on \( \Jxx \),  
ensuring that the effective Hamiltonian supports the desired localization and interference behavior.  
These parameter choices will be specified and justified in the subsequent analysis.

\begin{remark}
After Stage~0, the system evolves under the full Hamiltonian \( \mb{H}_1 \) defined above.  
However, our analysis focuses on an effective Hamiltonian \( \Heff \), Eq.~\eqref{eq:stage0-Heff},  
derived in Section~\ref{sec:stage-0}.  
This effective Hamiltonian captures the essential dynamics responsible for the algorithmic speedup  
and serves as the foundation for the analytical and numerical investigations presented in the following sections.
\end{remark}

%% % Part II.
\section{Anti-Crossing and the Two-Level Hamiltonian \( \mb{B}(w,x) \)}
\label{sec:anti-crossings}

The term \emph{anti-crossing} is sometimes used loosely, so we begin with a precise notion in the two-level case, then extend it to multi-level systems, and finally classify different types of anti-crossing.
We also introduce a canonical two-level Hamiltonian whose eigensystem will be used throughout our analysis. 

\subsection{Level Repulsion vs.~Anti-Crossing}
\label{subsec:level-repulsion}
We begin by distinguishing the concept of an \emph{anti-crossing} (also called \emph{avoided-crossing}) from level repulsion.

Consider a generic two-level Hamiltonian of the form
\[
\mb{H}(x) :=
\begin{bmatrix}
e_1(x) & v(x) \\
v(x) & e_2(x)
\end{bmatrix},
\]
where \( e_1(x) \), \( e_2(x) \), and \( v(x) \) are real-valued functions of a parameter \( x \). 

The eigenvalues of this Hamiltonian are
\(
\lambda_{\pm}(x) = \tfrac{e_1(x) + e_2(x)}{2} \pm \tfrac{1}{2} \sqrt{(e_1(x) - e_2(x))^2 + 4v(x)^2},
\)
and the energy gap between them is
\(
\Delta(x) := \lambda_+(x) - \lambda_-(x) = \sqrt{(e_1(x) - e_2(x))^2 + 4v(x)^2}.
\)
The off-diagonal term \( v(x) \) induces \emph{level repulsion}: if \( v(x) \neq 0 \), then the eigenvalues never cross, and the gap \( \Delta(x) \) remains strictly positive.
Thus, assuming the off-diagonal coupling \( v(x) \) is nonzero, level repulsion is always present.

\begin{definition}
We say that an \emph{anti-crossing} occurs when the two
unperturbed energy levels $e_1(x)$ and $e_2(x)$ cross, i.e.,
$e_1(x_*) = e_2(x_*)$ for some $x_*$, and the off-diagonal coupling
$v(x_*) \neq 0$. 
In this case the eigenvalue curves form an anti-crossing with gap
\(
\Delta_{\min} = 2|v(x_*)|.
\)
\end{definition}
The size of the anti-crossing gap depends on $|v(x_*)|$: stronger coupling leads to a larger gap, while weaker coupling results in a narrower one.

By contrast, if the two diagonal entries \( e_1(x) \) and \( e_2(x) \) remain well separated for all \( x \), then the system exhibits level repulsion but not an anti-crossing. Figure~\ref{fig:B2wx} illustrates an example of level repulsion without an anti-crossing. 

The eigenvectors of the two-level Hamiltonian are given by
\[
%\begin{cases}
\ket{\lambda_{-}(x)} = \cos\theta(x) \ket{0} + \sin\theta(x) \ket{1}, \quad
\ket{\lambda_{+}(x)} = -\sin\theta(x) \ket{0} + \cos\theta(x) \ket{1},
%\end{cases}
\]
where the mixing angle \( \theta(x) \) satisfies
\(
\tan(2\theta(x)) = \tfrac{2v(x)}{e_1(x) - e_2(x)}.
\)
Thus, near the anti-crossing point \( x = x_* \), the eigenstates interpolate between the unperturbed basis states.

\begin{remark}
The trigonometric expression for eigenvectors in terms of the mixing angle \( \theta(x) \),
is equivalent to the rational-form representation
\[
\ket{\lambda_{-}(x)} = \tfrac{1}{\sqrt{1 + \gamma(x)^2}} \left(\gamma(x) \ket{0} +  \ket{1} \right),\quad
\ket{\lambda_{+}(x)} = \tfrac{1}{\sqrt{1 + \gamma(x)^2}} \left( \ket{0} -\gamma(x) \ket{1} \right),
\]
where the two parametrizations are related by
\(
\gamma(x) = \tfrac{1}{\tan\theta(x)}, \text{and} \quad \tan(2\theta(x)) = \tfrac{2v(x)}{e_1(x) - e_2(x)}.
\)
This rational-form expression is particularly useful for our analysis, as it aligns directly with the basic matrix form introduced below.
\end{remark}

\begin{remark}
  While the explicit forms of the eigenvectors are not directly used in this paper, they are included here for completion, and used in the companion paper~\cite{Choi-Limitation} for
  bounding the anti-crossing gap.
\end{remark}
Our earlier work~\cite{Choi2021,Choi2020}, including the development of the original \DDD{} algorithm, was motivated by investigating the structural characteristics of eigenstates around the anti-crossing.

\subsubsection{Anti-crossing in a Multi-level System}
In a multi-level system, the notion of an anti-crossing extends naturally
by restricting the Hamiltonian to the two-dimensional subspace spanned by
the pair of eigenstates whose unperturbed energies intersect.
This reduction yields a $2 \times 2$ effective Hamiltonian that captures the
essential structure of the anti-crossing, including both the energy gap and
the interpolating behavior of the eigenstates.
Thus, the same framework as in the two-level case applies.

With this perspective, we refine the definition of an $(L,R)$-anti-crossing
given in \cite{Choi2021}. Recall that $E_0^A(t)$ denotes the ground state
energy of the Hamiltonian $\mb{H}_A(t)$ projected to the subspace spanned by
the subsets of $A$.
For simplicity, we will refer to this subpace as the subspace defined by $A$.
\begin{definition}
  We say an anti-crossing is an $(L,R)$-anti-crossing at $t_*$ if there exist
  bare energies $E_0^L(t)$ and $E_0^R(t)$ such that:
  \begin{enumerate}
    \item $E_0^L(t)$ and $E_0^R(t)$ approximate the unperturbed energy levels
    of the effective $2 \times 2$ Hamiltonian describing the anti-crossing
    for $t \in [t_* - \delta, t_* + \delta]$ for some small $\delta>0$;  and
    \item $E_0^L(t)$ and $E_0^R(t)$ cross at $t_*$, i.e.\ $E_0^L(t_*) = E_0^R(t_*)$.
  \end{enumerate}
\end{definition}
See Figure~\ref{fig:AC} for an illustration.

\begin{figure}[!htbp]
  \centering
  \includegraphics[width=0.5\textwidth]{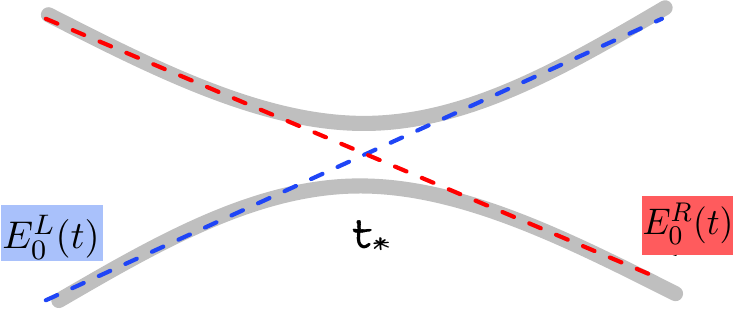}
  \caption{
Schematic of an $(L,R)$-anti-crossing.
Dashed lines: bare energies $E_0^L(t)$ and $E_0^R(t)$ crossing at $t_*$.
Solid gray curves: the two lowest eigenvalues of the Hamiltonian, showing
the avoided crossing that originates from this bare crossing.
  }
  \label{fig:AC}  
\end{figure}

In this sense, the anti-crossing is said to originate from the crossing of
$E_0^L(t)$ and $E_0^R(t)$. The unperturbed levels define the formal structure,
while the bare energies provide concrete analytic proxies that allow us to
identify and study the anti-crossing without explicitly constructing the
effective $2 \times 2$ Hamiltonian, which could be challenging.
In the companion paper~\cite{Choi-Limitation}, an effective $2 \times 2$ Hamiltonian
is explicitly constructed to derive a perturbative bound on the anti-crossing gap.

%% Informally, we say that an anti-crossing is caused by a (degenerate) local
%% minimum $\LM$ and the global minimum $\GM$ if the bare energy
%% $E_0^{\LM}(t)$ crosses the bare energy $E_0^{\GM}(t)$.

\subsubsection{Types of Anti-Crossing}
We classify anti-crossings according to whether the subspaces defined by $L$
and $R$, when embedded in the full Hilbert space, overlap or are disjoint.

\paragraph{Tunneling-Induced Anti-Crossing.}
If the embedded subspaces of $L$ and $R$ overlap, then $H_L$ and $H_R$
are submatrices of the same block in the Hilbert space decomposition.
In this case, off-diagonal terms within the block induce tunneling between
the two configurations. We refer to this as a
\emph{tunneling-induced anti-crossing}.

\paragraph{Block-Level Anti-Crossing.}
If the embedded subspaces of $L$ and $R$ are disjoint, then $H_L$ and $H_R$
belong to two distinct blocks of the Hamiltonian. In this case, the ground
states of the two blocks may become nearly degenerate, leading to a true
crossing if the blocks are decoupled, or to an avoided crossing if
there is weak but nonzero inter-block coupling.
We refer to this as a \emph{block-level anti-crossing}.

Figure~\ref{fig:AC} provides a generic schematic of an $(L,R)$-anti-crossing,
whether tunneling-induced or block-level. The interpretation depends on the
relation between the subspaces defined by $L$ and $R$. In the block-level
case, for example, the two levels correspond to the ground states of two
distinct blocks---the same-sign block $H_L$ (blue) and the opposite-sign
block $H_R$ (red). A weak inter-block coupling lifts their degeneracy and
produces an avoided crossing. When the resulting gap is small, the system
evolves almost as if the blocks were decoupled, remaining adiabatically in
the same-sign block and following the blue path.

The original double anti-crossing (DAC) observed in~\cite{Choi2021} is, in
fact, a \emph{double block-level anti-crossing}. Thus, instead of a diabatic transition---from ground to
excited to ground state---the evolution effectively remains confined to the
same-sign block, without any true inter-block transition. In this sense, the
block-level anti-crossing is dynamically bypassed.

There is also a more subtle variant, in which one of the two competing levels
is not an eigenstate of a single block, but a superposition of states from
different blocks. In this case, the coupling is not merely perturbative; the
anti-crossing reflects quantum interference between blocks. We refer to
this as an \emph{interference-involved block-level anti-crossing}, an example is
shown in Figure~\ref{fig:L2-fail}.

\subsection{The Basic Matrix \( \mb{B}(w,x) \): Eigenvalues and Eigenstates}
\label{subsec:basic-matrix}

We define the following effective two-level Hamiltonian, which will serve as a basic building block for our analysis throughout the paper:
\begin{equation}
\mb{B}(w, x) :=
\begin{bmatrix}
-w & -\tfrac{1}{2}x \\
-\tfrac{1}{2}x & 0
\end{bmatrix},
\label{eq:B-matrix}
\end{equation}
where \( w = w(t) \) and \( x = x(t) \) are real-valued parameters, typically derived from problem Hamiltonians and driver strengths. This is a special case of a spin-\( \tfrac{1}{2} \) system, with analytic eigenstructure.
The eigenvalues of $\mb{B}(w,x)$ are
\begin{equation}
\beta_k = -\tfrac{1}{2}\!\left(w + (-1)^k \sqrt{w^2 + x^2}\right), \quad k=0,1,
\label{eq:B-evals}
\end{equation}
with normalized eigenvectors
\begin{equation}
\ket{\beta_0} = \tfrac{1}{\sqrt{1+\gamma^2}} \bigl(\gamma\ket{0} + \ket{1}\bigr), 
\quad
\ket{\beta_1} = \tfrac{1}{\sqrt{1+\gamma^2}} \bigl(\ket{0} - \gamma\ket{1}\bigr),
\label{eq:B-evecs}
\end{equation}
where the mixing coefficient is
\(
\gamma = \tfrac{x}{w+\sqrt{w^2+x^2}}.
\)

\begin{remark}
For $w>0$, we use the standard basis
$\ket{0} = \bigl[1,0\bigr]^T$, $\ket{1} = \bigl[0,1\bigr]^T$.
For $w<0$, the roles flip:
$\ket{0}(w<0) = \ket{1}(w>0)$ and
$\ket{1}(w<0) = \ket{0}(w>0)$.
\end{remark}
Figure~\ref{fig:B2wx} visualizes the eigenspectrum and ground state behavior under variation of \( x \) and \( w \).

% [Your figure code unchanged]

\begin{figure}[h!]
  \centering
$$
  \begin{array}{cc}
    \includegraphics[width=0.48\textwidth]{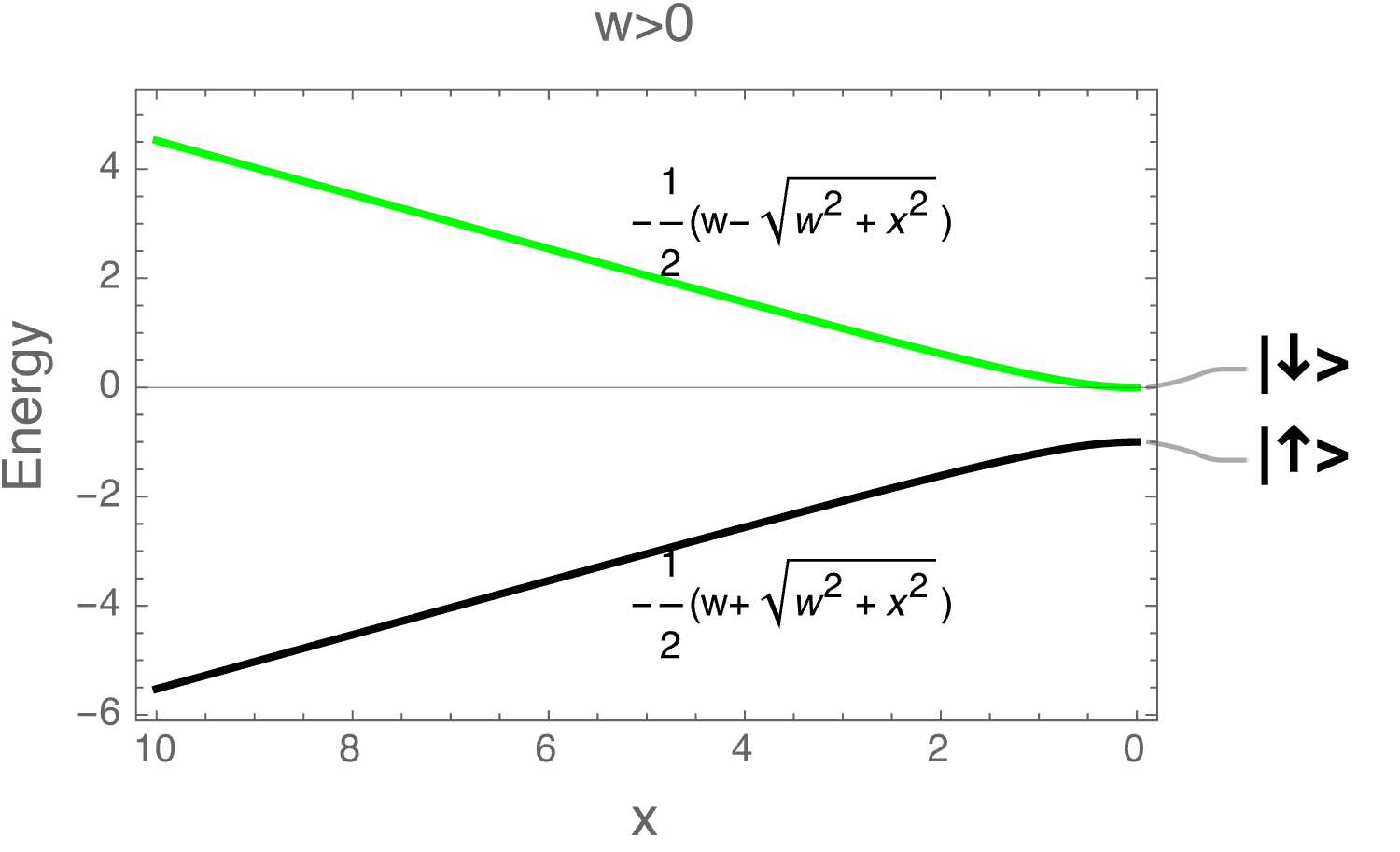} &
    \includegraphics[width=0.48\textwidth]{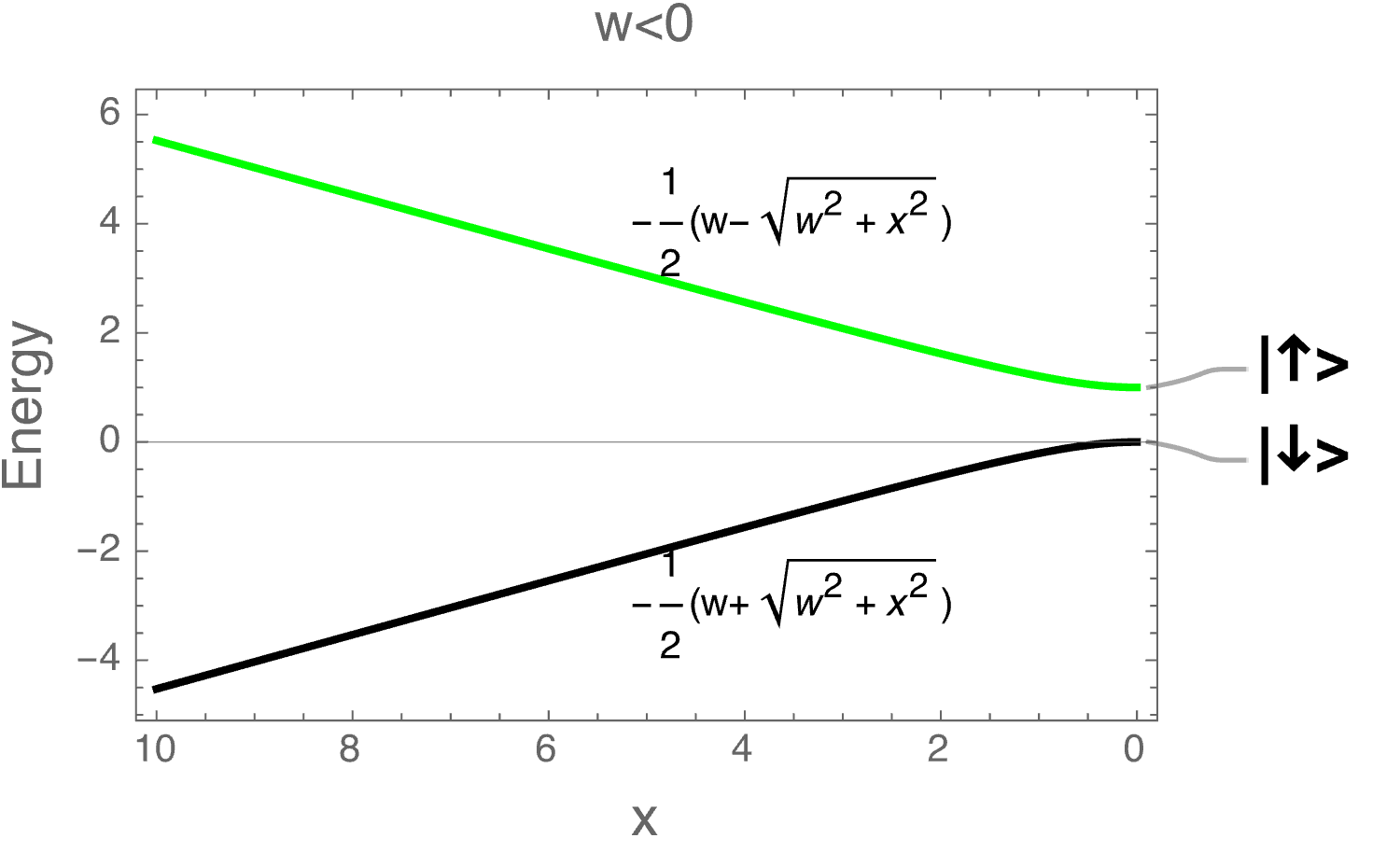} 
  \end{array}
$$
  \caption{Eigenvalues of the two-level Hamiltonian \( \mb{B}(w, x) \), where \( \beta_k = -\tfrac{1}{2} \left( w + (-1)^k \sqrt{w^2 + x^2} \right), \; k = 0, 1 \).
    The ground state energy (\( \beta_0 \)) is shown in black, and the first excited state energy (\( \beta_1 \)) is shown in blue.
    The energy gap widens as \( x \) increases---there is no anti-crossing in this case.
    The ground state is 
    \( \ket{\beta_0} = \tfrac{1}{\sqrt{1 + \gamma^2}} \left( \gamma \ket{0} + \ket{1} \right) \), with \( \gamma = \tfrac{x}{w + \sqrt{w^2 + x^2}} \).
    (a) \( w = 1 \); \quad (b) \( w = -1 \).
    For \( w > 0 \), we have \( \ket{0} = \ket{\uparrow} \) and \( \ket{1} = \ket{\downarrow} \) as in (a).  
Notice that the contents of \( \ket{0} \) and \( \ket{1} \) swap for \( w < 0 \) as in (b).
}
  \label{fig:B2wx}
\end{figure}

\section{A Single Clique: Decomposition and Low-Energy Subspace Analysis}
\label{sec:single-clique}

In this section, we focus on analyzing the Hilbert space associated with a single clique.  
We begin, in Section~\ref{sec:low-high}, by describing a natural decomposition of the Hilbert space into a low-energy subspace, consisting of independent-set states,  
and a high-energy subspace, consisting of dependent-set states.  
We describe the angular momentum structure and spectral properties of the low-energy subspace in Section~\ref{sec:ang-struct-clique}.

To fix notation, let \( G_c = (V, E) \) be a clique,  
where \( V = \{1, \dots, n_c\} \) is the set of vertices, and  
\( E = \{(i,j) : i < j, \; i, j \in V\} \) is the set of edges.  
We assume all vertices in the clique have the same weight \( w_c \),  
i.e., \( w_i = w_c \) for all \( i \in V \).
For clarity, we will omit the graph argument \( (G_c) \) in this section, with the understanding that all quantities refer to this single clique.

\subsection{Decomposition of the Clique Space: Low and High Energy Subspaces}
\label{sec:low-high}
The Hilbert space \( \mathcal{V}_c \) for a clique of size \( n_c \) consists of \( 2^{n_c} \) computational basis states, each corresponding to a binary string of length \( n_c \). Among these, only \( n_c + 1 \) basis states correspond to independent sets:
\[
\begin{array}{l}
\bst{0} = \texttt{00\ldots0} \\
\bst{1} = \texttt{10\ldots0} \\
\bst{2} = \texttt{01\ldots0} \\
\quad\vdots \\
\bst{n_c} = \texttt{00\ldots1}
\end{array}
\]
with 
\(
\mathbb{N}_{\text{ind}} = \left\{
\bst{0},
\bst{1},
\ldots,
\bst{n_c}
\right\},
\)
where each \( \bst{i} \) is a binary string of length \( n_c \) with a single 1 in position \( i \) (and 0s elsewhere), and \( \bst{0} \) is the all-zero string.

The energy associated with each singleton state is \( -w_c \), while the empty set has energy 0.  
In contrast, the energy of all other bit strings---which correspond to dependent sets---is at least \( (\Jzz^{\text{clique}} - 2w_c) \).  
By choosing \( \Jzz^{\text{clique}} \) sufficiently large, these dependent-set states become energetically inaccessible.
Hence, the Hilbert space admits a natural decomposition:
\begin{equation}
  \label{eq:clique-decomposition}
\mathcal{V}_c = \SLE \oplus \mathcal{H}^{\text{dep}},
\end{equation}
where \( \SLE \) is the low-energy subspace spanned by \( \mathbb{N}_{\text{ind}} \),  
and \( \mathcal{H}^{\text{dep}} \) is the high-energy subspace spanned by the dependent-set states.

The remainder of this section focuses on the analysis of the low-energy subspace \( \SLE \).

\subsection{Angular-Momentum Decomposition of the Low-Energy Subspace}
\label{sec:ang-struct-clique}

The low-energy subspace $\SLE$ of a clique can be naturally expressed in the
\emph{total angular momentum basis}.  
By changing to this basis, $\SLE$ decomposes into a single effective spin-$\tfrac{1}{2}$
(two-dimensional same-sign sector) together with $(n_c-1)$ spin-0
(one-dimensional opposite-sign) components.  
In this representation, the relevant spin operators take a block-diagonal form,
which induces a block-diagonal structure in the restricted Hamiltonian.  
This block-diagonal restricted Hamiltonian yields exact spectral formulas,
explains the $\Jxx$ see-saw effect between the same-sign and opposite-sign blocks,
and extends naturally to the case of coupled subcliques, a key step toward the global
block decomposition of shared-structure graphs.

\subsubsection{Low-Energy Subspace in Total Angular Momentum Basis $\Bc$}
\label{sec:total-angular-basis}
\begin{mdframed}
\begin{lemma}
  \label{single-clique-basis}
  The total angular momentum basis for \( \SLE \) consists of the states:
  \begin{align}
  \label{eq:Bc}
  (\Bc) \left\{
  \begin{array}{llll}
    \ket{s,-(s-1)}, &\ket{1,(s-1),-(s-1)},&\ldots&\ket{n_c-1,(s-1),-(s-1)}\\
    \ket{s, -s} & & &
  \end{array}
  \right.
  \end{align}
  where \( s = \tfrac{n_c}{2} \) is the total spin.

  Explicitly:
  \begin{itemize}
    \item \( \ket{s, -s} = \ket{\bst{0}} \), representing the empty set.
    \item \( \ket{s,-(s-1)} = \tfrac{1}{\sqrt{n_c}} \sum_{i=1}^{n_c} \ket{\bst{i}} \), a uniform superposition of all singletons with positive amplitudes.
    \item \( \ket{k, s-1, -(s-1)} \), for \( k = 1, \ldots, n_c - 1 \), consists of a superposition of singleton states with both positive and negative amplitudes.
  \end{itemize}
  Thus, \( \ket{s, -s} \) and \( \ket{s,-(s-1)} \) are same-sign states, while \( \ket{k, s-1, -(s-1)} \) are opposite-sign states.
\end{lemma}
\end{mdframed}

\begin{proof}
  The total Hilbert space of \( n_c \) spin-\(\tfrac{1}{2} \) particles decomposes into irreducible representations of total spin \( s = \tfrac{n_c}{2} \),
  followed by smaller spin sectors \( (s-1, s-2, \ldots) \). According to the Clebsch--Gordan decomposition of \( n_c \) spin-\( \tfrac{1}{2} \) particles, we have:
  \[
  \underbrace{\tfrac{1}{2} \otimes \tfrac{1}{2} \otimes \ldots \otimes \tfrac{1}{2}}_{n_c} = \tfrac{n_c}{2} \oplus
  \underbrace{\tfrac{n_c}{2} - 1 \oplus \ldots \oplus \tfrac{n_c}{2} - 1}_{n_c - 1} \oplus \ldots
  \]

This decomposition is a standard result in angular momentum theory and provides a natural organization of basis states according to their total spin quantum numbers. 
  The key observation is that the independent-set states reside in the lowest-weight subspace, corresponding to the smallest \( m_s \) values within these spin multiplets. Specifically, the independent-set states are spanned by the lowest-weight states \( (m_s = -s, -(s - 1)) \) in the spin-\( s \) multiplet and the last \( (m_s = -(s - 1)) \) states in spin-\( (s - 1) \) multiplets. This yields the following basis:
  \[
  (\Bc) \left\{
  \begin{array}{llll}
    \ket{s,-(s-1)}, &\ket{1,(s-1),-(s-1)},&\ldots&\ket{n_c-1,(s-1),-(s-1)}\\
    \ket{s, -s} & & &
  \end{array}
  \right.
  \]

  One can check that:
  \(
  \ket{s,-(s-1)} = \tfrac{1}{\sqrt{n_c}} \sum_{i=1}^{n_c} \ket{\bst{i}}, \quad
  \ket{s, -s} = \ket{\bst{0}}.
  \)

  For \( k = 1, \ldots, n_c - 1 \), let \( \ket{k, s-1, -(s-1)} \) denote the last (\( m_s = -(s - 1) \)) state of spin-\((s - 1)\). Since these states share the same \( m_s \) value, they must be orthogonal to \( \ket{s,-(s-1)} \):
  \(
  \ket{k, s-1, -(s-1)} = \sum_{i=1}^{n_c} a_{k_i} \ket{\bst{i}},
  \)
  where \( a_{k_i} \) are Clebsch--Gordan coefficients.

  For each \( k \geq 1 \), at least one coefficient \( a_{k,j} < 0 \) for some \( j \), indicating that the superposition contains both positive and negative amplitudes. Hence, these states are opposite-sign.
\end{proof}

\paragraph{Remark (Basis Reordering).}  
For convenience, we reorder the basis states in \( \Bc \) as follows:
\[
(\Bc')\quad \ket{s, -(s - 1)},\, \ket{s, -s},\, \ket{1, s - 1, -(s - 1)},\, \dots,\, \ket{n_c - 1, s - 1, -(s - 1)}.
\]
That is, we swap the order of the two same-sign basis states.
This ordering simplifies the representation of operators in the next steps.

\paragraph{Basis Transformation.}  
The transformation between the computational basis \( \{ \ket{\bst{i}} \} \) and the angular momentum basis $(\Bc')$ can be derived either from the Clebsch--Gordan coefficients or directly from the relationships established in the proof. Specifically:
\begin{itemize}
  \item The state \( \ket{s, -(s - 1)} \) is a uniform superposition over all singleton states \( \{ \ket{\bst{i}} \}_{i=1}^{n_c} \).
  \item The remaining states \( \ket{k, s - 1, -(s - 1)} \), for \( k = 1, \dots, n_c - 1 \), form an orthogonal complement to \( \ket{s, -(s - 1)} \) within the subspace spanned by \( \{ \ket{\bst{1}}, \dots, \ket{\bst{n_c}} \} \).
\end{itemize}

We denote the basis transformation matrix from the computational basis to the angular momentum basis by \( \Uc \).

\subsubsection{Spin Operators in the \(\Bc'\) Basis}
\label{sec:spin-operators}
Consider the spin operators on \( \mathcal{V}_c \):
\[
\Sop{\Z} = \tfrac{1}{2}\sum_{i=1}^{n_c} \sigma_i^z, \quad
\Sop{\sZ} = \sum_{i=1}^{n_c} \shz{i}, \quad
\Sop{\X} = \tfrac{1}{2}\sum_{i=1}^{n_c} \sigma_i^x, \quad
\Sop{\XX} = \tfrac{1}{4} \sum_{ij \in \edge(G_{\text{driver}})} \sigma_i^x \sigma_j^x,
\]
where \( G_{\text{driver}} = G_c \).

We project these operators onto the low-energy subspace \( \SLE \) using the projection operator \( \PLE \), and then transform them into the \( \Bc' \) basis via the basis transformation \( \Uc \). For any operator \( \mathrm{X} \), we use a bar to denote the operator:
\[
\bar{\mathrm{X}} = \Uc^{\dagger} \PLE \mathrm{X} \PLE \Uc.
\]

\begin{mdframed}
  \begin{theorem}
  \label{thm:Sop}
  The restricted operators in the \(\Bc'\) basis are block-diagonal and given explicitly by:
  \begin{align}
    \left\{
    \begin{array}{ll}
    & \bar{\Sop{\sZ}} = \shz{} \oplus 1 \oplus \cdots \oplus 1, \\[5pt]  
    & \bar{\Sop{\X}} =  \tfrac{\sqrt{n_c}}{2} \sigma^x \oplus 0 \oplus \cdots \oplus 0, \\[5pt]
      & \bar{\Sop{\XX}} = \left(\tfrac{n_c - 1}{4}\right)\shz{} \oplus \left(-\tfrac{1}{4}\right) \oplus \cdots \oplus \left(-\tfrac{1}{4}\right).
    \end{array}
    \right.
    \label{eq:transformed-operators}
  \end{align}
  where \( \shz{} \) and \( \sigma^x \) act on the effective spin-\( \tfrac{1}{2} \) (two-dimensional same-sign) subspace, while the scalars act on spin-0 (one-dimensional opposite-sign) subspaces.
\end{theorem}
\end{mdframed}

\begin{proof}
Recall the reordered basis \( \Bc' \) is:
\(\{
|s,-(s-1)\rangle,\quad |s,-s\rangle,\quad |1,s-1,-(s-1)\rangle,\dots,\ket{n_c - 1, s - 1, -(s - 1)}\}.
\)

In standard angular momentum theory, the eigenvalues of \( \Sop{\Z} \) for these states are:
\begin{itemize}
  \item \( \ket{s,-(s-1)} \): eigenvalue \( m_s = -(s-1) \)
  \item \( \ket{s,-s} \): eigenvalue \( -s \)
  \item \( \ket{k, s-1, -(s-1)} \): eigenvalue \( -(s-1) \) for all \( k \)
\end{itemize}

Thus:
\[
\bar{\Sop{\Z}} =
\begin{bmatrix}
-(s-1) & 0 & 0 & \cdots & 0 \\
0 & -s & 0 & \cdots & 0 \\
0 & 0 & -(s-1) & \cdots & 0 \\
\vdots & & \vdots & \ddots & \vdots \\
0 & \cdots & 0 & \cdots & -(s-1)
\end{bmatrix}
\]

and the shifted operator \( \bar{\Sop{\sZ}} \) becomes:
\begin{align*}
\bar{\Sop{\sZ}} =
\begin{bmatrix}
1 & 0 & 0 & \cdots & 0 \\
0 & 0 & 0 & \cdots & 0 \\
0 & 0 & 1 & \cdots & 0 \\
\vdots & & \vdots & \ddots & \vdots \\
0 & \cdots & 0 & \cdots & 1
\end{bmatrix}
= \shz{} \oplus 1 \oplus \cdots \oplus 1.
\end{align*}

To obtain \( \bar{\Sop{\X}} \), recall:
\(
\Sop{\X} = \tfrac{1}{2}(\Sop{+} + \Sop{-}),
\)
where \( \Sop{+} \) and \( \Sop{-} \) are the (raising and lowering) ladder operators.  
From angular momentum algebra:
\[
\braket{s,-(s-1)|\Sop{+}|s,-s} = \braket{s,-s|\Sop{-}|s,-(s-1)} = \sqrt{n_c},
\]
thus:
\(
\bar{\Sop{\X}} = \tfrac{\sqrt{n_c}}{2} \sigma^x \oplus 0 \oplus \cdots \oplus 0.
\)

Since \( \Sop{+}^2 = \Sop{-}^2 = 0 \) on \( \SLE \), it follows:
\(
\Sop{\X}^2 = \tfrac{1}{4}(\Sop{+}\Sop{-} + \Sop{-}\Sop{+}),
\)
and:
\begin{align*}
  \left\{
    \begin{array}{ll}
  & \Sop{\X}^2 \ket{s,-s} = \tfrac{n_c}{4} \ket{s,-s}, \\
  & \Sop{\X}^2 \ket{s,-(s-1)} = \tfrac{3n_c - 2}{4} \ket{s,-(s-1)}, \\
      & \Sop{\X}^2 \ket{s-1,-(s-1)} = \tfrac{n_c - 2}{4} \ket{s-1,-(s-1)}.
      \end{array}
    \right.
\end{align*}

Using:
\(
\Sop{\XX} = \tfrac{1}{2}(4 \Sop{\X}^2 - n_c),
\)
we conclude:
\(
\bar{\Sop{\XX}} =
\left(\tfrac{n_c - 1}{4}\right) \shz{} \oplus \left(-\tfrac{1}{4}\right) \oplus \cdots \oplus \left(-\tfrac{1}{4}\right).
\)
\end{proof}

\subsubsection{Decomposition into Spin-\(\tfrac{1}{2}\) and Spin-$0$ Components}
\label{sec:spin-structure}
The transformed operators given by Theorem~\ref{thm:Sop} offer a clear physical interpretation:
the low-energy subspace \( \SLE \) decomposes into a direct sum consisting of a single effective spin-\(\tfrac{1}{2}\) subsystem,
together with \( n_c - 1 \) (opposite-sign) spin-$0$ subsystems:
\begin{align}
  \SLE = \left[\tfrac{1}{2}\right]_{n_c} \oplus \underbrace{0 \oplus \cdots \oplus 0}_{n_c - 1}.
  \label{eq:single-clique}
\end{align}

The effective spin-\(\tfrac{1}{2}\) subsystem \( \left[\tfrac{1}{2}\right]_{n_c} \) is spanned by the two same-sign basis states:
\(
\left|\tfrac{1}{2}, -\tfrac{1}{2}\right\rangle = |s, -(s-1)\rangle, \quad 
\left|\tfrac{1}{2}, \tfrac{1}{2}\right\rangle = |s, -s\rangle.
\)
The spin-$0$ components correspond to the opposite-sign basis vectors:
\(
|k, s - 1, -(s - 1)\rangle,\) for \(k = 1, \dots, n_c - 1.
\)

Correspondingly, the Hamiltonian decomposes into a same-sign two-dimensional effective spin-\(\tfrac{1}{2}\) block and \( n_c - 1 \) opposite-sign spin-0 blocks.  
The system Hamiltonian \( \mb{H}_1 \), which governs the evolution during Stages~1 and~2, when restricted to the low-energy subspace \( \SLE \) of the clique \( G_c \), takes the form
\[
\bar{\mb{H}_1} 
= -\mt{x}\, \bar{\Sop{\X}}
+ \mt{jxx}\, \bar{\Sop{\XX}} 
- w_c\, \bar{\Sop{\sZ}}.
\]

Substituting the operator expressions from Eq.~\ref{eq:transformed-operators} in Theorem~\ref{thm:Sop}, we obtain
\begin{align}
\bar{\mb{H}_1} &= 
\left(
- \tfrac{\sqrt{n_c}}{2}\, \mt{x}\, \sigma^x
+ \left( - w_c + \tfrac{n_c - 1}{4} \mt{jxx} \right)\shz{}
\right)
\oplus
\underbrace{
\left[
- \left( w_c + \tfrac{1}{4} \mt{jxx} \right)
\right]
\oplus \cdots \oplus
\left[
- \left( w_c + \tfrac{1}{4} \mt{jxx} \right)
\right]
}_{n_c - 1}\\
&= \mb{B}(\mt{\weff_c}, \sqrt{n_c}\mt{x})
\oplus 
\underbrace{[-(w_c + \tfrac{1}{4} \mt{jxx})] 
\oplus \cdots 
\oplus [- (w_c + \tfrac{1}{4} \mt{jxx})]}_{n_c - 1}
\label{eq:Hsc}
\end{align}
where the effective weight is defined as
\(
\mt{\weff_c} = w_c - \tfrac{n_c - 1}{4} \mt{jxx}.
\)

An illustration of this basis transformation is provided in Figure~\ref{fig:M4},
which shows how the original product basis
is transformed into a direct-sum decomposition via total angular momentum.

\begin{figure}[h!]
  \centering
  \includegraphics[width=0.8\textwidth]{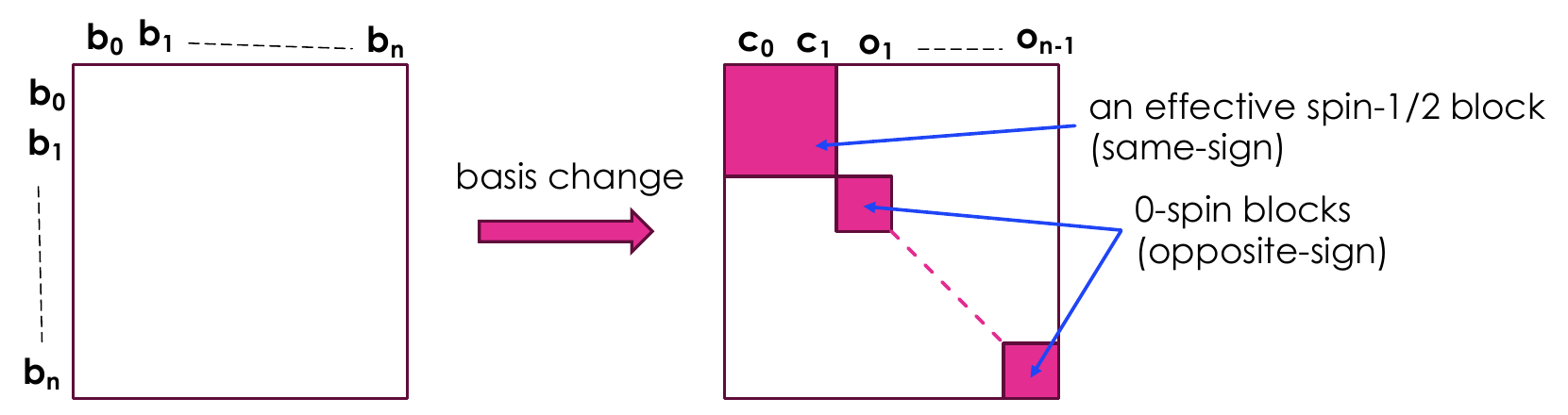}
  \caption{
    Basis transformation of \( \SLE \) from the product space  \( \{ \bst{0}, \bst{1}, \dots, \bst{n} \} \) to the direct-sum space (angular momentum basis).
    The same-sign states \( \{{\bf c_0, c_1} \} \) are the basis of the effective spin-\(\tfrac{1}{2}\) subspace, while the opposite-sign states \( \{ {\bf o_1, \dots, o_{n - 1} }\} \)
    are the bases of the spin-$0$ subspaces. The Hamiltonian matrix decomposes accordingly into a same-sign block and \( n - 1 \) opposite-sign blocks.
  }
  \label{fig:M4}
\end{figure}

\subsubsection{Full Spectrum and the $\Jxx$ See-Saw Effect}
\label{sec:see-saw-spectrum}
Since the eigensystem of $\mb{B}(w,x)$ is known analytically
(Eqs.~\eqref{eq:B-evals}, \eqref{eq:B-evecs}),
the full spectrum and eigenstates of the Hamiltonian $\bar{\mb{H}_1}$
for the single clique $G_c$ are also known analytically.
In particular, the eigenvalues of the same-sign block \( \mb{B}(\mt{\weff_c}, \sqrt{n_c} \mt{x}) \) are:
\begin{equation}
\beta_k = -\tfrac{1}{2}\!\left(\mt{\weff_c} + (-1)^k \sqrt{[\mt{\weff_c}]^2 + n_c[\mt{x}]^2}\right), \quad k=0,1.
\label{eq:beta-value}
\end{equation}
%% with corresponding eigenvectors:
%% \begin{align}
%% \begin{cases}
%% \ket{\beta_0} = \tfrac{1}{\sqrt{1 + \gamma^2}} \left( \gamma \ket{0} + \ket{1} \right), \\
%% \ket{\beta_1} = \tfrac{1}{\sqrt{1 + \gamma^2}} \left( \ket{0} - \gamma \ket{1} \right),
%% \end{cases}
%% \label{eq:beta-vector}
%% \end{align}
%% with the mixing coefficient
%% \begin{align}
%%   \gamma = \tfrac{ \sqrt{n_c}\, \mt{x} }{ \mt{\weff_c} + \sqrt{ \left[\mt{\weff_c}\right]^2 + n_c \left[\mt{x}\right]^2 } }.
%%   \label{eq:gamma}
%% \end{align}
Note that $\beta_0$ and $\beta_1$ are time-dependent through $\mt{x}(t)$.

\paragraph{Analytic Spectrum and the See-Saw Effect.}
The spectrum of the reduced Hamiltonian \( \bar{\mb{H}_1} \) consists of three distinct eigenvalues:
\(
\left\{ \beta_0,\, \beta_1,\, \theta \right\},
\)
where \( \{\beta_0, \beta_1\} \) arise from the effective spin-\(\tfrac{1}{2}\) (same-sign) block and
\(
\theta = -\left( w_c + \tfrac{1}{4} \mt{jxx} \right)
\)
is the eigenvalue shared by the \( (n_c - 1) \) degenerate spin-0 (opposite-sign) blocks; see Figure~\ref{fig:see-saw}(a).

We now examine how this spectrum changes with the strength \( \Jxx \).  
Let \( \Jxx = \alpha \Gamma_2 \) and \( \mt{jxx} = \alpha \mt{x} \), with \( \alpha > 0 \). Then the effective weight
\(
\mt{\weff_c}(\mt{x}) = w_c - \tfrac{n_c - 1}{4} \mt{jxx}
\)
crosses zero at the critical point
\(
x_0 = \tfrac{4 w_c}{\alpha (n_c - 1)}.
\)
We thus have
\[
\begin{cases}
\mt{\weff_c}(\mt{x}) > 0 & \text{for } \mt{x}<x_0, \\
\mt{\weff_c}(\mt{x}) = 0 & \text{at } \mt{x} = x_0, \\
\mt{\weff_c}(\mt{x}) < 0 & \text{for } \mt{x} > x_0,
\end{cases}
\]
At this point, the sign of $\mt{\weff_c}$  changes, and the roles of \(\ket{0}\) and \(\ket{1}\) in the ground state eigenvector exchange.
For \( \mt{x} > x_0 \), the ground state energy of the same-sign block approximates
\(
\beta_0(\mt{x}) \approx -\tfrac{1}{2\alpha}\, \mt{x},
\)
so its slope with respect to \( \mt{x} \) is approximately \( -\tfrac{1}{2\alpha} \), independent of \( n_c \).
In contrast, when \( \Jxx = 0 \) (i.e., \( \alpha = 0 \)), the same energy behaves as
\(
\beta_0(\mt{x}) \approx -\tfrac{1}{2} \left(w_c + \sqrt{n_c} \mt{x} \right).
\)
As \( \Jxx \) increases, the magnitude of this slope decreases---from approximately \( \sqrt{n_c} \) down to \( \tfrac{1}{\alpha} \)---flattening the spectral curve; see Figure~\ref{fig:see-saw}(c).

\paragraph{Crossover point $x_c$.} The crossing point \( x_c \) at which \( \beta_0(x_c) = \theta(x_c) \) satisfies
\begin{equation}
  x_c = \tfrac{4 \alpha}{4 - \alpha^2}.
  \label{eq:crossing-point}
  \end{equation}

To ensure that \( x_c \le \Gamma_2 \), we must restrict \( \alpha \) to satisfy
\begin{equation}
 \alpha < \alpha_{\max}(\Gamma_2) := \tfrac{-2 + 2\sqrt{1 + \Gamma_2^2}}{\Gamma_2},
\label{eq:alpha-max}
\end{equation}

At \( \mt{x} = x_c \), the ground state switches from the same-sign subspace to the opposite-sign subspace.  
This transition is a key feature of the algorithm and motivates the design of the two-stage schedule.  
The behavior of this crossing is shown in Figure~\ref{fig:see-saw}(b), and its dependence on \( \alpha \) across the full schedule is illustrated in Figure~\ref{fig:2-stage}.
This interplay---the rising of the same-sign energies and simultaneous lowering of the opposite-sign energy as \( \Jxx \) increases---manifests as a \emph{see-saw effect}, shown in Figure~\ref{fig:see-saw}(d).

\begin{figure}[!htbp]
  \centering
  $$
  \begin{array}{cc}
    \includegraphics[width=0.5\textwidth]{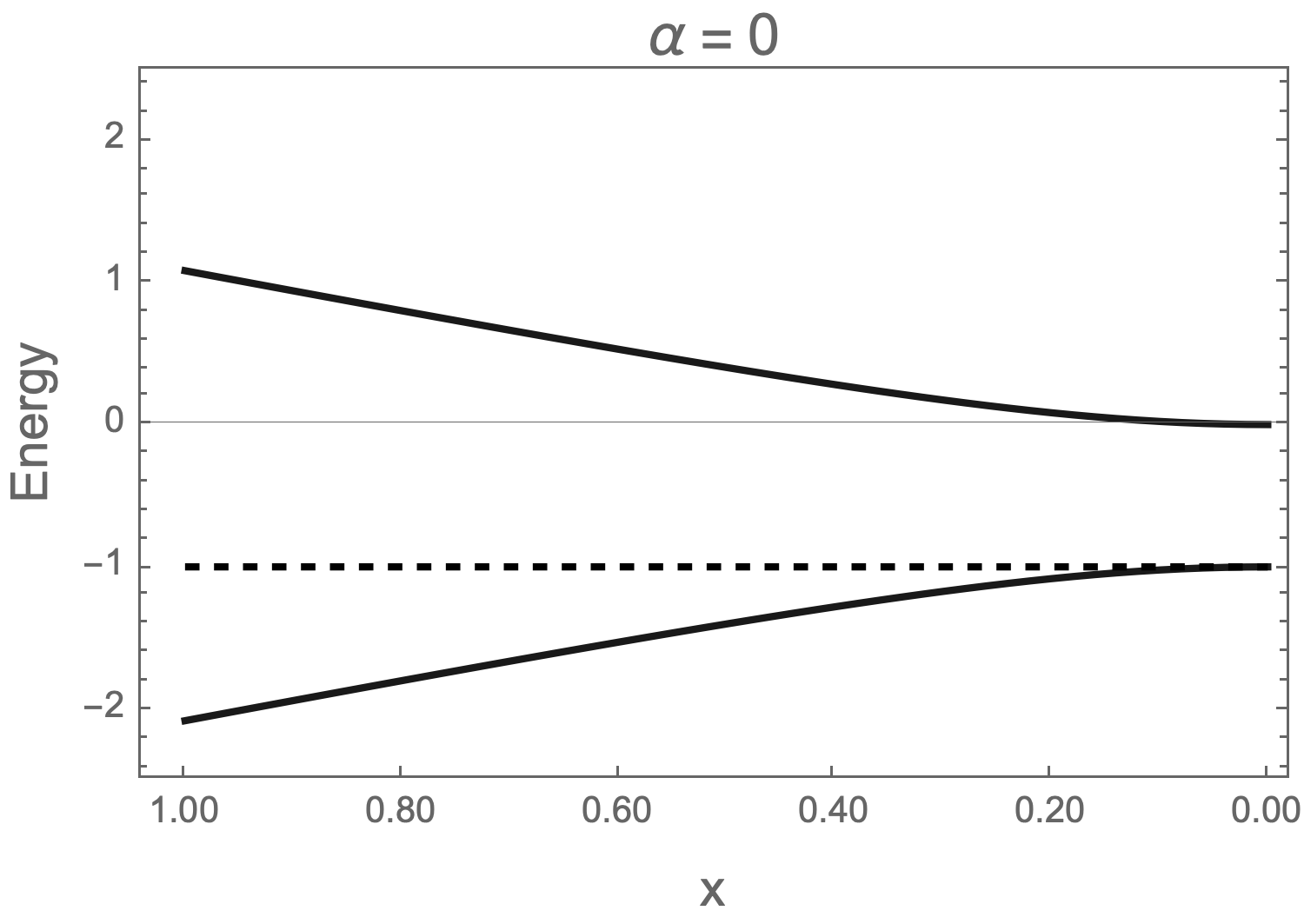} &
    \includegraphics[width=0.5\textwidth]{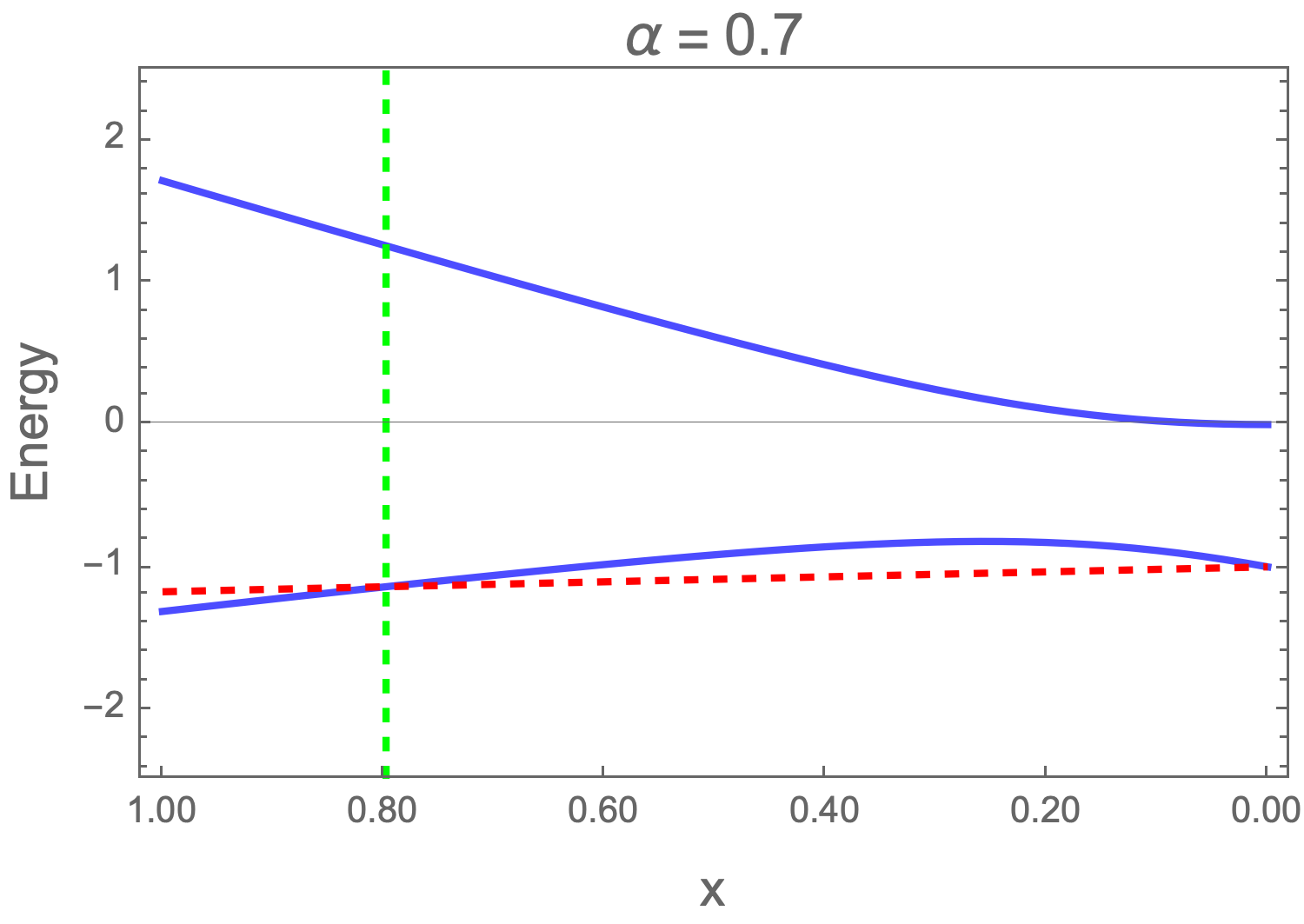} \\
    (a) & (b)\\
    \includegraphics[width=0.5\textwidth]{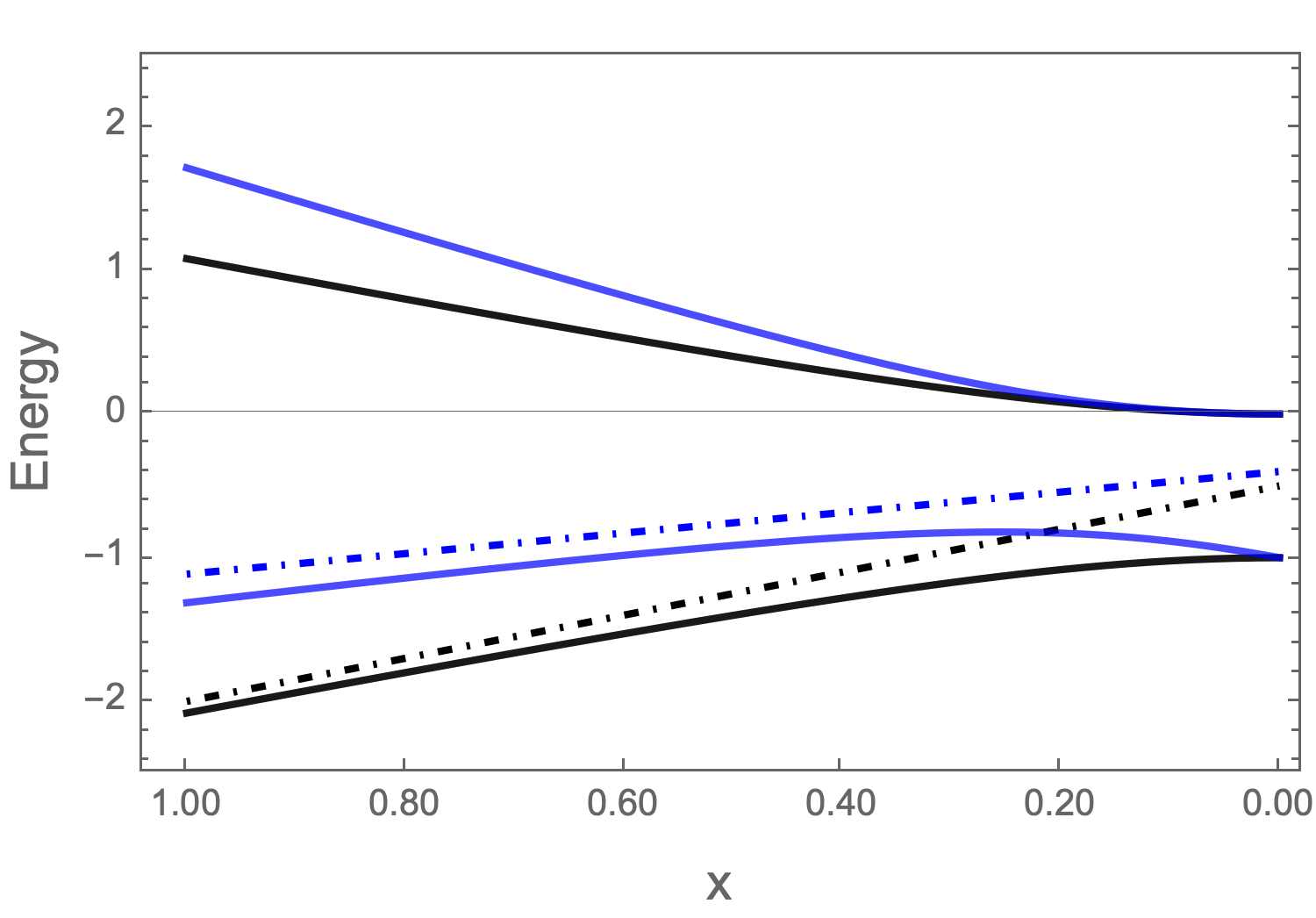} &
    \includegraphics[width=0.5\textwidth]{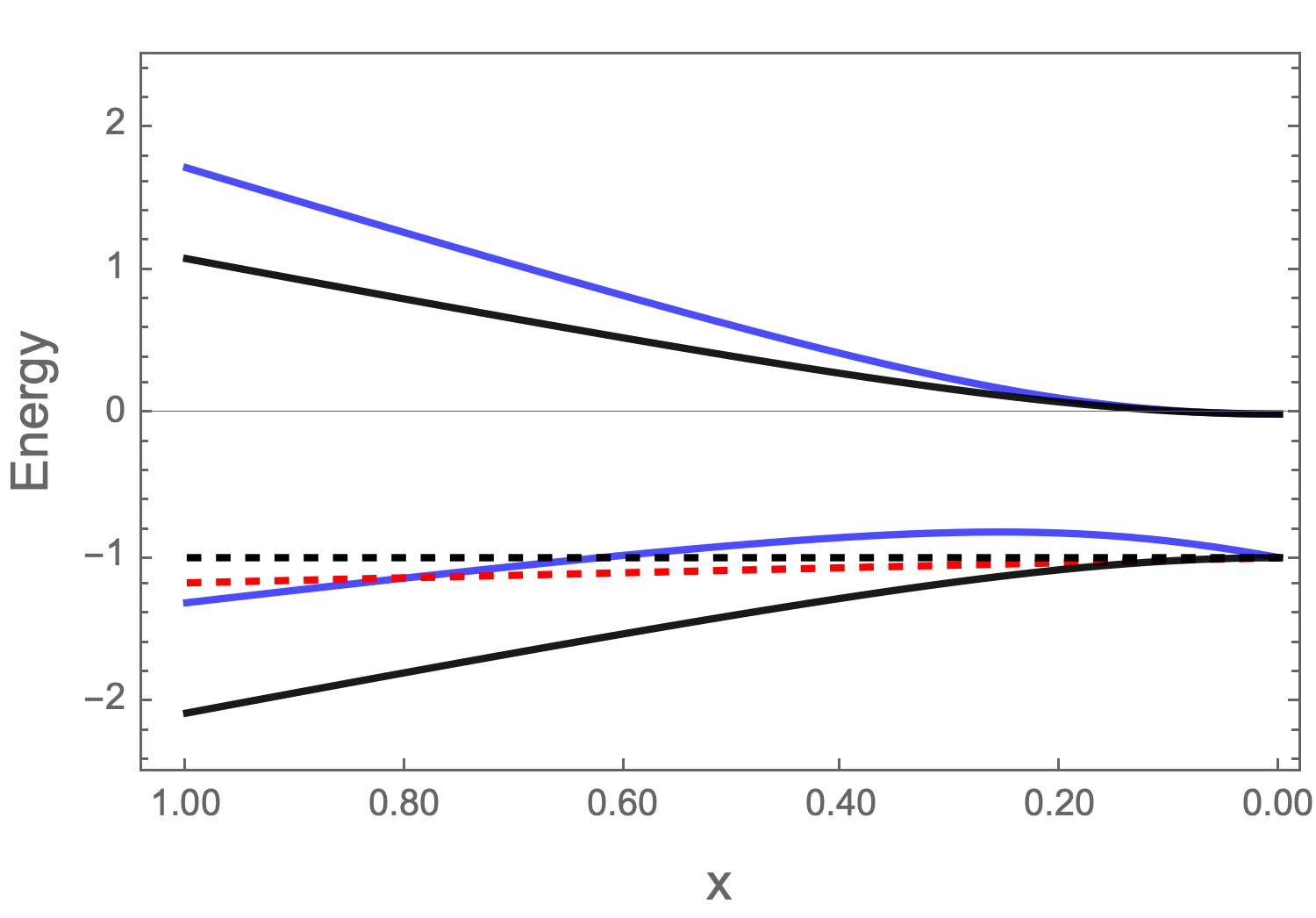} \\
    (c) & (d)    
  \end{array}
  $$
  \caption{Effect of $\Jxx$ on the eigenvalues of a single clique, where $\Jxx = \alpha \Gamma_2$.
We compare the case $\alpha = 0$ with $\alpha = 0.7$.\\
Each clique has three eigenvalues: $\{\beta_0(\mt{x}), \beta_1(\mt{x})\}$ from the effective spin-$\tfrac{1}{2}$ subspace, and $\theta(\mt{x})$ from the spin-0 state. The x-axis is the transverse field $\mt{x}$, which decreases from $\Gamma_2 = 1$ to $0$. Note: $\mt{jxx} = \alpha \mt{x}$. [$w_c = 1$, $n_c = 9$]\\
(a) \( \alpha = 0 \): \( \beta_0(\mt{x}) \), \( \beta_1(\mt{x}) \) in solid black; \( \theta(\mt{x}) \) in dashed black.\\
(b) \( \alpha = 0.7 \): \( \beta_0(\mt{x}) \), \( \beta_1(\mt{x}) \) in solid blue; \( \theta(\mt{x}) \) in dashed red.
The green dashed line marks the crossing point $x_c = \tfrac{4 \alpha}{4 - \alpha^2}$ where $\beta_0(x_c) = \theta(x_c)$.\\
(c) \textbf{Slope magnitude reduced:} as $\Jxx$ increases from $0$ to $0.7$, the same-sign energies $\beta_0$ and $\beta_1$ rise (black to blue).
The magnitude of the negative slope decreases; that is, the curve becomes less steep, flattening from approximately $-\sqrt{n_c}$ (black dotdashed) to $-\tfrac{1}{\alpha}$ (blue dotdashed).\\
(d) \textbf{See-saw effect:} increasing $\Jxx$ simultaneously raises the same-sign energies (black to blue) and lowers the opposite-sign energy $\theta(\mt{x})$ (black dashed to red dashed), in a see-saw-like fashion.
}
\label{fig:see-saw}
\end{figure}

\subsubsection{Transformation via Two Coupled Subcliques}
\label{sec:partial-coupling}
In this section, we reinterpret a single clique as two coupled subcliques and derive the transformation needed to express the partial clique $\ZZ$-coupling.

Consider a clique \( c \) composed of two subcliques \( a \) and \( b \), with \( n_c = n_a + n_b \).  
We apply the single-clique analysis from Eq.~\eqref{eq:single-clique} independently to each subclique, yielding:
\begin{align*}
  \mc{L}^a &= \left[\tfrac{1}{2}\right]_{n_a} \oplus \underbrace{0 \oplus \cdots \oplus 0}_{n_a - 1},\\
  \mc{L}^b &= \left[\tfrac{1}{2}\right]_{n_b} \oplus \underbrace{0 \oplus \cdots \oplus 0}_{n_b - 1}.
\end{align*}

The low-energy subspace \(\mc{L}^c\) of clique \( c \) is the restricted subspace within the tensor product of \(\mc{L}^a\) and \(\mc{L}^b\).  
Projecting onto the low-energy sector of two spin-\(\tfrac{1}{2}\) systems yields an effective spin-\(\tfrac{1}{2}\) for the full clique, along with an additional spin-0 component:
\[
\PLE\left(\left[\tfrac{1}{2}\right]_{n_a} \otimes \left[\tfrac{1}{2}\right]_{n_b}\right)\PLE = \left[\tfrac{1}{2}\right]_{n_c} \oplus 0.
\]

We now describe the corresponding Hamiltonians. For simplicity, we drop the time parameter \( (t) \). 

The Hamiltonians for the effective spin-\(\tfrac{1}{2}\) block are explicitly:
\[
\mathbb{B}_a = 
\begin{bmatrix}
 -\left(\mathtt{w} -\tfrac{n_a - 1}{4}\mathtt{jxx}\right)  & -\tfrac{\sqrt{n_a}}{2}\mathtt{x} \\[6pt]
 -\tfrac{\sqrt{n_a}}{2}\mathtt{x} & 0
\end{bmatrix}, \quad
\mathbb{B}_b = 
\begin{bmatrix}
 -\left(\mathtt{w} -\tfrac{n_b - 1}{4}\mathtt{jxx}\right) & -\tfrac{\sqrt{n_b}}{2}\mathtt{x} \\[6pt]
 -\tfrac{\sqrt{n_b}}{2}\mathtt{x} & 0
\end{bmatrix}.
\]

The combined low-energy Hamiltonian is:
\[
\mathbb{D}_c = \PLE \left(\mathbb{B}_a \otimes \mathbb{I}_b + \mathbb{I}_a \otimes \mathbb{B}_b + H_{\XX}\right)\PLE,
\]
which takes the explicit \(3 \times 3\) form:
\[
\mathbb{D}_c =
\begin{bmatrix}
 -\left(\mathtt{w} -\tfrac{n_a - 1}{4}\mathtt{jxx}\right)  & -\tfrac{\sqrt{n_a}}{2}\mathtt{x}  & \tfrac{\sqrt{n_a n_b}}{4}\mathtt{jxx} \\[6pt]
 -\tfrac{\sqrt{n_a}}{2}\mathtt{x} & 0 & -\tfrac{\sqrt{n_b}}{2}\mathtt{x} \\[6pt]
 \tfrac{\sqrt{n_a n_b}}{4}\mathtt{jxx} & -\tfrac{\sqrt{n_b}}{2}\mathtt{x} & -\left(\mathtt{w} -\tfrac{n_b - 1}{4}\mathtt{jxx}\right)
\end{bmatrix}.
\]

To recover the effective spin-\(\tfrac{1}{2}\) plus zero-spin structure, we apply the basis transformation:
\[
\Ucombine =
\begin{bmatrix}
 \sqrt{\tfrac{n_a}{n_c}} & 0 & -\sqrt{\tfrac{n_b}{n_c}} \\[6pt]
 0 & 1 & 0 \\[6pt]
 \sqrt{\tfrac{n_b}{n_c}} & 0 & \sqrt{\tfrac{n_a}{n_c}}
\end{bmatrix},
\]
where the columns correspond to the transformed basis states:
\[
\ket{1_c} = \sqrt{\tfrac{n_a}{n_c}}\ket{1_a0_b} + \sqrt{\tfrac{n_b}{n_c}}\ket{0_a1_b}, \quad
\ket{0_c} = \ket{0_a0_b}, \quad
\ket{\odot_q} = -\sqrt{\tfrac{n_b}{n_c}}\ket{1_a0_b} + \sqrt{\tfrac{n_a}{n_c}}\ket{0_a1_b}.
\]

In this basis, the Hamiltonian simplifies to:
\[
\Ucombine^\dagger \mathbb{D}_c \Ucombine =
\begin{bmatrix}
 -\left(\mathtt{w} -\tfrac{n_c - 1}{4}\mathtt{jxx}\right) & -\tfrac{\sqrt{n_c}}{2}\mathtt{x} & 0 \\[6pt]
 -\tfrac{\sqrt{n_c}}{2}\mathtt{x} & 0 & 0 \\[6pt]
 0 & 0 & -\left(\mathtt{w} +\tfrac{1}{4}\mathtt{jxx}\right)
\end{bmatrix},
\]
which clearly isolates the effective spin-\(\tfrac{1}{2}\) subsystem and the spin-0 subsystem, as described at the outset of this section.

\begin{remark}
If the vertex weights are not exactly the same, their difference will appear in the off-diagonal, weakly coupling the two blocks.
This is the case in our original example in \cite{Choi2021}, where the \LM{} is almost degenerate. 
The weight difference is so small (considered as almost-degenerate) that one can treat the blocks as if they were disjoint.
\end{remark}

\paragraph{Partial clique $\ZZ$-coupling.}

Suppose an external spin \(r\) $ZZ$-couples only to subclique \(a\) (and not to subclique \(b\)), described initially by the operator:
\(
\left(\sum_{i=1}^{n_a} \shz{i}\right)\shz{r}.
\)
Applying the single-clique result to subclique \(a\), this simplifies to:
\(
 \shz{a}\shz{r}.
\)

Now, under the transformation \(\Ucombine \), we explicitly obtain:
\[
\Ucombine^\dagger \shz{a} \Ucombine =
\begin{bmatrix}
\dfrac{n_a}{n_c} & 0 & -\dfrac{\sqrt{n_a n_b}}{n_c} \\[8pt]
0 & 0 & 0 \\[8pt]
-\dfrac{\sqrt{n_a n_b}}{n_c} & 0 & \dfrac{n_b}{n_c}
\end{bmatrix},
\]
demonstrating clearly the induced coupling between the same-sign two-dimensional subsystem and the zero-spin subsystem.
The coupling strength is explicitly given by:
\(
-\tfrac{\sqrt{n_a n_b}}{n_c}.
\)

In particular, we denote the transformed operator for \(n_a = n_c -1\) and \(n_b =1\) by
\begin{align}
  \mb{T} =
\begin{bmatrix}
\dfrac{n_c -1}{n_c} & 0 & -\dfrac{\sqrt{n_c -1}}{n_c} \\[6pt]
0 & 0 & 0 \\[6pt]
-\dfrac{\sqrt{n_c -1}}{n_c} & 0 & \dfrac{1}{n_c}
\end{bmatrix}
\label{eq:opT}
\end{align}
which will be used in our later analysis.

%% We will show that if \(\Jzz\) is sufficiently small, the coupling between the effective-spin and
%% zero-spin subsystems can be treated perturbatively, simplifying the subsequent analysis.
%% Thus, the analysis of the full Hamiltonian can be approximated by considering the two blocks independently.

\section{Analytical Solution for the Bare Subsystem (\MIC{})}
\label{sec:bare-sub}
\label{sec:bare}

We define the \emph{bare subsystem} as the system restricted to a given \MIC{},  
that is, a subgraph consisting of independent cliques \( \{ \text{Clique}(w_i, n_i) \}_{i=1}^k \).
A \MIC{} defines either a critical local minima (\LM{}), or the global minimum (\GM{}) where each clique has size one.
It is a \emph{subsystem} because we consider only the subgraph induced by this \MIC{},  
and it is \emph{bare} because it is fully decoupled from the rest of the graph.

We refer to the bare subsystem induced by the set of cliques in \( L \) (generating the local minima \LM{})  
as the \emph{\( L \)-bare subsystem}, denoted by \( \mb{H}_L^{\ms{bare}} \). 
Similarly, the bare subsystem induced by the subset \( R \subseteq \GM \) is called the \emph{\( R \)-bare subsystem},  
denoted by \( \mb{H}_R^{\ms{bare}} \), and referred to as the \emph{GM-bare subsystem} when \( R = \GM \).
In our analysis in later sections, we will express the total Hamiltonian  
in terms of the \( L \)- and \( R \)-bare subsystems together with their coupling.
Throughout, we distinguish bare subsystem operators from their full-system counterparts 
by using the superscript \(^{\ms{bare}}\); the same symbol without the superscript 
refers to the corresponding operator in the full system.

The bare subsystem admits a fully analytical treatment due to its tensor-product structure across independent cliques.  Our analysis proceeds in three steps, and the resulting closed-form solutions will later be used to derive analytical bounds on \( \Jxx \).
\begin{itemize}
  \item First, we decompose the low-energy subspace of the bare subsystem into three types of sectors:  
  the \emph{same-sign sector} \( \mc{C}^{\ms{bare}} \), the \emph{all-spin-zero (AS0) opposite-sign sector} \( \mc{Q}^{\ms{bare}} \), and the \emph{intermediate opposite-sign sectors}.  
  We define corresponding Hamiltonians for each: \( \mb{H}_{\mc{C}^{\ms{bare}}} \), \( \mb{H}_{\mc{Q}^{\ms{bare}}} \), and \( \mb{H}_{\mc{W}^{\ms{bare}}} \), where \( \mc{W}^{\ms{bare}} \) denotes a typical intermediate sector.

  \item Second, we derive a closed-form solution for the same-sign block \( \mb{H}_{\mc{C}^{\ms{bare}}} \),  
  including explicit expressions for its eigenvalues and eigenstates.  
  In the case of uniform clique sizes, we perform a symmetric-subspace reduction based on angular momentum structure,  
  reducing the Hilbert space dimension from \( 2^m \) to \( m + 1 \).

  \item Third, the ground-state energies of the blocks are totally ordered.  
  Depending on the relative energy of the spin-\( \tfrac{1}{2} \) and spin-$0$ components, the ordering takes one of two forms:  
  either the same-sign sector has the lowest energy and the AS0 opposite-sign sector the highest,  
  or the order is reversed.
\end{itemize}

\subsection{Low-Energy Structure of the Bare Subsystem: Same-Sign and Opposite-Sign Sectors}

Let \( G_{\ms{mic}} \) consist of \( m \) independent cliques \( \ms{Clique}(w_i, n_i) \),  
where \( w_i \equiv 1 \) denotes the vertex weight and \( n_i \) the size of the \( i \)th clique.  
This graph admits \( \prod_{i=1}^{m} n_i \) degenerate maximal independent sets of size \( m \).

From the single-clique result (Eq.~\eqref{eq:single-clique}), the low-energy subspace of each clique decomposes as:
\[
\mc{L}_i = \left[\tfrac{1}{2}\right]_{n_i} \oplus \underbrace{0 \oplus \cdots \oplus 0}_{n_i - 1}.
\]

Hence, the total low-energy subspace \( \mc{L} \) of the bare subsystem is given by:
\[
\mc{L} = \bigotimes_{i=1}^{m} \mc{L}_i 
= \bigotimes_{i=1}^{m} \left( \left[\tfrac{1}{2}\right]_{n_i} \oplus \underbrace{0 \oplus \cdots \oplus 0}_{n_i - 1} \right)
= \bigoplus \left( \bigotimes_{i=1}^{m} \left[\tfrac{1}{2} \text{ or } 0 \right] \right),
\]
where the direct sum ranges over all tensor products selecting one spin-\( \tfrac{1}{2} \) and \( n_i - 1 \) spin-zeros from each clique.

This decomposition yields \( \prod_{i=1}^{m} n_i \) block subspaces. Among them, the block
\[
\mc{C}^{\ms{bare}} := \bigotimes_{i=1}^{m} \left[\tfrac{1}{2}\right]_{n_i}
\]
is composed entirely of spin-\( \tfrac{1}{2} \) factors and is referred to as the \emph{same-sign sector}.

At the opposite extreme, there are \( \prod_{i=1}^{m} (n_i - 1) \) blocks composed entirely of spin-$0$ factors:
\[
\mc{Q}^{\ms{bare}} := \bigotimes_{i=1}^{m} 0,
\]
referred to as the \emph{all-spin-zero (AS0) opposite-sign sector}.

The remaining opposite-sign blocks, each containing a mixture of spin-\( \tfrac{1}{2} \) and spin-$0$ components, are called the \emph{intermediate sectors}.
Let \( \mc{W}^{\ms{bare}} \) denote one such intermediate sector.

We denote the restriction of the Hamiltonian \( \mb{H}_1 \) to each sector by:
\[
\mb{H}_{\mc{C}^{\ms{bare}}} := \text{projection of } \mb{H}_1 \text{ to } \mc{C}^{\ms{bare}}, \quad
\mb{H}_{\mc{Q}^{\ms{bare}}} := \text{projection of } \mb{H}_1 \text{ to } \mc{Q}^{\ms{bare}}, \quad
\mb{H}_{\mc{W}^{\ms{bare}}} := \text{projection of } \mb{H}_1 \text{ to } \mc{W}^{\ms{bare}}. 
\]

The same-sign block Hamiltonian is given by:
\begin{align}
\mb{H}_{\mc{C}^{\ms{bare}}} = \sum_{i=1}^{m} \mb{B}_i \left( \mt{\weff_i},\, \sqrt{n_i}\, \mt{x} \right),
\label{eq:MIC-cl}
\end{align}
where
\(
\mt{\weff_i} = \left( w_i - \tfrac{n_i - 1}{4} \mt{jxx} \right),
\)
and each \( \mb{B}_i \) is a two-level block operator defined as:
\[
\mb{B}_i = \mb{I}_2 \otimes \cdots \otimes \underbrace{\mb{B}}_{\text{\( i \)th position}} \otimes \cdots \otimes \mb{I}_2,
\]
with \( \mb{B} \) being the two-level Hamiltonian from Eq.~\ref{eq:B-matrix}, and \( \mb{I}_2 \) the \( 2 \times 2 \) identity matrix.

The AS0 opposite-sign block Hamiltonian is:
\[
\mb{H}_{\mc{Q}^{\ms{bare}}} = \sum_{i=1}^{m} -\left( w_i + \tfrac{1}{4} \mt{jxx} \right),
\]
which evaluates to a single-value offset.

The intermediate block Hamiltonian \( \mb{H}_{\mc{W}^{\ms{bare}}} \) can be expressed in terms of a same-sign block Hamiltonian with fewer active \( \mb{B}_i \) terms, accompanied by an appropriate energy shift.

In the following, we present a closed-form solution for the same-sign block; analogous expressions for the intermediate blocks follow similarly.

\subsection{Closed-Form Solution of the Same-Sign Block \( \mb{H}_{\mc{C}^{\ms{bare}}} \)}
\label{sec:same-sign-block}

Since each \( \mb{B}_i \) in $\mb{H}_{\mc{C}^{\ms{bare}}}$ (in Eq.~\eqref{eq:MIC-cl}) 
acts on a disjoint spin-\(\tfrac{1}{2}\) subsystem, the full Hamiltonian \( \mb{H}_{\mc{C}^{\ms{bare}}} \) is frustration-free.
Its eigenstates are tensor products of the eigenstates of the individual \( \mb{B}_i \), and its eigenvalues are additive:
\[
\mb{H}_{\mc{C}^{\ms{bare}}} \ket{\psi_1} \otimes \cdots \otimes \ket{\psi_m}
= \left( \sum_{i=1}^{m} \beta_{k_i}^{(i)} \right) \ket{\psi_1} \otimes \cdots \otimes \ket{\psi_m},
\]
where \( \ket{\psi_i} \in \{ \ket{\beta_0^{(i)}},\, \ket{\beta_1^{(i)}} \} \) are the eigenstates of the \( i \)th two-level
block \( \mb{B}_i \), and \( \beta_k^{(i)} \) are the corresponding eigenvalues
(Eqs.~\eqref{eq:B-evals}, \eqref{eq:B-evecs}).

As a result, the full spectrum and eigensystem of the same-sign block \( \mb{H}_{\mc{C}^{\ms{bare}}} \) are exactly solvable.
\begin{mdframed}
\begin{theorem}[Exact Spectrum and Ground State of the Same-Sign Block]
\label{thm:same-sign-spectrum}
Let \( \mb{H}_{\mc{C}^{\ms{bare}}} \) be given as in  Eq.~\eqref{eq:MIC-cl}.
Then all eigenvalues and eigenstates of \( \mb{H}_{\mc{C}^{\ms{bare}}} \) are exactly solvable.  
Each eigenstate is indexed by a bit string \( z = z_1 z_2 \ldots z_m \in \{0,1\}^m \), with:
\begin{align}
  \left\{
  \begin{array}{ll}
  E_z &= \sum_{i=1}^{m} \beta_{{z_i}}^{(i)}, \\
  \ket{E_z} &= \bigotimes_{i=1}^{m} \ket{\beta_{{z_i}}^{(i)}},
  \end{array}
  \right.
\end{align}
where \( \beta_k^{(i)}\) and \( \ket{\beta_k^{(i)}} \), for \( k = 0, 1 \), are the eigenvalues and eigenvectors of the local two-level subsystem \( \mb{B}_i \), given by:
\begin{align*}
  \begin{cases}
\beta_k^{(i)} &= -\tfrac{1}{2} \left( \mt{\weff_i} + (-1)^k \sqrt{ \left[\mt{\weff_i}\right]^2 + n_i\, \mt{x}^2 } \right), \\
\ket{\beta_0^{(i)}} &= \tfrac{1}{\sqrt{1 + \gamma_i^2}} \left( \gamma_i\ket{0}_i +  \ket{1}_i \right), \quad
\ket{\beta_1^{(i)}} = \tfrac{1}{\sqrt{1 + \gamma_i^2}} \left( \ket{0}_i - \gamma_i\ket{1}_i \right),
\end{cases}
  \end{align*}
with mixing ratio
\(
\gamma_i = \tfrac{ \sqrt{n_i}\, \mt{x} }{ \mt{\weff_i} + \sqrt{ [\mt{\weff_i}]^2 + n_i\, \mt{x}^2 } }.
\)

In particular, the ground state corresponds to the all-zero bit string \( z = 0 \ldots 0 \), and is given by:
\begin{align*}
  \begin{cases}
\ket{E_0} &= \bigotimes_{i=1}^{m} \ket{\beta_0^{(i)}}, \\[4pt]
E_0 &= \sum_{i=1}^{m} \beta_0^{(i)}
= -\tfrac{1}{2} \sum_{i=1}^{m} \left( \mt{\weff_i} + \sqrt{ \left[\mt{\weff_i}\right]^2 + n_i\, \mt{x}^2 } \right).
  \end{cases}
  \end{align*}
\end{theorem}
\end{mdframed}

\subsection{Uniform Clique Size: Symmetric-Subspace Reduction of the Same-Sign Block}
\label{sec:reduction-symmetric}
In the case when all cliques have the same size \( n_i = n_c \), one can further reduce the same-sign block. 
The same-sign block \( \mb{H}_{\mc{C}^{\ms{bare}}} \), defined on \( m \) effective spin-\( \tfrac{1}{2} \) subsystems, acts on a Hilbert space of dimension \( 2^m \).  
Using the same technique of transforming to the total angular momentum basis, this space decomposes into a direct sum of angular momentum sectors.
In particular, 
\[
\left[ \tfrac{1}{2} \right]_{n_c}^{\otimes m} =\underbrace{ \tfrac{1}{2} \otimes \cdots \otimes \tfrac{1}{2} }_{m}
= \tfrac{m}{2} \oplus \underbrace{ \left( \tfrac{m}{2} - 1 \right) \oplus \cdots \oplus \left( \tfrac{m}{2} - 1 \right) }_{m - 1} \oplus \cdots
\]

Let \( s = \tfrac{m}{2} \). The highest-spin sector \( s \) corresponds to the fully symmetric subspace.  
Its basis vectors \( \ket{s, m_s} \) are uniform superpositions over computational basis states with fixed Hamming weight, as illustrated in Table~\ref{t1}.
\begin{align}
  \boxed{
\begin{array}{l}
\ket{s, s} = \ket{11 \ldots 1} \quad \text{(Hamming weight \( m \))} \\[6pt]
\ket{s, s - 1} = \tfrac{1}{\sqrt{m}} \left( \ket{01\ldots1} + \cdots + \ket{11\ldots0} \right) \\[6pt]
\quad \vdots \\[6pt]
\ket{s, -(s - 1)} = \tfrac{1}{\sqrt{m}} \left( \ket{10\ldots0} + \cdots + \ket{00\ldots1} \right) \\[6pt]
\ket{s, -s} = \ket{00 \ldots 0} \quad \text{(Hamming weight 0)}
\end{array}
}
\label{t1}
\end{align}
Each \( \ket{s, m_s} \) is supported on all bitstrings of Hamming weight \( s + m_s \).  
We interpret \( 1 \) as spin-up (included vertex), and \( 0 \) as spin-down (excluded vertex).

\paragraph{Operators in the Symmetric Subspace.}
Within the spin-\( s = \tfrac{m}{2} \) subspace, the operators \( \Sop{\Z} \), \( \Sop{\sZ} \), and \( \Sop{\X} \) reduce to:
\begin{align}
\label{eq:CSZ}
\mb{C}\Sop{\Z}(m) &=
\begin{bmatrix}
s & 0 & \cdots & 0 \\
0 & s - 1 & \cdots & 0 \\
\vdots & \vdots & \ddots & \vdots \\
0 & 0 & \cdots & -s
\end{bmatrix}, \quad
\mb{C}\Sop{\sZ}(m) = \mb{C}\Sop{\Z}(m) + \tfrac{m}{2} \mathbb{I} =
\begin{bmatrix}
m & 0 & \cdots & 0 \\
0 & m - 1 & \cdots & 0 \\
\vdots & \vdots & \ddots & \vdots \\
0 & 0 & \cdots & 0
\end{bmatrix}.
\end{align}

%\paragraph{Transverse Field \( \Sop{\X} \).}
Recall that the transverse operator \( \Sop{\X} = \tfrac{1}{2}(\Sop{+} + \Sop{-}) \) where
\(
\Sop{\pm} \ket{s, m_s} = \sqrt{s(s+1) - m_s(m_s \pm 1)}\, \ket{s, m_s \pm 1}.
\)
Letting \( m_s = s - a \), where 
$a$ counts the number of spin flips downward, for \( a = 0, 1, \ldots, 2s \), this gives:
\(
\braket{s, m_s | \Sop{\X} | s, m_s - 1} = \tfrac{1}{2} \sqrt{(m - a)(a + 1)},
\)
where \( m = 2s \) is the total number of spins.

Hence,
\begin{align}
\label{eq:CSX}
\mb{C}\Sop{\X}(m) =
\begin{bmatrix}
0 & \tfrac{\sqrt{m}}{2} & 0 & \cdots & 0 \\
\tfrac{\sqrt{m}}{2} & 0 & \tfrac{\sqrt{2(m - 1)}}{2} & \cdots & 0 \\
0 & \tfrac{\sqrt{2(m - 1)}}{2} & 0 & \cdots & 0 \\
\vdots & \vdots & \vdots & \ddots & \tfrac{\sqrt{m}}{2} \\
0 & 0 & 0 & \tfrac{\sqrt{m}}{2} & 0
\end{bmatrix}.
\end{align}

One can again use CG transformation to obtain the restricted same-sign block to the symmetric subspace, which is also known as spanned by \emph{Dicke states}
and can be obtained directly through the explicit transformation.

\paragraph{Explicit Transformation to the Symmetric Subspace.}

The symmetric subspace of \( m \) spin-\( \tfrac{1}{2} \) particles corresponds to the totally symmetric sector of the full Hilbert space \( (\mathbb{C}^2)^{\otimes m} \).  
This subspace has dimension \( m+1 \), and is spanned by the \emph{Dicke states}, which are uniform superpositions over computational basis states with fixed Hamming weight.

\paragraph{Dicke Basis.}  
For each \( k = 0, 1, \dotsc, m \), define the Dicke state:
\(
\ket{s, m_s = k - \tfrac{m}{2}} = \tfrac{1}{\sqrt{\binom{m}{k}}} \sum_{\substack{\vec{z} \in \{0,1\}^m \\ \text{Hamming weight} = k}} \ket{\vec{z}}.
\)
These form an orthonormal basis for the symmetric subspace, where \( s = \tfrac{m}{2} \) and \( m_s = -s, \dotsc, s \).

\paragraph{Transformation Matrix.}  
Let \( \mr{U}_{\ms{Dicke}} \in \mathbb{C}^{2^m \times (m+1)} \) denote the matrix whose columns are the normalized Dicke states (embedded in the full space).  
Then \( \mr{U}_{\ms{Dicke}}^\dagger \mr{U}_{\ms{Dicke}} = \mathbb{I}_{m+1} \), and the following operator reductions hold:
\begin{align*}
  \begin{cases}
  \mb{C}\Sop{\sZ}(m) &= \mr{U}_{\ms{Dicke}}^\dagger\, \Sop{\sZ}(m)\, \mr{U}_{\ms{Dicke}}, \\
  \mb{C}\Sop{\X}(m)  &= \mr{U}_{\ms{Dicke}}^\dagger\, \Sop{\X}(m)\, \mr{U}_{\ms{Dicke}}, \\
  \mb{H}_{\mc{C}^{\ms{bare}}}^{\ms{sym}} &= \mr{U}_{\ms{Dicke}}^\dagger\, \mb{H}_{\mc{C}^{\ms{bare}}}\, \mr{U}_{\ms{Dicke}} 
  = - \sqrt{n_c}\, \mt{x}\, \mb{C}\Sop{\X}(m) - \mt{\weff_c}\, \mb{C}\Sop{\sZ}(m).
  \end{cases}
  \end{align*}

\begin{theorem}[Reduction to the Symmetric Subspace via Dicke Basis]
\label{thm:explicit-symmetric-reduction}
Let \( \mb{H}_{\mc{C}^{\ms{bare}}} \in \mathbb{R}^{2^m \times 2^m} \) be the same-sign block Hamiltonian acting on \( m \) effective spin-\( \tfrac{1}{2} \) subsystems, each corresponding to a clique of size \( n_c \).  
Then its restriction to the symmetric subspace is given by:
\[
\mb{H}_{\mc{C}^{\ms{bare}}}^{\ms{sym}}
= - \sqrt{n_c}\, \mt{x}\, \mb{C}\Sop{\X}(m) - \mt{\weff_c}\, \mb{C}\Sop{\sZ}(m),
\]
where
\(
\mt{\weff_c} = \left( w_c - \tfrac{n_c - 1}{4} \mt{jxx} \right).
\)
\end{theorem}

As a direct consequence, the symmetric-subspace Hamiltonian \( \mb{H}_{\mc{C}^{\ms{bare}}}^{\ms{sym}} \) is a special tridiagonal matrix with reflective off-diagonal entries. This structure admits a closed-form solution for the eigensystem.

\begin{mdframed}
  \begin{corollary}
The eigensystem of the following tridiagonal matrix:
\begin{align}
\label{eq:M}
\mathcal{M}_{m}(w,x) =
\begin{bmatrix}
   -m w & -\tfrac{\sqrt{m}}{2} x & 0 & \cdots & 0 \\
   -\tfrac{\sqrt{m}}{2} x & -(m-1) w & -\tfrac{\sqrt{2(m-1)}}{2} x & \cdots & 0 \\
   0 & -\tfrac{\sqrt{2(m-1)}}{2} x & -(m-2) w & \cdots & 0 \\
   \vdots & \vdots & \vdots & \ddots & -\tfrac{\sqrt{m}}{2} x \\
   0 & 0 & 0 & -\tfrac{\sqrt{m}}{2} x & 0
\end{bmatrix}
\end{align}
is given in closed form by:
\(
\lambda_{m_s} = -s w + m_s \sqrt{w^2 + x^2}, \quad \text{where } s = \tfrac{m}{2},\ m_s = -s, -s+1, \ldots, s.
\)
In particular, the lowest eigenvalue occurs at \( m_s = -s \), yielding:
\(
\lambda_{\min} = -s \left( w + \sqrt{w^2 + x^2} \right).
\)
\end{corollary}
\end{mdframed}
To the best of our knowledge, this closed-form eigensystem for \( \mathcal{M}_m(w,x) \) has not been known for general \( m > 1 \).

Figure~\ref{fig:MIC-see-saw} illustrates the three eigenvalues of the same-sign block in the symmetric subspace \( \mb{H}_{\mc{C}^{\ms{bare}}}^{\ms{sym}} \), along with the AS0 energy level of \( \mb{H}_{\mc{Q}^{\ms{bare}}} \), under different values of \( \Jxx \).  
 The plots reveal a clear see-saw effect: as \( \Jxx \) increases, the same-sign energy levels are lifted upward,  
while the AS0 energy level is lowered.  

\begin{figure}[!htbp]
  \centering
  \begin{tabular}{cc}
    \includegraphics[width=0.48\textwidth]{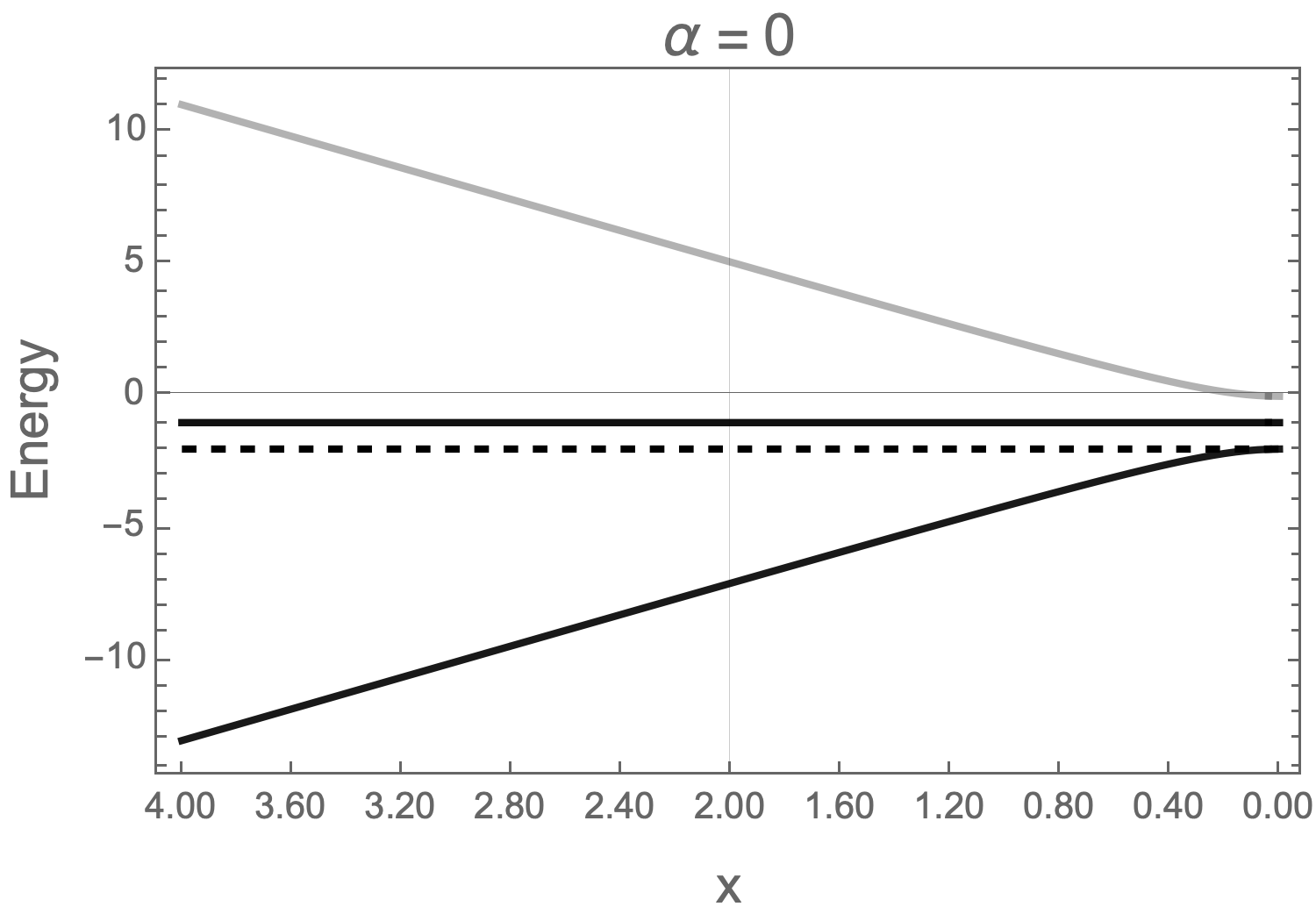} &
    \includegraphics[width=0.48\textwidth]{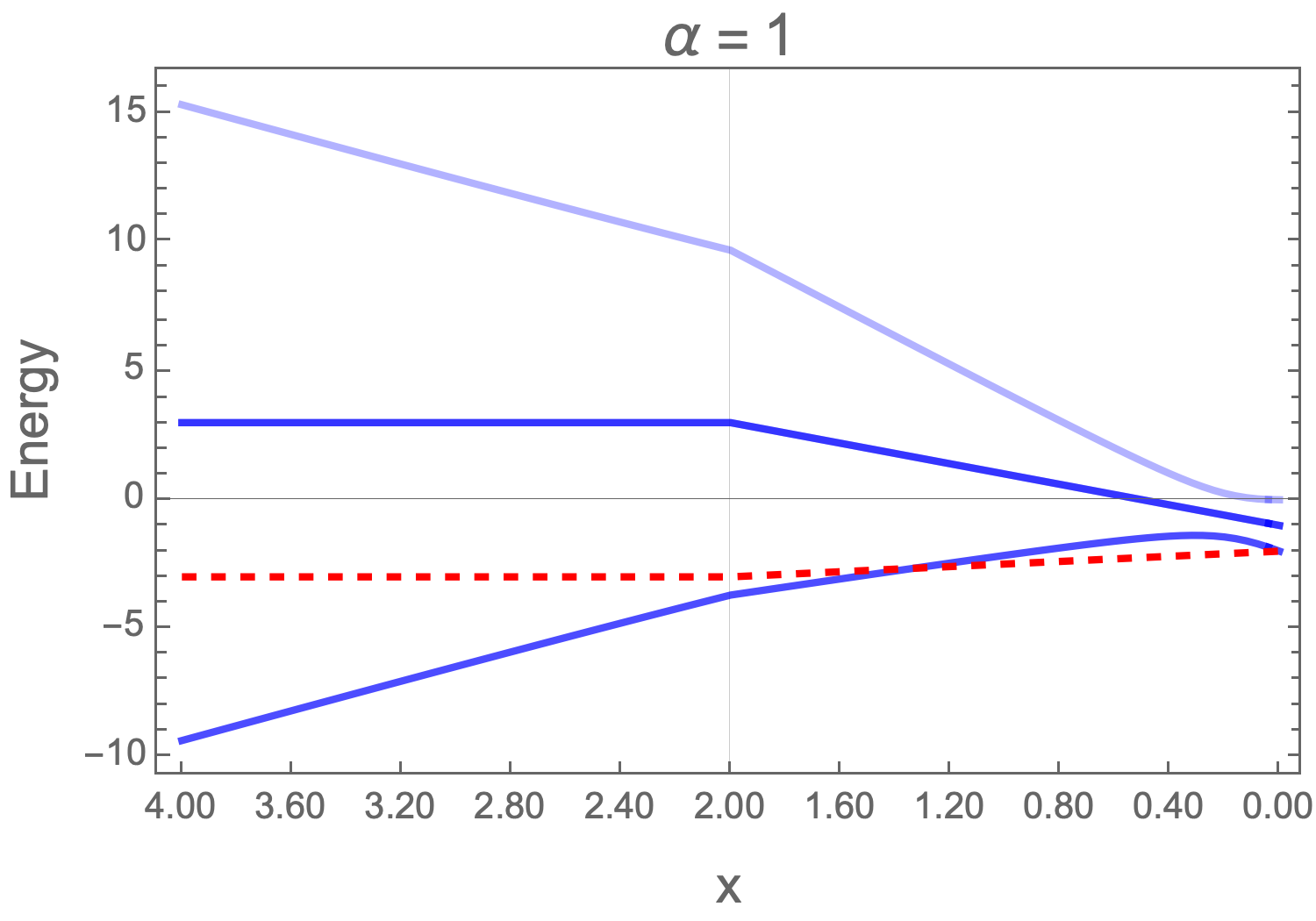}
  \end{tabular}
  \caption{
Energy spectrum of a bare subsystem (\MIC{}) under different values of \( \Jxx = \alpha \Gamma_2 \), with \( m = 2 \), \( n_c = 9 \).  
We set \( \Gamma_2 = m \) and \( \Gamma_1 = 2 \Gamma_2 \).  
The coupling schedule is piecewise-defined: \( \mt{jxx} = \Jxx \) for \( \mt{x} \ge \Gamma_2 \), and \( \mt{jxx} = \alpha \mt{x} \) for \( \mt{x} < \Gamma_2 \).\\
(a) \( \alpha = 0 \): three eigenvalues of the symmetric-subspace same-sign block \( \mb{H}_{\mc{C}^{\ms{bare}}}^{\ms{sym}} \) are shown in black;  
the energy of the AS0 block \( \mb{H}_{\mc{Q}^{\ms{bare}}} \) is shown as a black dashed line.\\
(b) \( \alpha = 1 \): same-sign eigenvalues are shown in blue; the AS0 energy is shown as a red dashed line.\\
The see-saw effect is evident: the same-sign energies (black) rise with \( \Jxx \), while the AS0 energy (black dashed) drops.
}
\label{fig:MIC-see-saw}
\end{figure}

\subsection{Block Energy Ordering}
\label{sec:block-ordering}
Since the ground state of each block is a tensor product of the eigenstates of the individual \( i \)-spins, and the corresponding eigenvalues are additive, the ground-state energy of each block is simply the sum of its local ground-state energies. Each block contains exactly \( m \) spins (either spin-\( \tfrac{1}{2} \) or spin-0), so the blocks can be directly compared based on their composition. As a result, the ground-state energies of the blocks are totally ordered.

If the ground-state energy of the spin-\( \tfrac{1}{2} \) component, denoted \( \beta_0^{(i)} \), is lower than that of the spin-0 component, denoted \( \theta^{(i)} \), then the same-sign block has the lowest energy and the AS0 opposite-sign block the highest. The ordering is reversed if \( \beta_0^{(i)} > \theta^{(i)} \). The relative ordering of \( \beta_0^{(i)} \) and \( \theta^{(i)} \) switches at the crossing point \( x_c = \tfrac{4\alpha}{4 - \alpha^2} \) (see Eq.~\eqref{eq:crossing-point} in Section~\ref{sec:single-clique}), as illustrated in Figure~\ref{fig:2-stage}.

\begin{figure}[!htbp]
  \centering
  $$
  \begin{array}{cc}
    \includegraphics[width=0.5\textwidth]{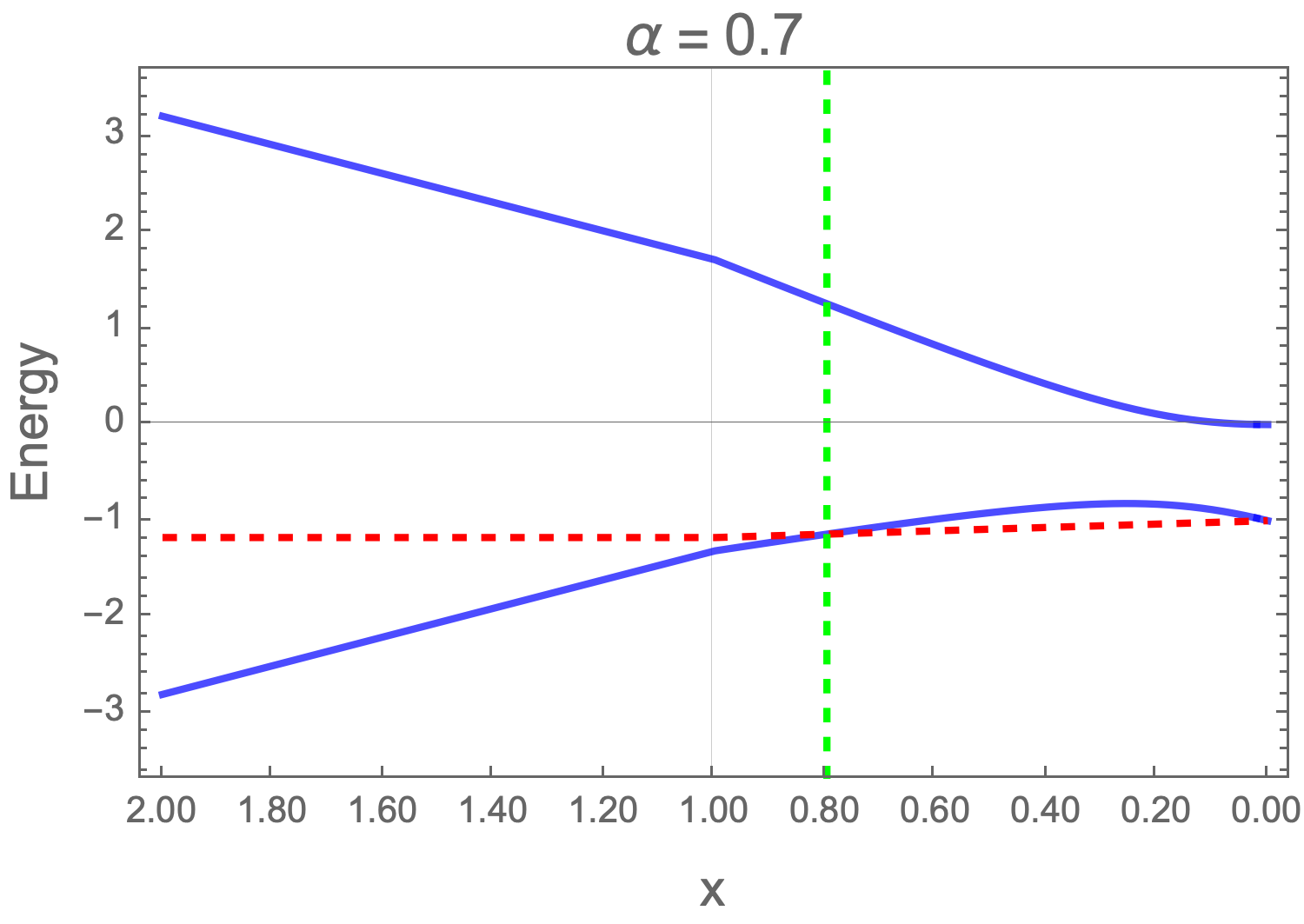} &
    \includegraphics[width=0.5\textwidth]{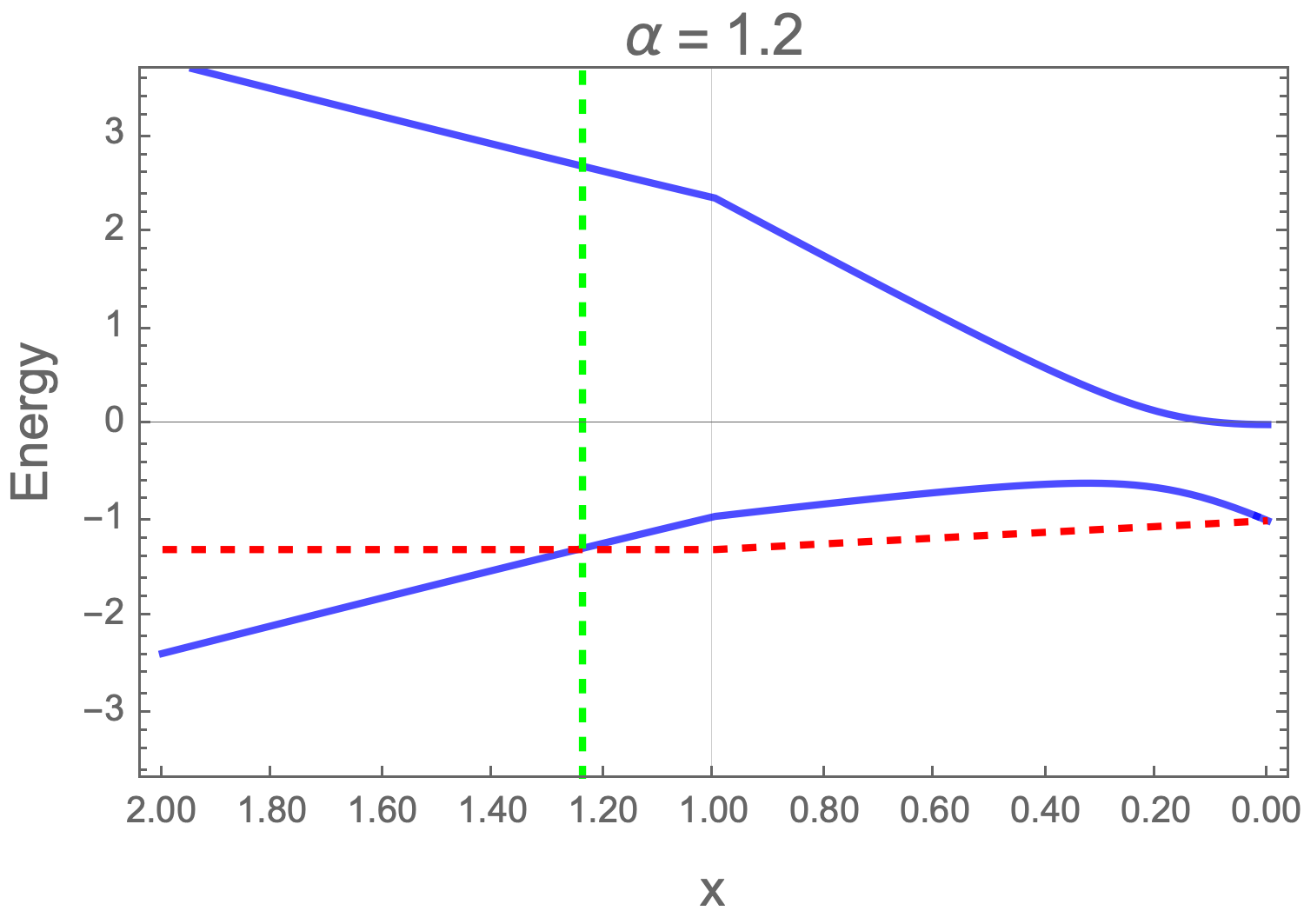}\\
    (a) & (b)
  \end{array}
  $$
  \caption{Two-stage eigenvalue evolution for a single clique under different values of \( \Jxx = \alpha \Gamma_2 \), with \( \Gamma_2 = 1 \), \( \Gamma_1 = 2 \), \( w_c = 1 \), and \( n_c = 9 \).  
  The coupling schedule is piecewise-defined: \( \mt{jxx} = \Jxx \) for \( \mt{x} \ge \Gamma_2 \), and \( \mt{jxx} = \alpha \mt{x} \) for \( \mt{x} < \Gamma_2 \).  
  (a) \( \alpha = 0.7 < \alpha_{\ms{max}}(\Gamma_2) \): the crossing occurs during Stage~2.  
  (b) \( \alpha = 1.2 > \alpha_{\ms{max}}(\Gamma_2) \): the crossing occurs during Stage~1.  
  Same-sign eigenvalues \( \beta_0 \), \( \beta_1 \) are shown as solid lines; the opposite-sign eigenvalue \( \theta \) is dashed.  
  The dashed green vertical line indicates the crossing point \( x_c = \tfrac{4\alpha}{4 - \alpha^2} \).
  }
  \label{fig:2-stage}
\end{figure}

Thus, to determine the overall block ordering, it suffices to compare the two extreme cases: the same-sign block \( \mb{H}_{\mc{C}^{\ms{bare}}} \) and the AS0 opposite-sign block \( \mb{H}_{\mc{Q}^{\ms{bare}}} \).  
In the case \( \Jxx = 0 \), i.e., \( \alpha = 0 \), we have \( \beta_0^{(i)} < \theta^{(i)} \) for all \( \mt{x} \in [0, \Gamma_1] \), so the same-sign block remains lowest in energy throughout the schedule.
In contrast, when $\alpha>0$, and in particular when $\alpha<\alpha_{\ms{max}}(\Gamma_2)$, the spectrum exhibits a transition at $x_c$: 
for $\mt{x}>x_c$, $\beta_0^{(i)}(\mt{x})<\theta^{(i)}(\mt{x})$; at $\mt{x}=x_c$, $\beta_0^{(i)}=\theta^{(i)}$; and for $\mt{x}<x_c$, $\beta_0^{(i)}>\theta^{(i)}$.
Thus, the energy ordering reverses after \( x_c \), and the AS0 opposite-sign block becomes the lowest in energy.  
This reversal is illustrated in Figure~\ref{fig:block-order}.

\begin{figure}[!htbp]
  \centering
  \begin{tabular}{cc}
    \includegraphics[width=0.45\textwidth]{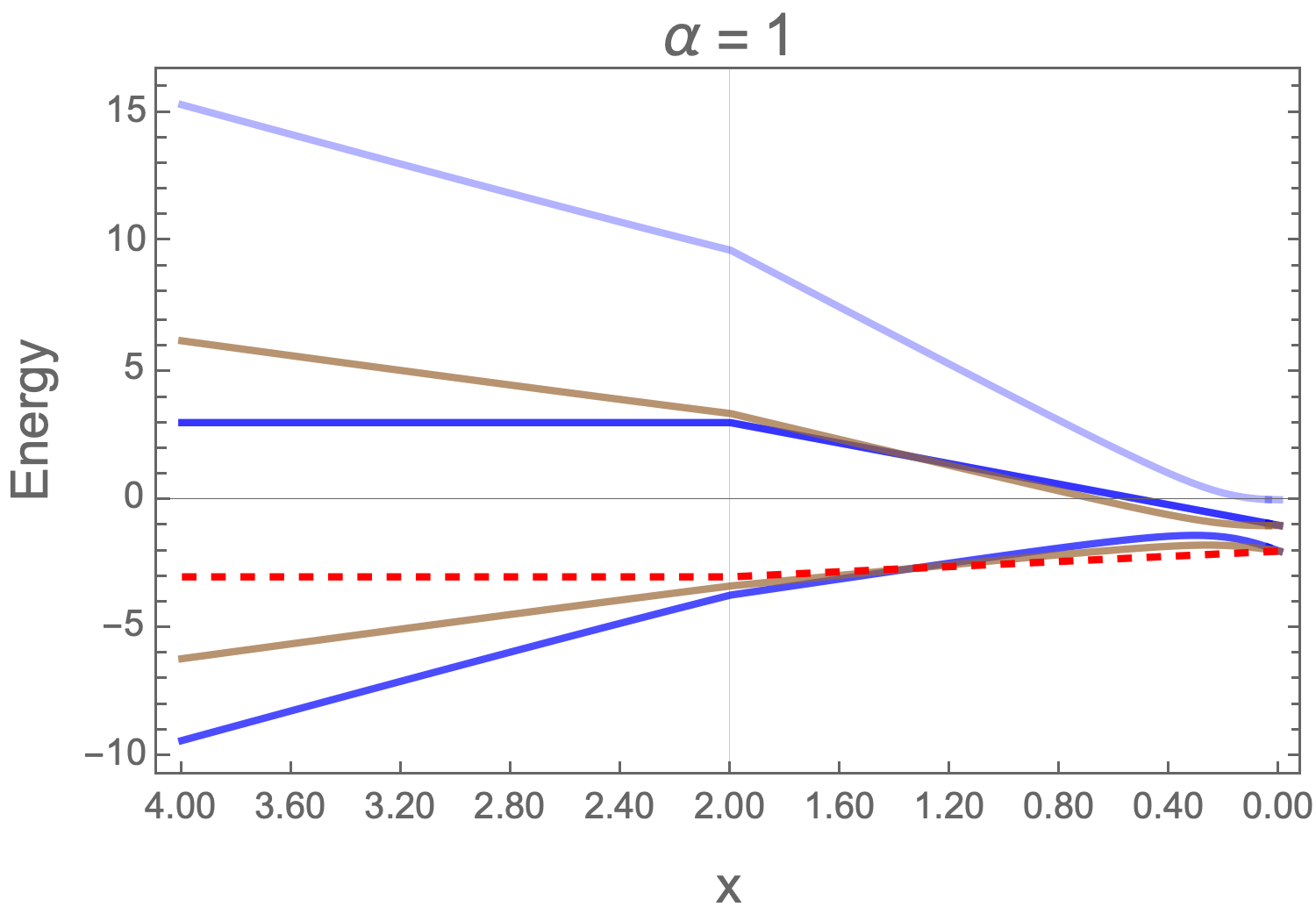} &
    \includegraphics[width=0.45\textwidth]{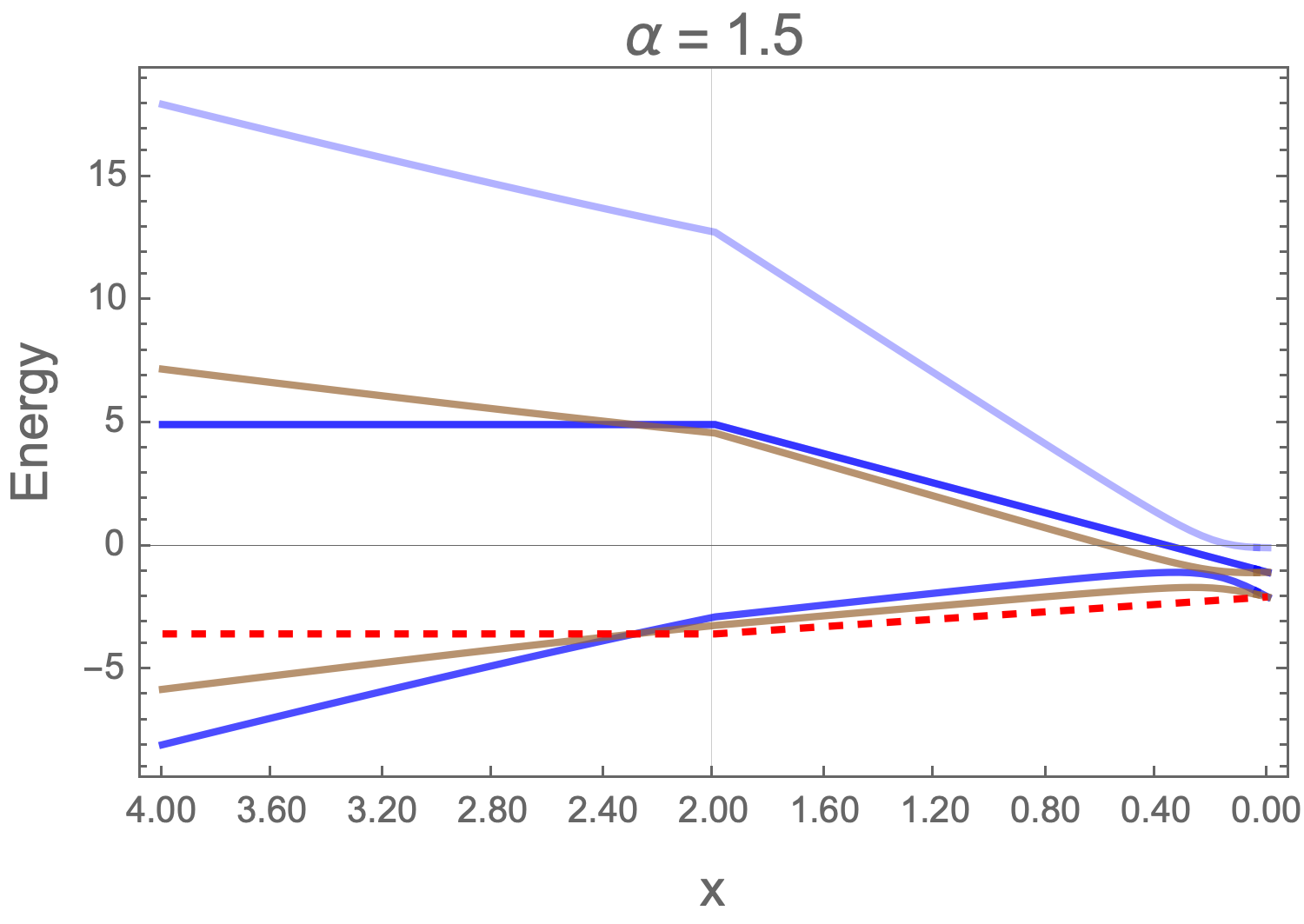}
  \end{tabular}
\caption{
Block energy ordering for \( m = 2 \), \( n_c = 9 \), with \( \Gamma_2 = m \).  
Same-sign block energy in blue; intermediate opposite-sign block in brown; AS0 block in red dashed.  
(a) \( \alpha = 1.0 < \alpha_{\ms{max}}(\Gamma_2) \): the crossing occurs during Stage~2, after which the AS0 block becomes lowest in energy.  
(b) \( \alpha = 1.5 > \alpha_{\ms{max}}(\Gamma_2) \): the crossing occurs during Stage~1, and the AS0 block becomes lowest earlier in the evolution.
}
\label{fig:block-order}
\end{figure}

%% %Part III
%% %4.\subsection{Analysis of Stage 0}
\section{Analysis of Stage 0}
\label{sec:stage-0}

Recall that the Hamiltonian for Stage~0 is
\[
\mb{H}_0(t) = \x{t} \mb{H}_{\ms{X}} + \jxx{t} \mb{H}_{\ms{XX}} + \mt{p}(t)\mb{H}_{\ms{problem}}, \quad t \in [0,1],
\]
with
\(
\x{t} = (1 - t)(\Gamma_0 - \Gamma_1) + \Gamma_1, 
\jxx{t} = t \Jxx, 
\mt{p}(t) = t.
\)
During this phase, the problem parameters \( w_i \), \( \Jzz^{\ms{clique}} \), and \( \Jzz \)  
are gradually ramped to their final values as \( \mt{p}(t) \) increases from 0 to 1,  
while \( \jxx{t} \) increases linearly to its target value \( \Jxx \).

\begin{definition}
  Given a Hamiltonian \( \mb{H} \) acting on a Hilbert space \( \mc{V} \),  
and a subspace \( \mc{L} \subset \mc{V} \) with orthogonal projector \( \Pi_{\mc{L}} \),  
we define the \emph{restricted Hamiltonian} on \( \mc{L} \) as
\(
\mb{H}|_{\mc{L}} := \Pi_{\mc{L}} \mb{H} \Pi_{\mc{L}}.
\)
\end{definition}

The main goal of this section is to derive the effective Hamiltonian at the end of Stage~0.  
The foundational idea follows~\cite{KKR2006}, which presents a systematic approach to constrained Hamiltonians.  
The Projection Lemma in \cite{KKR2006} shows that when a large penalty is assigned to a subspace,  
the lowest eigenvalues of the full Hamiltonian are close to those of the restricted Hamiltonian.  
To improve this approximation,~\cite{KKR2006} further derives an effective Hamiltonian  
by incorporating perturbative corrections from the high-energy subspace.  
We adopt a similar approach: by assigning a large penalty \( \Jzz^{\ms{clique}} \) to states involving edges within a clique,  
we energetically suppress their contribution and restrict the evolution to a reduced subspace.  
We then apply perturbation theory to characterize the resulting effective Hamiltonian.

In the following, we first introduce a decomposition of the Hilbert space based on energy penalties from the problem Hamiltonian.  
We then derive the effective Hamiltonian and argue that the spectral gap remains large throughout Stage~0,  
justifying the adiabatic elimination of high-energy states.

\subsection{Decomposition of the Hilbert Space}
Recall the MIS-Ising problem Hamiltonian defined in Eq.~\ref{eq:problem-Ham}:
  \[
\mb{H}_{\ms{problem}} = \sum_{i \in \ver(G)} (-w_i) \shz{i} 
+  \Jzz^{\ms{clique}} 
\sum_{(i,j) \in \edge(G_{\ms{driver}})} \shz{i} \shz{j} 
+ \Jzz \sum_{(i,j) \in \edge(G) \setminus \edge(G_{\ms{driver}})} \shz{i} \shz{j}.
\]
The vertex set of the graph is \( \ver(G) = V_L \,\cup\, R \),  
where \(V_L\) is the union of all cliques in \(L\) and \(R\) is the set of vertices outside these cliques.  
The corresponding Hilbert space factors as
\(
\mc{V} = \mc{V}_L \otimes \mc{V}_R,
\)
where \(\mc{V}_L\) is the Hilbert space of all vertices in \(V_L\) and  
\(\mc{V}_R\) is that of the vertices in \(R\).

The parameter \( \Jzz^{\ms{clique}} \) penalizes states that include occupied edges within a clique.  
By setting \( \Jzz^{\ms{clique}} \) sufficiently large, the Hilbert space \(\mc{V}_L\) separates into a low-energy subspace  
(spanned by all independent-set states within $L$) and a high-energy subspace (spanned by all dependent-set states within $L$).

We define
\[
\mc{L}_{-} := (\text{low-energy subspace of }\mc{V}_L) \otimes \mc{V}_R, 
\quad
\mc{L}_{+} := (\text{high-energy subspace of }\mc{V}_L) \otimes \mc{V}_R.
\]
Here \(\mc{L}_{+}\) corresponds to high-energy states---each such state containing at least one intra-clique edge incurring the \( \Jzz^{\ms{clique}} \) penalty.

Let \( \Pi_{-} \) and \( \Pi_{+} \) denote the orthogonal projectors onto \( \mc{L}_{-} \) and \( \mc{L}_{+} \), respectively.  
With respect to this decomposition, the system Hamiltonian at the end of Stage~0 can be written in block form:
\[
  \mb{H}_0(1) \;=\;
  \begin{pmatrix}
    \mb{H}^{\ms{low}} & \mb{V} \\
    \mb{V}^\dagger & \mb{H}^{\ms{high}}
  \end{pmatrix},
  \]
  with
\(
{\mb{H}}^{\ms{low}} = \Pi_{-} \mb{H}_0(1) \Pi_{-},  \mb{H}^{\ms{high}}  = \Pi_{+} \mb{H}_0(1) \Pi_{+}.
\)
At the end of Stage~0, the full Hamiltonian takes the form
\[
\mb{H}_0(1) = 
\left(\Gamma_1 \mb{H}_{\ms{X}} + \Jxx \mb{H}_{\ms{XX}} + \mb{H}_{\ms{prob}}^{\ms{low}} \right)
+ \mb{H}_{\ms{prob}}^{\ms{high}},
\]
where
\(
\mb{H}_{\ms{problem}} = \mb{H}_{\ms{prob}}^{\ms{low}} + \mb{H}_{\ms{prob}}^{\ms{high}},
\)
with
\(
\mb{H}_{\ms{prob}}^{\ms{low}} = \Pi_{-} \mb{H}_{\ms{problem}} \Pi_{-}, \mb{H}_{\ms{prob}}^{\ms{high}}  = \Pi_{+} \mb{H}_{\ms{problem}} \Pi_{+}.
\)
By construction, \( \mb{H}_{\ms{prob}}^{\ms{high}}  \) annihilates the low-energy subspace:
\(
\mb{H}_{\ms{prob}}^{\ms{high}}  \Pi_{-} = \Pi_{-} \mb{H}_{\ms{prob}}^{\ms{high}}  = 0.
\)
Thus
$\mb{H}^{\ms{low}} = \Pi_{-} \mb{H}_0(1) \Pi_{-} =  \Pi_{-} (\mb{M}+ \mb{H}_{\ms{prob}}^{\ms{high}})  \Pi_{-} = \Pi_{-} \mb{M} \Pi_{-}$,
where $\mb{M}=\Gamma_1 \mb{H}_{\ms{X}} + \Jxx \mb{H}_{\ms{XX}} + \mb{H}_{\ms{prob}}^{\ms{low}}$.

In the following, we will show that by the end of Stage~0, if
$\Jzz^{\ms{clique}}$ is chosen sufficiently large, the Stage~1
evolution is governed exactly by the restricted Hamiltonian
$\Heff = \mb{H}^{\ms{low}} =\Pi_{-} \mb{M} \Pi_{-}$.

\subsection{Effective Hamiltonian at the End of Stage~0}

\begin{definition}
Let \( \lambda_{\ms{cutoff}} \in \mathbb{R} \) be a cutoff energy, and let \( \epsilon > 0 \).  
We say that \( \mb{H}_{\ms{eff}} \) is an \( \epsilon \)-close effective low-energy Hamiltonian of \( \mb{H} \)  
if for every eigenstate \( \mb{H} \ket{\psi} = \lambda \ket{\psi} \) with \( \lambda < \lambda_{\ms{cutoff}} \),  
there exists a state \( \ket{\psi'} \) such that \( \mb{H}_{\ms{eff}} \ket{\psi'} = \lambda' \ket{\psi'} \), and:
\begin{enumerate}
    \item \( |\lambda - \lambda'| = \mathcal{O}(\epsilon) \) \quad (energy closeness),
    \item \( \| \ket{\psi'} - \Pi_{-} \ket{\psi} \| = \mathcal{O}(\epsilon) \) \quad (state closeness).
\end{enumerate}
\label{def:epsilon-effective}
\end{definition}

\noindent
Intuitively, \( \mb{H}_{\ms{eff}} \) provides an approximation of \( \mb{H} \)  
that accurately captures the system's behavior within the low-energy subspace \( \mc{L}_{-} \).  
This ensures that both the spectrum and eigenstates of \( \mb{H} \) below \( \lambda_{\ms{cutoff}} \)  
are preserved up to an error of order \( \mathcal{O}(\epsilon) \).

The following Theorem~\ref{thm2} can be derived as a consequence of Theorem~6.2 in~\cite{KKR2006}.  
However, we provide a direct derivation based on projection properties,  
and explicitly establish the correspondence between eigenstates.

\begin{theorem}
\label{thm2}
Let $\mb{H} = \mb{M} + \mb{H}_{\ms{prob}}^{\ms{high}}$, where 
$\Pi_{-}\mb{H}_{\ms{prob}}^{\ms{high}} = \mb{H}_{\ms{prob}}^{\ms{high}}\Pi_{-} = 0$
and all eigenvalues of $\mb{H}_{\ms{prob}}^{\ms{high}}$ are at least 
$\lambda_{\ms{high}} > \lambda_{\ms{cutoff}}$.
Furthermore, to ensure the effective Hamiltonian approximation remains within \( \mathcal{O}(\epsilon) \),  
we impose the stronger condition:
\(
\lambda_{\ms{high}} \geq \tfrac{1}{\epsilon} \| \mb{M} \|^2.
\)
Then the projected Hamiltonian
\(
\mb{H}_{\ms{eff}} := \Pi_{-} \mb{M} \Pi_{-}
\)
is an \( \epsilon \)-close effective low-energy Hamiltonian of \( \mb{H} \), as defined in Definition~\ref{def:epsilon-effective}.
\end{theorem}

\begin{proof}
Let \( \mb{H} \ket{\psi} = \lambda \ket{\psi} \) be an eigenstate of \( \mb{H} \) with \( \lambda < \lambda_{\ms{cutoff}} \).  
We express \( \mb{H} \) in the block basis defined by \( \Pi_{-} \) and \( \Pi_{+} \):
\[
\mb{H} = 
\begin{pmatrix}
\mb{H}_{++} & \mb{H}_{+-} \\
\mb{H}_{-+} & \mb{H}_{--}
\end{pmatrix},
\quad \text{where } \mb{H}_{ab} := \Pi_a \mb{H} \Pi_b \text{ for } a,b \in \{+, -\}.
\]
Similarly, we decompose \( \mb{M} \) as
\[
\mb{M} = 
\begin{pmatrix}
\mb{M}_{++} & \mb{M}_{+-} \\
\mb{M}_{-+} & \mb{M}_{--}
\end{pmatrix}.
\]

Let \( \ket{\psi_{+}} := \Pi_{+} \ket{\psi} \) and \( \ket{\psi_{-}} := \Pi_{-} \ket{\psi} \).  
The eigenvalue equation becomes:
\[
\begin{pmatrix}
\mb{H}_{++} & \mb{H}_{+-} \\
\mb{H}_{-+} & \mb{H}_{--}
\end{pmatrix}
\begin{pmatrix}
\ket{\psi_{+}} \\
\ket{\psi_{-}}
\end{pmatrix}
= \lambda
\begin{pmatrix}
\ket{\psi_{+}} \\
\ket{\psi_{-}}
\end{pmatrix}.
\]

Solving for \( \ket{\psi_{+}} \), we obtain
\(
\ket{\psi_{+}} = \mb{G}(\lambda) \mb{H}_{+-} \ket{\psi_{-}},
\)
where the Green's function is defined as
\(
\mb{G}(\lambda) := (\lambda I - \mb{H}_{++})^{-1},
\)
which corresponds to \( G_{+}(z) \) in~\cite{KKR2006}.

Substituting back, we obtain an effective energy-dependent equation for \( \ket{\psi_{-}} \):
\(
\mb{H}_{\ms{eff}}^{\ms{exact}}(\lambda) \ket{\psi_{-}} = \lambda \ket{\psi_{-}},
\)
where the exact effective Hamiltonian is given by
\[
\mb{H}_{\ms{eff}}^{\ms{exact}}(\lambda) := \mb{H}_{--} + \mb{H}_{-+} \mb{G}(\lambda) \mb{H}_{+-}
= \mb{M}_{--} + \mb{M}_{-+} \mb{G}(\lambda) \mb{M}_{+-}.
\]

To approximate this, we expand \( \mb{G}(\lambda) \) using a resolvent expansion:
\[
\mb{G}(\lambda) = \left( \lambda I - \mb{H}_{\ms{prob}}^{\ms{high}}  - \mb{M}_{++} \right)^{-1} 
= \mb{D}^{-1} \left( I - \mb{M}_{++} \mb{D}^{-1} \right)^{-1},
\]
where \( \mb{D} := \lambda I - \mb{H}_{\ms{prob}}^{\ms{high}}  \), and hence \( \| \mb{D}^{-1} \| \leq 1 / (\lambda_{\ms{high}} - \lambda_{\ms{cutoff}}) \).

Expanding in powers of \( \mb{M}_{++} \mb{D}^{-1} \), we obtain:
\begin{align*}
\mb{H}_{\ms{eff}}^{\ms{exact}}(\lambda)
&= \mb{M}_{--} + \mb{M}_{-+} \mb{D}^{-1} \mb{M}_{+-} + \mb{M}_{-+} \mb{D}^{-1} \mb{M}_{++} \mb{D}^{-1} \mb{M}_{+-} + \ldots \\
&= \mb{M}_{--} + \mb{M}_{-+} \mb{D}^{-1} \mb{M}_{+-} + \mathcal{O}\left( \tfrac{\| \mb{M} \|^3}{(\lambda_{\ms{high}} - \lambda_{\ms{cutoff}})^2} \right).
\end{align*}

Under the condition \( \lambda_{\ms{high}} \ge \tfrac{1}{\epsilon} \| \mb{M} \|^2 \),  
the leading correction is bounded by \( \mathcal{O}(\epsilon) \), and we may set
\(
\mb{H}_{\ms{eff}} := \mb{M}_{--}
\)
as an \( \epsilon \)-close effective low-energy Hamiltonian of \( \mb{H} \).

Finally, letting $\ket{\psi'} := \ket{\psi_{-}}$, we see that both the energy 
and state closeness conditions in Definition~\ref{def:epsilon-effective} 
are satisfied up to $\mathcal{O}(\epsilon)$, completing the proof.
%% Finally, letting \( \ket{\psi'} := \ket{\psi_{-}} \), we have:
%% \begin{enumerate}
%%     \item \( |\lambda - \lambda'| = \mathcal{O}(\epsilon) \), \quad where \( \lambda' := \bra{\psi'} \mb{H}_{\ms{eff}} \ket{\psi'} \),
%%     \item \( \| \ket{\psi'} - \Pi_{-} \ket{\psi} \| = \mathcal{O}(\epsilon) \),
%% \end{enumerate}
%% completing the proof.
\end{proof}

\begin{remark}
  The structure of this proof closely resembles methods commonly associated
  with the Schrieffer--Wolff transformation,  
  a technique originally developed to perturbatively eliminate high-energy subspaces
  via a unitary block-diagonalization.  
  Although the present argument is based on resolvent expansion and projection---rather than an explicit unitary transformation---the resulting effective Hamiltonian is often referred to in the literature, particularly in physics,  
  as a Schrieffer--Wolff-type reduction; see, e.g.,~\cite{SWT-2011}.
\end{remark}

Applying Theorem~\ref{thm2}, we have
the effective low-energy Hamiltonian 
\begin{equation}
\mb{H}_0(1)^{\ms{eff}} = \Pi_{-} \mb{M} \Pi_{-}
= \Pi_{-} \left( \Gamma_1 \mb{H}_{\ms{X}} \right) \Pi_{-} + \Jxx \mb{H}_{\ms{XX}} + \mb{H}_{\ms{prob}}^{\ms{low}}.
\label{eq:stage0-Heff}
  \end{equation}
by setting
$\Jzz^{\ms{clique}} = \Omega(\| \mb{M} \|^2 / \epsilon).$
The ground state of $\mb{H}_0(1)^{\ms{eff}}$ is approximately the ground state of the transverse-field term $\Pi_{-} ( \Gamma_1 \mb{H}_{\ms{X}} ) \Pi_{-}$ as $\left\| \Jxx \mb{H}_{\ms{XX}} + \mb{H}_{\ms{prob}}^{\ms{low}} \right\| \ll \left\| \Pi_{-} ( \Gamma_1 \mb{H}_{\ms{X}} ) \Pi_{-} \right\|$ when $\Gamma_1 = O\left( N \Jzz + 1 + \tfrac{(n_c - 1)}{4} \Jxx \right)$.

\subsection{Spectral Gap Behavior in Stage~0}
The initial Hamiltonian \( \Gamma_0 \mb{H}_{\ms{X}} \) (with \( \Gamma_0 = 2\Gamma_1 \)) has a spectral gap of \( 2\Gamma_0 \).  
We argue that this spectral gap decreases smoothly from \( 2\Gamma_0 \) to \( 2\Gamma_1 \), without encountering any small-gap regime.

The initial ground state is the uniform superposition over all computational basis states,  
corresponding to the ground state of \( \mb{H}_{\ms{X}} \).  
As the system evolves, the problem Hamiltonian \( \mb{H}_{\ms{problem}}(t) \) and the XX-driver \( \mb{H}_{\ms{XX}} \)  
are gradually turned on, while the transverse field \( \mt{x}(t) \) is smoothly reduced.

Throughout Stage~0, the transverse field remains large, and the suppression of high-energy states in \( \mc{L}_+ \)  
helps ensure that the spectral gap decreases in a controlled manner, without abrupt transitions.  
This can be understood as follows:
\begin{enumerate}
    \item For small \( t \), the transverse field term \( \mt{x}(t) \mb{H}_{\ms{X}} \) dominates,  
    maintaining a large gap between the ground and first excited states.
    \item As \( \mb{H}_{\ms{XX}} \) gradually turns on, it remains small relative to the transverse field.
    \item At intermediate \( t \), the increasing energy penalty \( t \cdot \Jzz^{\ms{clique}} \)  
    progressively separates the \( \mc{L}_+ \) subspace from the low-energy dynamics.
    \item For large \( t \), the system is effectively projected into \( \mc{L}_- \),  
    and the spectral gap approaches \( 2\Gamma_1 \), determined by the final transverse field strength.
\end{enumerate}

Thus, the gap transitions from \( 2\Gamma_0 \) to \( 2\Gamma_1 \) smoothly,  
allowing the adiabatic evolution in Stage~0 to proceed efficiently without risk of gap closure.

%% %\subsection{Analysis of Main 2-stage}
\section{Main Analysis of the Two-Stage Dynamics on Bipartite Structures}
\label{sec:stage-main}

We now analyze the main two-stage dynamics of the system Hamiltonian for the two bipartite substructures introduced in Section~\ref{sec:GIC}:  
the \emph{disjoint-structure graph} \( \Gdis \) and the \emph{shared-structure graph} \( \Gshare \), as shown in Figure~\ref{fig:dis-share-graph}.

To simplify notation, we denote the effective Hamiltonian at the start of Stage~1 by
$\Heff(0) := \mb{H}_0(1)^{\ms{eff}}$.
From this point onward, the system evolves under the time-dependent Hamiltonian \( \Heff(t) \), which incorporates the Stage~1 and~2 annealing schedules.  
\emph{When there is no ambiguity, we refer to \( \Heff \) as the full Hamiltonian for Stages~1 and~2,} and all subsequent block decompositions---into same-sign and opposite-sign components---are understood to apply to \( \Heff \).

This section is organized as follows:
\begin{enumerate}
  \item In Section~\ref{sec:sub1}, we derive the block decomposition of the full Hamiltonian into same-sign and opposite-sign components  
  through local angular momentum transformations applied to each clique.
  
  \item In Section~\ref{sec:sub2}, we describe two ways of structurally decomposing the same-sign block Hamiltonian during Stage~1,  
  corresponding to the \( L \)-inner and \( R \)-inner block structures.

  \item In Section~\ref{sec:sub4}, we analyze the two-stage evolution and derive feasibility bounds on \( \Jxx \).  
  We show that, within the feasible regime, the system evolves to the global minimum without encountering an anti-crossing.
  
  \item In Section~\ref{sec:V3}, we illustrate the quantum interference effects using a minimal three-vertex conceptual model (V3),  
  which may also serve as a testbed for experimental validation.
\end{enumerate}

\subsection{Block Decomposition of the Full Hamiltonian: Same-Sign vs.~Opposite-Sign Blocks}
\label{sec:sub1}

We begin by expressing the Hamiltonian \( \Heff \)  
in the angular momentum basis induced by the \XX-driver graph.  
This basis is constructed locally within each clique using projection and transformation operators, as shown in Section~\ref{sec:single-clique}.  
The resulting global basis, denoted \( \Ba \), yields a block-diagonal decomposition of the full Hamiltonian  
into one \emph{same-sign block} and multiple \emph{opposite-sign blocks}.

The angular momentum basis \( \Ba \) is constructed by applying a local transformation within each clique \( C_i \),  
which maps the computational basis on \( C_i \) to an angular momentum basis \( {\Bc'}_i \).  
Because the total Hilbert space is a tensor product,  
these local transformations can be applied independently to each clique,  
with the identity acting on the remainder of the system.  
%This application is made possible by the uniform external interaction between cliques and the rest of the graph.  
These local transformations are then combined to define a global transformation on the full system,  
preserving the tensor-product structure and yielding a basis constructed from the local angular momentum decompositions of the cliques defined by the \(\XX\)-driver graph.

The resulting block decomposition emerges hierarchically:  
each clique yields one same-sign sector (an effective spin-\( \tfrac{1}{2} \)) and several opposite-sign sectors (spin-0).  
At the next level, the set of cliques in \( L \)---those forming the \LM{}---induces a bare subsystem with  
the \emph{same-sign sector} \( \mc{C}^{\ms{bare}} \), the \emph{intermediate opposite-sign sectors} \( \mc{W}^{\ms{bare}} \),
and the \emph{all-spin-zero (AS0) opposite-sign sector} \( \mc{Q}^{\ms{bare}} \).
Finally, the full same-sign sector \( \mc{C} \) of the system is defined as
\[
\mc{C} = \mc{C}^{\ms{bare}} \otimes \left( \mathbb{C}^2 \right)^{\otimes m_r},
\]
where \( m_r = |R| \).  
Similarly, the full opposite-sign sectors \( \mc{W} \) and \( \mc{Q} \) are defined by
\[
\mc{W} = \mc{W}^{\ms{bare}} \otimes \left( \mathbb{C}^2 \right)^{\otimes m_r}, \quad
\mc{Q} = \mc{Q}^{\ms{bare}} \otimes \left( \mathbb{C}^2 \right)^{\otimes m_r}.
\]
Thus, $\mc{Q}$ is also referred to as the (AS0) opposite-sign sector of the full system.
The full Hamiltonian then decomposes into the same-sign block \( \mb{H}_{\mc{C}} \) and multiple opposite-sign blocks \( \mb{H}_{\mc{Q}}, \mb{H}_{\mc{W}} \),  
which are either decoupled (in the disjoint case) or coupled (in the shared case).

\paragraph{Notation for Operators in \( \Ba \).}
A bar denotes an operator expressed in the angular momentum basis \( \Ba \), e.g., 
\(\bar{\mb{H}_{\ms{X}}}\) for the transverse-field term, 
\(\bar{\mb{H}_{\ms{P}}}\) for the problem Hamiltonian, 
and \(\bar{\mb{H}_{\XX}}\) for the \(\XX\)-driver term.

The effective Hamiltonian in \( \Ba \) is given by:
\[
\bar{\Heff} := \mt{x} \bar{\mb{H}}_{\ms{X}} + \mt{jxx} \bar{\mb{H}_{\XX}} + \bar{\mb{H}_{\ms{P}}}.
\]

%% \paragraph{Block Labels.}
%% We adopt the following notation for the key angular momentum sectors:
%% \begin{itemize}
%%   \item \( \mc{C} \): the same-sign block, where all cliques are in the spin-\( \tfrac{1}{2} \) sector;
%% %  \item \( \mc{D} \): all opposite-sign blocks;
%%   \item \( \mc{Q} \): the all-spin-zero block, where all cliques are in the spin-0 sector;
%% %  \item \( \mc{W} \): the intermediate opposite-sign blocks, containing mixed configurations of spin-\( \tfrac{1}{2} \) and spin-0 sectors.
%% \end{itemize}

First, we state the resulting block structure in Theorem~\ref{thm:transformed-hamiltonian}, illustrated in Figure~\ref{fig:block-diag}.
We prove the theorem by deriving \( \bar{\Heff} \) for the disjoint-structure graph \( \Gdis \) in Section~\ref{sec:dis}.  
We then describe the modifications required for the shared-structure graph \( \Gshare \) in Section~\ref{sec:share},  
highlighting the critical differences introduced by shared connectivity.

\begin{mdframed}
\begin{theorem}[Block Structure of the Transformed Hamiltonian]
\label{thm:transformed-hamiltonian}
Under the angular momentum basis transformation,  
the effective Hamiltonian \( \bar{\Heff} \) becomes block-diagonal in the basis \( \Ba \):
\[
\bar{\Heff} = \mb{H}_{\mc{C}} \oplus \cdots \oplus \mb{H}_{\mc{Q}} + \Hinter
\]
where
$\Hinter \equiv 0$ in the disjoint case, while in the shared case
\[
\Hinter = \Jzz \sum_{(i,j) \in L \times R} \tfrac{\sqrt{n_i - 1}}{n_i} \mb{T}^{\ms{cq}}_i \shz{j},
\]
where the matrix $\mb{T}^{\ms{cq}}$ (as in Eq.~\eqref{eq:Tcq}) is a $3 \times 3$ off-diagonal operator 
that mixes same-sign and opposite-sign components.
The coupling strength in \( \Hinter \) depends on \( \Jzz \) and clique size \( n_i \),  
but not on the transverse field \( \mt{x} \).

In both the disjoint and shared cases, the Hamiltonians \( \mb{H}_{\mc{C}} \) and \( \mb{H}_{\mc{Q}} \) admit unified expressions:
\begin{align*}
\mb{H}_{\mc{C}} &=
\sum_{i \in L} \mb{B}_i \!\left( \mt{\weff_i},\, \sqrt{n_i}\,\mt{x} \right)
+ \sum_{j \in R} \mb{B}_j \!\left( w_j,\, \mt{x} \right)
+ \Jzz \sum_{(i,j) \in L \times R} f_i^{\ms{C}} \shz{i} \shz{j}, \\
\mb{H}_{\mc{Q}} &=
-\sum_{i \in L} \left( w_i + \tfrac{1}{4}\mt{jxx} \right)
+ \sum_{j \in R} \mb{B}_j \!\left( -w_j + \Jzz \sum_{i \in L} f_i^{\ms{Q}},\, \mt{x} \right),
\end{align*}
with effective coefficients
\[
\mt{\weff_i} = w_i - \tfrac{n_i - 1}{4}\mt{jxx}, \quad
f_i^{\ms{C}} = 
\begin{cases}
1 & \text{for } \Gdis, \\
\tfrac{n_i - 1}{n_i} & \text{for } \Gshare,
\end{cases} \qquad
f_i^{\ms{Q}} = 
\begin{cases}
1 & \text{for } \Gdis, \\
\tfrac{1}{n_i} & \text{for } \Gshare.
\end{cases}
\]

The initial ground state resides in $\mb{H}_{\mc{C}}$ at the beginning of Stage~1. 
\end{theorem}
\end{mdframed}

\begin{figure}[!htbp]
  \centering
  \includegraphics[width=0.62\linewidth]{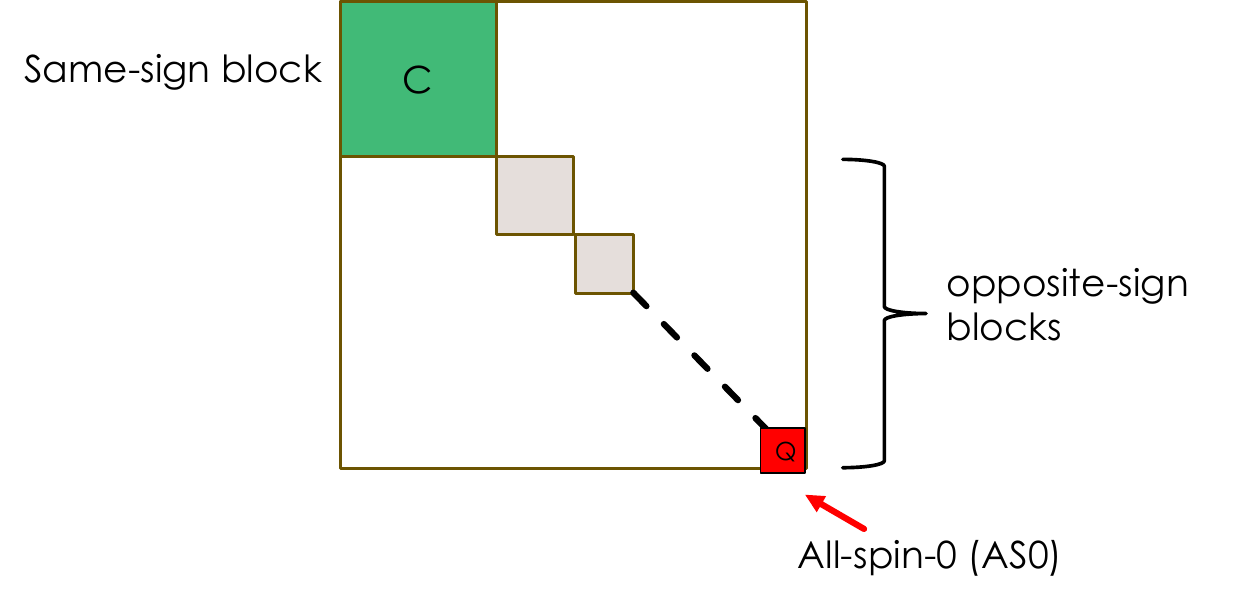}
  \caption{Block-diagonal structure of the transformed Hamiltonian $\bar{\Heff}$.
The same-sign block $\mb{H}_{\mc{C}}$ (green) and the AS0 block $\mb{H}_{\mc{Q}}$ (red)
are the two dominant blocks in the low-energy analysis. Intermediate
opposite-sign blocks are shown in neutral color. In the disjoint case
the blocks are decoupled, while in the shared case they are coupled via
the inter-block term $\Hinter$.}
  \label{fig:block-diag}
\end{figure}

The following corollary follows directly from Theorem~\ref{thm:transformed-hamiltonian}:
\begin{mdframed}
\begin{corollary}[Bare Subsystem Decomposition of \( \mb{H}_{\mc{C}} \)]
\label{cor:bare-decomposition}
The same-sign block Hamiltonian \( \mb{H}_{\mc{C}} \) decomposes
into two bare subsystems coupled via $\ZZ$-couplings:
\begin{align}
\mb{H}_{\mc{C}} = \mb{H}_L^{\ms{bare}} \otimes \mb{I}_R 
+ \mb{I}_L \otimes \mb{H}_R^{\ms{bare}} 
+ \mb{H}_{LR},
\label{eq:same-sign-bare}
\end{align}
where
\[
\mb{H}_L^{\ms{bare}} = \sum_{i \in L} \mb{B}_i^L(\mt{\weff},\, \sqrt{n_c}\,\mt{x}), 
\quad
\mb{H}_R^{\ms{bare}} = \sum_{j \in R} \mb{B}_j^R(w,\,\mt{x}),
\quad
\mb{H}_{LR} = \Jzz \sum_{(i,j) \in L \times R} f_i^{\ms{C}} \shz{i} \shz{j},
\]
with the effective weight \( \mt{\weff} = w - \tfrac{n_c - 1}{4}\,\mt{jxx} \).  
The superscripts in \( \mb{B}_i^L \) and \( \mb{B}_j^R \) indicate the subspace of the tensor product on which each two-level Hamiltonian acts, and are omitted when the context is clear.
\end{corollary}
\end{mdframed}

\subsubsection{The Disjoint-Structure Graph: Block-Diagonal Structure via Clique Contraction}
\label{sec:dis}

We now turn to the disjoint-structure graph \( \Gdis \),  
illustrated in Figure~\ref{fig:dis-share-graph} of Section~\ref{sec:GIC},  
which serves as the foundation for the block decomposition stated in Theorem~\ref{thm:transformed-hamiltonian}.  
The key idea is that the local angular momentum transformation within each clique  
induces an effective contraction of the graph structure,  
allowing us to derive the block Hamiltonians explicitly.

Recall that within the low-energy subspace of a clique of size \( n_c \),  
the transformed spin operators \( \bar{\Sop{\sZ}}, \bar{\Sop{\X}}, \bar{\Sop{\XX}} \),  
given in Equation~\eqref{eq:transformed-operators} of Theorem~\ref{thm:Sop},  
are block-diagonal in the angular momentum basis.

This local transformation naturally induces a \emph{clique contraction}:  
each clique, whose vertices interact identically with the rest of the graph,  
is collapsed into a single \emph{super-vertex}, carrying the effective operators described above.  
Under this contraction, the graph \( \Gdis \) becomes a bipartite graph \( \Gcnt \),  
in which each clique is replaced by a super-vertex,  
as illustrated in Figure~\ref{fig:contract}.  
Thus, \( L \) consists of the set of super-vertices.

\begin{figure}[h]
  \centering
  \includegraphics[width=0.7\textwidth]{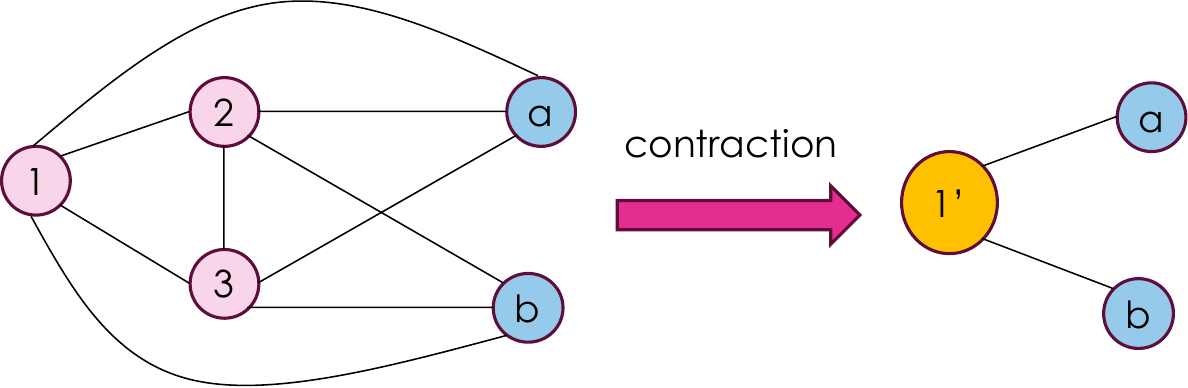} 
  \caption{Clique contraction:  
  A clique (\(\{1,2,3\}\)) whose vertices share identical external connectivity  
  is contracted into a single \emph{super-vertex} (\(1'\)).}
  \label{fig:contract}
\end{figure}

\paragraph{The Transformed Transverse-Field Term.}
For the transverse-field Hamiltonian \( \mb{H}_{\ms{X}} \),  
the transformed operator in \( \Ba \) takes the form:
\begin{align*}
\bar{\mb{H}}_{\ms{X}} &=
\sum_{i \in L} \left( \tfrac{\sqrt{n_i}}{2} \sigma^x_i \oplus 0 \oplus \cdots \oplus 0 \right)
+ \tfrac{1}{2} \sum_{i \in R} \sigma^x_i \\[4pt]
&=
\left( \sum_{i \in L} \tfrac{\sqrt{n_i}}{2} \sigma^x_i + \tfrac{1}{2} \sum_{i \in R} \sigma^x_i \right)
\oplus \cdots \oplus \left( \sum_{i \in L'} \tfrac{\sqrt{n_i}}{2} \sigma^x_i + \tfrac{1}{2} \sum_{i \in R} \sigma^x_i \right)
\oplus \cdots \oplus
\left( \sum_{i \in L} 0 + \tfrac{1}{2} \sum_{i \in R} \sigma^x_i \right),
\end{align*}
where the direct sum runs over the block decomposition induced by \( \Ba \), and the contributions are grouped as follows:
\begin{itemize}
  \item \( \bar{\mb{H}}_{\ms{X}}^{\mc{C}} := \sum_{i \in L} \tfrac{\sqrt{n_i}}{2} \sigma^x_i + \tfrac{1}{2} \sum_{i \in R} \sigma^x_i \)  
    corresponds to the same-sign block \( \mc{C} \).

  \item \( \bar{\mb{H}}_{\ms{X}}^{\mc{Q}} := \tfrac{1}{2} \sum_{i \in R} \sigma^x_i \)  
    corresponds to the AS0 opposite-sign block \( \mc{Q} \), which receives no transverse contribution from \( L \).

  \item \( \bar{\mb{H}}_{\ms{X}}^{\mc{W}} := \sum_{i \in L'} \tfrac{\sqrt{n_i}}{2} \sigma^x_i + \tfrac{1}{2} \sum_{i \in R} \sigma^x_i \)  
  corresponds to the intermediate opposite-sign blocks \( \mc{W} \),  
  where some cliques are in spin-0 sectors and others in spin-\( \tfrac{1}{2} \), denoted by the subset \( L' \subset L \).
\end{itemize}

\paragraph{The Transformed Problem Hamiltonian.}
The Ising interaction term,
\(
\sum_{(i,j) \in \edge(G)} \shz{i} \shz{j},
\)
is mapped under clique contraction to:
\(
\sum_{(i',j') \in \edge(G_{\ms{contract}})} \left( \shz{i'} \oplus 1 \oplus \cdots \oplus 1 \right) \shz{j'},
\)
where \( \edge(G_{\ms{contract}}) = L \times R \).
This yields a block structure:
\begin{align*}
\sum_{(i',j') \in \edge(G_{\ms{contract}})} \shz{i'} \shz{j'} 
\oplus \cdots \oplus 
m_l \sum_{j \in R} \shz{j},
\end{align*}
which highlights the distinct contributions of the same-sign and opposite-sign sectors.

The transformed problem Hamiltonian becomes:
\begin{align*}
\bar{\mb{H}}_{\ms{P}} &=
\sum_{i \in L} (-w_i) \left( \shz{i} \oplus 1 \oplus \cdots \oplus 1 \right)
+ \sum_{i \in R} (-w_i) \shz{i}
%\\[5pt] &\quad
+ \Jzz \left(
\sum_{(i,j) \in \edge(G_{\ms{contract}})} \shz{i} \shz{j}
\oplus \cdots \oplus m_l \sum_{j \in R} \shz{j}
\right) \\[5pt]
&=
\left( \sum_{i \in L \cup R} (-w_i) \shz{i}
+ \Jzz \sum_{(i,j) \in \edge(G_{\ms{contract}})} \shz{i} \shz{j} \right)
\oplus \cdots \oplus
\left( \sum_{i \in L} (-w_i) + \sum_{i \in R} (-w_i + m_l \Jzz) \shz{i} \right).
\end{align*}

\noindent
Thus, the same-sign sector behaves as a standard Ising Hamiltonian on the contracted graph \( G_{\ms{contract}} \),  
while each opposite-sign block contributes an energy shift and modifies only the \( R \)-part of the spectrum.

\paragraph{The Transformed \( \XX \)-Driver Term.}
The \( \XX \)-driver modifies the diagonal energies of super-vertices:
\begin{align}
\bar{\mb{H}_{\XX}} =
\sum_{i \in L} \left( \tfrac{n_i - 1}{4} \right) \shz{i}
\oplus \left( -\tfrac{1}{4} \right) \oplus \cdots \oplus \left( -\tfrac{1}{4} \right).
\end{align}

\noindent
This shifts effective vertex-weights as follows:
\begin{itemize}
  \item In the same-sign block:  
  \[
  -w_i \longmapsto -\left( w_i - \tfrac{n_i - 1}{4} \mt{jxx} \right) \quad \text{(energy lifted)};
  \]
  \item In the AS0 opposite-sign block:  
  \[
  -w_i \longmapsto -\left( w_i + \tfrac{1}{4} \mt{jxx}  \right) \quad \text{(energy lowered)}.
  \]
\end{itemize}

\begin{remark}
In the disjoint-structure case, the opposite-sign blocks are, in principle, decoupled from the same-sign block and can be ignored.  
In contrast, in the shared-structure case, weak to moderate couplings emerge between the same-sign and opposite-sign sectors,  
necessitating their inclusion in the analysis.

Because of the block ordering in Section~\ref{sec:block-ordering}, the energy levels of the intermediate opposite-sign blocks are
sandwiched between those of the same-sign block and the AS0 block.
Therefore, for the purpose of analyzing the low-energy spectrum and the dynamical behavior near critical anti-crossings,  
it is sufficient to retain only the same-sign block and the AS0 block.  
All other opposite-sign blocks can be safely ignored.
\end{remark}

\paragraph{Same-Sign Block Hamiltonian.}  
We denote by \( \mb{H}_{\mc{C}}^{\ms{dis}} \) the effective Hamiltonian restricted to the same-sign sector \( \mc{C} \) for the disjoint-structure graph case:  
%i.e., the upper-left block of the transformed Hamiltonian \( \bar{\Heff} \):
\begin{align*}
  \mb{H}_{\mc{C}}^{\ms{dis}} &=
  - \mt{x} \bar{\mb{H}}_{\ms{X}}^{\ms{C}} 
  + \bar{\mb{H}}_{\ms{Z}}^{\ms{C}}(\mt{jxx}) 
  + \Jzz \sum_{(i,j) \in \edge(G_{\ms{contract}})} \shz{i} \shz{j},
\end{align*}
where
%\begin{align*}
  \[
  \bar{\mb{H}}_{\ms{Z}}^{\ms{C}}(\mt{jxx}) :=
  \sum_{i \in L} \left( -w_i + \tfrac{n_c - 1}{4} \mt{jxx} \right) \shz{i}
  + \sum_{j \in R} (-w_j) \shz{j}
  \]
%\end{align*}
collects all diagonal terms linear in \( \shz{i} \), including both the bare vertex weights and the \XX-induced energy shifts.

\paragraph{AS0 Block Hamiltonian.}  
We denote by \( \mb{H}_{\mc{Q}}^{\ms{dis}} \) the effective Hamiltonian restricted to the AS0 sector \( \mc{Q} \) for the disjoint-structure graph case:   
\begin{align*}
  \mb{H}_{\mc{Q}}^{\ms{dis}} &=
  - \mt{x} \bar{\mb{H}}_{\ms{X}}^{\ms{Q}} 
  + \bar{\mb{H}}_{\ms{Z}}^{\ms{Q}}(\mt{jxx}),
\end{align*}
where
\[
\bar{\mb{H}}_{\ms{X}}^{\ms{Q}} = \tfrac12 \sum_{j \in R} \sigma^x_j, 
\quad
\bar{\mb{H}}_{\ms{Z}}^{\ms{Q}}(\mt{jxx}) := \sum_{i \in L} \!\left( -w_i - \tfrac14 \mt{jxx} \right)
  + \sum_{j \in R} \!\left( -w_j + m_l \Jzz \right) \shz{j}.
\]

\noindent
The diagonal term \( \bar{\mb{H}}_{\ms{Z}}^{\ms{Q}} \) reflects a constant energy shift from the \XX-driver applied to \( L \),  
and a coupling-induced energy shift on the \( R \)-vertices, scaled by the number of cliques \( m_l \).  
Note that \( \mb{H}_{\mc{Q}}^{\ms{dis}} \) contains no coupling,  
as the spin-0 components contribute only a scalar constant to the interaction term. \( \mb{H}_{\mc{Q}}^{\ms{dis}} \) is analytically solvable.

\begin{remark}
The differences between the same-sign block $\mb{H}_{\mc{C}}^{\ms{dis}}$  and the AS0 block \( \mb{H}_{\mc{Q}}^{\ms{dis}} \) are:
  \begin{itemize}
    \item The effective transverse-field contribution from \( L \) vanishes:  
    \( \sqrt{n_i} \sigma^x \mapsto 0 \).
    
    \item The operator \( \shz{i} \) is replaced by the scalar \( 1 \)  
    (i.e., \( \shz{i} \mapsto 1 \)), contributing only constant energy shifts.

    \item The \XX-driver term, which raises the energy of each vertex \( i \in L \)  
    by \( \tfrac{n_i - 1}{4} \Jxx \) in the same-sign block,  
    instead \emph{lowers} the energy in the all-spin-zero block by \( \tfrac{1}{4} \Jxx \) per clique.
  \end{itemize}

\end{remark}

This completes the block-level analysis of the disjoint-structure case.  
%in which the same-sign and opposite-sign blocks evolve independently throughout the annealing process.  
We now turn to the shared-structure case, where weak inter-block couplings emerge and modify this picture.

\subsubsection{The Shared-Structure Graph: Modification to the Disjoint Case}
\label{sec:share}

The \emph{shared-structure graph} \( \Gshare \), shown in Figure~\ref{fig:dis-share-graph},  
differs from the disjoint-structure graph \( \Gdis \) primarily in its inter-subsystem \(\ZZ\)-couplings.

In the transformed Hamiltonian, the only modification appears in the \(\ZZ\)-interaction terms.  
Unlike the fully connected case in \( \Gdis \), only \( n_c - 1 \) out of \( n_c \) vertices in each clique of size \( n_c \)  
are coupled to each vertex \( j \in R \).  
As shown in Section~\ref{sec:partial-coupling}, this partial adjacency induces a modified coupling structure in the angular momentum basis.

%% The interaction between an external vertex and the internal basis of a clique of size \( n_c \) is represented by the matrix $\mb{T}$ in Eq.~\eqref{eq:opT}, which decomposes as
%% $\mb{T} = \left( \tfrac{n_c - 1}{n_c} \shz{} \right) \oplus \left( \tfrac{1}{n_c} \right) + \mb{T}^{\ms{cq}}$ with 
%% \begin{equation}
%%  \mb{T}^{\ms{cq}} =
%% \begin{bmatrix}
%% 0 & 0 & -\tfrac{\sqrt{n_c - 1}}{n_c} \\[4pt]
%% 0 & 0 & 0 \\[4pt]
%% -\tfrac{\sqrt{n_c - 1}}{n_c} & 0 & 0
%% \end{bmatrix}
%% \label{eq:Tcq}
%% \end{equation}
%%  represents mixing between same-sign and opposite-sign components.
%% The coupling magnitude is
%% \(
%% \tfrac{\sqrt{n_c - 1}}{n_c} \;\approx\; \tfrac{1}{\sqrt{n_c + 1}},
%% \)
%% which decreases with increasing clique size \( n_c \).

The interaction between an external vertex and the internal basis of a clique of size \( n_c \) is represented by the matrix $\mb{T}$ in Eq.~\eqref{eq:opT}, which decomposes as
\(
\mb{T} = \left( \tfrac{n_c - 1}{n_c} \shz{} \right) \oplus \left( \tfrac{1}{n_c} \right) + \tfrac{\sqrt{n_c - 1}}{n_c} \mb{T}^{\ms{cq}},
\)
with
\begin{equation}
 \mb{T}^{\ms{cq}} =
\begin{bmatrix}
0 & 0 & -1 \\[0pt]
0 & 0 & 0 \\[0pt]
-1 & 0 & 0
\end{bmatrix},
\label{eq:Tcq}
\end{equation}
which mixes same-sign and opposite-sign components.
The coupling magnitude is
\(
\tfrac{\sqrt{n_c - 1}}{n_c} \;\approx\; \tfrac{1}{\sqrt{n_c + 1}},
\)
decreasing with clique size \( n_c \).

The resulting decomposition and sector-wise coupling structure are illustrated in Figure~\ref{fig:Couple-CQ}.
\begin{figure}[!htbp]
  \centering
  \includegraphics[width=0.7\textwidth]{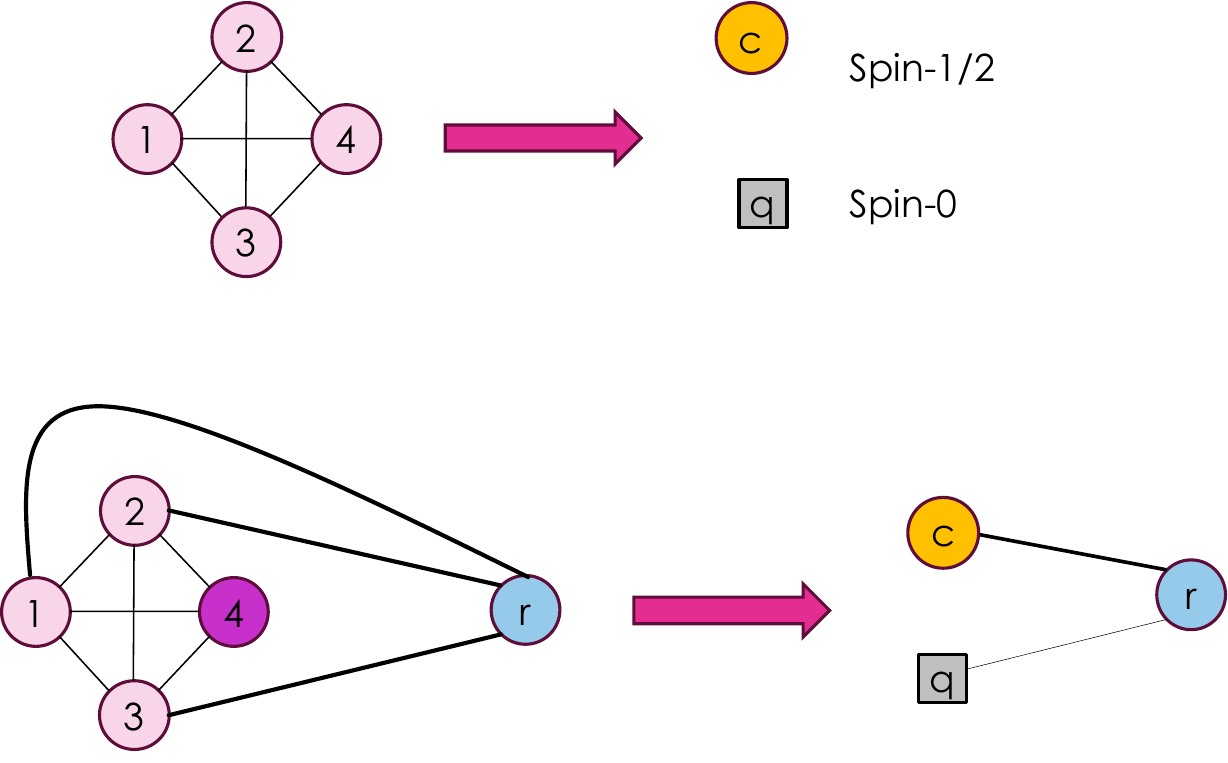} 
  \caption{
  Clique-to-vertex coupling under partial adjacency.  
  A clique (\(\{1,2,3,4\}\)) is transformed into an effective spin-\( \tfrac{1}{2} \) component labeled \(c \) (orange)  
  and a spin-0 component labeled \( q \) (gray square).  
  The external vertex \( r \) (blue) connects to all but one clique vertex (purple).
  The resulting \(\ZZ\)-coupling is split across sectors:  
  \( \tfrac{n_c - 1}{n_c} \Jzz \) between \( c \) and \( r \),  
  and \( \tfrac{1}{n_c} \Jzz \) between \( q \) and \( r \).  
  An additional off-diagonal coupling term of strength \( -\tfrac{\sqrt{n_c - 1}}{n_c} \Jzz \),  
  not shown in the figure, induces mixing between the same-sign and opposite-sign sectors.
  }
  \label{fig:Couple-CQ}
\end{figure}

\paragraph{Modification of \(\ZZ\) Terms.}

The Ising interaction term
\[
\sum_{(i,j) \in \edge(G)} \shz{i} \shz{j}
= \sum_{j \in R} \sum_{i \in L} \sum_{i_k=1}^{n_i -1} \shz{i_k} \shz{j}
\]
is modified under contraction due to the partial connectivity between \( R \) and cliques in \( L \).  
For each \( i \in L,\, j \in R \), the interaction becomes:
\[
\sum_{i_k=1}^{n_i -1} \shz{i_k} \shz{j}
\;\Longrightarrow\; 
\left( \mb{T}_i \oplus 1 \oplus \cdots \oplus 1 \right) \shz{j}
= \mb{T}_i \shz{j} \oplus \underbrace{ \shz{j} \oplus \cdots \oplus \shz{j} }_{n_i - 2 \text{ times}},
\]
%where \( \mb{T}_i \) denotes the effective operator acting on the same-sign and the opposite-sign blocks (as derived in S6.).

Therefore, the full Ising interaction term is transformed into:
\begin{align*}
\sum_{j \in R} \sum_{i \in L} \sum_{i_k=1}^{n_i -1} \shz{i_k} \shz{j}
&= \sum_{j \in R} \sum_{i \in L} \mb{T}_i \shz{j}
\oplus \cdots \oplus m_l \sum_{j \in R} \shz{j}.
\end{align*}

The transformed problem Hamiltonian thus becomes:
\begin{align*}
\bar{\mb{H}}_{\ms{P}}^{\ms{share}} &=
\sum_{i \in L} (-w_i) \left( \shz{i} \oplus 1 \oplus \cdots \oplus 1 \right)
+ \sum_{j \in R} (-w_j) \shz{j}
%\\
%&\quad
+ \Jzz \left(
\sum_{j \in R} \sum_{i \in L} \left( \mb{T}_i \shz{j} \right)
\oplus \cdots \oplus m_l \sum_{j \in R} \shz{j}
\right) \\
&=
\left( \sum_{i \in L \cup R} (-w_i) \shz{i}
+ \Jzz \sum_{(i,j) \in L \times R} \mb{T}_i \shz{j} \right)
\oplus \cdots \oplus
\left( \sum_{i \in L} (-w_i) + \sum_{j \in R} (-w_j + m_l \Jzz) \shz{j} \right).
\end{align*}

\noindent
In summary, the interaction term \( \shz{i} \shz{j} \) is replaced by the effective operator \( \mb{T}_i \shz{j} \).

\begin{mdframed}[linewidth=0.5pt]
The full Hamiltonian for the shared-state case is thus:
\begin{align*}
\bar{\Heff} 
  &=
- \mt{x} \left( \bar{\mb{H}}_{\ms{X}}^{\mc{C}} \oplus \ldots  \oplus \bar{\mb{H}}_{\ms{X}}^{\mc{Q}} \right)
+ \left( \bar{\mb{H}}_{\ms{Z}}^{\mc{C}}(\mt{jxx}) \oplus \ldots \oplus \bar{\mb{H}}_{\ms{Z}}^{\mc{Q}}(\mt{jxx}) \right) 
+ \Jzz \sum_{(i,j) \in L \times R} \left( \tfrac{n_i - 1}{n_i} \shz{i} \oplus \tfrac{1}{n_i} + \mb{T}^{\ms{cq}}_i \right) \shz{j} \\
&= \left(\mb{H}_{\mc{C}}^{\ms{share}} \oplus \ldots \oplus \mb{H}_{\mc{Q}}^{\ms{share}} \right)
+ \Hinter
\end{align*}
%% with
%% \[
%% \mb{H}_{\mc{C}}^{\ms{share}} =
%% - \mt{x} \bar{\mb{H}}_{\ms{X}}^{\mc{C}} 
%% + \bar{\mb{H}}_{\ms{Z}}^{\mc{C}}(\mt{jxx}) 
%% +   \Jzz \sum_{(i,j) \in \edge(G_{\ms{contract}})} \tfrac{n_i - 1}{n_i} \shz{i} \shz{j},
%% \]
%% and 
%% \[
%% \mb{H}_{\mc{Q}}^{\ms{share}} =
%% - \mt{x} \bar{\mb{H}}_{\ms{X}}^{\mc{Q}} 
%% + \sum_{i \in L} \left( -w_i - \tfrac{1}{4} \mt{jxx} \right)
%% + \sum_{j \in R} \left( -w_j +   \Jzz \sum_{i \in L}  \tfrac{1}{n_i} \right) \shz{j}
%% \]
%% and the coupling term:
%% \[
%% \Hinter=\Jzz \sum_{(i,j) \in L \times R} \mb{T}_i \shz{j}
%% \]
with
\[
\begin{aligned}
\mb{H}_{\mc{C}}^{\ms{share}} &=
- \mt{x} \bar{\mb{H}}_{\ms{X}}^{\mc{C}} 
+ \bar{\mb{H}}_{\ms{Z}}^{\mc{C}}(\mt{jxx}) 
+   \Jzz \sum_{(i,j) \in \edge(G_{\ms{contract}})} \tfrac{n_i - 1}{n_i} \shz{i} \shz{j}, \\[4pt]
\mb{H}_{\mc{Q}}^{\ms{share}} &=
- \mt{x} \bar{\mb{H}}_{\ms{X}}^{\mc{Q}} 
+ \sum_{i \in L} \left( -w_i - \tfrac{1}{4} \mt{jxx} \right)
+ \sum_{j \in R} \left( -w_j +   \Jzz \sum_{i \in L}  \tfrac{1}{n_i} \right) \shz{j}, \\[4pt]
\Hinter &= \Jzz \sum_{(i,j) \in L \times R} \tfrac{\sqrt{n_i - 1}}{n_i} \mb{T}^{\ms{cq}}_i \shz{j}.
\end{aligned}
\]

The inter-block coupling magnitude depends on \( \Jzz \), \( \tfrac{1}{\sqrt{n_i + 1}} \), and \( m_l \) (size of $L$), but notably is independent of the time or the transverse field.
\end{mdframed}
 
\subsection{Inner Decompositions of the Same-Sign Block during Stage~1}
\label{sec:sub2}

This subsection develops two inner decompositions of the same-sign block \( \mb{H}_{\mc{C}} \) ($L$-inner and $R$-inner)\footnote{The decompositions in Section~\ref{sec:sub1} and Section~\ref{sec:sub2} are both structural,  
but in fundamentally different senses:  
the former is physical, reflecting a true partition of the Hilbert space into dynamically distinct sectors;  
the latter is combinatorial, reorganizing the matrix to reveal subsystem indexing, without affecting the underlying dynamics.}, illustrates them through a concrete example, and introduces the corresponding notions of $L$- and $R$-localization.

\subsubsection{Two Block Decompositions of \( \mb{H}_{\mc{C}} \): \( L \)-Inner vs.~\( R \)-Inner}
Physically, the same-sign block Hamiltonian \( \mb{H}_{\mc{C}} \) is symmetric under interchange of the \( L \) and \( R \) subsystems: permuting the tensor factors leaves the spectrum unchanged.  
Mathematically, this corresponds to a permutation similarity transformation: reordering the basis to place \( L \) before \( R \), or vice versa, yields a matrix with identical eigenvalues.
Combinatorially, this symmetry allows the matrix to be organized into a two-layer block structure with either subsystem---\( L \) or \( R \)---serving as the ``inner'' block.  
That is, \( \mb{H}_{\mc{C}} \) can be expressed either as a collection of \( L \)-blocks indexed by the states from \( R \), or as \( R \)-blocks indexed by the states from \( L \).

For illustration, we assume uniform clique size \( n_i = n_c \) for all \( i \).  
We apply Theorem~\ref{thm:explicit-symmetric-reduction} to restrict the same-sign block Hamiltonian
in Eq.~\eqref{eq:same-sign-bare} to the symmetric subspace:
\begin{align*}
  \mb{H}_{\mc{C}}^{\ms{sym}}
&= \mb{H}_L^{\ms{bare}} \otimes \mb{I}_R + \mb{I}_L \otimes \mb{H}_R^{\ms{bare}} + \mb{H}_{LR},
\end{align*}
where
%% \begin{align*}
%% \mb{H}_L^{\ms{bare}} &= - \sqrt{n_c}\, \mt{x}\, \mb{C}\Sop{\X}(m) - \weff\, \mb{C}\Sop{\sZ}(m), \\
%% \mb{H}_R^{\ms{bare}} &= - \mt{x}\, \mb{C}\Sop{\X}(m_r) - w\, \mb{C}\Sop{\sZ}(m_r), \\
%% \mb{H}_{LR} &= \Jzz^{\ms{C}}\, \mb{C}\Sop{\sZ}(m) \cdot \mb{C}\Sop{\sZ}(m_r)
%% \end{align*}
\[
\mb{H}_L^{\ms{bare}} = - \sqrt{n_c}\,\mt{x}\,\mb{C}\Sop{\X}(m) - \weff\,\mb{C}\Sop{\sZ}(m),
\quad
\mb{H}_R^{\ms{bare}} = - \mt{x}\,\mb{C}\Sop{\X}(m_r) - w\,\mb{C}\Sop{\sZ}(m_r),
\]
and
\[
\mb{H}_{LR} = \Jzz^{\ms{C}}\, \mb{C}\Sop{\sZ}(m)\,\mb{C}\Sop{\sZ}(m_r), \quad \mbox{ with }
\Jzz^{\ms{C}} = \Jzz f_c^{\ms{C}} 
 = 
\begin{cases}
\Jzz & \text{for } \Gdis, \\
\tfrac{n_c - 1}{n_c} \Jzz & \text{for } \Gshare.
\end{cases}
\]
Here, \( \mb{C}\Sop{\sZ}(m) \) and \( \mb{C}\Sop{\X}(m) \) denote the collective Pauli operators restricted to the symmetric
subspace of \( m \) spins, as defined in Eqs.~\eqref{eq:CSZ} and~\eqref{eq:CSX}, and all \(^{\ms{bare}}\) operators are understood in this context to be restricted to the symmetric subspace.

We illustrate this dual block structure through a concrete example with \( m = 2 \) and \( m_r = 3 \) in the symmetric subspace, showing (i) the basis structure (Figure~\ref{fig:L-R-inner}) and (ii) the explicit matrix form.

\paragraph{Visual and Combinatorial Structure.}
Combinatorially, the symmetric subspace contains \( (m + 1)(m_r + 1) \) basis states.  
These states can be organized either as the \( L \)-inner decomposition  
(grouping by \( L \), with blocks indexed by spin-up count in \( R \)),  
or as the \( R \)-inner decomposition (grouping by \( R \), with blocks indexed by spin-up count in \( L \)).

See Figure~\ref{fig:L-R-inner} for an illustration.
\begin{figure}[!htbp]
  \centering
  \includegraphics[width=0.6\textwidth]{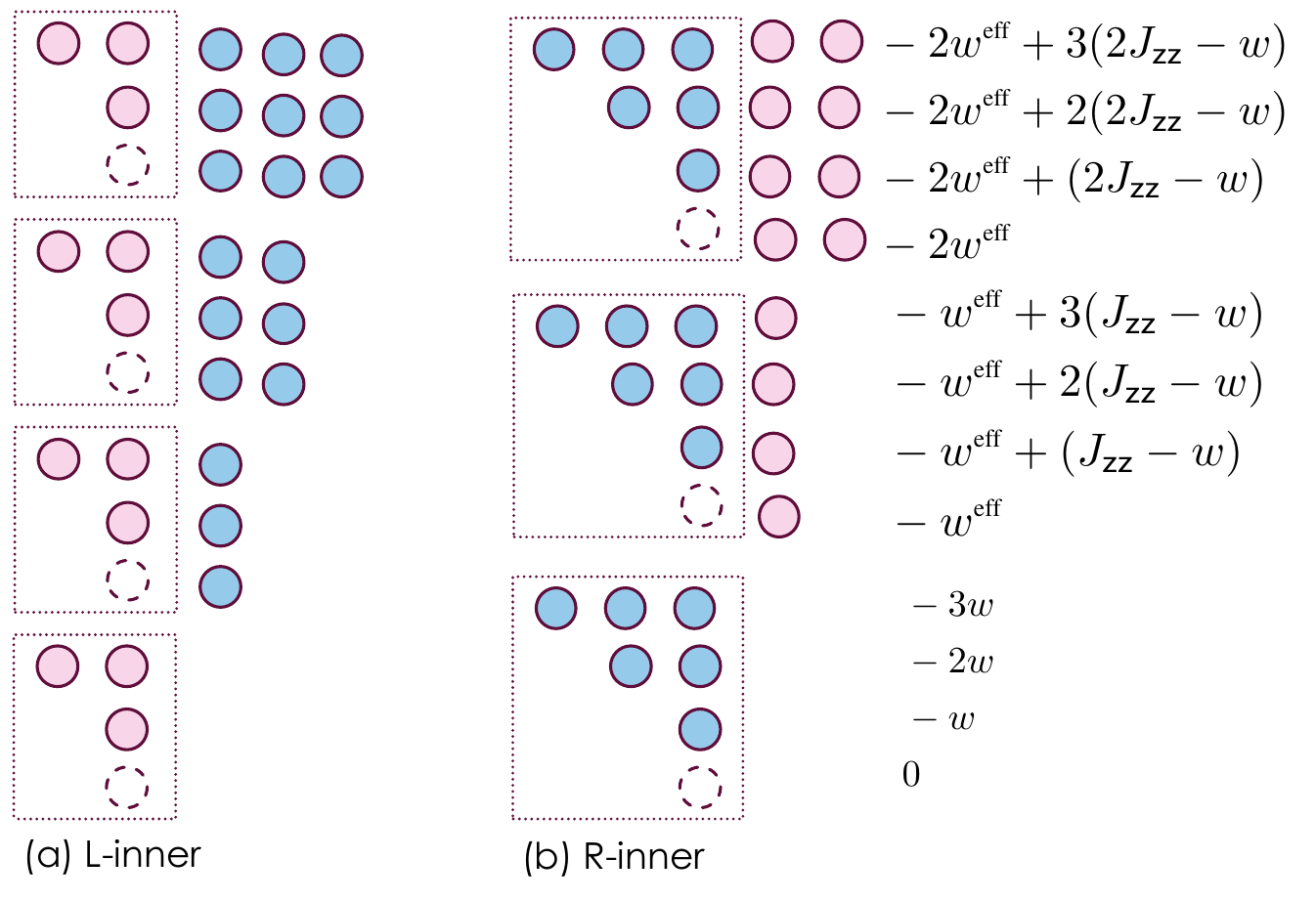}
\caption{
  Two possible orderings of the basis states in the symmetric subspace of the same-sign block \(
  \mb{H}_{\mc{C}}^{\ms{sym}} \), illustrated for \( m = 2 \), \( m_r = 3 \).  
(a) \emph{\( L \)-inner ordering}: basis states are grouped by the number of spin-up qubits in the \( R \) subsystem (blue),  
with each group forming a block over the \( L \) subsystem (pink).  
(b) \emph{\( R \)-inner ordering}: basis states are grouped by spin-up count in \( L \) (pink),  
with each group forming a block over \( R \) (blue).  
The dashed circle marks the empty-set basis state (no spin-ups).  
Each blue circle contributes energy \( -w \); each pink circle contributes \( -\weff \);  
each blue-pink pair incurs a coupling penalty of \( +\Jzz \).  
The total energy of each basis state is shown on the right. Energies within each block vary monotonically.  
To ensure that the block energies in (b) decrease monotonically from top to bottom, we require  
\( - m \weff > -\weff + m_r(\Jzz - w) \).  
This condition forms the basis for deriving the Steering lower bound \( \Jxxsteer \).
}
\label{fig:L-R-inner}
\end{figure}

\paragraph{Explicit Matrix Representation.}
We now present the explicit matrix form of \( \mb{H}_{\mc{C}}^{\ms{sym}} \) in the symmetric subspace,  
corresponding to the two decompositions shown in Figure~\ref{fig:L-R-inner}.  
In both cases, the matrix has dimension \( 12 \times 12 \),  
but the block layout differs depending on the ordering of basis states.

Each diagonal entry reflects the total energy of a basis state.
Off-diagonal entries arise from the transverse-field term, which connects basis states differing by a single spin flip.

\paragraph{\( L \)-inner block decomposition:}
There are \( m_r + 1 = 4 \) outer-layer blocks, each a \( (m+1) \times (m+1) \) matrix acting on the \( L \) subsystem.
The full matrix representation is denoted \( \mb{H}_{\mc{C}}^{{L\!-\!\ms{inner}}} =\)

\begin{center}
  \scalebox{0.55}{$
%    \mb{H}_{\mc{C}}^{L-\ms{inner}}=
\left(
\begin{array}{ccc|ccc|ccc|ccc}
 6 \Jzz - 3w - 2\weff & -\frac{\sqrt{\text{nc}}\, x}{\sqrt{2}} & 0 & -\frac{1}{2} \sqrt{3}x & 0 & 0 & 0 & 0 & 0 & 0 & 0 & 0 \\
 -\frac{\sqrt{\text{nc}}\, x}{\sqrt{2}} & 3 \Jzz - 3w - \weff & -\frac{\sqrt{\text{nc}}\, x}{\sqrt{2}} & 0 & -\frac{1}{2} \sqrt{3}x & 0 & 0 & 0 & 0 & 0 & 0 & 0 \\
 0 & -\frac{\sqrt{\text{nc}}\, x}{\sqrt{2}} & -3w & 0 & 0 & -\frac{1}{2} \sqrt{3}x & 0 & 0 & 0 & 0 & 0 & 0 \\
 \hline
 -\frac{1}{2} \sqrt{3}x & 0 & 0 & 4 \Jzz - 2w - 2\weff & -\frac{\sqrt{\text{nc}}\, x}{\sqrt{2}} & 0 & -x & 0 & 0 & 0 & 0 & 0 \\
 0 & -\frac{1}{2} \sqrt{3}x & 0 & -\frac{\sqrt{\text{nc}}\, x}{\sqrt{2}} & 2 \Jzz - 2w - \weff & -\frac{\sqrt{\text{nc}}\, x}{\sqrt{2}} & 0 & -x & 0 & 0 & 0 & 0 \\
 0 & 0 & -\frac{1}{2} \sqrt{3}x & 0 & -\frac{\sqrt{\text{nc}}\, x}{\sqrt{2}} & -2w & 0 & 0 & -x & 0 & 0 & 0 \\
 \hline
 0 & 0 & 0 & -x & 0 & 0 & 2 \Jzz - w - 2\weff & -\frac{\sqrt{\text{nc}}\, x}{\sqrt{2}} & 0 & -\frac{1}{2} \sqrt{3}x & 0 & 0 \\
 0 & 0 & 0 & 0 & -x & 0 & -\frac{\sqrt{\text{nc}}\, x}{\sqrt{2}} & \Jzz - w - \weff & -\frac{\sqrt{\text{nc}}\, x}{\sqrt{2}} & 0 & -\frac{1}{2} \sqrt{3}x & 0 \\
 0 & 0 & 0 & 0 & 0 & -x & 0 & -\frac{\sqrt{\text{nc}}\, x}{\sqrt{2}} & -w & 0 & 0 & -\frac{1}{2} \sqrt{3}x \\
 \hline
 0 & 0 & 0 & 0 & 0 & 0 & -\frac{1}{2} \sqrt{3}x & 0 & 0 & -2\weff & -\frac{\sqrt{\text{nc}}\, x}{\sqrt{2}} & 0 \\
 0 & 0 & 0 & 0 & 0 & 0 & 0 & -\frac{1}{2} \sqrt{3}x & 0 & -\frac{\sqrt{\text{nc}}\, x}{\sqrt{2}} & -\weff & -\frac{\sqrt{\text{nc}}\, x}{\sqrt{2}} \\
 0 & 0 & 0 & 0 & 0 & 0 & 0 & 0 & -\frac{1}{2} \sqrt{3}x & 0 & -\frac{\sqrt{\text{nc}}\, x}{\sqrt{2}} & 0 \\
\end{array}
\right)
$}
\end{center}

\paragraph{\( R \)-inner block decomposition:} There are \( m + 1 = 3 \) outer-layer blocks, each a \( (m_r+1) \times (m_r+1) \) matrix
acting on the \( R \) subsystem.
The full matrix representation is denoted \( \mb{H}_{\mc{C}}^{{R\!-\!\ms{inner}}} =\)
\begin{center}
  \scalebox{0.55}{$
%    \mb{H}^{R-\ms{inner}}=
\left(
\begin{array}{cccc|cccc|cccc}
 6 \Jzz - 3w - 2\weff & -\frac{1}{2} \sqrt{3}x & 0 & 0 & -\frac{\sqrt{\text{nc}}\, x}{\sqrt{2}} & 0 & 0 & 0 & 0 & 0 & 0 & 0 \\
 -\frac{1}{2} \sqrt{3}x & 4 \Jzz - 2w - 2\weff & -x & 0 & 0 & -\frac{\sqrt{\text{nc}}\, x}{\sqrt{2}} & 0 & 0 & 0 & 0 & 0 & 0 \\
 0 & -x & 2 \Jzz - w - 2\weff & -\frac{1}{2} \sqrt{3}x & 0 & 0 & -\frac{\sqrt{\text{nc}}\, x}{\sqrt{2}} & 0 & 0 & 0 & 0 & 0 \\
 0 & 0 & -\frac{1}{2} \sqrt{3}x & -2\weff & 0 & 0 & 0 & -\frac{\sqrt{\text{nc}}\, x}{\sqrt{2}} & 0 & 0 & 0 & 0 \\
 \hline
 -\frac{\sqrt{\text{nc}}\, x}{\sqrt{2}} & 0 & 0 & 0 & 3 \Jzz - 3w - \weff & -\frac{1}{2} \sqrt{3}x & 0 & 0 & -\frac{\sqrt{\text{nc}}\, x}{\sqrt{2}} & 0 & 0 & 0 \\
 0 & -\frac{\sqrt{\text{nc}}\, x}{\sqrt{2}} & 0 & 0 & -\frac{1}{2} \sqrt{3}x & 2 \Jzz - 2w - \weff & -x & 0 & 0 & -\frac{\sqrt{\text{nc}}\, x}{\sqrt{2}} & 0 & 0 \\
 0 & 0 & -\frac{\sqrt{\text{nc}}\, x}{\sqrt{2}} & 0 & 0 & -x & \Jzz - w - \weff & -\frac{1}{2} \sqrt{3}x & 0 & 0 & -\frac{\sqrt{\text{nc}}\, x}{\sqrt{2}} & 0 \\
 0 & 0 & 0 & -\frac{\sqrt{\text{nc}}\, x}{\sqrt{2}} & 0 & 0 & -\frac{1}{2} \sqrt{3}x & -\weff & 0 & 0 & 0 & -\frac{\sqrt{\text{nc}}\, x}{\sqrt{2}} \\
 \hline
 0 & 0 & 0 & 0 & -\frac{\sqrt{\text{nc}}\, x}{\sqrt{2}} & 0 & 0 & 0 & -3w & -\frac{1}{2} \sqrt{3}x & 0 & 0 \\
 0 & 0 & 0 & 0 & 0 & -\frac{\sqrt{\text{nc}}\, x}{\sqrt{2}} & 0 & 0 & -\frac{1}{2} \sqrt{3}x & -2w & -x & 0 \\
 0 & 0 & 0 & 0 & 0 & 0 & -\frac{\sqrt{\text{nc}}\, x}{\sqrt{2}} & 0 & 0 & -x & -w & -\frac{1}{2} \sqrt{3}x \\
 0 & 0 & 0 & 0 & 0 & 0 & 0 & -\frac{\sqrt{\text{nc}}\, x}{\sqrt{2}} & 0 & 0 & -\frac{1}{2} \sqrt{3}x & 0 \\
\end{array}
\right)
$}
\end{center}

Notice that each inner-block Hamiltonian is a linear shift of a base Hamiltonian:
\begin{align}
\mb{H}_L^{(r)} &= \mb{H}_L^{(0)} + r \cdot \mb{S}_L, \quad
\mb{H}_R^{(l)} = \mb{H}_R^{(0)} + l \cdot \mb{S}_R,
\label{eq:inner-block-linear-shift}
\end{align}
for \( r \in \{0, 1, \dots, m_r\} \) and \( l \in \{0, 1, \dots, m\} \), where the shift matrices are given by:
\begin{align*}
\mb{S}_L =
\begin{bmatrix}
m \Jzz - w  & \cdots & 0 & 0 \\
\vdots  & \ddots & \vdots & \vdots \\
0  & \cdots & \Jzz - w & 0 \\
0  & \cdots & 0 & -w
\end{bmatrix}, \quad
\mb{S}_R =
\begin{bmatrix}
m_r \Jzz - \weff  & \cdots & 0 & 0 \\
\vdots  & \ddots & \vdots & \vdots \\
0  & \cdots & \Jzz - \weff & 0 \\
0  & \cdots & 0 & -\weff
\end{bmatrix}.
\end{align*}

\begin{remark}
The zero-indexed blocks correspond directly to the bare subsystems:
\[
\mb{H}_R^{(0)} = \ket{0}_L\bra{0} \otimes \mb{H}_R^{\ms{bare}}, 
\qquad
\mb{H}_L^{(0)} = \mb{H}_L^{\ms{bare}} \otimes \ket{0}_R\bra{0}.
\]
Thus there are two possible inner decompositions of \( \mb{H}_{\mc{C}} \):  
the \( L \)-inner blocks \( \mb{H}_L^{(r)} \) indexed by \( r \),  
and the \( R \)-inner blocks \( \mb{H}_R^{(l)} \) indexed by \( l \).  
For brevity, we will often refer to these simply as the \(L\)-blocks and \(R\)-blocks.
\end{remark}

\subsubsection{Definition of Structural Localization}
\label{sec:struct-localization}

We define a quantitative notion of localization
of the ground state of the same-sign block \( \mb{H}_{\mc{C}} \), based on the inner decompositions.
We begin with the \emph{\( R \)-localization} defined via the \( R \)-inner decomposition; the definition of \( L \)-localization is analogous.

Recall from Eq.~\eqref{eq:inner-block-linear-shift} that
\[
\mb{H}_R^{(l)} = \mb{H}_R^{(0)} + l \cdot \mb{S}_R,
\]
with the block energies ordered in decreasing order when \( -\weff= -w + \tfrac{n_c - 1}{4} \Jxx > 0 \) or \( \Jxx \) is sufficiently large.

\begin{definition}[Structural Localization]
Let \( \{P_R^{(l)}\} \) denote the projectors onto the \( \mb{H}_R^{(l)} \)-inner block, for \(l=0,\ldots,m\).
For parameters \( k \ll m \) and tolerance \( \epsilon > 0 \), we say the ground state 
\( \ket{\psi(t)} \) of \( \mb{H}_{\mc{C}}(t) \) is 
\emph{\( R \)-localized up to depth \( k \)} if the cumulative overlap
\[
\sum_{l=0}^{k-1} \| P_R^{(l)} \ket{\psi(t)} \|^2 \;\geq\; 1 - \epsilon.
\]
\end{definition}

\noindent
In words, the ground state amplitude is almost entirely concentrated within the lowest \( k \) 
\( R \)-blocks.

In Section~\ref{sec:stage1}, we will show that if \( \Jxx \) satisfies the \emph{Steering} lower bound (derived in Section~\ref{sec:Jxx-bounds}), the ground state is smoothly steered into the \( R \)-localized region by the end of Stage~1.  
In contrast, for \( \Jxx = 0 \),
the ground state exhibits \( L \)-localization with dominant weight on \( \mb{H}_L^{(0)} \).
This case is analyzed in detail in the companion paper~\cite{Choi-Limitation}.

 %inner-L vs inner-R
\subsection{Two-Stage Evolution and Feasibility Bounds on $\Jxx$}
\label{sec:sub4}

Our goal is to show that, for an appropriate choice of \( \Jxx \) and \( \Jzz \),
the system evolves from the initial ground state of $\Heff$ to the global minimum (\GM{}) without encountering an anti-crossing.
The correctness of the algorithm critically depends on
four feasibility bounds on \( \Jxx \)(Table~\ref{tab:JxxBounds}),
%\emph{Stage-Separation Upper Bound} ($\Jxxsep$), \emph{Lifting Lower Bound} ($\Jxxlift$), 
%\emph{Steering Lower Bound} ($\Jxxsteer$), and \emph{Sinking Upper Bound} ($\Jxxsink$),
together with two upper bounds on \( \Jzz \): \( \Jzzinter \) and \( \Jzzsteer \).
%Here \( \Jzzinter \) ensures that inter-block coupling remains weak during the early stage (when \( \mt{x} \) is large), 
%while \( \Jzzsteer \) is required for the smooth steering condition (and is the bound on which $\Jxxsteer$ depends). 

In Section~\ref{sec:Jxx-bounds}, we derive these bounds analytically, and show that for \(\Jzz \le \Jzzsteer\) 
there exists a nonempty feasible window for \( \Jxx \).
To give intuition before the detailed derivations, we briefly summarize how each bound 
contributes to maintaining a smooth two-stage evolution.
%By construction, as long as \(\Jzz \le \Jzzinter\)  and \( \Jxx \) lies within this window,
%the evolution path remains smooth:
The Stage-Separation bound together with \(\Jzzinter\) ensure that Stage~1 is effectively confined to the same-sign block, 
allowing Stage~1 and Stage~2 to be analyzed separately; 
the Lifting bound ensures the original anti-crossing is removed; 
the Steering bound directs the system smoothly into the \(R\)-region (bypassing tunneling) during Stage~1,
and Stage~2 is secured by the Sinking bound, which prevents the emergence of a new anti-crossing when the lowest opposite-sign block participates. 
%Thus, with all four bounds satisfied, the two-stage evolution is smooth throughout.
The analysis, supported by numerical results, of Stage~1 and Stage~2 is presented in Section~\ref{sec:stage1} and Section~\ref{sec:stage2}, respectively.

\subsubsection{Analytical Bounds on \(\Jxx\) and the Feasibility Window}
\label{sec:Jxx-bounds}
Recall that \(\Jxx\) produces a \emph{see-saw effect}:  
it raises the energy associated with \LM{} in the same-sign block while simultaneously lowering the energy in the opposite-sign blocks.  
This effect introduces a fundamental trade-off: if \( \Jxx \) is too small, the system fails to lift the local minima in the same-sign block high enough;  
if it is too large, it lowers the opposite-sign state too far, introducing new anti-crossings.  
As a result, the success of the algorithmic design depends critically on choosing an appropriate value of \( \Jxx \).
These considerations motivate four analytical feasibility bounds on \( \Jxx \)
 summarized in Table~\ref{tab:JxxBounds}.
\begin{table}[h!]
\centering
\small
\begin{tabular}{@{}lll@{}}
\toprule
\textbf{Name} & \textbf{Notation} & \textbf{Purpose} \\
\midrule
\textbf{(i) Stage-Separation Upper Bound} & \( \Jxxsep \) &
Ensure the ground state of the same-sign block is the lowest during Stage~1 \\

\textbf{(ii) Lifting Lower Bound} & \( \Jxxlift \) &
Lift \LM{} energy in the same-sign block high enough during Stage~2 \\

\textbf{(iii) Steering Lower Bound} & \( \Jxxsteer \) &
Ensure smooth transition into the \( R \)-localized region during Stage~1 \\

\textbf{(iv) Sinking Upper Bound} & \( \Jxxsink \) &
Prevent the AS0 opposite-sign block from dropping too low in Stage~2 \\
\bottomrule
\end{tabular}
\caption{Summary of the four feasibility bounds on \( \Jxx \). (Note: the Sinking Upper Bound is required only in the shared-structure case.)}
\label{tab:JxxBounds}
\end{table}

For simplicity, we assume uniform clique size, \( n_i = n_c \) for all \( i \),
and remark on the changes needed for each bound when clique sizes vary.
We also set \( w = 1 \) in the following analysis.

We derive the following explicit bounds (below) in terms of \( \Gamma_2 \), \( \Jzz \), \( m \), and the unknown quantities \( m_r \) and \( m_g \):
\begin{align}
\left\{
\begin{aligned}
\Jxxlift &=   2 \tfrac{m}{m_g} \Gamma_2 , \\
\Jxxsteer &= \tfrac{4}{n_c - 1}[m_r(\Jzz - 1) + 1], \\
\Jxxsep &=  2(\Gamma_2 - 1), \\
\Jxxsink &=  2(\Gamma_2 - 1) + \tfrac{2 m_r}{n_c} \Jzz.
\end{aligned}
\right.
\tag{Jxx-Bounds} \label{eq:Jxx-bounds}
\end{align}

These analytical bounds are approximate but conservative, derived from the bare (decoupled) block energies and perturbative shifts.
They suffice to demonstrate a nonempty feasible window for \( \Jxx \).  
In practice, additional values may succeed beyond the bounds and can be identified through numerical exploration or adaptive tuning.

\subsubsection*{Feasible \(\Jxx\) Window}
Notice that the lower bound \( \Jxxsteer \) depends on \( \Jzz \).  
If \( \Jzz \) is too large, the required \( \Jxx \) may exceed the upper bounds, resulting in an empty feasible window.  
To avoid this, we now show that the four bounds are compatible under a reasonable condition on \( \Jzz \), thereby ensuring a nonempty feasible window.

\begin{mdframed}[linewidth=1pt]
\begin{theorem}
\label{thm:feasible-window}
Assume uniform clique size \( n_i = n_c \) for all \( i \), and suppose \( m \geq 3 \) and \( m_r \geq 2 \).  
Use \( \Gamma_2 = m \).  
If
\(
\Jzz \le \Jzzsteer := 1 + \tfrac{(n_c - 1)(m - 1) - 2}{2m_r},
\)
then setting
\(
\Jxx = 2(m - 1),
\)
yields
\(
\max\{ \Jxxlift,\, \Jxxsteer \} \leq \Jxx \leq \min\{\Jxxsep, \Jxxsink\},
\)
where explicit formulas for all bounds are derived in equations~\eqref{eq:Jxx-bounds}.
\end{theorem}
\end{mdframed}
The choice \( \Jxx = 2(m - 1) \) is not unique; it simply serves as a constructive witness that the feasible window is nonempty whenever the condition on \( \Jzz \) holds.

\paragraph{Feasibility Condition on \(\Jzz\).}
The upper bound \( \Jzzsteer \) in the theorem depends on the unknown parameter \( m_r \).  
To ensure feasibility without knowing \( m_r \) precisely, we apply a conservative upper bound:
\(
m_r < \sqrt{n_c} m.
\)
Substituting into the definition of \( \Jzzsteer \), we obtain:
\(
\Jzz \le 1 + \tfrac{(n_c - 1)(m - 1) - 2}{2 \sqrt{n_c} m}.
\)
In practice, it suffices to enforce the relaxed condition:
\(
\Jzz \lesssim 1 + \tfrac{\sqrt{n_c} + 1}{2}.
\)
%Thus, for moderate values of \( \Jzz \), a feasible \( \Jxx \) regime always exists.

We now present the detailed derivations of the four analytical bounds on $\Jxx$.
Each bound is introduced by its guiding idea, followed by an explicit derivation of the corresponding formula.

\subsubsection*{(i) Stage-Separation Upper Bound and $\Jzzinter$}
The Stage-Separation requirement explains the rationale behind the two-stage design of the coupling schedule.
The objective is to ensure that the system's ground state starts and remains confined to the same-sign block during Stage~1,  
and involves opposite-sign blocks only during Stage~2.  
To achieve this, two conditions are needed:  
(a) the lowest energy level of the full system must lie within the same-sign block throughout Stage~1, and  
(b) the inter-block coupling between same-sign and opposite-sign blocks must remain weak enough to be neglected.

\paragraph{Condition (a): Crossover within Stage~2.}
This motivates a two-stage schedule in which the \XX-coupling \( \mt{jxx} \) is held constant during Stage~1  
and scaled linearly with the transverse field during Stage~2.  
Using the block energy ordering results of the $L$-bare subsystem from Section~\ref{sec:block-ordering} (Figure~\ref{fig:block-order}),  
we find that the crossover between the same-sign and opposite-sign blocks occurs at
\(
\mt{x}_c = \tfrac{4\alpha}{4 - \alpha^2}.
\)
This lies within Stage~2 (i.e., \( \mt{x}_c \le \Gamma_2 \)) if and only if \( \alpha \le \alpha_{\max}(\Gamma_2) \),  
where \( \alpha_{\max}(\Gamma_2) \) is defined in Eq.~\ref{eq:alpha-max}.  
For $\Gamma_2 \ge 2$, one has
$
\alpha_{\max}(\Gamma_2) \gtrsim \tfrac{2(\Gamma_2 - 1)}{\Gamma_2}.
$
This yields the Stage-Separation bound on the constant \XX-coupling during Stage~1:
\[
\Jxxsep = 2(\Gamma_2 - 1).
\]

\paragraph{Condition (b): Weak inter-block coupling.}
In addition, recall from Theorem~\ref{thm:transformed-hamiltonian} that the effective Hamiltonian decomposes as
\[
\bar{\Heff} = \mb{H}_{\mc{C}} \oplus \cdots \oplus \mb{H}_{\mc{Q}} + \Hinter,
\]
where $\Hinter = 0$ in the disjoint case, while in the shared case
\[
\Hinter = \Jzz \sum_{(i,j) \in L \times R} \tfrac{\sqrt{n_i - 1}}{n_i} \mb{T}^{\ms{cq}}_i \shz{j}.
\]
Here $\mb{T}^{\ms{cq}}$ (Eq.~\eqref{eq:Tcq}) is a $3 \times 3$ off-diagonal operator mixing same-sign and opposite-sign components.  
The strength of this coupling depends on \( \Jzz \) and clique size \( n_i \), but not on the transverse field \( \mt{x} \).

Since we take $\Gamma_2 = m$, and Stage~1 is defined by $\mt{x} \in [\Gamma_2, \Gamma_1]$, 
the inter-block coupling is negligible whenever \( \Jzz \ll m \).  
We summarize this condition by writing \( \Jzz \le \Jzzinter \).  
In particular, this requires that $\Jzz$ be chosen as a constant (independent of $m$), 
so that the condition continues to hold as the problem size grows.
In practice, since $\Jzzsteer \le \Jzzinter$ by construction, 
it suffices to impose $\Jzz \le \Jzzsteer$.

\paragraph{Conclusion.}
Together, the Stage-Separation bound \( \Jxxsep \) and the inter-block coupling bound \( \Jzzinter \)  
ensure that the low-energy dynamics of Stage~1 are effectively confined to the same-sign block.  
This joint condition justifies treating Stage~1 and Stage~2 as separate regimes in the analysis.

\subsubsection*{(ii) Lifting Lower Bound \( \Jxxlift \)}

\textbf{Idea:}  
This bound ensures that the energy associated with \LM{} in the same-sign block is raised high enough,
during Stage~2.
The key idea is to use exact analytical expressions for the bare energy levels, \( \ELz \) and \( \EGz \),  
derived in Section~\ref{sec:bare-sub}, to approximate the true energies associated with \LM{} and \GM{}.  
We then determine the condition under which \( \ELz \) remains strictly above \( \EGz \) throughout Stage~2.

This is a perturbative approximation.  
By second-order perturbation theory, the energy associated with \LM{} is shifted upward,  
while the energy associated with \GM{} is shifted downward.  
Therefore, if the bare energy levels do not cross, the true levels will not have an anti-cross either.

The resulting threshold defines the \emph{Lifting Lower Bound} \( \Jxxlift \):  
the minimum value of \( \Jxx \) required to ensure that the energy level associated with \LM{} in the same-sign block remains sufficiently lifted throughout Stage~2.  
This approximation is further supported numerically.

\noindent
\textbf{Derivation:}  
The (negative) slope magnitude of \( \ELz(\mt{x}) \) decreases from
\[
\tfrac{\sqrt{n_c}}{2} \cdot m \quad \text{(at } \Jxx = 0\text{)} \quad \text{to} \quad \tfrac{1}{\alpha} \cdot m \quad \text{(at } \Jxx = \alpha \Gamma_2\text{).}
\]
This is illustrated in Figure~\ref{fig:see-saw}(c), and notably, the slope bound is independent of clique size \( n_c \).
In contrast, the (negative) slope magnitude of \( \EGz(\mt{x}) \) remains approximately constant:
\(
\tfrac{m_g}{2}.
\)
To maintain \( \ELz(\mt{x}) > \EGz(\mt{x}) \) throughout Stage~2, it suffices to require 
\( \tfrac{m}{\alpha} < \tfrac{m_g}{2} \), i.e.\ \( \alpha > 2 \tfrac{m}{m_g} \).
Thus, the bound becomes:
\[
\Jxxlift = 2 \cdot \tfrac{m}{m_g} \cdot \Gamma_2.
\]

\subsubsection*{(iii) Steering Lower Bound \( \Jxxsteer \)}

\textbf{Idea:}  
The lower bound \( \Jxxsteer \) is defined as the minimal value of \( \Jxx \) such that  
the diagonal entries of the \( R \)-inner block decomposition Hamiltonian are  
strictly decreasing from top to bottom.  
This condition ensures that the system can smoothly localize into the lowest \( R \)-blocks.

The key idea is that the coupling strength \( \Jxx \) must be large enough  
so that any configuration involving vertices from \( L \) incurs a high energy penalty,  
thereby suppressing support on the \( L \)-blocks throughout Stage~1.
See Figure~\ref{fig:L-R-inner} for an illustration.

\noindent
\textbf{Derivation:}  
Specifically, we require the inequality
\(
-2\,\weff > -(\weff + m_r(1 - \Jzz)),
\)
where \(
\weff = 1 - \tfrac{n_c - 1}{4} \Jxx.
\) is the effective vertex weight.
Solving the inequality yields the lower bound:
\[
\Jxxsteer = \tfrac{4}{n_c - 1} \left[ m_r(\Jzz - 1) + 1 \right].
\]

\begin{remark}
Since \(\Jxxsteer\) depends on \( n_c \),  
a conservative generalization for non-uniform clique sizes is to replace \( n_c \) with \( \min_i n_i \).  
\end{remark}

\subsubsection*{(iv) Sinking Upper Bound \( \Jxxsink \)}

\textbf{Idea:}
This is analogous to the Lifting Lower Bound \( \Jxxlift \).
We want to make sure that the AS0 energy level is not dropped too low to form a new anti-crossing.
We use the bare (decoupled) energy levels \( \EGz \) (the ground state energy of \( \mb{H}_{\GM}^{\ms{bare}} \)) and \( \ESz \) (the ground state energy of \( \mb{H}_{\mc{Q}} \))  
to approximate the true energy levels.
We then impose an upper bound on \( \Jxx \) to ensure that \( \ESz > \EGz \) throughout Stage~2.

By the same perturbative argument applied when the blocks are coupled, while using the decoupled bare levels as reference points,  
the true energy level corresponding to \( \ESz \) is shifted upward, while the true energy level corresponding to \( \EGz \) is shifted downward.  
Therefore, as long as \( \ESz > \EGz \), no anti-crossing will be formed.
This approximation is further supported numerically.

\noindent
\textbf{Derivation:}  
We require that
\(
\ESz(\mt{x}) > \EGz(\mt{x}),
\) during Stage~2,
and in particular, we enforce this at the beginning of Stage~2, with \( x = \Gamma_2 \):
\begin{align*}
\ESz(\Gamma_2) &\approx -m\left(1 + \tfrac{\Jxx}{4}\right) - \tfrac{m_r}{2}\left(1 - \tfrac{m}{n_c} \Jzz + \Gamma_2 \right), \\
\EGz(\Gamma_2) &\approx -\tfrac{m + m_r}{2}(1 + \Gamma_2).
\end{align*}
Solving this inequality yields the upper bound:
\[
\Jxxsink = 2(\Gamma_2 - 1) + \tfrac{2 m_r}{n_c} \Jzz.
\]
\begin{remark}
Since \(\Jxxsink\) depends on \(n_c\), a conservative generalization for non-uniform clique sizes is to replace \(n_c\) with \(\max_i n_i\).
\end{remark}

Having established the four feasibility bounds on $\Jxx$
 and the condition for a nonempty window, we now turn to the detailed analysis of Stage~1 and Stage~2.
 %Jxx
\subsubsection{Stage~1: Structural Steering}
\label{sec:stage1}

The purpose of Stage~1 is twofold:  
(1) to ensure that the evolution is confined to the same-sign block, and  
(2) to demonstrate successful \emph{structural steering} within this block.  
Confinement is guaranteed by the Stage-Separation bound \(\Jxxsep\) together with the requirement \(\Jzz \le \Jzzinter\) (in the shared case).

We define \emph{structural steering} as the mechanism by which the evolving ground state is directed smoothly into the \(R\)-localized region (the global-minimum-supporting region). 
In our design, steering is achieved by ordering the energies of the \(R\)-inner blocks of the same-sign Hamiltonian so that the ground state is guided directly into the \(R\)-localized region, rather than first becoming trapped in the \(L\)-region and only later tunneling into the \(R\)-region, as happens in the stoquastic annealing case.  
Formally, this mechanism is guaranteed by two feasibility bounds: the Lifting bound \(\Jxxlift\), which forces the ground state toward the lowest \(R\)-blocks; and the Steering bound \(\Jxxsteer\), which ensures that this localization proceeds smoothly.  
This analytically established behavior is further confirmed by numerical evidence.

\subsubsection*{Numerical Confirmation of Structural Steering}

We confirm structural steering numerically by tracking structural localization (as defined in Section~\ref{sec:struct-localization}), 
namely the cumulative projection weight of the ground state onto the lowest $R$-blocks.

Since this analysis involves only the same-sign block, and assuming uniform clique size, 
we may employ the symmetric-subspace reduction, which enables large-scale calculations.
These calculations determine the appropriate value of $k$, corresponding to the minimum number of $R$-blocks that capture nearly all of the ground-state weight.
In practice, when $m_r > m$, we typically find $k \approx 2$, 
whereas when $2 \le m_r \ll m$ (possible in the shared-structure case) \footnote{In this case, only the condition $\Jzz \le \Jzzinter$ is required (not $\Jzz \le \Jzzsteer$).}
 a slightly larger $k$ (upto $0.2m$) may be required.

For the disjoint case, Figure~\ref{fig:dis-steering}(a) shows smooth steering into the lowest $k=2$ blocks (with $m_r = m_g > m$). 
For contrast, Figure~\ref{fig:dis-steering}(b) shows the evolution when the steering condition is violated (with $\Jzz = 1000 \gg \Jzzsteer$, resulting in $\Jxx < \Jxxsteer$). 
Here the ground state first localizes in the $L$-region and only later tunnels into the $R$-region, 
highlighting the tunneling bottleneck that structural steering avoids.

Figure~\ref{fig:steering-m30} shows successful steering for $m = 30$ with small $m_r = 5$. 
By the end of Stage~1, more than $90\%$ of the ground state amplitude is concentrated within just six blocks.

\begin{figure}[!htbp]
\centering
\begin{subfigure}{0.45\textwidth}
  \includegraphics[width=\linewidth]{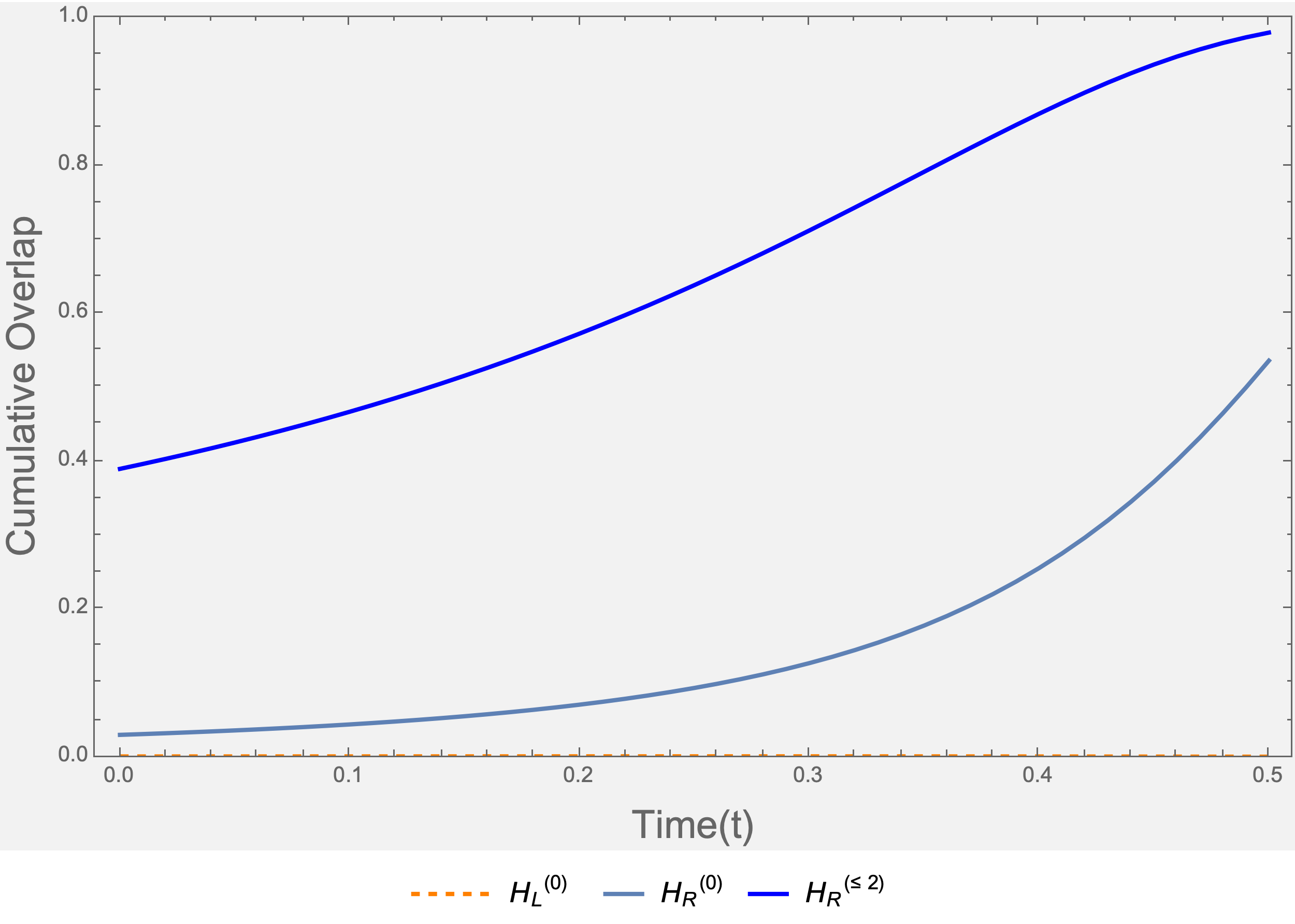}
  \caption{Successful structural steering with $\Jxx = 2(m-1)$, $\Jzz = 1 + \tfrac{\sqrt{n_c} + 1}{2}$.}
\end{subfigure}
\hspace{0.05\textwidth}
\begin{subfigure}{0.45\textwidth}
  \includegraphics[width=\linewidth]{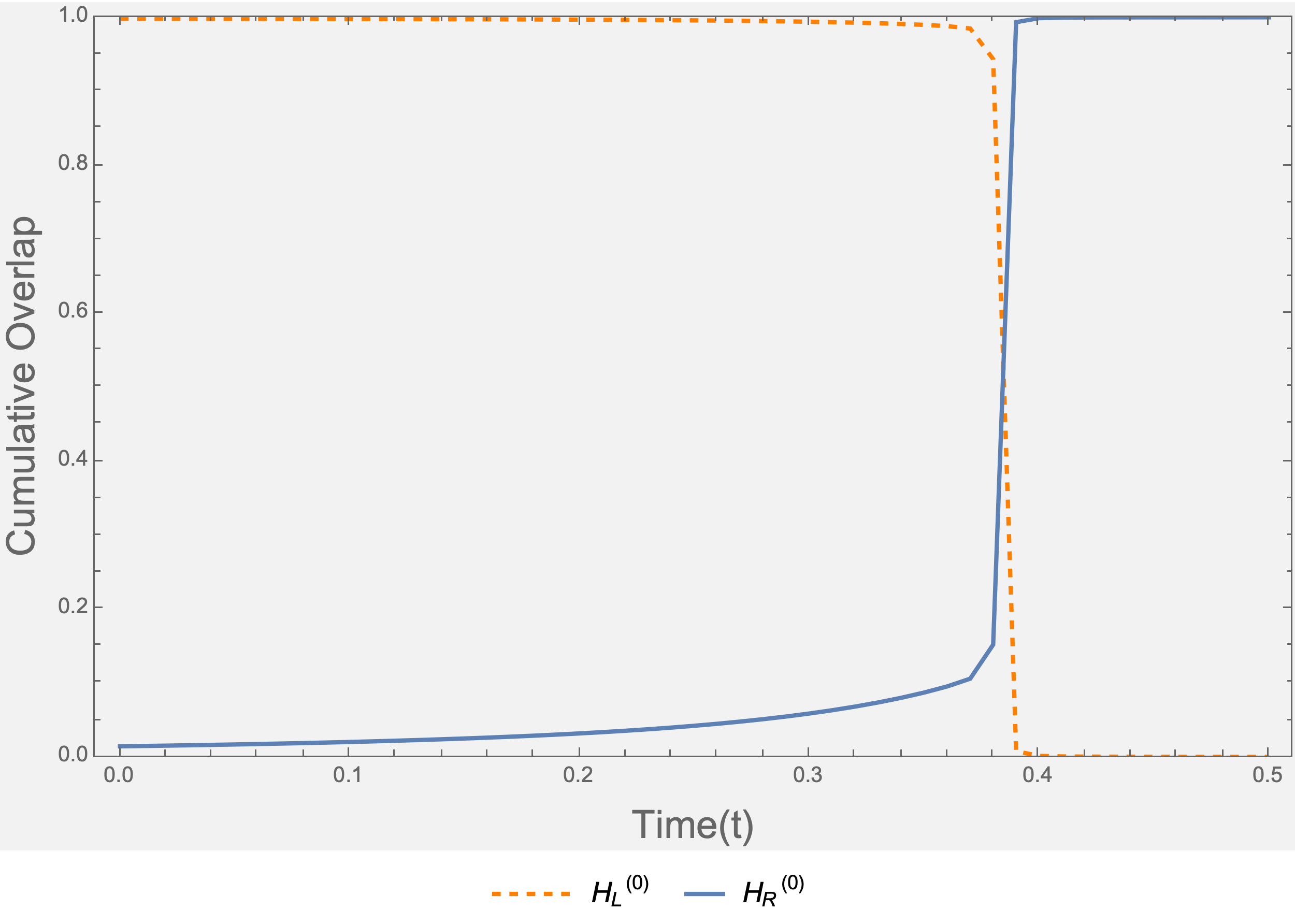}
  \caption{Failure of structural steering when $\Jzz = 1000 \gg \Jzzsteer$.}
\end{subfigure}
\caption{
Stage~1 evolution for a disjoint-structure graph with $m=10$, $m_g=m_r=15$, $n_c=9$. Here we use $\Gamma_1=4 \Gamma_2$ with $\Gamma_2=m$.
Orange (dashed) indicates the $L$-region block \(H_L^{(0)}\), gray the lowest $R$-block \(H_R^{(0)}\), and blue the cumulative overlap with the lowest two $R$-blocks \(H_R^{(\leq 2)}\). 
In (a), by the end of Stage~1 (\( t = 0.5 \)), more than 90\% of the ground state amplitude is concentrated within \(H_R^{(\leq 2)}\), demonstrating effective structural steering. 
In (b), $\Jzz = 1000 \gg \Jzzsteer$, and thus $\Jxx \ll \Jxxsteer$. The ground state remains localized in the $L$-region until $t \approx 0.38$, then abruptly tunnels into the $R$-region, indicating an exponentially small-gap tunneling-induced anti-crossing.}
\label{fig:dis-steering}
\end{figure}

\begin{figure}[!htbp]
  \centering
  \includegraphics[width=0.5\textwidth]{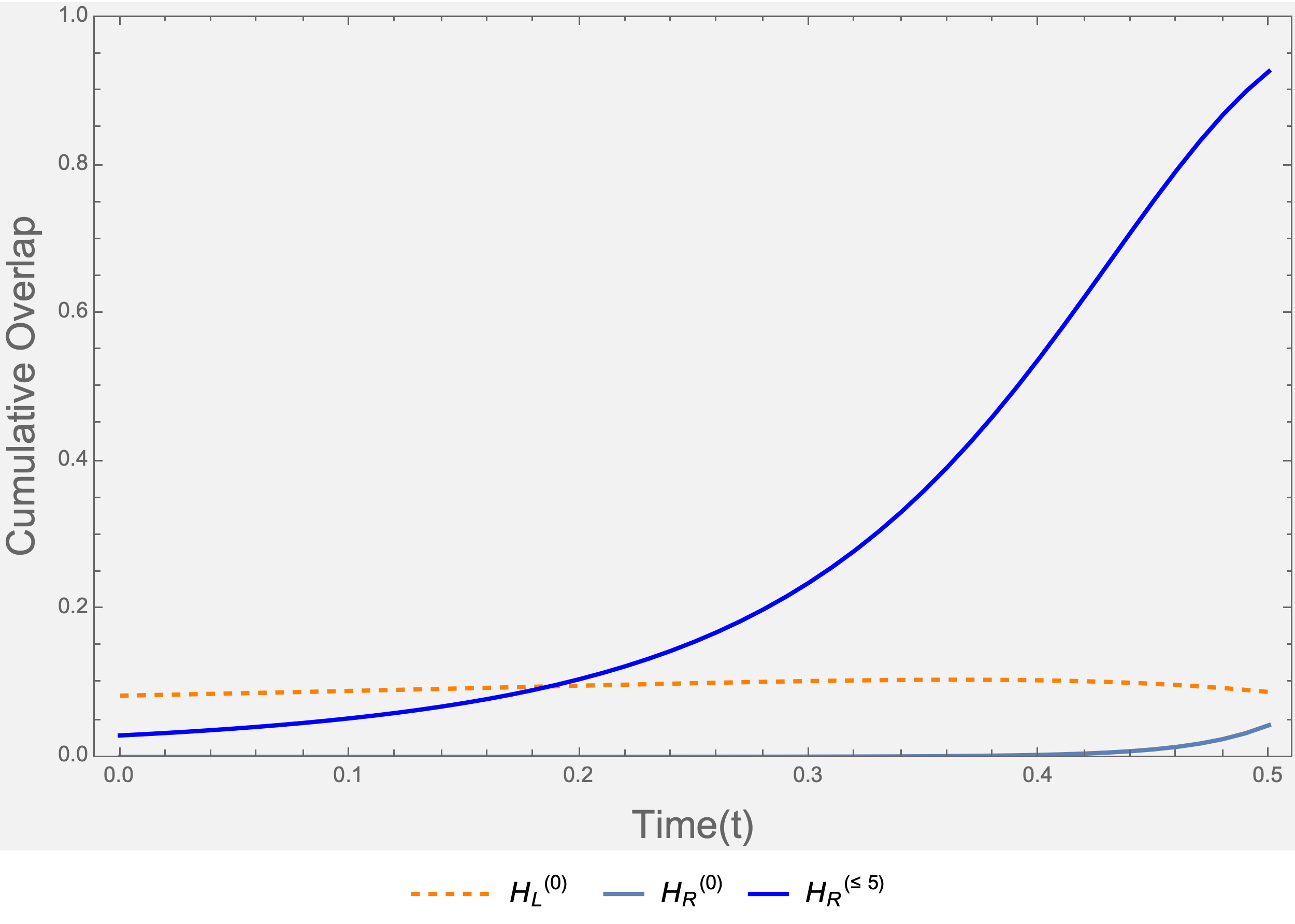}
  \caption{
    Stage~1 structural steering for a shared-structure graph with \( m = 30, m_r = 5, n_c = 9 \), and \( \Jxx = 2(m{-}1) \) and
    $\Jzz = 1 + \tfrac{\sqrt{n_c} + 1}{2}$. Here we use $\Gamma_1=4 \Gamma_2$ with $\Gamma_2=m$.
    The cumulative projection weight of the ground state onto the lowest energy-ordered blocks is shown as a function of time \(t\).
    Orange (dashed) indicates the $L$-region block \(H_L^{(0)}\), gray the lowest $R$-block \(H_R^{(0)}\), and blue the cumulative overlap with the lowest five $R$-blocks \(H_R^{(\leq 5)}\).
    By the end of Stage~1 (\( t = 0.5 \)), more than 90\% of the ground state amplitude is concentrated within \(H_R^{(\leq 5)}\), demonstrating effective structural steering even for small \( m_r \) relative to \( m \).
  }
  \label{fig:steering-m30}
\end{figure}

Overall, throughout Stage~1 the ground state evolves smoothly from a uniform superposition into a compressed state supported by the lowest few energy-ordered $R$-blocks, whenever $m_r \ge 2$. 
This continuous structural compression avoids abrupt transitions between disjoint regions of the Hilbert space, thereby preventing tunneling-induced anti-crossings and ensuring that the spectral gap remains large.
While a full proof of gap preservation is left for future work, these numerical results provide strong supporting evidence for the effectiveness and scalability of structural steering in Stage~1.

In summary, Stage~1 achieves smooth and scalable localization into the $R$-region without tunneling.

\subsubsection{Stage~2: Smooth Evolution to \GM{} via an Opposite-Sign Path}
\label{sec:stage2}

The purpose of Stage~2 is to show that the evolution proceeds from the \(R\)-localized region at the end of Stage~1 to \GM{} without the emergence of any new anti-crossing. 
We analyze this stage under two structural cases: the \emph{disjoint-structure graph} \( \Gdis \) and the \emph{shared-structure graph} \( \Gshare \).

The disjoint-structure case serves mainly for illustration and comparison: 
the evolution remains entirely within the same-sign block, and the dynamics proceed smoothly.
In the shared-structure case, the lowest opposite-sign block (AS0) participates in the evolution. 
By construction, the Sinking bound \(\Jxx \le \Jxxsink\) ensures that no new anti-crossing emerges, so the evolution remains smooth. 
Moreover, the dynamics can be explained by \emph{sign-generating interference}, which produces negative amplitudes in the computational basis of the ground state. 
Thus the system follows an opposite-sign path, a route not accessible in the stoquastic case.

\subsubsection*{Disjoint-Structure Case}

In this case, the evolution remains entirely within the same-sign block---more precisely, it proceeds from the \(R\)-localized region into the \( \mb{H}_R^{(0)} \) block during Stage~2.  
The spectral gap within \( \mb{H}_R^{(0)} \) is given analytically by \( \sqrt{1 + \mt{x}^2} \), 
and thus remains bounded below by 1 throughout Stage~2.  
Figure~\ref{fig:L1} compares the energy spectra of the same-sign block for \( \Jxx = 0 \) and \( \Jxx = 2(m{-}1) \).

Because the same-sign block admits a symmetric-subspace reduction, its dimension scales linearly with \(m\).  
This allows efficient numerical calculations on large instances.  
Figure~\ref{fig:L1-large} presents a large-scale example with \( m = 25 \) and \( m_g = 35 \), confirming that the algorithm continues to perform correctly at scale.

In this disjoint-structure case, there is no requirement for the Sinking bound.  
Large values of \( \Jxx \) are also admissible, in which case double block-level true crossings occur but are dynamically bypassed.  
See Figure~\ref{fig:L1-DAC}.

\begin{figure}[!htbp]
  \centering
  \begin{subfigure}[b]{0.48\textwidth}
    \includegraphics[width=\textwidth]{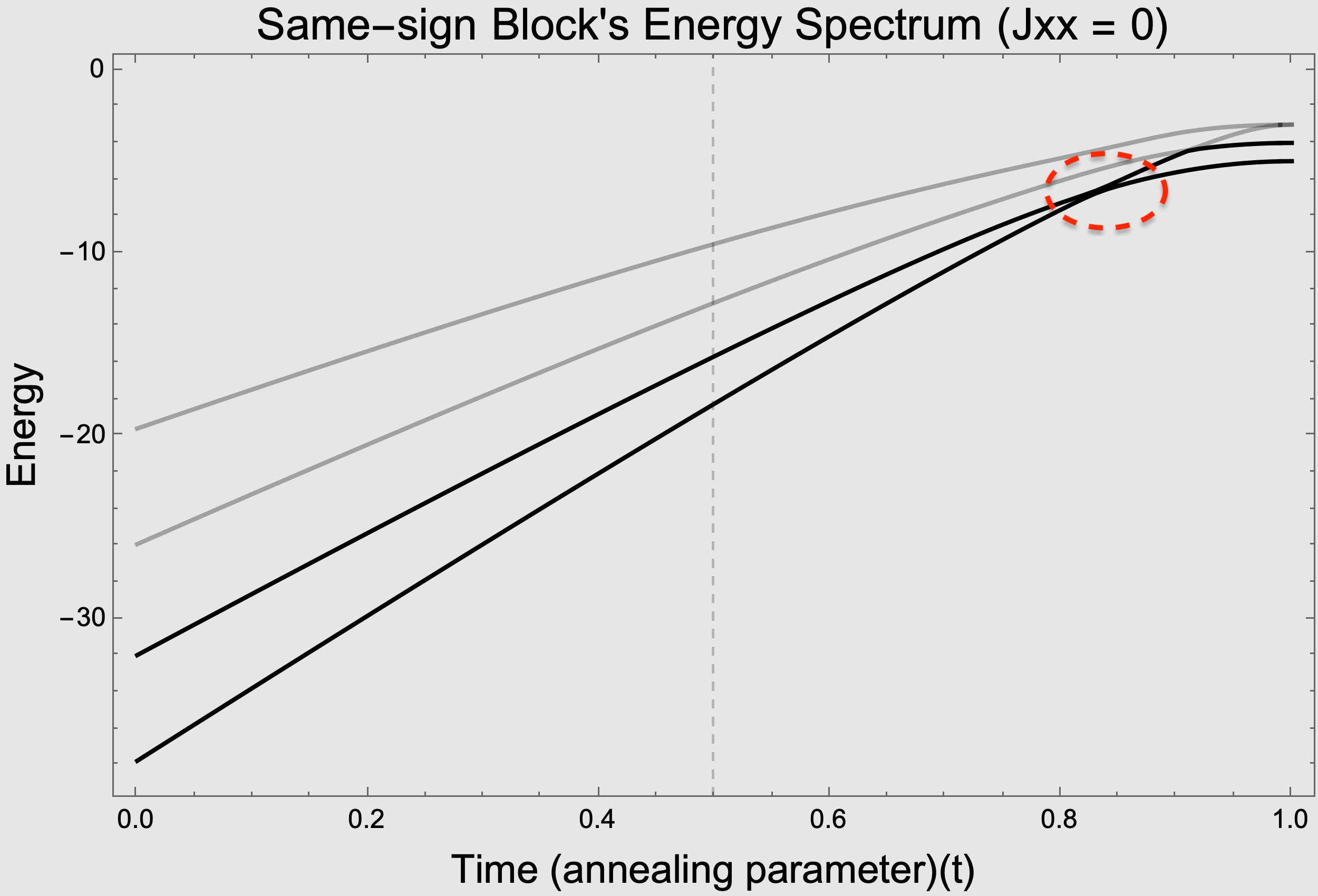}
    \caption{}
  \end{subfigure}
  \hfill
  \begin{subfigure}[b]{0.48\textwidth}
    \includegraphics[width=\textwidth]{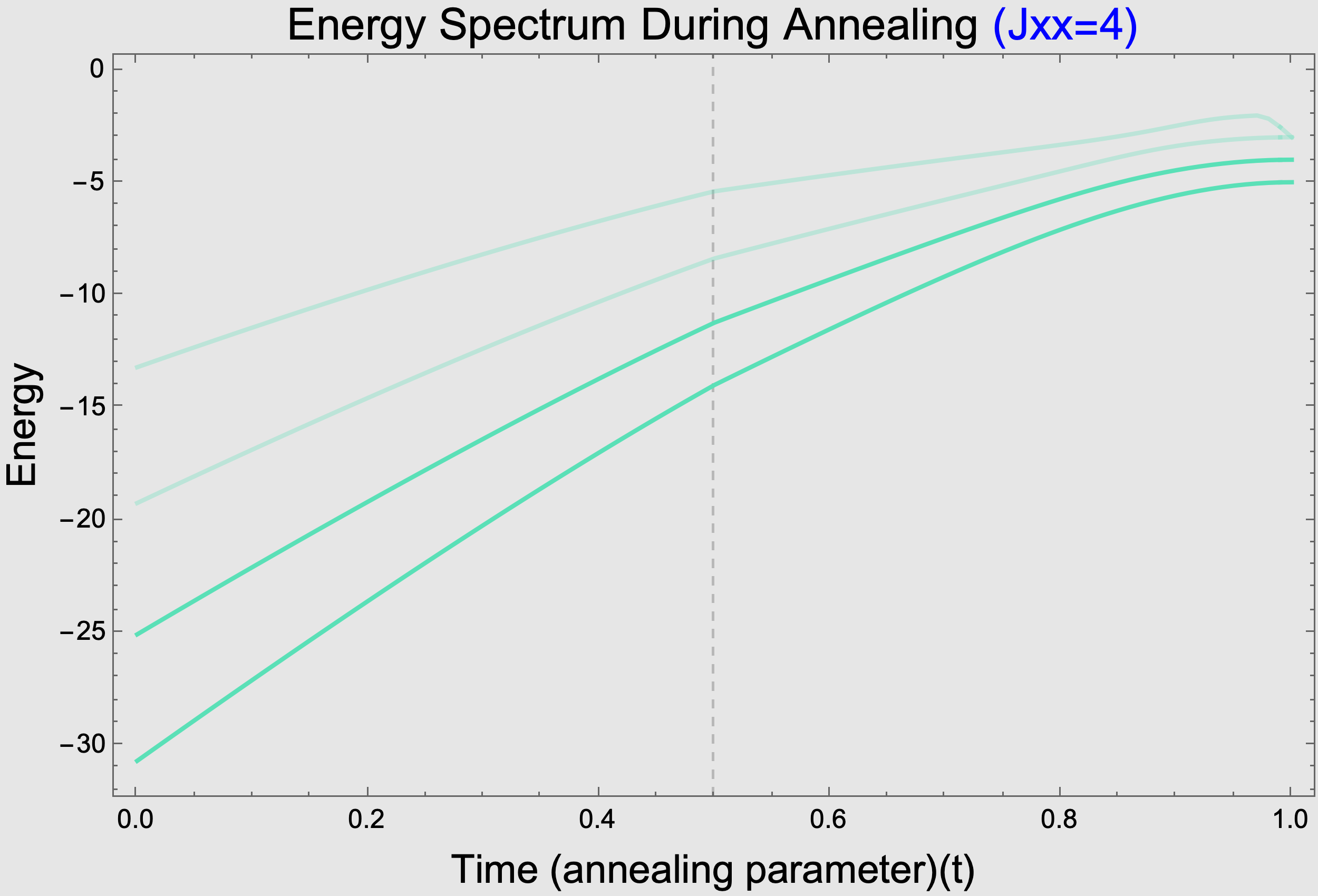}
    \caption{}
  \end{subfigure}

  \vspace{1em}

  \begin{subfigure}[b]{0.48\textwidth}
    \includegraphics[width=\textwidth]{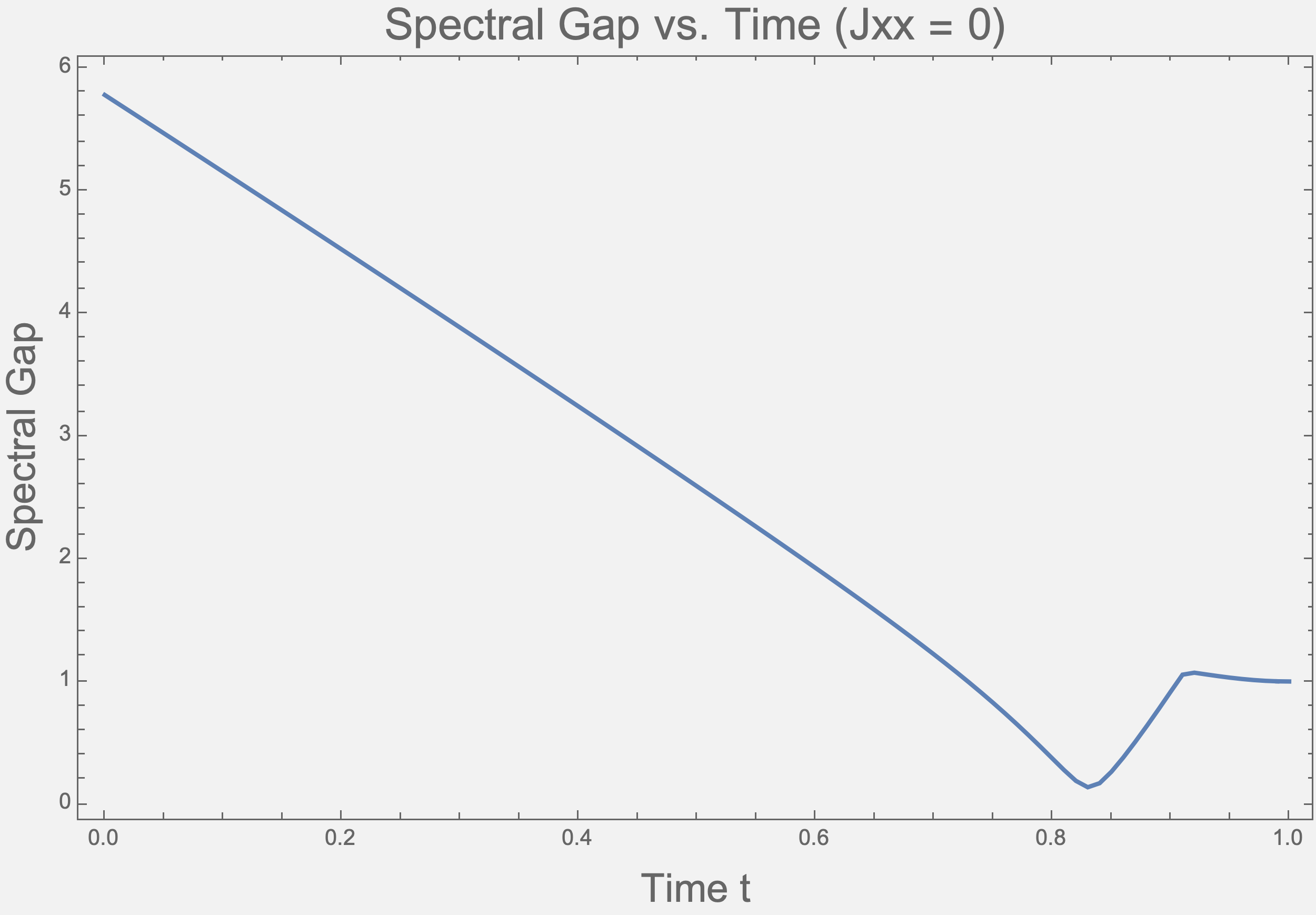}
    \caption{}
  \end{subfigure}
  \hfill
  \begin{subfigure}[b]{0.48\textwidth}
    \includegraphics[width=\textwidth]{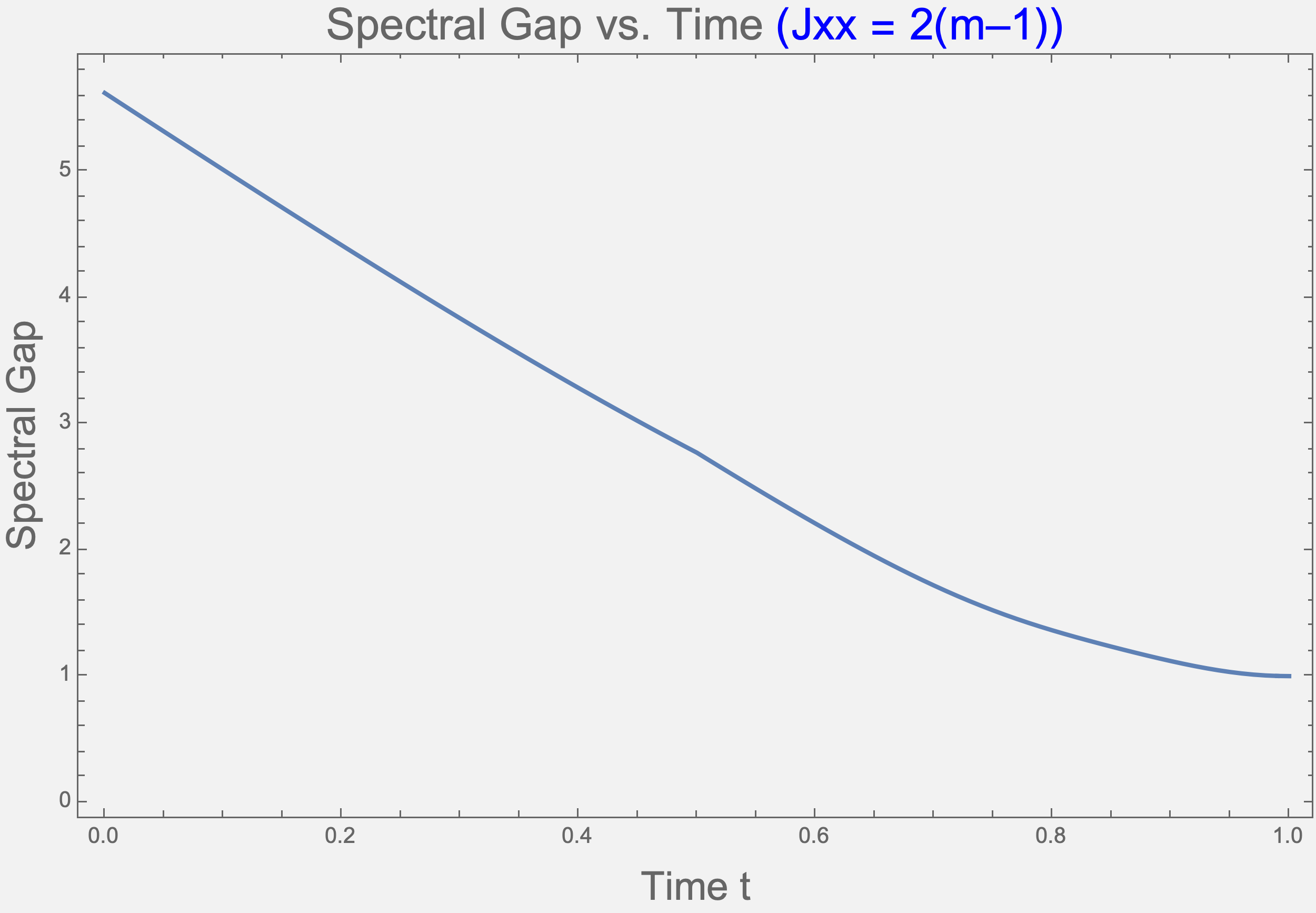}
    \caption{}
  \end{subfigure}
\caption{
Energy spectra for the disjoint-structure graph \(\Gdis\), with \( m = m_l = 3 \), \( m_g = 5 \), and clique size \( n_c = 9 \).\\
(a) TFQA spectrum showing a small-gap anti-crossing; the corresponding gap profile is shown in (c).\\
(b) \DDD{} spectrum with \(\XX\)-driver at \(\Jxx = 2(m{-}1) = 4\); the corresponding gap profile is shown in (d).
}
\label{fig:L1}
\end{figure}

\begin{figure}[!htbp]
  \centering
  \begin{subfigure}[b]{0.48\textwidth}
    \includegraphics[width=\textwidth]{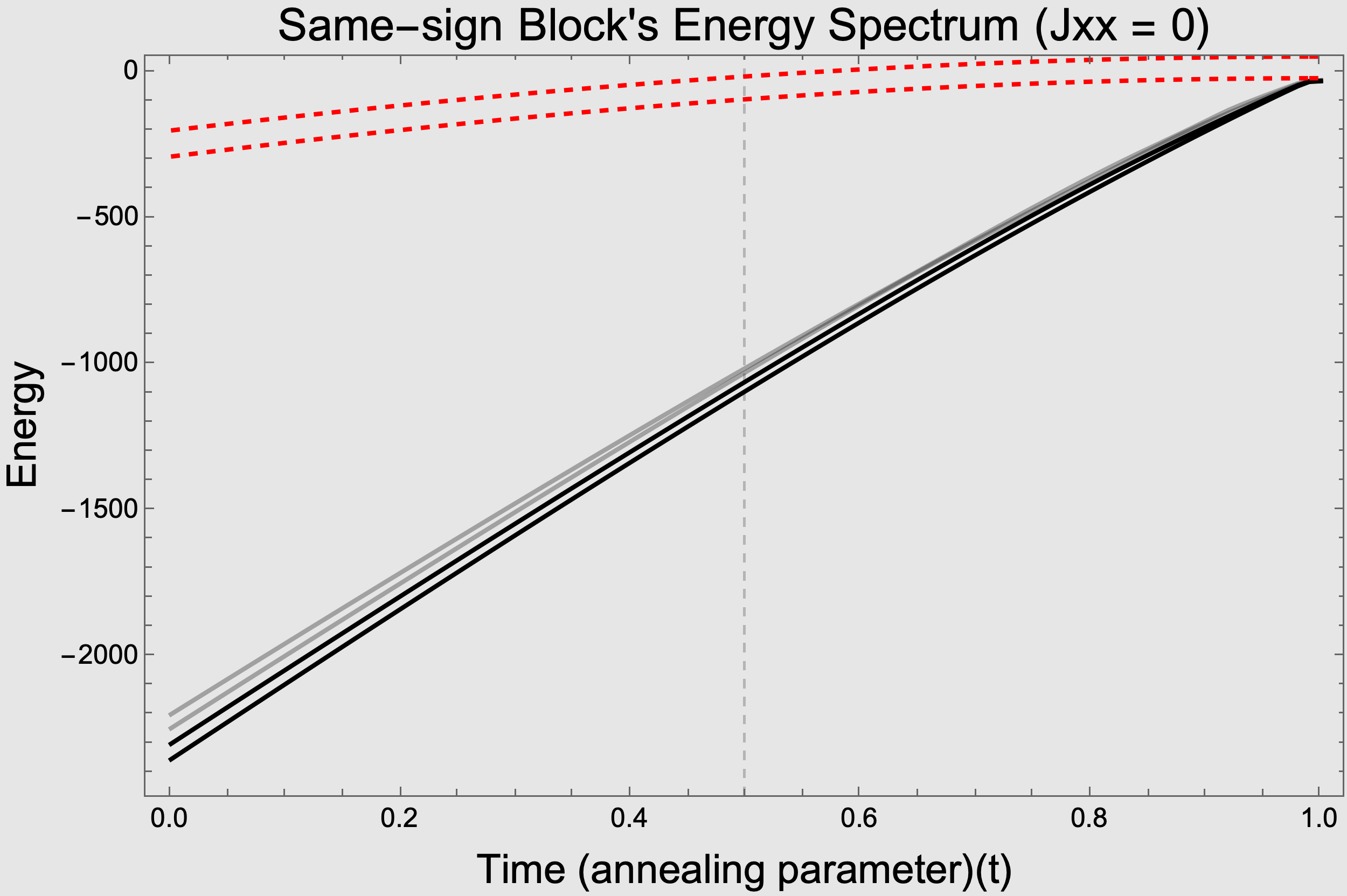}
    \caption{}
  \end{subfigure}
  \hfill
  \begin{subfigure}[b]{0.48\textwidth}
    \includegraphics[width=\textwidth]{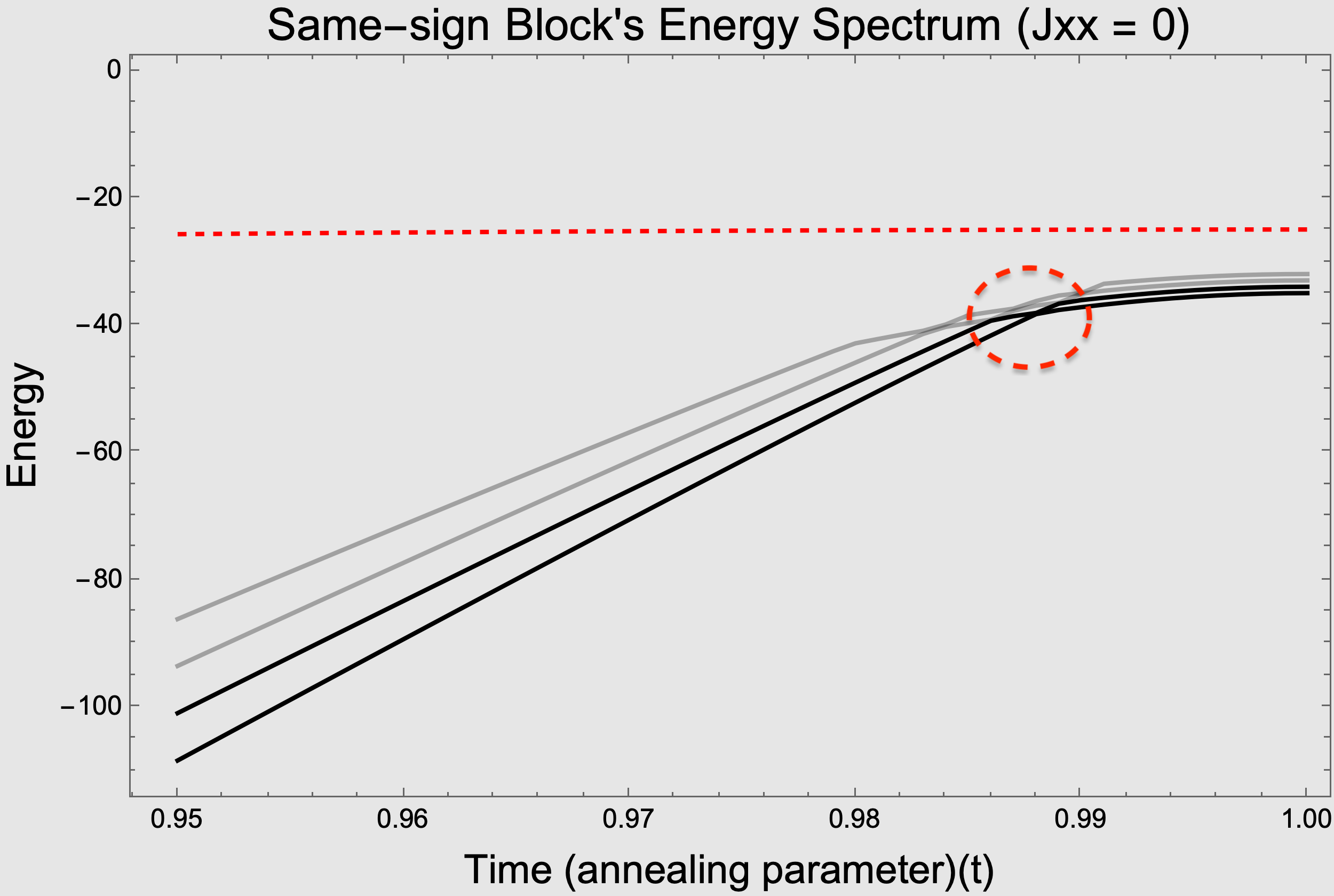}
    \caption{}
  \end{subfigure}

  \vspace{1em}

  \begin{subfigure}[b]{0.48\textwidth}
    \includegraphics[width=\textwidth]{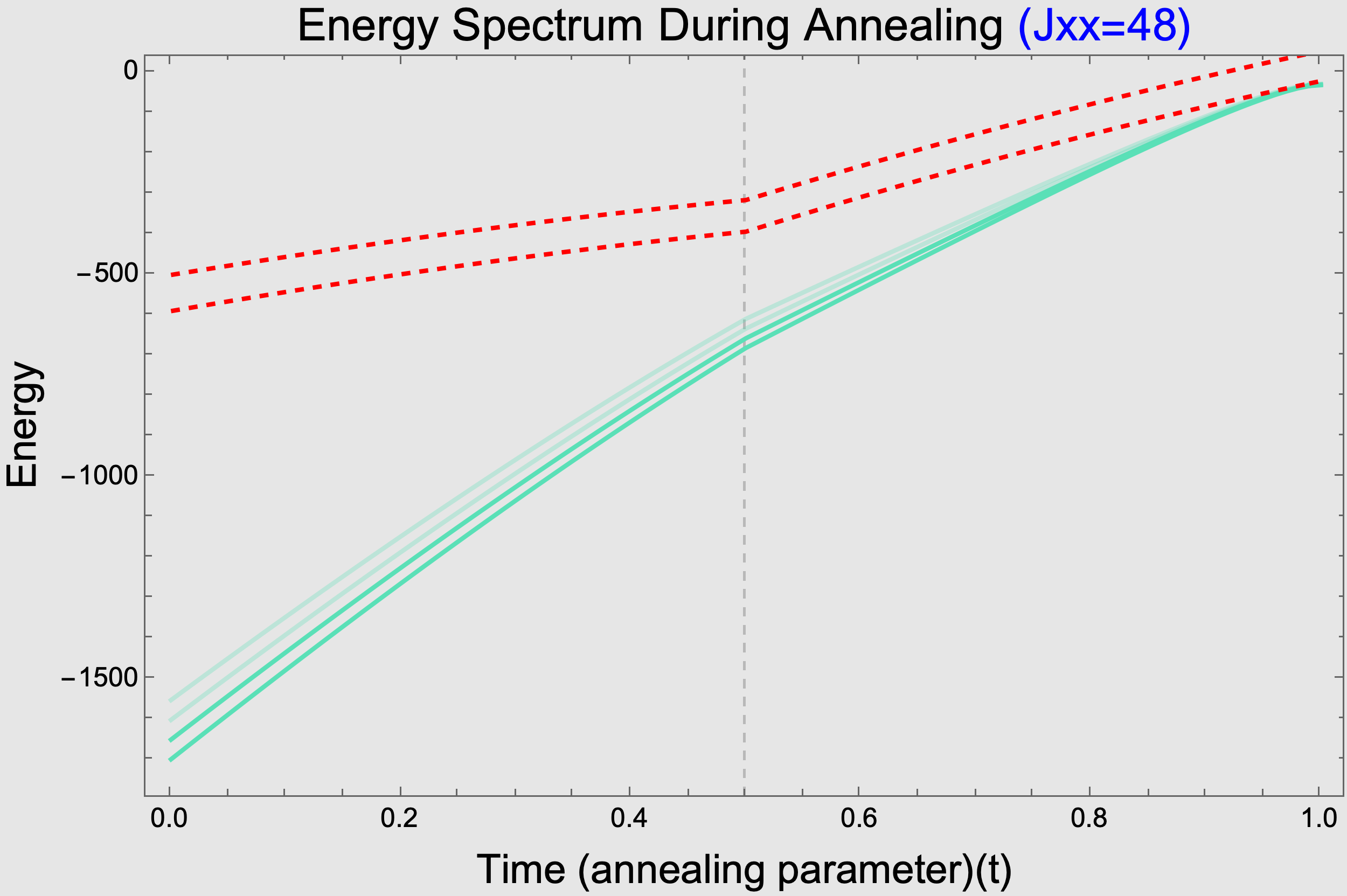}
    \caption{}
  \end{subfigure}
  \hfill
  \begin{subfigure}[b]{0.48\textwidth}
    \includegraphics[width=\textwidth]{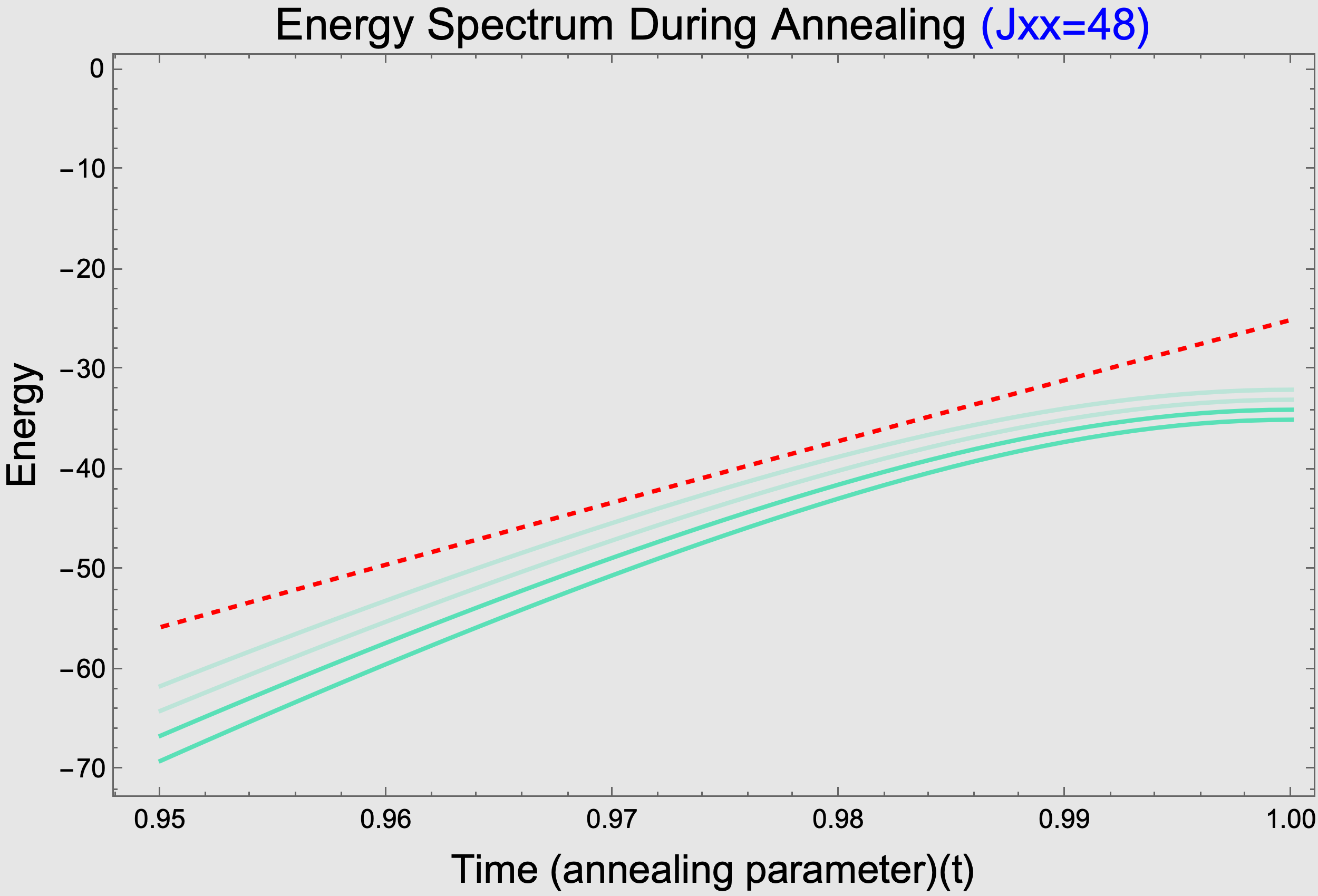}
    \caption{}
  \end{subfigure}
\caption{
Spectral decomposition for the disjoint-structure graph \(\Gdis\), with \( m = m_l = 25 \), \( m_g = 35 \), and clique size \( n_c = 9 \).\\
(a) Energy spectrum of the same-sign block under TFQA (\(\Jxx = 0\)), showing an anti-crossing near \( t \approx 0.985 \). \\
(b) Zoom of panel (a) over \( t \in [0.95, 1] \). \\
(c) Full spectrum under \DDD{} with \(\Jxx = 2(m{-}1) = 48\): the opposite-sign energies (dashed red) are lowered by the \(\XX\)-driver but remain above the same-sign ground state (green), confirming successful evolution with no anti-crossing. \\
(d) Zoom of panel (c) over \( t \in [0.95, 1] \).
}
\label{fig:L1-large}
\end{figure}

\begin{figure}[!htbp]
  \centering
  \begin{subfigure}[b]{0.48\textwidth}
    \includegraphics[width=\textwidth]{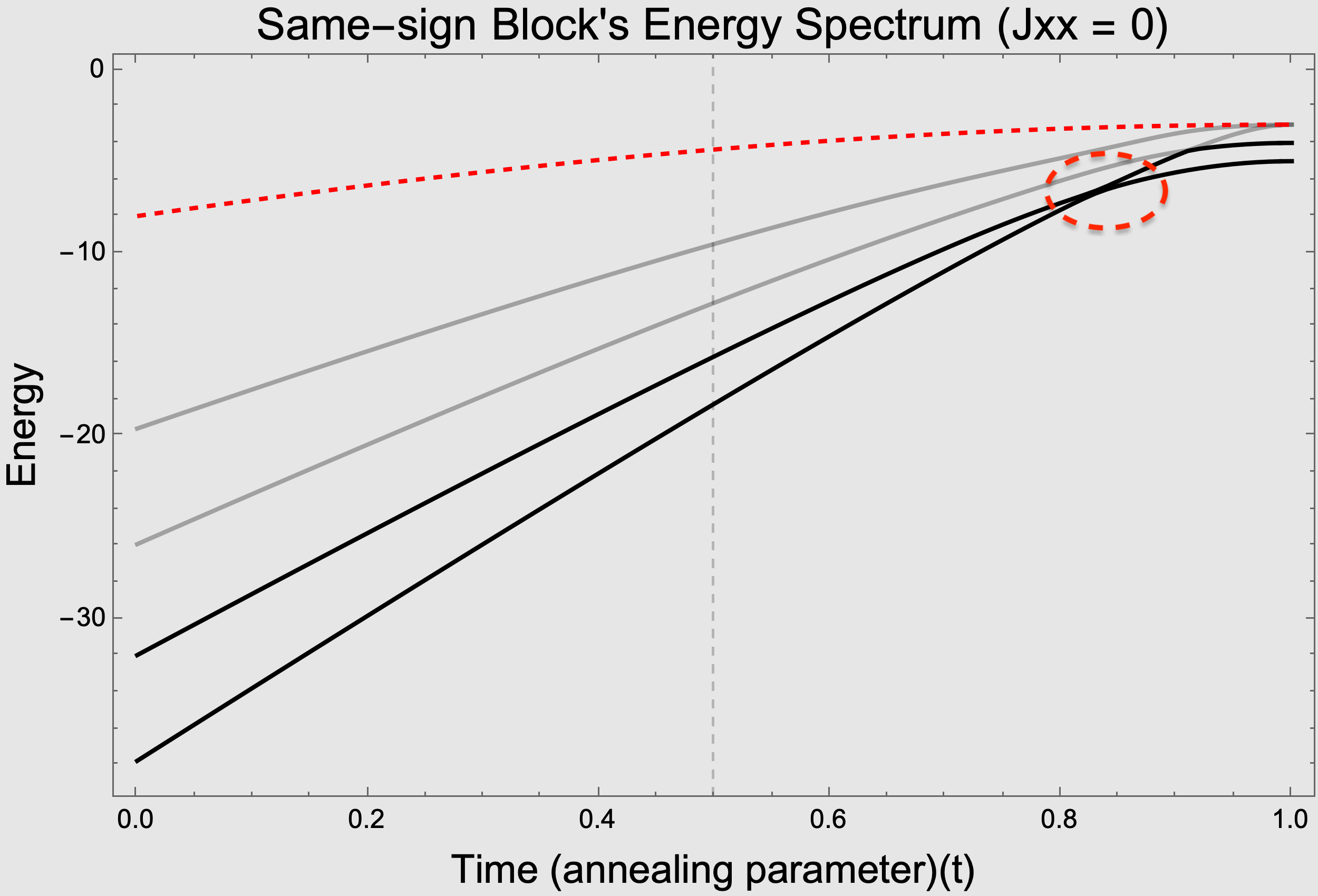}
    \caption{}
  \end{subfigure}
  \hfill
  \begin{subfigure}[b]{0.48\textwidth}
    \includegraphics[width=\textwidth]{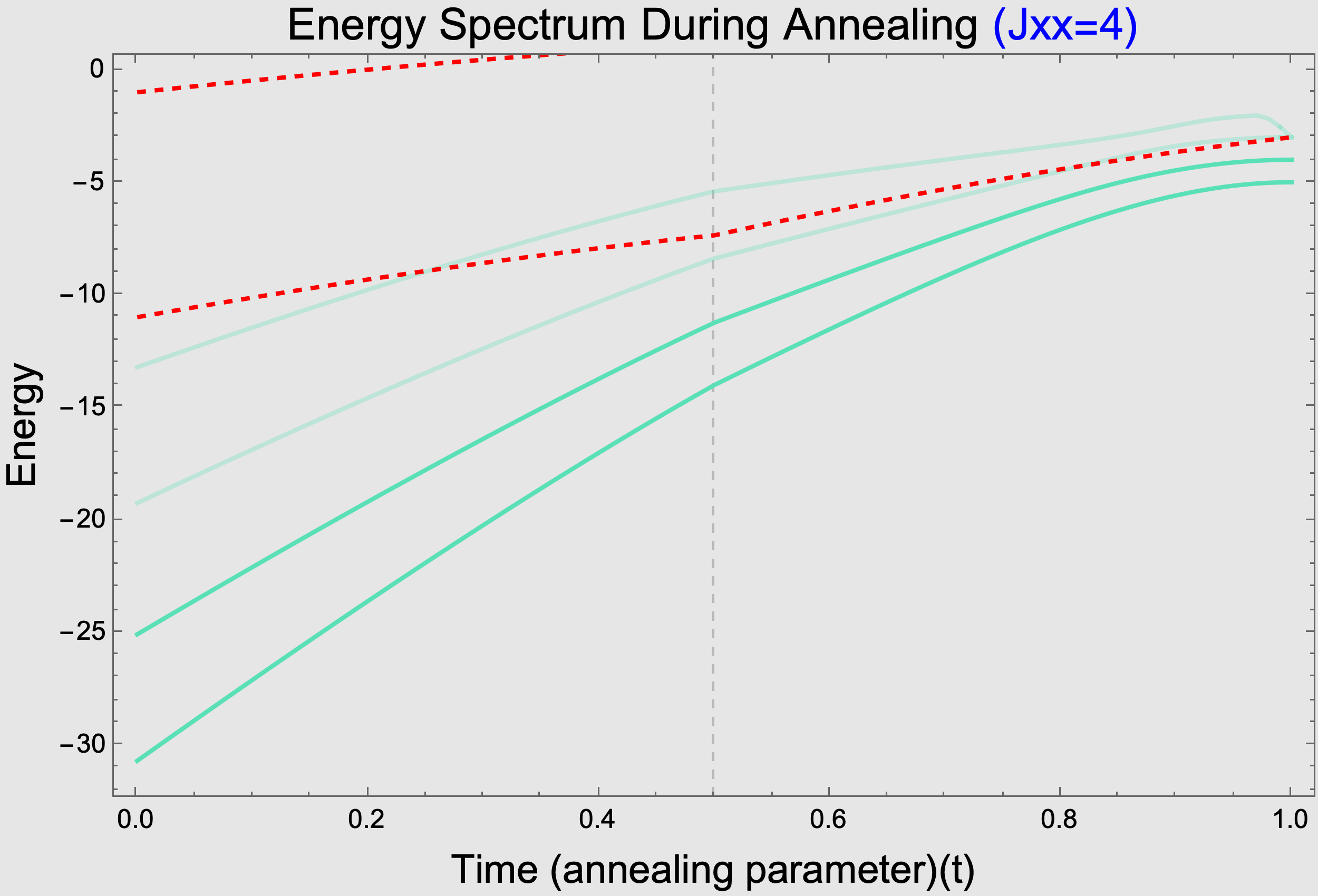}
    \caption{}
  \end{subfigure}

  \vspace{1em}

  \begin{subfigure}[b]{0.48\textwidth}
    \includegraphics[width=0.8\textwidth]{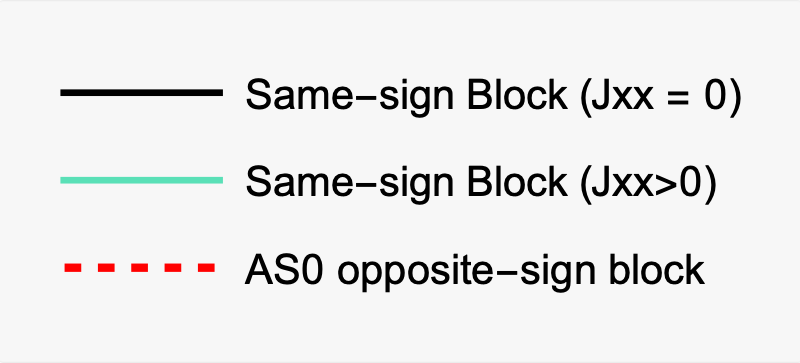}
    \caption{}
  \end{subfigure}
  \hfill
  \begin{subfigure}[b]{0.48\textwidth}
    \includegraphics[width=\textwidth]{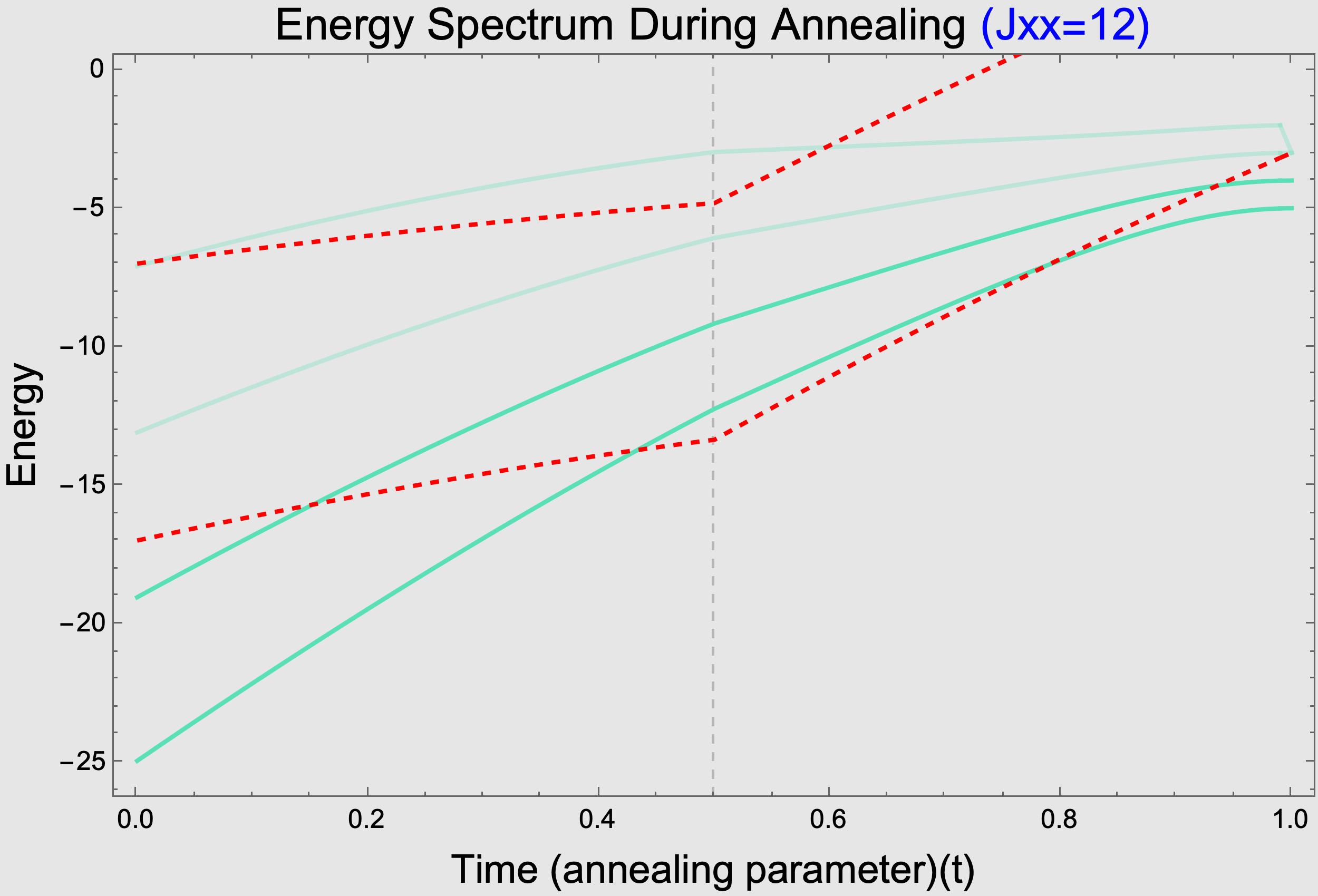}
    \caption{}
  \end{subfigure}
\caption{
Spectral decomposition for the disjoint-structure graph \(\Gdis\), with \( m = m_l = 3 \), \( m_g = 5 \), and clique size \( n_c = 9 \). \\
(a) Energy spectrum of the same-sign block under TFQA (\(\Jxx = 0\)), showing an anti-crossing near \( t \approx 0.85 \). \\
(b) Full spectrum under \DDD{} with \(\Jxx = 4\): the opposite-sign energies (dashed red) remain above the same-sign ground state (green), confirming successful evolution with no anti-crossing. \\
(c) Legend identifying curve types across panels. \\
(d) Large-\(\Jxx\) case (\(\Jxx = 12\)):
two true block-level crossings occur as the lowest opposite-sign energy drops below the same-sign ground state energy
near \( t \approx 0.4 \) and again near \( t \approx 0.7 \).
}
\label{fig:L1-DAC}
\end{figure}

\subsubsection*{Shared-Structure Case}
We now analyze the shared-structure case in detail.  
The goal is to show explicitly how the participation of opposite-sign blocks leads to both positive and negative amplitudes in the ground state.  
Our approach is twofold: first we establish the sign pattern mathematically, and then we interpret it physically in terms of \emph{sign-generating quantum interference}.
We begin with the mathematical analysis.

\subsubsection*{Emergence of Negative Amplitudes}
In the angular-momentum basis, Stage~2 involves the lowest $R$-supporting block together with the opposite-sign blocks that couple to it.
We consider the following restricted subspace in the computational basis (transforming back), which consists of two parts:
\begin{itemize}
  \item $\mc{M}$: the set of $\GM$-supporting states containing $R$,  
  \(
  \mc{M} \;=\; \{\, M \subset \GM : R \subset M \,\}.
  \)
  \item $\mc{D}$: the set of dependent-set states paired to $\mc{M}$ via $\XX$-couplers.  
For each $M \in \mc{M}$ and each $b \in M \setminus R$, there is a unique partner $a$ such that $(a,b)$ is an $\XX$-coupled pair.  
The corresponding dependent state $D_{M,b} \in \mc{D}$ is obtained by flipping $(a,b)$.  
Thus each $M$ has $|M\setminus R|$ associated dependent-set states, while each $D_{M,b}$ is uniquely coupled to its parent $M$.
\end{itemize}

Restricting the Hamiltonian to this subspace yields the block form
\[
\mb{H}_{\ms{eff}}^{(\mc{M}\mc{D})} \;=\;
\begin{pmatrix}
\mb{H}_{\mc{M}} & \mb{V} \\
\mb{V}^\top & \mb{H}_{\mc{D}}
\end{pmatrix},
\]
where $\mb{H}_{\mc{M}}$ and $\mb{H}_{\mc{D}}$ are both stoquastic (off-diagonals $\le 0$),
with $\mb{H}_{\mc{D}}$ has strictly positive diagonal entries,
and $\mb{V}$ consists of only nonnegative entries $\mt{jxx}$ (as each $D\in\mc{D}$ is coupled only to one $M\in\mc{M}$).

To identify the ground-state structure, we apply the Rayleigh variational principle to a trial vector 
\((u_M,u_D)\), 
with \(u_M \ge 0\) on \(\mc{M}\).  
The minimum is attained when
\(
u_D \;=\; -\,H_{\mc{D}}^{-1}\mb{V}^\top u_M.
\)
Since $\mb{V}^\top u_M\ge 0$ and $H_{\mc{D}}^{-1}\ge 0$ entrywise,
the Stage~2 ground state necessarily takes the form
$(u_M,\,u_D)$, with  $u_M \ge 0$ and  $u_D \le 0$.
In other words, the ground state carries positive amplitudes on the $\GM$-supporting sector and negative amplitudes on the dependent-set sector,
establishing the structural origin of negative components in the Stage~2 ground state.
This analytical prediction is confirmed numerically: Figure~\ref{fig:neg-amp} plots the total fraction of negative amplitudes as a function of time.

\begin{figure}[!htbp]
  \centering
  \includegraphics[width=0.5\textwidth]{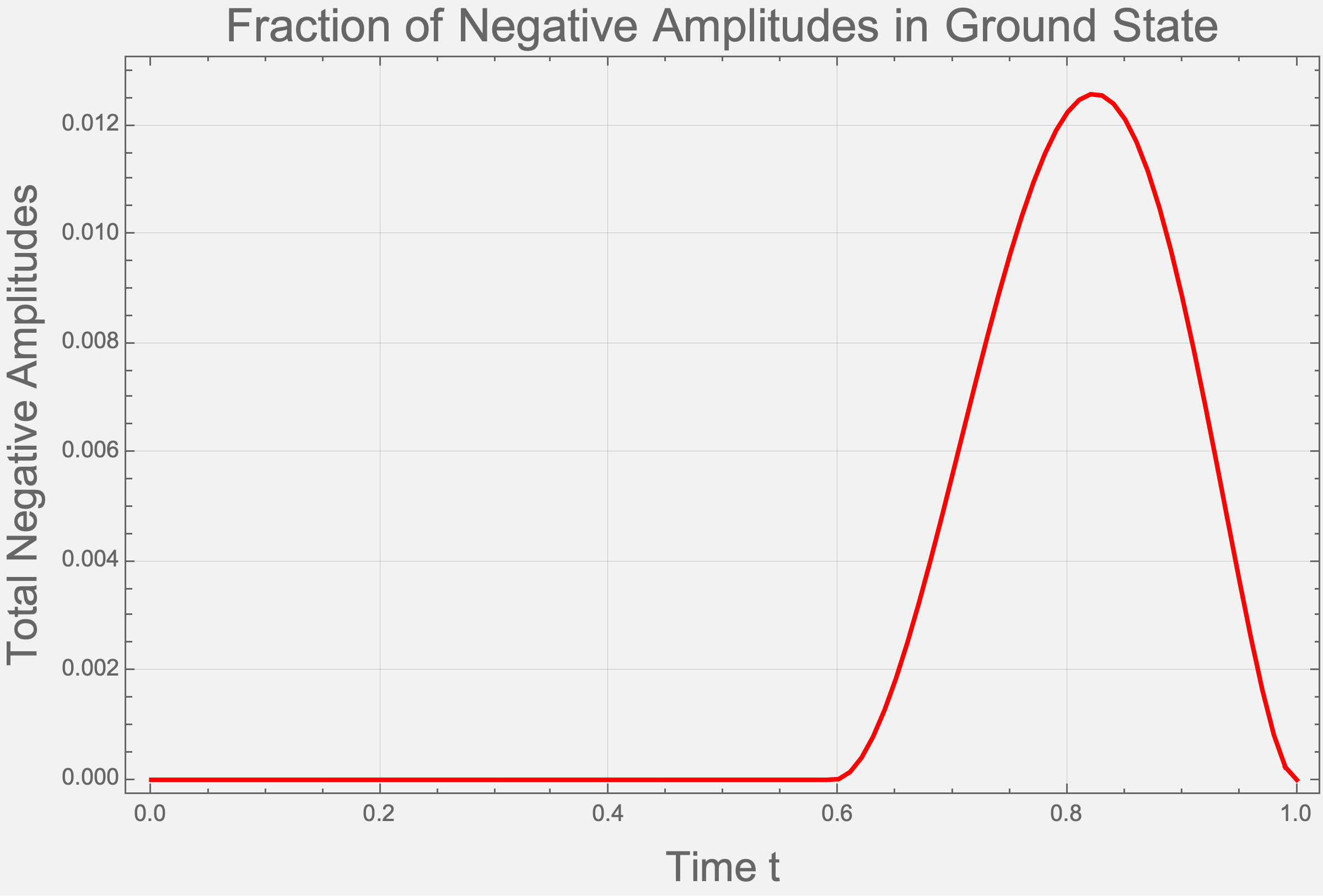}
  \caption{
    \textbf{Emergence of Negative Amplitudes during Stage~2.}
    The plot shows the total fraction of negative amplitudes in the ground state as a function of annealing time \( t \).  
    Throughout Stage~1, the amplitude distribution remains strictly non-negative, consistent with confinement to the same-sign block.  
    As Stage~2 begins (\( t = 0.5 \)), negative amplitudes emerge around \( t \approx 0.6 \).  
%    indicating the inclusion of opposite-sign basis states and the onset of quantum interference across blocks.  
  }
  \label{fig:neg-amp}
\end{figure}

\paragraph{Interpretation as sign-generating interference.}  
The negative amplitudes identified above can be understood in terms of
constructive and destructive interference between same-sign and opposite-sign components.
This provides the physical mechanism we call \emph{sign-generating quantum interference}.  

We illustrate this mechanism in detail using the three-vertex conceptual model (V3) in Section~\ref{sec:V3},  
which may also serves as a minimal testbed for probing sign-generating interference experimentally.

\paragraph{Numerical analysis.}
Figure~\ref{fig:L2} compares the energy spectra for the shared-structure case under TFQA and \DDD{}.  
In TFQA $(\Jxx = 0)$, a small-gap anti-crossing is visible, while in \DDD{} with $\Jxx = 2(m{-}1)$ this anti-crossing is dissolved.  
The mechanism is clarified in Figure~\ref{fig:L2-analysis}, which decomposes the spectrum into same-sign and opposite-sign blocks:  
the same-sign component is lifted while the opposite-sign component is lowered, producing destructive interference that avoids any new anti-crossing.  
Finally, Figure~\ref{fig:L2-fail} illustrates the large-$\Jxx$ regime, where double anti-crossings appear and may cause failure.

In summary, Stage~2 succeeds through sign-generating interference along an opposite-sign path, whereby the tunneling-induced anti-crossing is dissolved without introducing new ones.

\begin{figure}[!htbp]
  \centering
  \begin{subfigure}[b]{0.48\textwidth}
    \includegraphics[width=\textwidth]{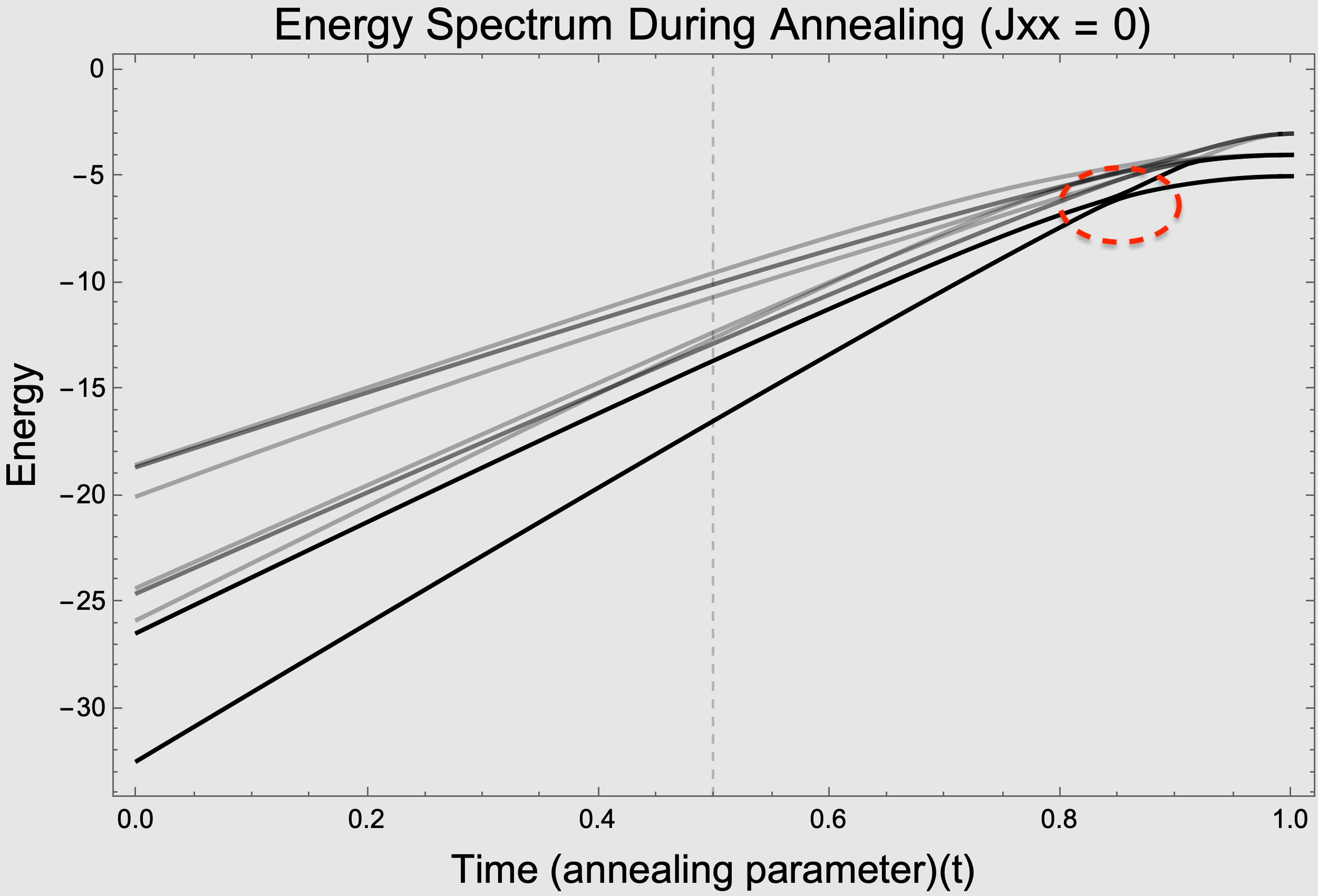}
    \caption{}
  \end{subfigure}
  \hfill
  \begin{subfigure}[b]{0.48\textwidth}
    \includegraphics[width=\textwidth]{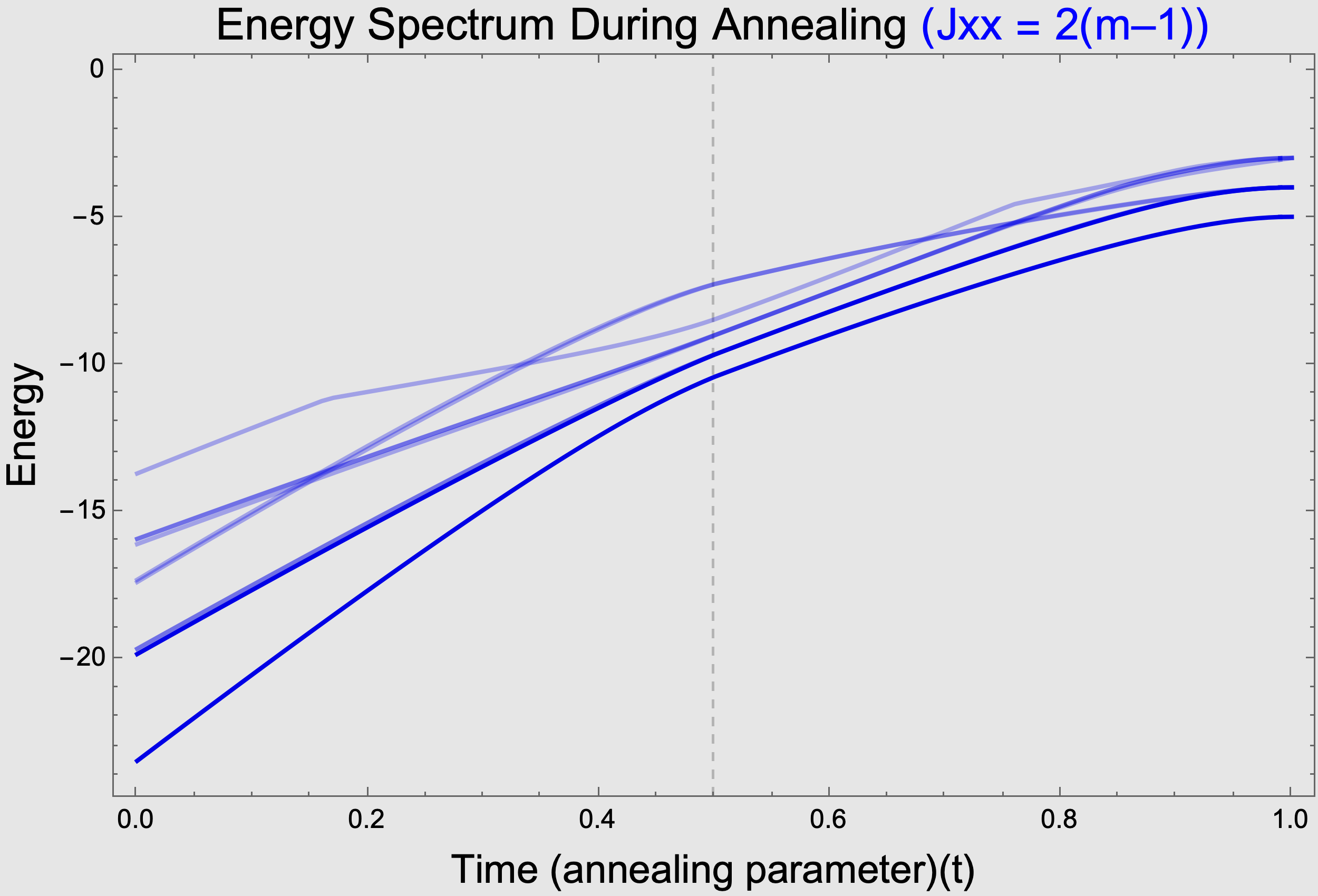}
    \caption{}
  \end{subfigure}

  \vspace{1em}

  \begin{subfigure}[b]{0.48\textwidth}
    \includegraphics[width=\textwidth]{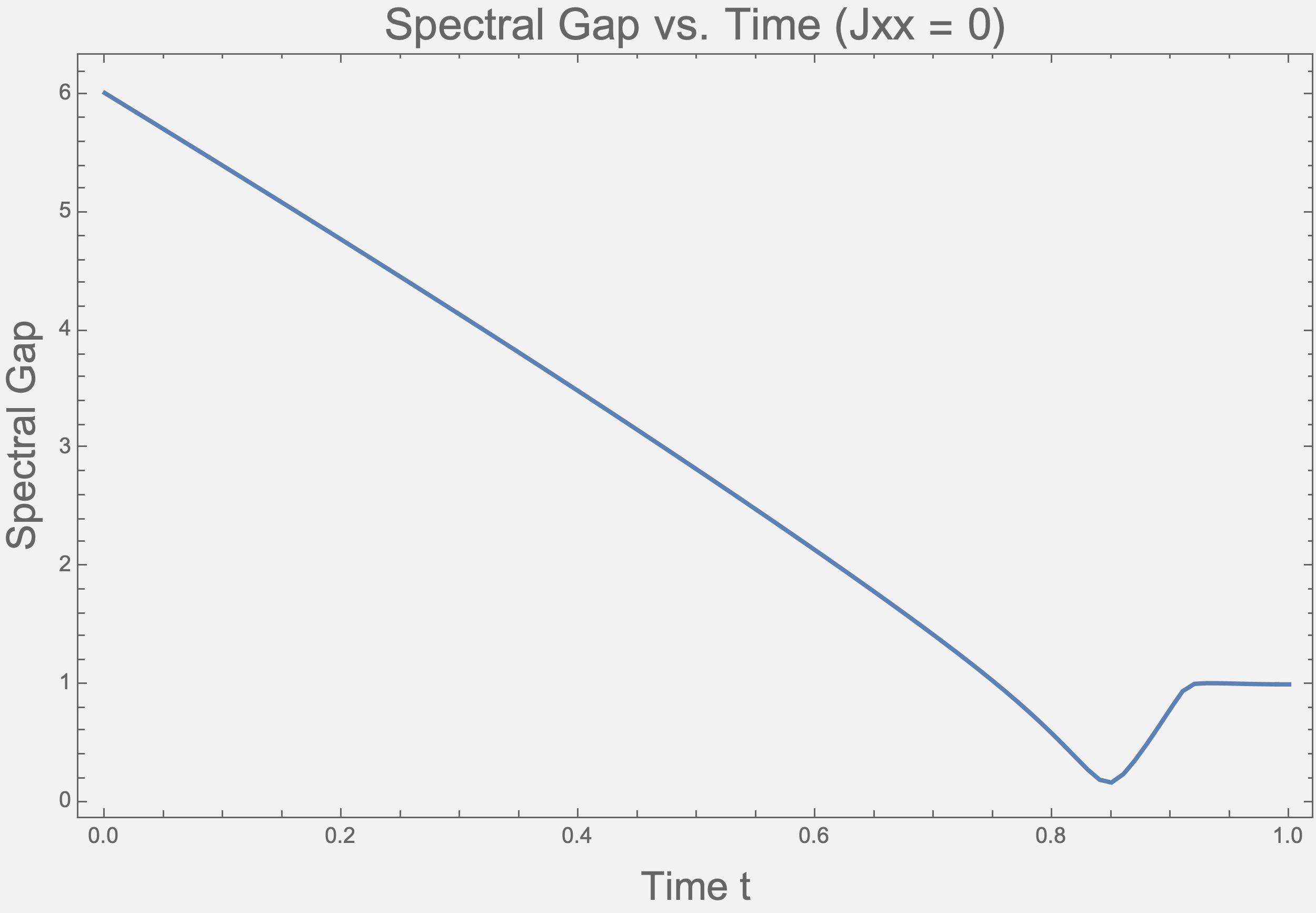}
    \caption{}
  \end{subfigure}
  \hfill
  \begin{subfigure}[b]{0.48\textwidth}
    \includegraphics[width=\textwidth]{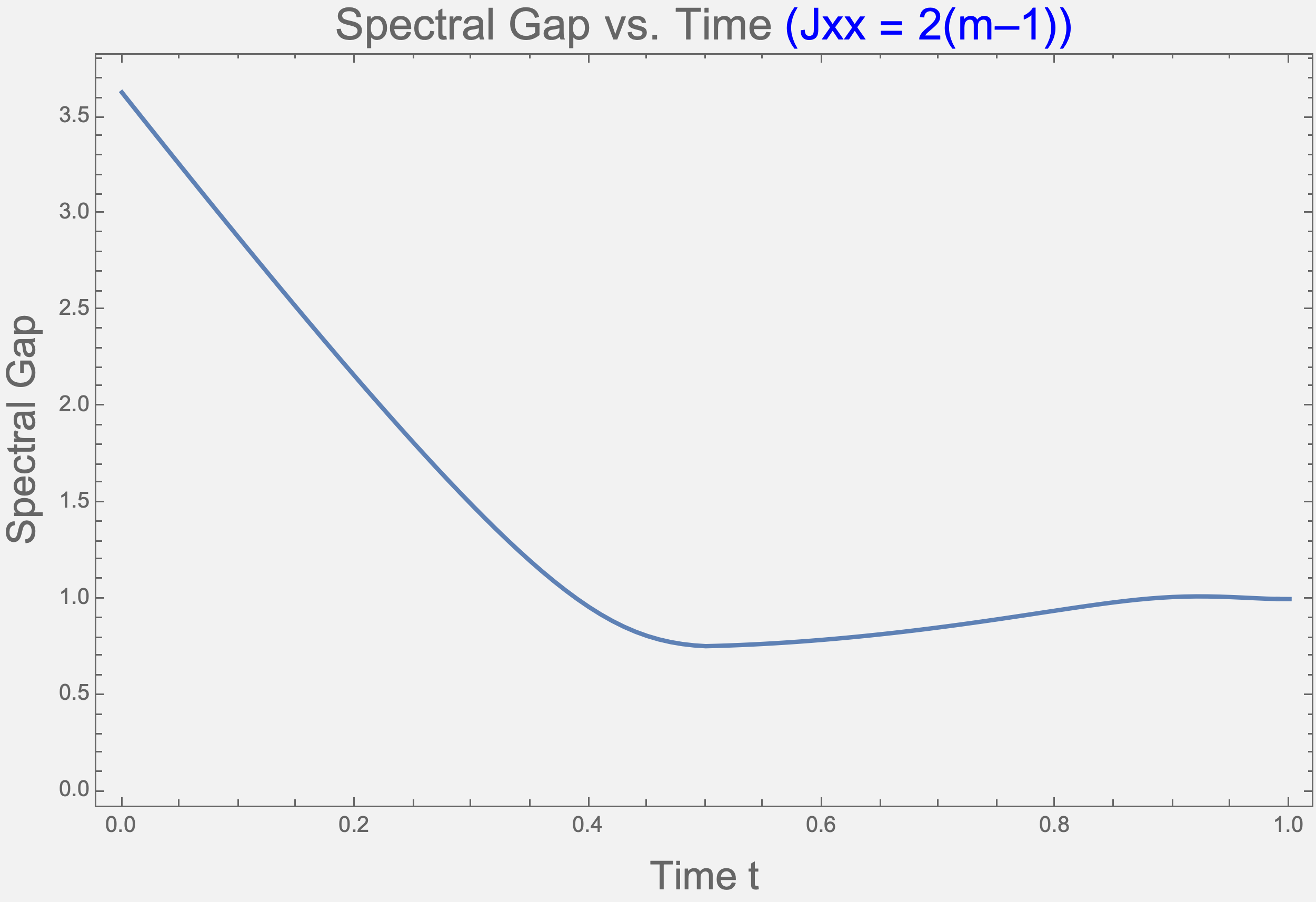}
    \caption{}
  \end{subfigure}
\caption{
Energy spectra for the shared-structure graph \(\Gshare\), with \( m = m_l = 3 \), \( m_r = 2 \), \( m_g = 5 \), and clique size \( n_c = 9 \).\\
(a) TFQA spectrum showing a small-gap anti-crossing; the corresponding gap profile is shown in (c).\\
(b) \DDD{} spectrum with \(\XX\)-driver at \( \Jxx = 2(m{-}1) \); the corresponding gap profile is shown in (d).
}
\label{fig:L2}
\end{figure}

\begin{figure}[!htbp]
  \centering
  \begin{subfigure}[b]{0.48\textwidth}
    \includegraphics[width=\textwidth]{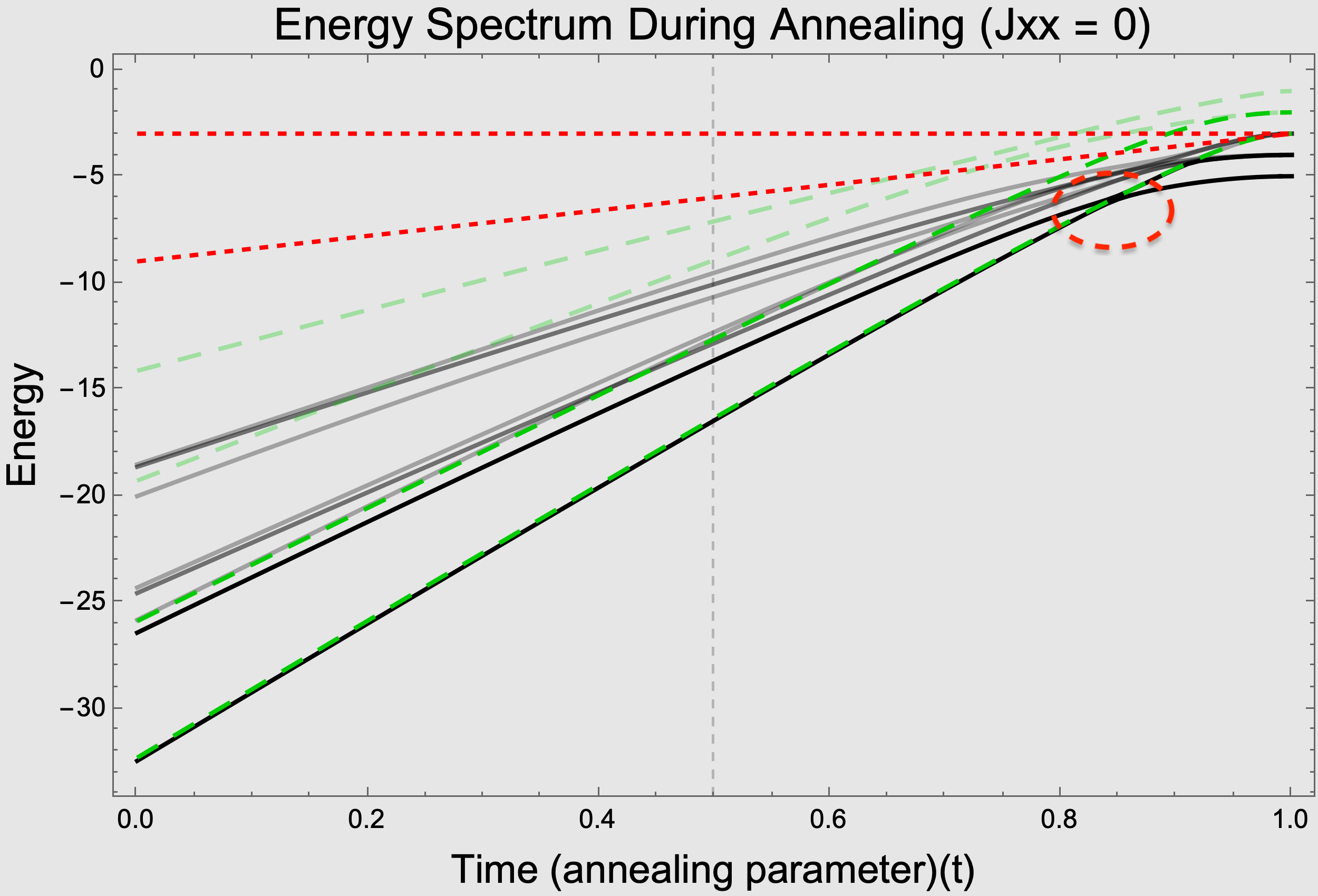}
    \caption{}
  \end{subfigure}
  \hfill
  \begin{subfigure}[b]{0.48\textwidth}
    \includegraphics[width=\textwidth]{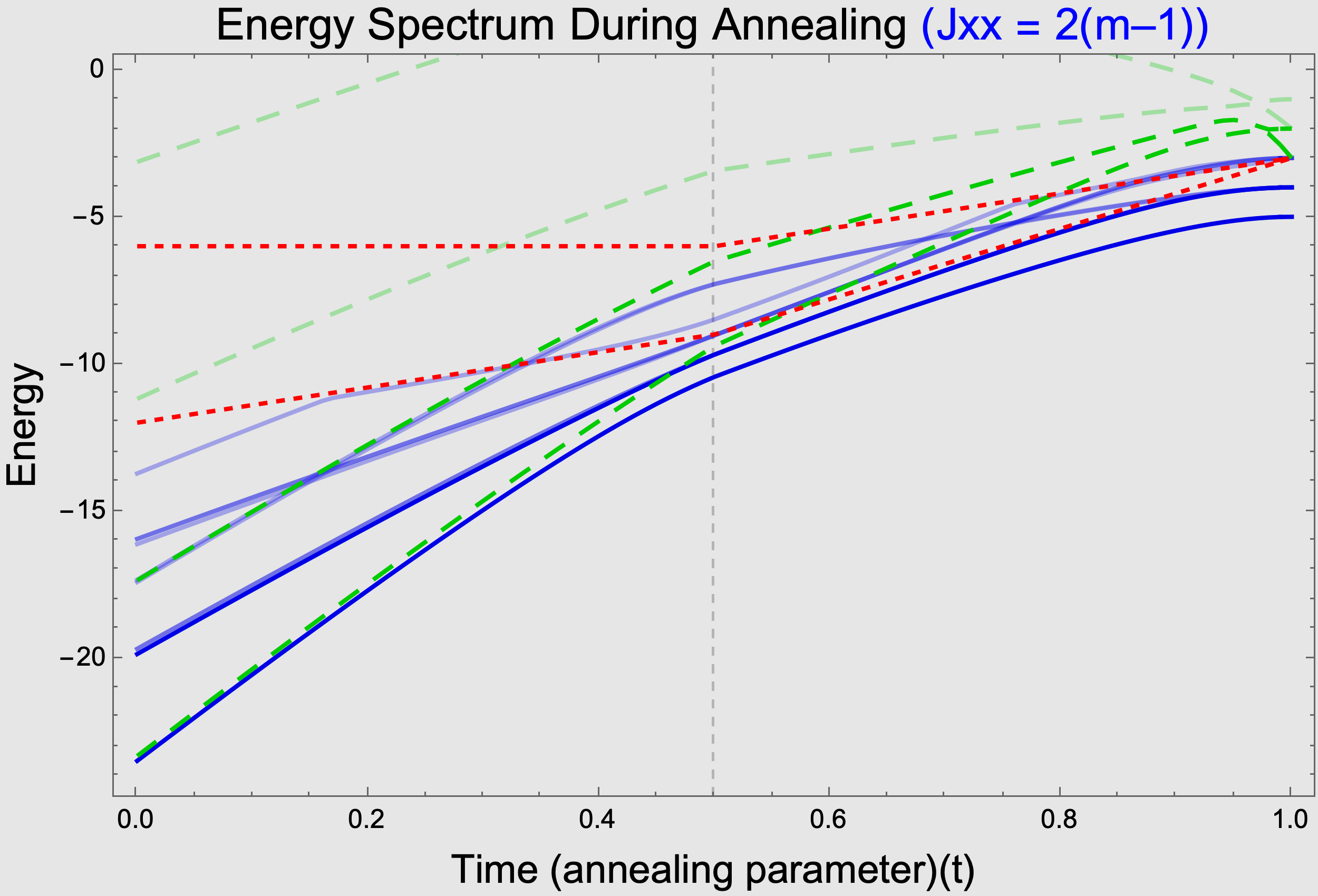}
    \caption{}
  \end{subfigure}

  \vspace{1em}

  \includegraphics[width=0.45\textwidth]{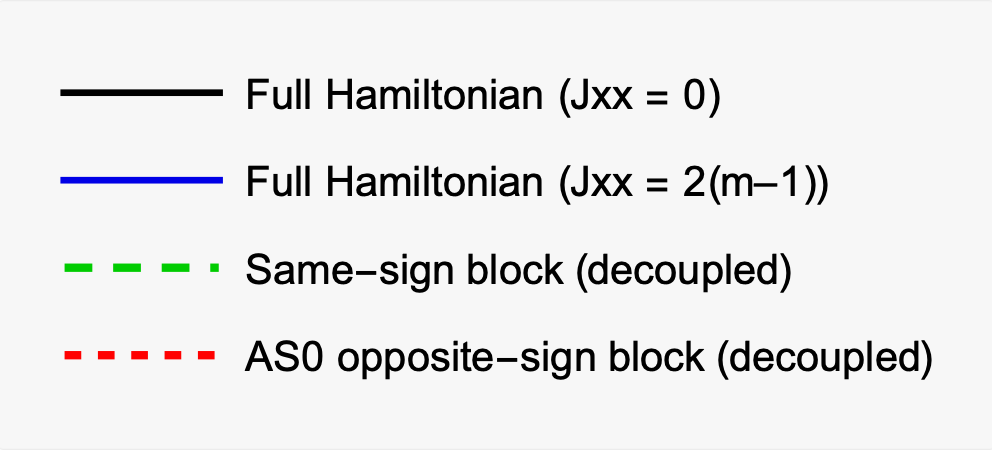}

\caption{
Spectral decomposition for the shared-structure graph \(\Gshare\), with \( m = m_l = 3 \), \( m_r = 2 \), \( m_g = 5 \), and clique size \( n_c = 9 \).  
In TFQA (a), the ground state resides entirely in the same-sign block (green dashed), producing a tunneling-induced anti-crossing.  
In \DDD{} (b), this anti-crossing is dissolved through a see-saw effect: the same-sign block (green dashed) is lifted while the opposite-sign block (red dashed) is lowered.
}
\label{fig:L2-analysis}
\end{figure}

\begin{figure}[!htbp]
  \centering
  \begin{subfigure}[b]{0.48\textwidth}
    \includegraphics[width=\textwidth]{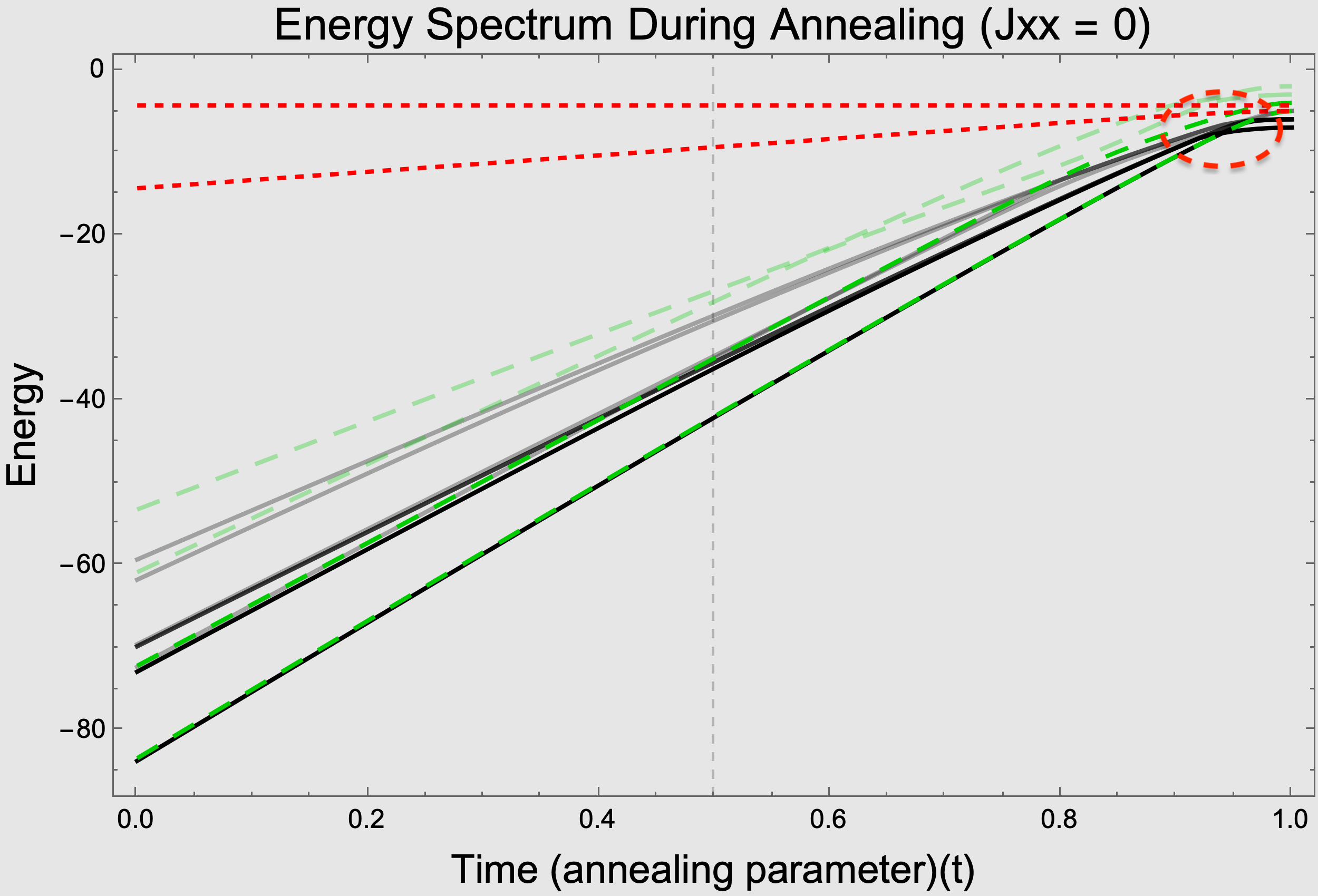}
    \caption{}
  \end{subfigure}
  \hfill
  \begin{subfigure}[b]{0.48\textwidth}
    \includegraphics[width=\textwidth]{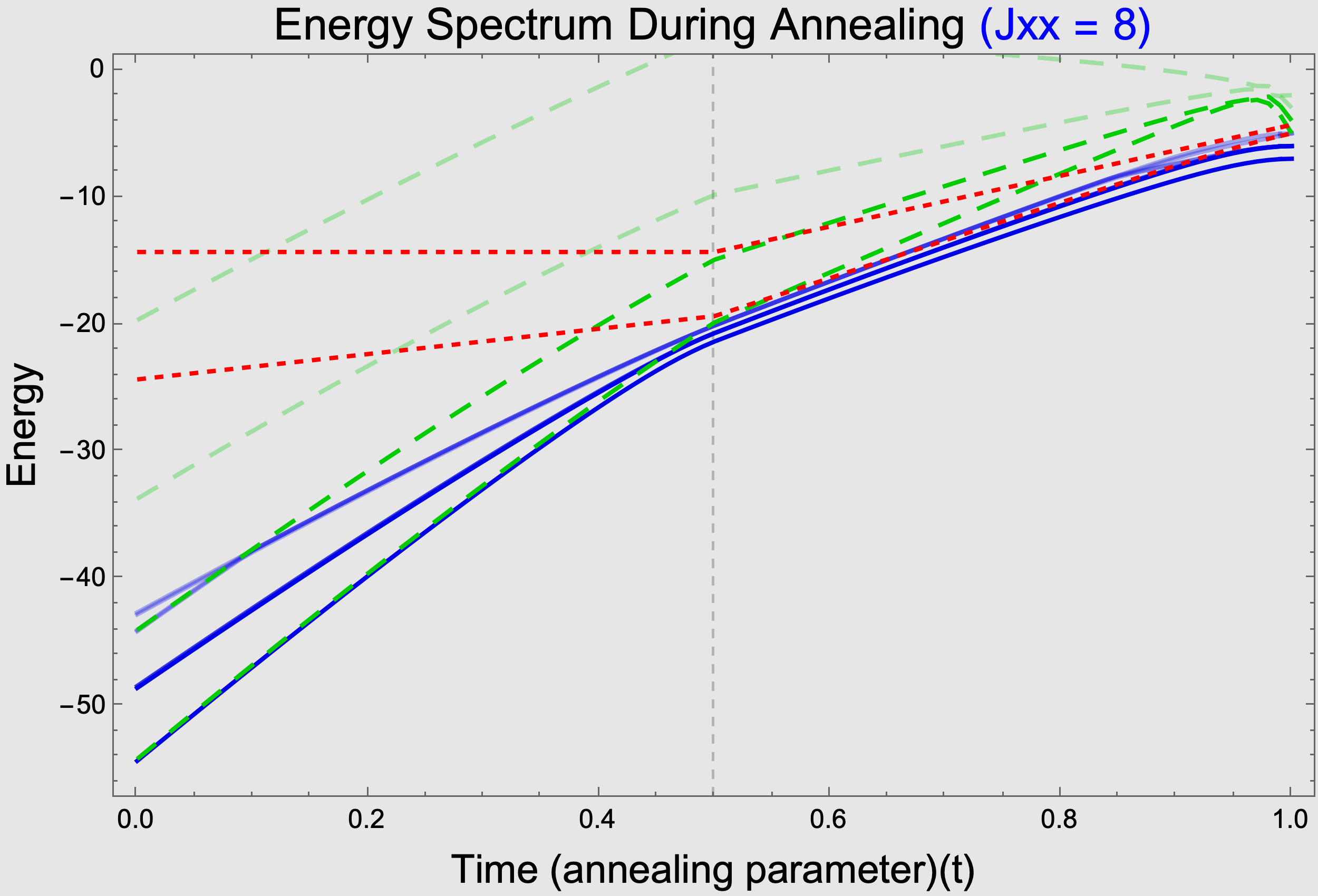}
    \caption{}
  \end{subfigure}

  \vspace{1em}

  \begin{subfigure}[b]{0.48\textwidth}
    \includegraphics[width=0.8\textwidth]{legendPanel.png}
    \caption{}
  \end{subfigure}
  \hfill
  \begin{subfigure}[b]{0.48\textwidth}
    \includegraphics[width=\textwidth]{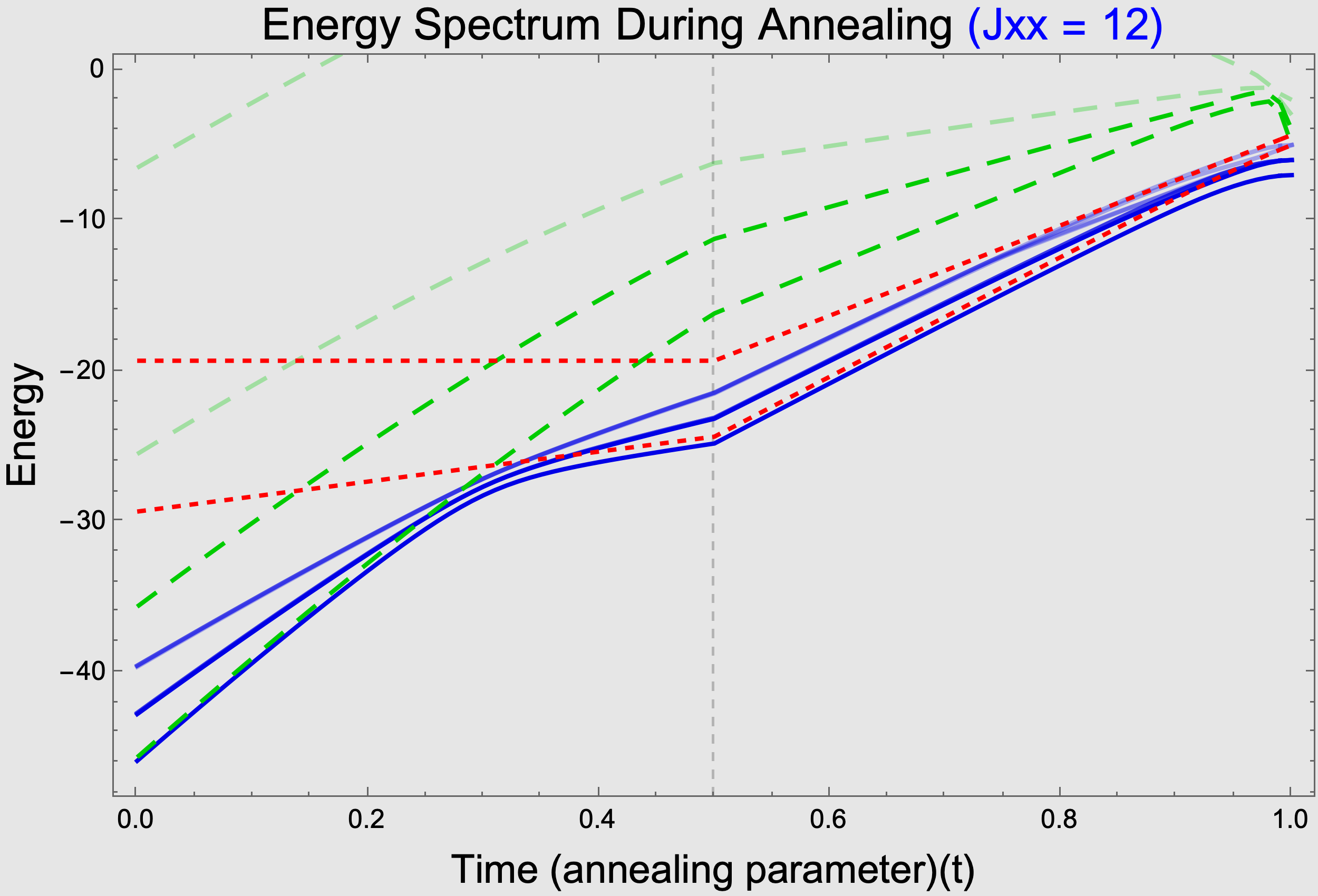}
    \caption{}
  \end{subfigure}
\caption{
Spectral decomposition for the shared-structure graph \(\Gshare\), with \( m = m_l = 5 \), \( m_g = 7 \), and clique size \( n_c = 9 \).\\
(a) Energy spectrum of the same-sign block under TFQA (\(\Jxx = 0\)), showing an anti-crossing near \( t \approx 0.9 \). \\
(b) Full spectrum under \DDD{} with \(\Jxx = 8\): the opposite-sign energies (dashed red) remain above the same-sign ground state (green), confirming successful evolution with no anti-crossing. \\
(c) Legend identifying curve types across panels. \\
(d) Large-\(\Jxx\) case (\(\Jxx = 12\)): the system exhibits two anti-crossings. A block-level anti-crossing appears near \( t \approx 0.3 \) (Stage~1), followed by an interference-involved block-level anti-crossing near \( t \approx 0.8 \) (Stage~2). Both gaps may not be small enough to allow the evolution to bypass them. While the first gap can be small, the evolution remains in the same-sign block and follows the first excited state instead of the true ground state, which is overtaken by the opposite-sign ground state. The second gap, however, can be large, preventing a transition back to the true ground state. As a result, the evolution may terminate in the first excited state.
}
\label{fig:L2-fail}
\end{figure}

\subsection{Three-Vertex Conceptual Model: Illustrating Quantum Interference}
\label{sec:V3}

We define \emph{quantum interference} as the superposition of components in a quantum system that project onto the same basis state and either interfere destructively---reducing amplitude through opposite-sign cancellation---or constructively---increasing amplitude through same-sign contributions.

To illustrate this mechanism, we introduce the three-vertex model (V3), a minimal example in which interference arises from basis rotation and occurs in both the stoquastic case (\( \Jxx = 0 \)) and the non-stoquastic case (\( \Jxx > 0 \)).
Within this framework, we emphasize the distinctive feature of \emph{sign-generating} interference: it produces negative amplitudes in the computational basis.
%and lowers the energy of the interfered state, enabling it to become the true ground state.
%Our goal is to justify the inequality \( \EGtrue \le \EGz \) through the mechanism of sign-generating quantum interference.

This section is organized as follows:
\begin{itemize}
\item \textbf{Model Description:} Formulation of the full Hamiltonian in both the computational and angular-momentum bases (Section~\ref{sec:model}).
\item \textbf{Quantum Interference via Basis Rotation:} How interference emerges from a basis change (Section~\ref{sec:interference-basis}).
\item \textbf{Numerical Illustration of Sign-Generating Interference:} Emergence of negative amplitudes in the non-stoquastic case (Section~\ref{sec:sign-generating}).
\end{itemize}

\subsubsection{Model Description}
\label{sec:model}
The model consists of three vertices \( V = \{a, b, r\} \), with associated weights \( w_a \), \( w_b \), and \( w_r \), respectively.  
The graph includes two edges, \( \{ab, ar\} \), each associated with a \( \ZZ \)-coupling, denoted \( \Jzz^{ab} \) and \( \Jzz^{ar} \).  
Additionally, there is an \( \XX \)-coupling on the edge \( ab \), denoted \( \Jxx \).

We bipartition \( V \) into \( L = \{a, b\} \) and \( R = \{r\} \).  
The set \( L \) corresponds to the local minima (\LM{}), while \( \GM = \{b, r\} \)---the maximum independent set---corresponds to the global minimum.  
See Figure~\ref{fig:V3} for a visualization of this setup.
\begin{figure}[H]
  \centering
  \includegraphics[width=0.8\textwidth]{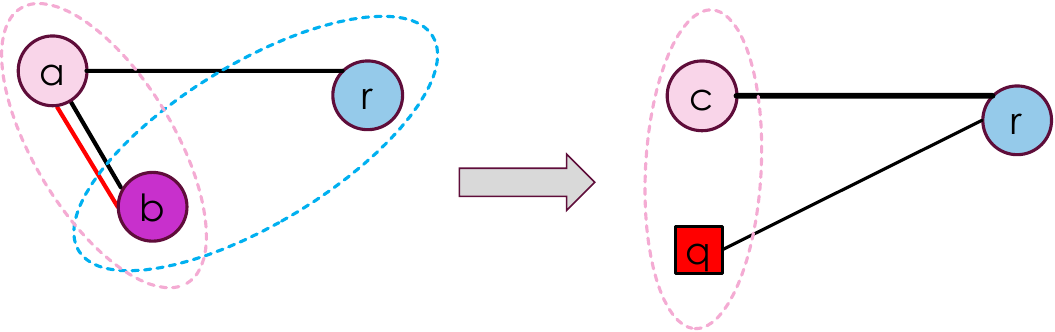}
  \caption{
    (Left) A weighted graph with three vertices $\{a, b, r\}$.
    The graph has two $\ZZ$-couplings (black edges) on $ab$ and $ar$, and one $\XX$-coupling (red edge) on $ab$.\\
    The vertex set is bipartitioned as $L = \{a, b\}$ and $R = \{r\}$.
    The set $L$ (pink dashed oval) corresponds to the local minima (\LM{}),
    while $\GM = \{b, r\}$ (blue dashed oval)---the maximum independent set---represents the global minimum.
    A local basis transformation is applied to the $L$-subsystem.\\
    (Right) The pair $\{a,b\}$ is transformed into $\{c,q\}$, where $c$ is a spin-$\tfrac{1}{2}$ subsystem,
    and $q$ is a spin-0 subsystem corresponding to an opposite-sign state \( \ket{\odot_q} \).
  }
  \label{fig:V3}
\end{figure}

\begin{remark}
  The V3 model is not meant to demonstrate structural steering in Stage~1. Since the right component contains only a single vertex (\( m_r = 1 \)), structrual steering fails, as it requires \( m_r \geq 2 \). The system therefore undergoes a tunneling-induced anti-crossing in Stage~1. At the end of Stage~1, we assume that the ground state has successfully tunneled and is now localized on the state \(\{r\}\), which is a subset of the global minimum configuration \(\GM = \{b, r\}\).
\end{remark}

\subsubsection*{Hamiltonian and Basis Construction}
Each vertex \( k \in \{a, b, r\} \) is modeled as a spin-\(\tfrac{1}{2}\) system with transverse field strength proportional to \( \sqrt{n_k} \). The local Hamiltonian at vertex \( k \), in the computational basis \( \{ \ket{1_k}, \ket{0_k} \} \), is given by
\[
\mathbb{B}_k = 
\begin{bmatrix}
 -w_k  & -\tfrac{\sqrt{n_k}}{2} \mathtt{x} \\[6pt]
 -\tfrac{\sqrt{n_k}}{2} \mathtt{x} & 0
\end{bmatrix}, \qquad w_k = w - \tfrac{n_k - 1}{4} \mt{jxx}.
\]

The \( L \)-subsystem consists of vertices \( a \) and \( b \), coupled via both ZZ- and XX-terms:
\begin{align*}
\mb{H}_L =
\mb{B}_a \otimes \mb{I}_2 + \mb{I}_2 \otimes \mb{B}_b + \Jzz^{ab} \shz{a}\shz{b} +  \tfrac{\sqrt{n_a n_b}}{4} \mt{jxx} \sigma_a^{x} \sigma_b^{x}.
\end{align*}
In the computational basis \( \{ \ket{1_a1_b}, \ket{1_a0_b}, \ket{0_a1_b}, \ket{0_a0_b} \} \), this becomes:
\[
\mb{H}_L =
\begin{blockarray}{ccccc}
& \ket{1_a1_b} & \ket{1_a0_b} & \ket{0_a1_b} & \ket{0_a0_b} \\
  \begin{block}{c(cccc)}
    \cline{2-5} 
\ket{1_a1_b} & \Jzz^{ab} - w_a - w_b & -\tfrac{\sqrt{n_b}}{2} \mt{x} & -\tfrac{\sqrt{n_a}}{2} \mt{x} & \tfrac{\sqrt{n_a n_b}}{4} \mt{jxx} \\
\ket{1_a0_b} & -\tfrac{\sqrt{n_b}}{2} \mt{x} & -w_a & \tfrac{\sqrt{n_a n_b}}{4} \mt{jxx} & -\tfrac{\sqrt{n_a}}{2} \mt{x} \\
\ket{0_a1_b} & -\tfrac{\sqrt{n_a}}{2} \mt{x} & \tfrac{\sqrt{n_a n_b}}{4} \mt{jxx} & -w_b & -\tfrac{\sqrt{n_b}}{2} \mt{x} \\
\ket{0_a0_b} & \tfrac{\sqrt{n_a n_b}}{4} \mt{jxx} & -\tfrac{\sqrt{n_a}}{2} \mt{x} & -\tfrac{\sqrt{n_b}}{2} \mt{x} & 0 \\
\end{block}
\end{blockarray}
\]

Assuming the coupling \( \Jzz^{ab} \) is large, we restrict to the low-energy subspace spanned by \( \ket{1_a0_b}, \ket{0_a1_b}, \ket{0_a0_b} \), yielding:
\[
\bar{\mb{H}}_L = \PLE \mb{H}_L \PLE =
\begin{blockarray}{cccc}
& \ket{1_a0_b} & \ket{0_a1_b} & \ket{0_a0_b} \\
  \begin{block}{c(ccc)}
    \cline{2-4}
\ket{1_a0_b} & -w_a & \tfrac{\sqrt{n_a n_b}}{4} \mt{jxx} & -\tfrac{\sqrt{n_a}}{2} \mt{x} \\
\ket{0_a1_b} & \tfrac{\sqrt{n_a n_b}}{4} \mt{jxx} & -w_b & -\tfrac{\sqrt{n_b}}{2} \mt{x} \\
\ket{0_a0_b} & -\tfrac{\sqrt{n_a}}{2} \mt{x} & -\tfrac{\sqrt{n_b}}{2} \mt{x} & 0 \\
\end{block}
\end{blockarray}
\]

This restriction removes the high-energy state \(\ket{1_a1_b} \), which is energetically penalized by large
\( \Jzz^{ab} \).

To uncover the effective angular momentum structure, we apply the basis transformation:
\[
\Ucombine =
\begin{blockarray}{cccc}
& \ket{1_c} & \ket{0_c} & \ket{\odot_q} \\
\begin{block}{c(ccc)}
  \cline{2-4}
\ket{1_a0_b} & \sqrt{\tfrac{n_a}{n_c}} & 0 & -\sqrt{\tfrac{n_b}{n_c}} \\
\ket{0_a1_b} & \sqrt{\tfrac{n_b}{n_c}} & 0 & \sqrt{\tfrac{n_a}{n_c}} \\
\ket{0_a0_b} & 0 & 1 & 0 \\
\end{block}
\end{blockarray}
\quad \text{where } n_c = n_a + n_b.
\]

The corresponding transformed basis states are:
\begin{align*}
  \begin{cases}
  \ket{0_c} &= \ket{0_a0_b},\\
\ket{1_c} &= \sqrt{\tfrac{n_a}{n_c}}\ket{1_a0_b} + \sqrt{\tfrac{n_b}{n_c}}\ket{0_a1_b}, \\
\ket{\odot_q} &= -\sqrt{\tfrac{n_b}{n_c}}\ket{1_a0_b} + \sqrt{\tfrac{n_a}{n_c}}\ket{0_a1_b}.
  \end{cases}
\end{align*}

In this basis, the restricted Hamiltonian becomes:
\[
\bar{\mb{H}}_L =
\begin{blockarray}{cccc}
& \ket{1_c} & \ket{0_c} & \ket{\odot_q} \\
\begin{block}{c(ccc)}
  \cline{2-4}
\ket{1_c} & -\left(w - \tfrac{n_c - 1}{4} \mt{jxx} \right) & -\tfrac{\sqrt{n_c}}{2} \mt{x} & 0 \\
\ket{0_c} & -\tfrac{\sqrt{n_c}}{2} \mt{x} & 0 & 0 \\
\ket{\odot_q} & 0 & 0 & -\left(w + \tfrac{1}{4} \mt{jxx}\right) \\
\end{block}
\end{blockarray}
\]

%The state \( \ket{\odot_q} \) lies entirely in the orthogonal complement of the symmetric spin-\(\tfrac{1}{2}\) subspace and represents an opposite-sign configuration.

We specialize to the case \( n_b = 1 \), \( n_a = n_c - 1 \), we have $w_b=w_r=w$,
while $w_a=  w - \tfrac{n_c - 2}{4} \mt{jxx}$.
The corresponding full transformed basis states are:
\begin{align*}
\begin{array}{rcl@{\qquad}rcl}
  \ket{0_c0_r} &=& \ket{0_a0_b0_r}, &
  \ket{0_c1_r} &=& \ket{0_a0_b1_r}, \\
  \ket{1_c0_r} &=& \sqrt{\tfrac{n_c - 1}{n_c}} \ket{1_a0_b0_r}
                + \sqrt{\tfrac{1}{n_c}} \ket{0_a1_b0_r}, &
  \ket{1_c1_r} &=& \sqrt{\tfrac{n_c - 1}{n_c}} \ket{1_a0_b1_r}
                + \sqrt{\tfrac{1}{n_c}} \ket{0_a1_b1_r}, \\
  \ket{\odot_q0_r} &=& -\sqrt{\tfrac{1}{n_c}} \ket{1_a0_b0_r}
                     + \sqrt{\tfrac{n_c - 1}{n_c}} \ket{0_a1_b0_r}, &
  \ket{\odot_q1_r} &=& -\sqrt{\tfrac{1}{n_c}} \ket{1_a0_b1_r}
                     + \sqrt{\tfrac{n_c - 1}{n_c}} \ket{0_a1_b1_r}.
\end{array}
\end{align*}

We now explicitly write down the full Hamiltonian in the two bases:
\paragraph{Full Hamiltonian in the Computational Basis.}
\begin{center}
\scalebox{0.9}{$
%  \[
  \mb{H}_{\ms{full}}^{\ms{comp}}=
\begin{blockarray}{ccccccc}
& \ket{1_a0_b1_r} & \ket{1_a0_b0_r} & \ket{0_a1_b1_r} & \ket{0_a1_b0_r} & \ket{0_a0_b1_r} & \ket{0_a0_b0_r} \\
  \begin{block}{c(cccccc)}
     \cline{2-7}
\ket{1_a0_b1_r} & - 2w +\tfrac{n_c - 2}{4} \mt{jxx} + \Jzz  & -\tfrac{1}{2}\mt{x} & \tfrac{\sqrt{n_c - 1}}{4} \mt{jxx} & 0 & -\tfrac{\sqrt{n_c - 1}}{2} \mt{x} & 0 \\
\ket{1_a0_b0_r} & -\tfrac{1}{2}\mt{x} & - w+\tfrac{n_c - 2}{4} \mt{jxx}  & 0 & \tfrac{\sqrt{n_c - 1}}{4} \mt{jxx} & 0 & -\tfrac{\sqrt{n_c - 1}}{2} \mt{x} \\
\ket{0_a1_b1_r} & \tfrac{\sqrt{n_c - 1}}{4} \mt{jxx} & 0 & -2w & -\tfrac{1}{2}\mt{x} & -\tfrac{1}{2}\mt{x} & 0 \\
\ket{0_a1_b0_r} & 0 & \tfrac{\sqrt{n_c - 1}}{4} \mt{jxx} & -\tfrac{1}{2}\mt{x} & -w & 0 & -\tfrac{1}{2}\mt{x} \\
\ket{0_a0_b1_r} & -\tfrac{\sqrt{n_c - 1}}{2} \mt{x} & 0 & -\tfrac{1}{2}\mt{x} & 0 & -w & -\tfrac{1}{2}\mt{x} \\
\ket{0_a0_b0_r} & 0 & -\tfrac{\sqrt{n_c - 1}}{2} \mt{x} & 0 & -\tfrac{1}{2}\mt{x} & -\tfrac{1}{2}\mt{x} & 0 \\
\end{block}
\end{blockarray}
%\]
$}
\end{center}

\paragraph{Full Hamiltonian in the Angular Momentum Basis.}
\begin{center}
\scalebox{0.85}{$
%  \[
  \mb{H}_{\ms{full}}^{\ms{ang}}=
\begin{blockarray}{ccccccc}
& \ket{1_c1_r} & \ket{1_c0_r} & \ket{0_c1_r} & \ket{0_c0_r} & \ket{\odot_q1_r} & \ket{\odot_q0_r} \\
  \begin{block}{c(cccc|cc)}
       \cline{2-7}
\ket{1_c1_r} & - 2w + \tfrac{n_c-1}{4} \mt{jxx} + \tfrac{n_c-1}{n_c}\Jzz & -\tfrac{1}{2}\mt{x} & -\tfrac{\sqrt{n_c}}{2} \mt{x}& 0 & - \sqrt{\tfrac{n_c-1}{n_c^2}}\Jzz & 0 \\
\ket{1_c0_r} & -\tfrac{1}{2}\mt{x} & - w +\tfrac{n_c-1}{4} \mt{jxx} & 0 & -\tfrac{\sqrt{n_c}}{2} \mt{x}& 0 & 0 \\
\ket{0_c1_r} & -\tfrac{\sqrt{n_c}}{2} \mt{x} & 0 & -w & -\tfrac{1}{2} \mt{x}& 0 & 0 \\
\ket{0_c0_r} & 0 & -\tfrac{\sqrt{n_c}}{2}\mt{x} & -\tfrac{1}{2}\mt{x} & 0 & 0 & 0 \\
\cline{2-5}
 \ket{\odot_q1_r} & {-\sqrt{\tfrac{n_c-1}{n_c^2}}\Jzz}  & 0 & 0 & 0 & - 2w  -\tfrac{1}{4}\mt{jxx} + \tfrac{1}{n_c}\Jzz & -\tfrac{1}{2}\mt{x} \\
\ket{\odot_q0_r} & 0 & 0 & 0 & 0 & -\tfrac{1}{2}\mt{x} &  - w  -\tfrac{1}{4}\mt{jxx}\\
\end{block}
\end{blockarray}
%\]
$}
\end{center}
The first four basis states span the same-sign block; the last two belong to the opposite-sign block. Coupling between blocks arises via the off-diagonal term proportional to \( \Jzz=\Jzz^{\ms{ar}} \).

\subsubsection{Explaining Quantum Interference via Basis Rotation}
\label{sec:interference-basis}

At the start of Stage~2, the evolving ground state has been steered into the
\(R\)-localized region (i.e., \( \ket{0_c1_r} \)).
As the evolution progresses, the relevant subspace is projected to
\(\mb{I}-\ket{0_r}\!\bra{0_r}=\ket{1_r}\!\bra{1_r}\), which is spanned by
\[
\big\{ \ket{1_c1_r},\ \ket{0_c1_r},\ \ket{\odot_q1_r} \big\},
\]
with effective Hamiltonian
\begin{align}
\mb{H}_{\eff}^{\ms{ang}}=
\begin{blockarray}{cccc}
& \ket{1_c1_r} & \ket{0_c1_r} & \ket{\odot_q1_r} \\
\begin{block}{c(ccc)}
\ket{1_c1_r} & - 2w + \tfrac{n_c-1}{4}\mt{jxx} + \tfrac{n_c-1}{n_c}\Jzz & -\tfrac{\sqrt{n_c}}{2}\mt{x} & - \sqrt{\tfrac{n_c-1}{n_c^2}}\Jzz \\
\ket{0_c1_r} & -\tfrac{\sqrt{n_c}}{2}\mt{x} & -w & 0 \\
\ket{\odot_q1_r} & -\sqrt{\tfrac{n_c-1}{n_c^2}}\Jzz & 0 & - 2w - \tfrac{1}{4}\mt{jxx} + \tfrac{1}{n_c}\Jzz \\
\end{block}
\end{blockarray}
\label{eq:V3-eff}
\end{align}

The ground state evolves into a superposition of the form
\[
\ket{\psi(t)} = \psi_{cr}(t)\ket{1_c1_r} + \psi_{qr}(t)\ket{\odot_q1_r} + \psi_{\bar{c}r}(t)\ket{0_c1_r},
\]
where all coefficients are nonnegative throughout the evolution as the Hamiltonian is stoquastic (in the angular-momentum basis).
The interference mechanism arises from the angular-momentum basis rotation:
\begin{align*}
\ket{1_c1_r} &= \sqrt{\tfrac{n_c-1}{n_c}}\,\ket{1_a0_b1_r} + \sqrt{\tfrac{1}{n_c}}\,\ket{0_a1_b1_r}, \\
\ket{\odot_q1_r} &= -\sqrt{\tfrac{1}{n_c}}\,\ket{1_a0_b1_r} + \sqrt{\tfrac{n_c-1}{n_c}}\,\ket{0_a1_b1_r}.
\end{align*}

Restricting to the subspace spanned by these orthonormal components, we write
\[
\ket{\psi(t)} = \psi_{cr}(t)\ket{1_c1_r} + \psi_{qr}(t)\ket{\odot_q1_r}.
\]
Expressed in the computational basis, this becomes
\begin{equation}
\label{eq:psi-decomp}
\ket{\psi(t)} = \alpha(t)\ket{0_a1_b1_r} + \beta(t)\ket{1_a0_b1_r},
\end{equation}
where
\begin{equation}
  \label{eq:alpha-beta}
  \alpha(t) = \psi_{cr}(t)\sqrt{\tfrac{1}{n_c}} + \psi_{qr}(t)\sqrt{\tfrac{n_c-1}{n_c}}, 
  \qquad
  \beta(t)  = \psi_{cr}(t)\sqrt{\tfrac{n_c-1}{n_c}} - \psi_{qr}(t)\sqrt{\tfrac{1}{n_c}}.
\end{equation}
Here \( \psi_{cr}(t) \) and \( \psi_{qr}(t) \) are the amplitudes on the same-sign basis state \( \ket{1_c1_r} \) and the opposite-sign basis state \( \ket{\odot_q1_r} \), respectively.

\begin{center}
\fbox{\parbox{0.99\linewidth}{
\textbf{Interference picture.}
The orthogonal components \( \ket{1_c1_r} \) and \( \ket{\odot_q1_r} \) interfere \emph{constructively} on
\( \ket{0_a1_b1_r} \), increasing \( \alpha(t) \) (supporting the global minimum), 
and \emph{destructively} on \( \ket{1_a0_b1_r} \), decreasing \( \beta(t) \) (dependent-set cancellation).
}}
\end{center}
In the stoquastic case (\( \Jxx=0 \)), the interference produces only nonnegative amplitudes, and both \(\alpha(t)\) and \(\beta(t)\) remain nonnegative.  
In the non-stoquastic case (\( \Jxx>0 \)), however, destructive interference can drive \(\beta(t)\) negative, providing evidence of \emph{sign-generating quantum interference}.

To show that \(\beta(t)<0\), we express \(\mb{H}_{\eff}^{\ms{ang}}\) back in the computational basis on the reduced subspace:
\begin{align}
\mb{H}_{\eff}^{\ms{comp}}=
\begin{blockarray}{ccc}
& \ket{0_a1_b1_r} & \ket{1_a0_b1_r} \\
\begin{block}{c(cc)}
\ket{0_a1_b1_r} & -2w & \tfrac{\sqrt{n_c-1}}{4}\mt{jxx} \\
\ket{1_a0_b1_r} & \tfrac{\sqrt{n_c-1}}{4}\mt{jxx} &
-2w + \tfrac{n_c-2}{4}\mt{jxx} + \Jzz \\
\end{block}
\end{blockarray}
\label{eq:V3-eff-comp}
\end{align}
Consistent with the Rayleigh variational argument for $\mb{H}_{\ms{eff}}^{(\mc{M}\mc{D})}$ above,
the ground state of \(\mb{H}_{\eff}^{\ms{comp}}\) has positive amplitude on the \GM{} state $\ket{0_a1_b1_r}$, namely \(\alpha(t) > 0\), and negative amplitude on the dependent-set state $\ket{1_a0_b1_r}$, namely \(\beta(t) < 0\), throughout Stage~2 (when the effective Hamiltonian is valid).

\subsubsection{Numerical Illustration of Sign-Generating Interference}
\label{sec:sign-generating}
To visualize the mechanism described above, we numerically compute the instantaneous ground state of the V3 Hamiltonian over the course of the annealing evolution during Stage~2. 
We plot its amplitude components in both the angular momentum basis and the computational basis, showing how quantum interference arises in both the stoquastic case (\(\Jxx = 0\)) and the non-stoquastic case (\(\Jxx = 0.6\)), and how sign generation emerges only in the latter.

The comparison highlights the essential distinction between tunneling-based and interference-based evolution.  
In the stoquastic case, the system undergoes a tunneling-induced anti-crossing as it transitions from support on local minima to support on the global minimum.  
Although quantum interference is present, it remains sign-preserving, and the system evolves entirely within the non-negative cone.  
In contrast, when \(\Jxx = 0.6\), negative amplitudes appear---evidence of sign-generating quantum interference.
Figures~\ref{fig:energy-spectrum-comparison} and~\ref{fig:matrixplot-combined} show the corresponding energy spectra and signed probability evolutions.

\begin{figure}[!htbp]
\centering
\begin{tabular}{cc}
\includegraphics[width=0.48\textwidth]{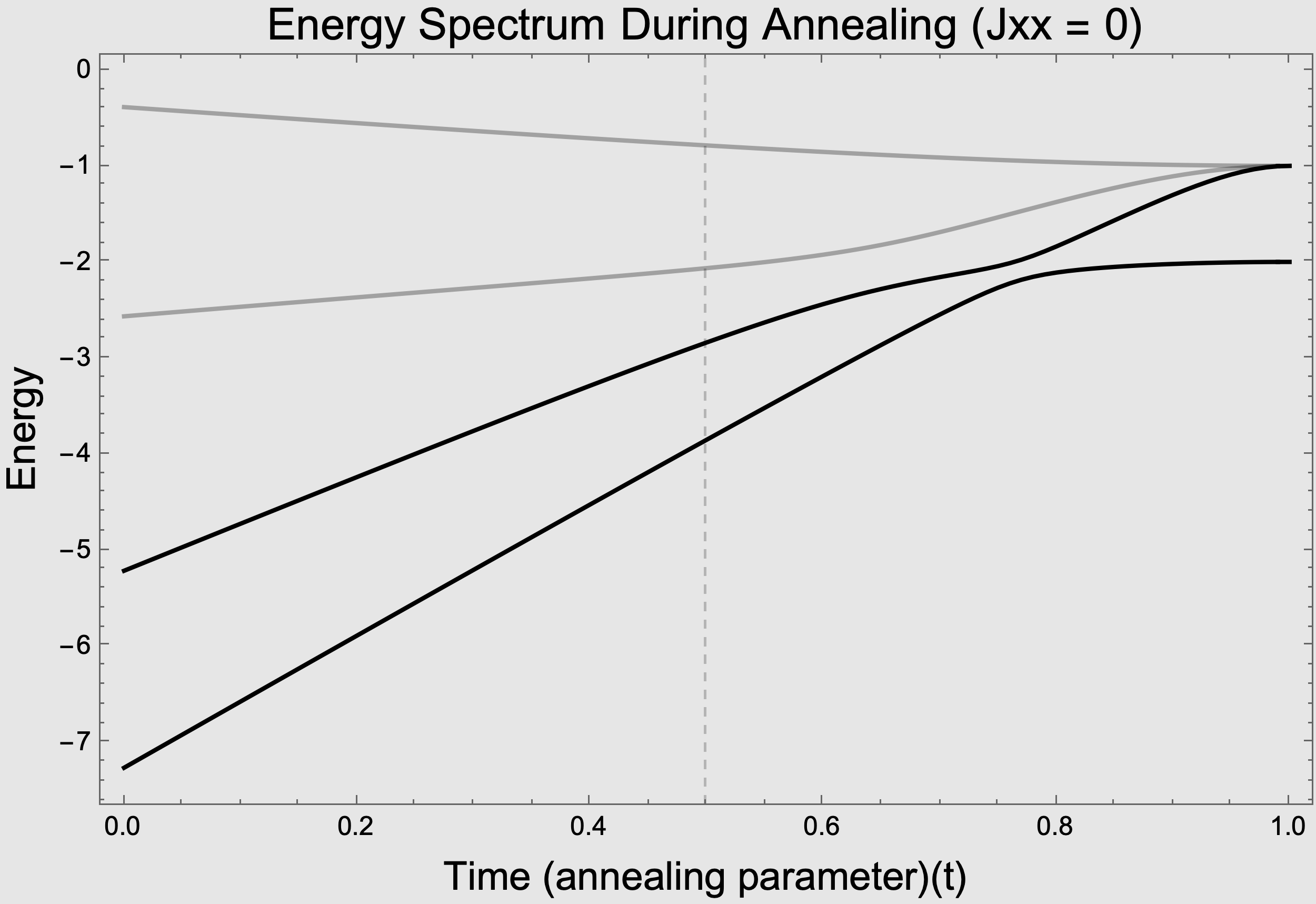} &
\includegraphics[width=0.48\textwidth]{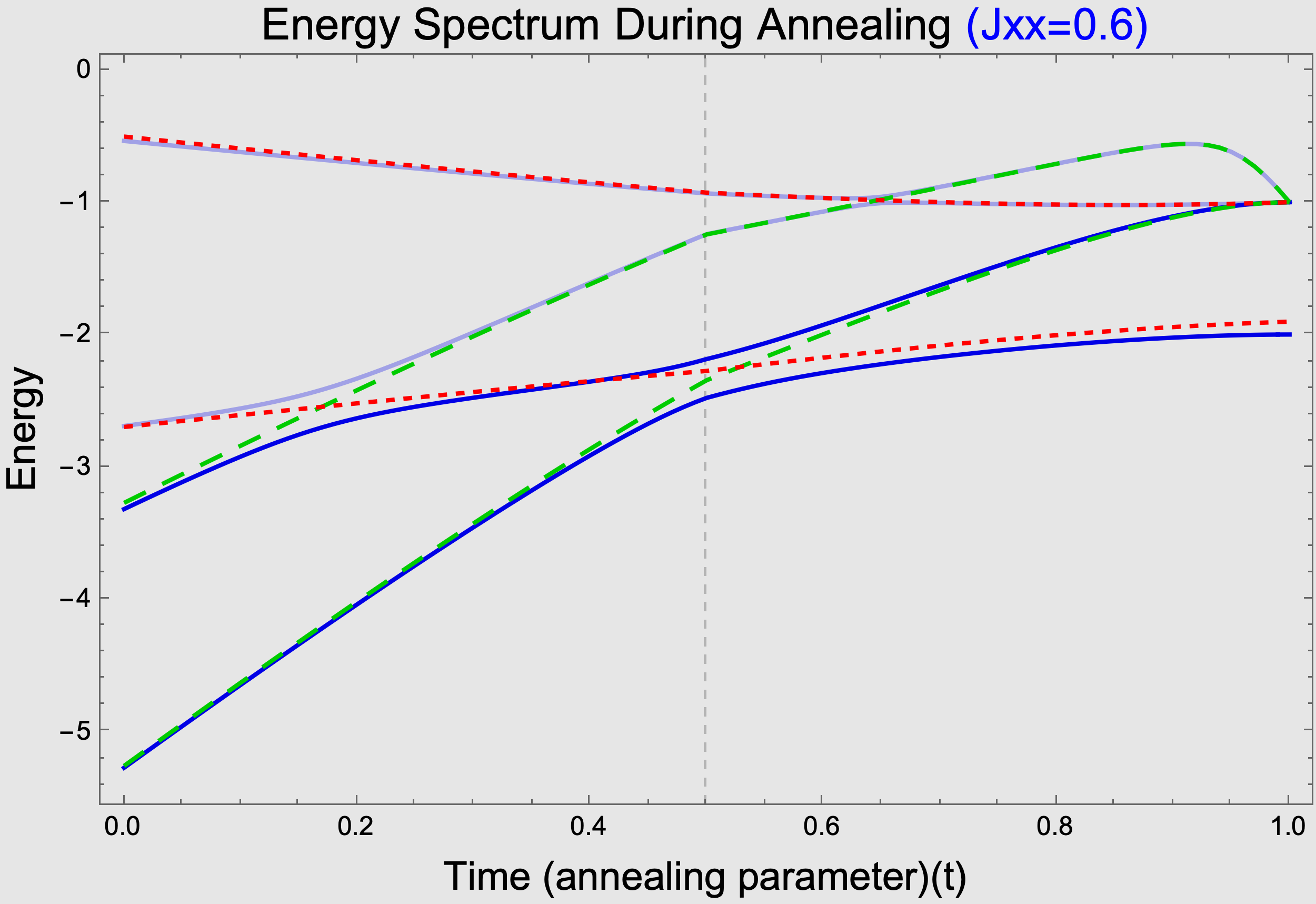} \\
\textbf{(a)} \(\Jxx = 0\) (stoquastic case) &
\textbf{(b)} \(\Jxx = 0.6\) (non-stoquastic case)
\end{tabular}
\caption{
\textbf{Energy spectrum during annealing for the V3 model at two values of \(\Jxx\)}.
(a) For \(\Jxx = 0\), the spectrum exhibits a tunneling-induced anti-crossing near \( t \approx 0.75 \), where the ground state transitions from support on local minima to support on the global minimum.
\\
(b) For \(\Jxx = 0.6\), no anti-crossing arises in Stage~2. The ground state energy of the opposite-sign block (\(\ESz\), red dotdashed) remains above the true ground state \(\EGtrue\).
}
\label{fig:energy-spectrum-comparison}
\end{figure}

\begin{figure}[!htbp]
\centering
%% \begin{tabular}{cc}
%%   \textbf{(a)}
%%   \includegraphics[width=0.8\textwidth]{MP0.png}\\
%%   \textbf{(b)} 
%% \includegraphics[width=0.8\textwidth]{TP0.png} \\
%% \hline\\
%% \hline
%% \textbf{(c)}
%% \includegraphics[width=0.8\textwidth]{MPxx06.png} \\
%% \textbf{(d)}
%% \includegraphics[width=0.8\textwidth]{TPxx06.png} 
%% \end{tabular}
\begin{tabular}{c}
  \textbf{(a)} $\Jxx=0$ \\
\includegraphics[width=0.8\textwidth]{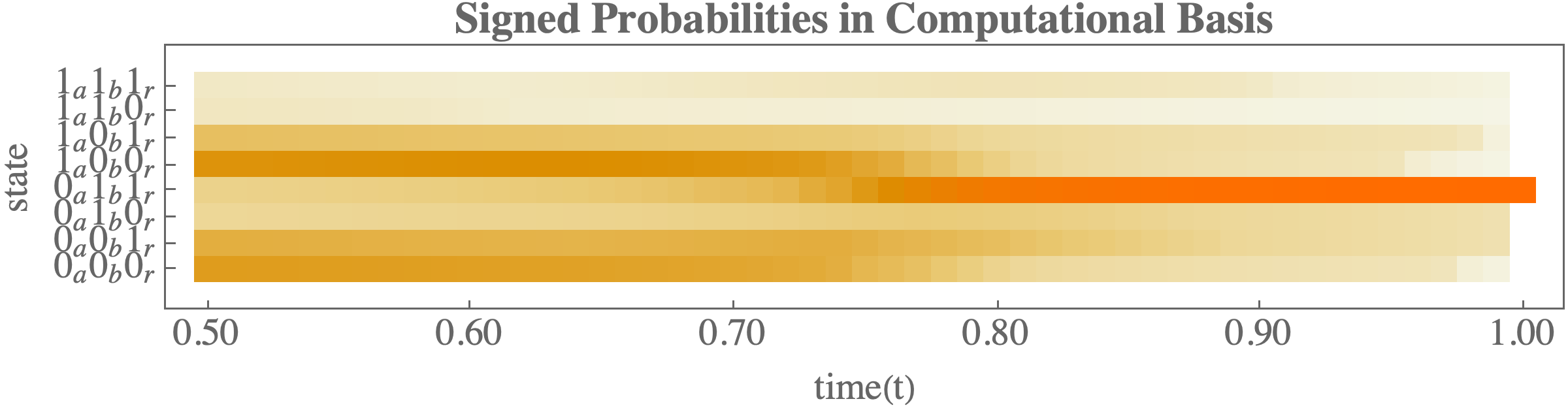} \\
\textbf{(b)} $\Jxx=0$ \\
\includegraphics[width=0.8\textwidth]{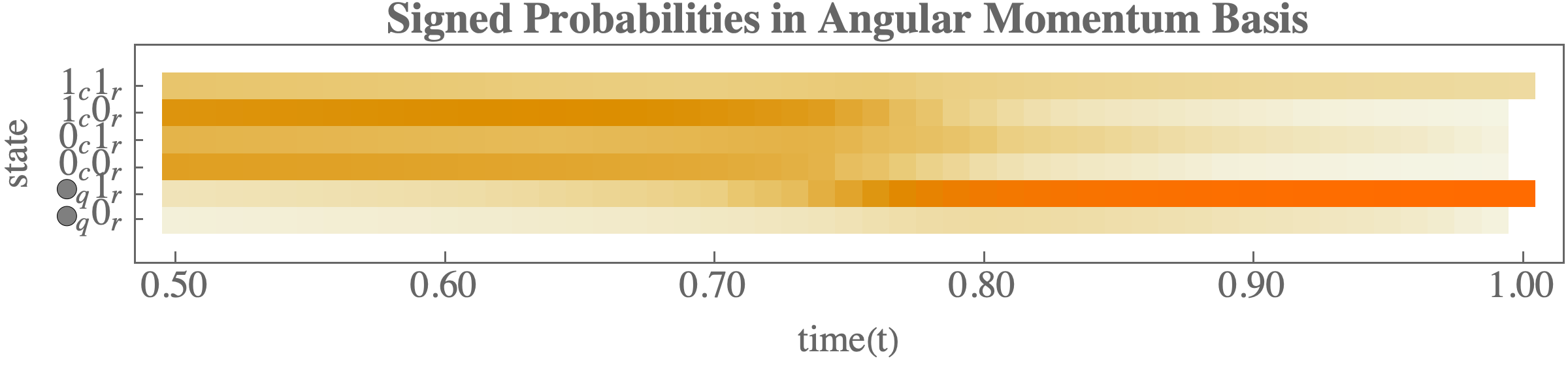} \\
\hline
\textbf{(c)} $\Jxx=0.6$ \\
\includegraphics[width=0.8\textwidth]{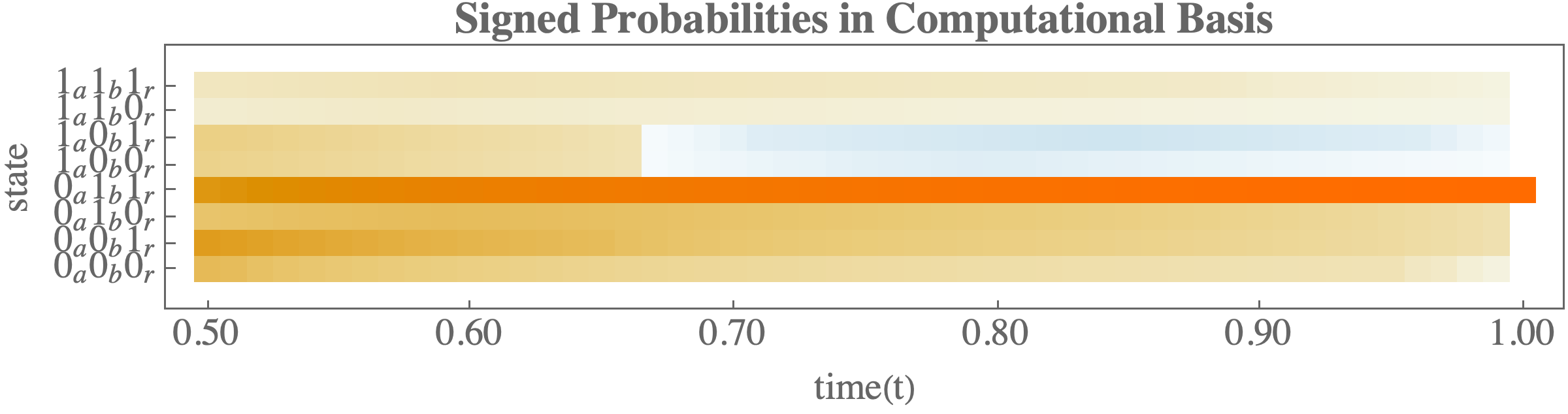} \\
\textbf{(d)} $\Jxx=0.6$ \\
\includegraphics[width=0.8\textwidth]{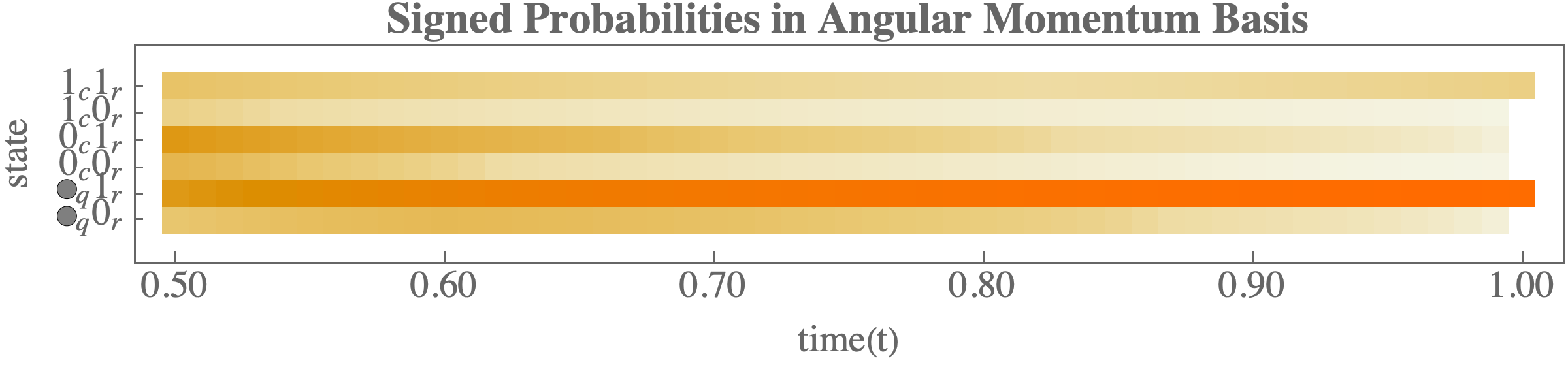} \\
\end{tabular}

\caption{
\textbf{Signed probability evolution of the instantaneous ground state in the V3 model, shown in two bases (computational and angular-momentum).}  
 Each entry shows \( \mathrm{sign}(\psi_i)\cdot|\psi_i|^2 \), where \( \psi_i \) is the amplitude on basis state \( i \). Color shading encodes signed magnitude: positive (orange), negative (blue), and near-zero (white). \\
 (a,b) Stoquastic case (\(\Jxx=0\)). In (a), all amplitudes remain non-negative, with no blue region. Around \(t \approx 0.75\),
a sharp transition occurs, indicative of a tunneling-induced anti-crossing. In (b), interference between \(\ket{1_c1_r}\) and the opposite-sign
component \(\ket{\odot_q1_r}\) is present but remains sign-preserving.
This shows that quantum interference is present even in the stoquastic case, but it remains sign-preserving.\\
(c,d) Non-stoquastic case (\(\Jxx=0.6\)). In (c), the transition from orange to blue on the dependent-set state \(\ket{1_a0_b1_r}\) reflects destructive interference and the onset of a negative amplitude, while constructive interference amplifies the global-minimum state \(\ket{0_a1_b1_r}\). In (d), the same evolution is shown in the angular-momentum basis: the first four states belong to the same-sign block, and the last two to the opposite-sign block.
}
\label{fig:matrixplot-combined}
\end{figure}

%% %Full system
\section{Full-System Extension via Iterative Application}
\label{sec:full-system}

In this section, we extend the analysis to the full system by iteratively applying the two-phase procedure to successive bipartite substructures.
In each iteration, a driver subgraph \( G_{\ms{driver}}^{(k)} \)---and thus the parameters \( m_k \) and \( n_c^{(k)} \)---is identified.

The correctness of applying the algorithm iteratively---by incrementally building up the driver graph, i.e., \( \cup_{k=1}^{it} G_{\ms{driver}}^{(k)} \)---follows from the following three facts:
\begin{enumerate}
    \item The transverse field is global: \( \mt{x} = (1 - t)\Gamma_1 \).
    \item Each iteration has its own stage-separation parameter \( \Gamma_2^{(k)} \).
    \item We use the latest \( \Jzz^{(k)} \).
\end{enumerate}

\begin{remark}
The upper bound on \( \Jzz \), $\Jzzsteer$, which ensures that the ground state steers into the $\GM$-supporting region during Stage~1, needs to be enforced only in the \emph{last iteration}. In earlier iterations, the algorithm localizes into another $\LM$-supporting region (associated with the current critical local minima), so the \( \Jzzsteer \) (and thus $\Jxxsteer$) bound may be violated without affecting the algorithm's performance. Each iteration applies its own \XX-driver and chooses a \( \Jzz \) value appropriate to that iteration's structure. It is only in the final iteration---when the goal is to steer into the $\GM$-supporting region---that the upper bound on \( \Jzzsteer \) must be respected. This allows for adaptive, iteration-dependent tuning of \( \Jzz \), even when the clique sizes vary or structural conditions differ across iterations.
\end{remark}

We set \( \Gamma_2^{(k)} = m_k \), and define the corresponding transition point
\(
t_k := 1 - \tfrac{\Gamma_2^{(k)}}{\Gamma_1} \leq \tfrac{1}{2}.
\)
For each iteration \( k \), we define
\(
\alpha_k := 2 \cdot \tfrac{\Gamma_2^{(k)} - 1}{\Gamma_2^{(k)}}, \mbox{ and }\Jxx^{(k)} := \alpha_k \, \Gamma_2^{(k)}.
\)
To simplify implementation across multiple iterations, we fix
\(
\Gamma_1 := K \cdot \max_k \Gamma_2^{(k)},
\)
for some constant \( K \).
This choice guarantees that the transverse field schedule for the two-stage evolution,
\(
\x{t} = (1 - t)\Gamma_1,
\)
can be reused identically across all iterations. While each iteration may use a different structure-dependent value \( \Gamma_2^{(k)} \), they all share the same global annealing profile governed by \( \Gamma_1 \).

In particular, the system Hamiltonian builds up incrementally.
The Hamiltonian at iteration \( it \) for Stage~0 is given by
\[
\mb{H}_0(t) = \x{t} \mb{H}_{\ms{X}} + \sum_{k=1}^{it} \left( \jxx{t}^{(k)}
\sum_{(i,j) \in \edge(G_{\ms{driver}}^{(k)})} \sigma_i^x \sigma_j^x \right)
+ \mt{p}(t) \mb{H}_{\ms{problem}}^{(k)}, \quad t \in [0,1],
\]
with
\(
\x{t} = (1 - t)(\Gamma_0 - \Gamma_1) + \Gamma_1, 
\jxx{t}^{(k)} = t \Jxx^{(k)}, 
\mt{p}(t) = t.
\)
The transverse-field term is
\(
\mb{H}_{\ms{X}} = - \sum_{i} \sigma_i^x,
\)
and the problem Hamiltonian is
\[
\mb{H}_{\ms{problem}}^{(k)} = \sum_{i \in \ver(G)} (-w_i)\, \shz{i}
+ \sum_{k=1}^{p} \left(
\Jzz^{\ms{clique}} \sum_{(i,j) \in \edge(\cup_{k=1}^{it} G_{\ms{driver}}^{(k)})} \shz{i} \shz{j} \right)
+ \Jzz^{(k)} \sum_{(i,j) \in \edge(G) \setminus \edge(\cup_{k=1}^{it} G_{\ms{driver}}^{(k)})} \shz{i} \shz{j}.
\]
Here, we take
\(
\Jzz^{(k)} :=  \left(1 + \tfrac{\sqrt{n_c^{(k)}} + 1}{2} \right).
\)

%\paragraph{Iteration-Dependent \(\Jzz\) Condition (Upper Bound).}

The system Hamiltonian during the main two-stage evolution is
\[
\mb{H}_1(t) = \x{t} \mb{H}_{\ms{X}} + \sum_{k=1}^{it} \left( \jxx{t}^{(k)}
\sum_{(i,j) \in \edge(G_{\ms{driver}}^{(k)})} \sigma_i^x \sigma_j^x \right)
+ \mb{H}_{\ms{problem}}^{(k)}, \quad t \in [0,1].
\]
The transverse field is given by
\(
\x{t} = (1 - t)\Gamma_1.
\)

The time-dependent coupling \( \jxx{t}^{(k)} \) for iteration \( k \) is given by
\[
\jxx{t}^{(k)} =
\begin{cases}
{\Jxx}^{(k)}, & \text{for } t \in [0, t_k] \quad \Leftrightarrow \quad \x{t} \in [\Gamma_2^{(k)}, \Gamma_1], \\
\alpha_k \, \x{t}, & \text{for } t \in [t_k, 1] \quad \Leftrightarrow \quad \x{t} \in [0, \Gamma_2^{(k)}].
\end{cases}
\]
Figure~\ref{fig:two-iteration} illustrates how the algorithm progressively removes small-gap obstructions by lifting one critical local minimum per iteration.
\begin{figure}[p]
  \centering
  \includegraphics[width=0.53\textwidth]{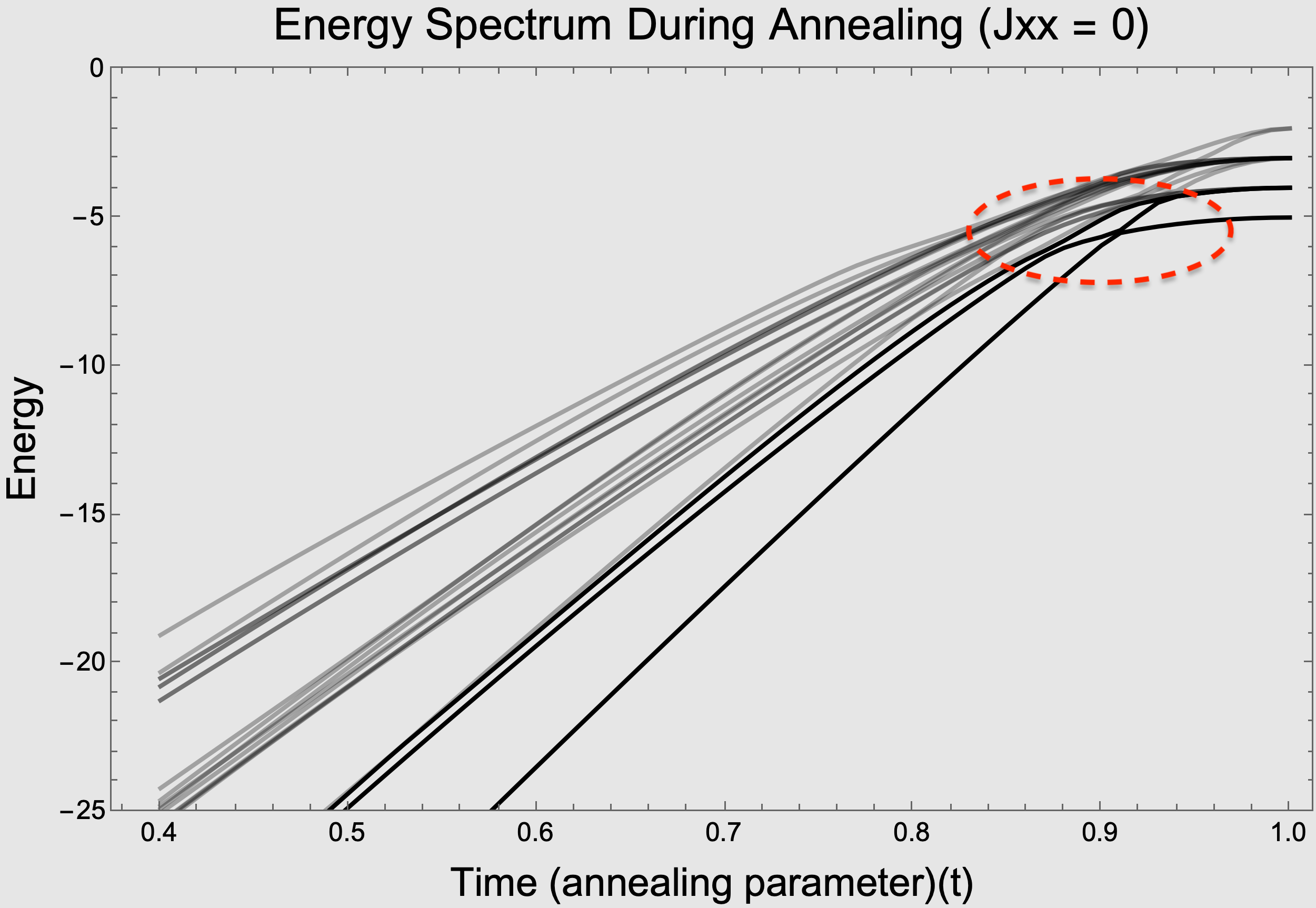} \\[1ex]
  \includegraphics[width=0.53\textwidth]{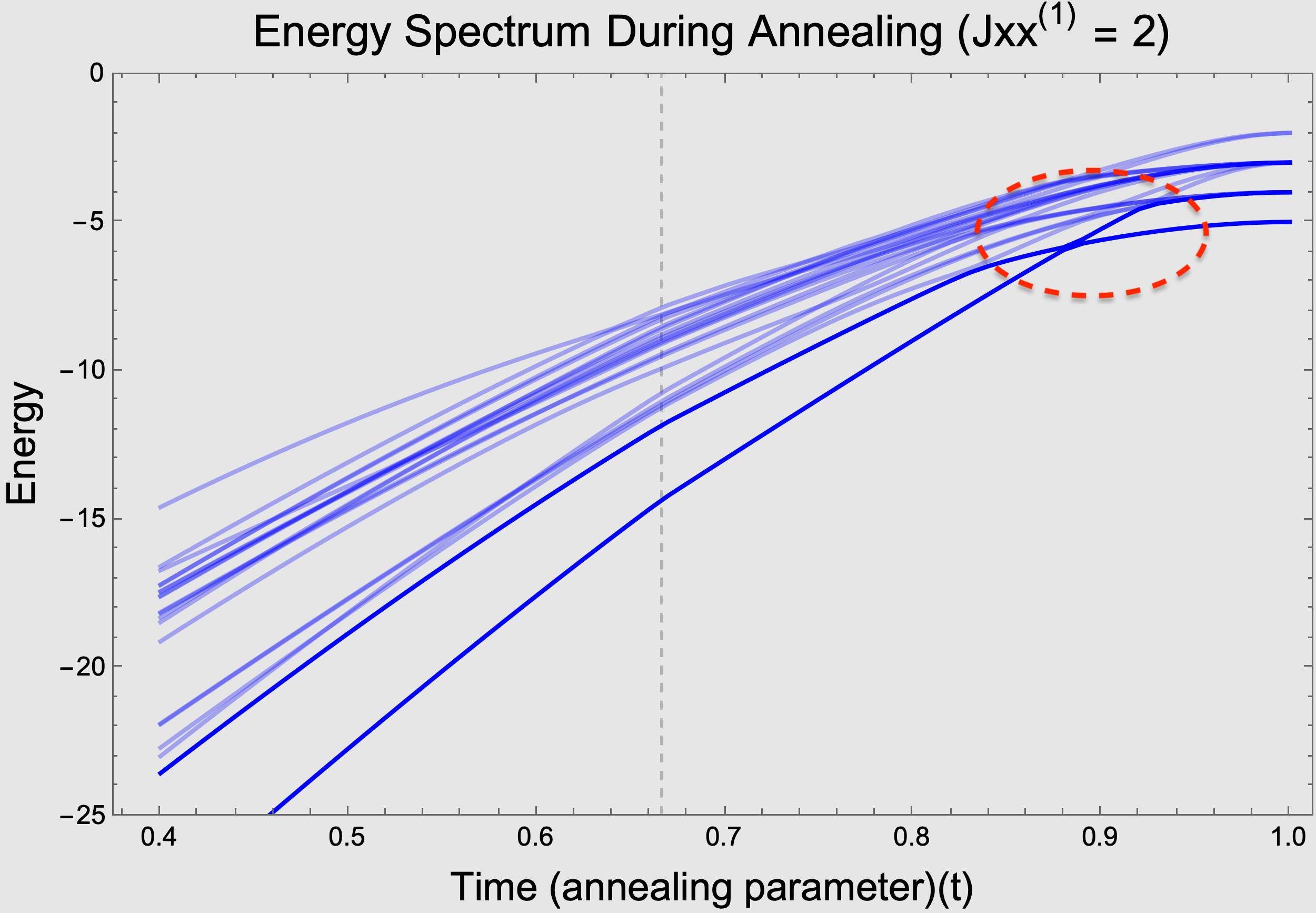} \\[1ex]
  \includegraphics[width=0.53\textwidth]{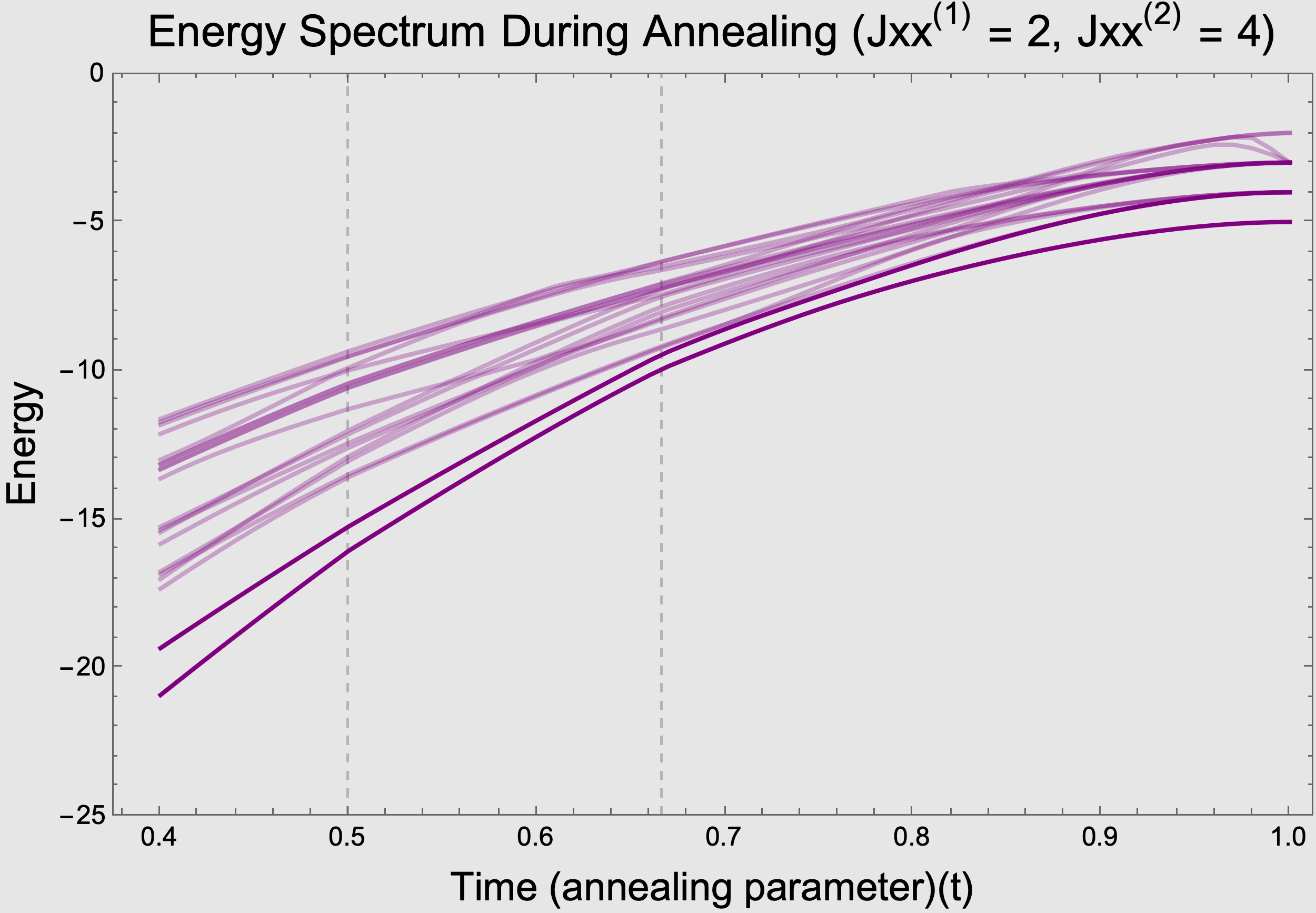}
  \caption{
  \textbf{Lifting Anti-Crossings in Two Iterations.}
  Top: \(\Jxx = 0\). The energy level associated with the global minimum (\(\GM\)) exhibits two anti-crossings with those of two local minima, highlighted by the red dotted oval. The first LM has \(m_l = 2\), \(n_c = 30\), and shares both vertices with \(\GM\); the second has \(m_1 = 3\), \(n_c = 10\), and shares one vertex with \(\GM\). The global minimum has \(m_g = 5\).
  Middle: \({\Jxx}^{(1)} = 2\); the first anti-crossing is lifted, but the second remains (dotted oval).
  Bottom: \({\Jxx}^{(2)} = 4\); both anti-crossings are lifted, and the system evolves smoothly to the global minimum.
  }
\label{fig:two-iteration}
\end{figure}

%% %Discussion

\section{Conclusion and Outlook}
\label{sec:discussion}

We have developed a systematic framework for analyzing a non-stoquastic
quantum annealing algorithm, \DDD{}, which achieves exponential speedup
on a structured family of Maximum Independent Set (MIS) instances. Our
approach departs from the traditional spectral-gap perspective and
instead measures the efficiency of the algorithm by the presence or absence of an
anti-crossing.
The key idea is to infer it
directly from the crossing 
behavior of bare energy levels of relavent subsystems,
without explicitly constructing
the effective two-level Hamiltonian.

The analytical foundation of our approach rests on three structural 
components, each supported by analytical derivations and numerical 
evidence. First, a block decomposition of the Hamiltonian into 
same-sign and opposite-sign components, derived from the angular momentum 
structure of the \MIC{} underlying the $\XX$-driver graph. Second, the 
identification of a see-saw energy effect induced by the non-stoquastic 
parameter $\Jxx$, which raises the energy associated with local minima in the 
same-sign block while lowering that of opposite-sign blocks. Third, 
the derivation of analytical bounds on $\Jxx$ (together with upper 
bounds on $\Jzz$) that support a two-stage annealing schedule ensuring 
smooth ground state evolution without anti-crossings.

Together, these structural components reveal two key mechanisms 
responsible for the speedup through a smooth evolution path:
\begin{itemize}
\item \textbf{Structural steering}: energy-guided localization within the same-sign 
block that steers the ground state smoothly into the $\GM$-supporting 
region, bypassing tunneling-induced anti-crossings.

  \item \textbf{Sign-generating quantum interference}: production of negative amplitudes that enables an opposite-sign path through destructive interference in the computational basis.
\end{itemize}

The analysis of these mechanisms is supported by both analytical derivations 
and numerical validation, and can, in principle, be made rigorous with 
further work.
We now summarize the key structural assumptions and aspects of the analysis that may be further formalized.

\subsection*{Foundational Assumptions and Future Refinement}
\begin{itemize}
  \item \textbf{Iterative correctness.} 
  We justify that the full-system evolution can be accurately approximated 
  by analyzing each bipartite substructure in isolation.

  \item \textbf{Worst-case structure.} 
  We argue that the shared-structure graph represents the worst-case 
  configuration within each bipartite subproblem. This allows us to 
  conservatively bound the relevant parameters.

  \item \textbf{Bare energy approximation.}
  Our key idea is to infer the presence or absence of anti-crossings 
  directly from the crossing behavior of bare energy levels of relevant 
  subsystems (\LM{} and \GM{}), without explicitly constructing the effective 
  two-level Hamiltonian. This reduction makes the problem tractable and 
  captures the essential structure. It can be partially justified by
  effective Hamiltonian arguments, but a fully rigorous treatment would require 
  bounding the error of this approximation.

  \item \textbf{Inter-block coupling.} 
  We argue that Stage~1 is confined to the same-sign block by assuming 
  the inter-block coupling is weak. Our numerical results confirm that 
  taking $\Jzz=\Jzzsteer$ suffices in practice, but deriving a sharp 
  analytical bound on $\Jzzinter$ remains an open direction for 
  future refinement.

  \item \textbf{Structural steering.} 
  We argue that Stage~1 achieves structural steering by ordering the 
  energies of the $R$-inner blocks so that the ground state is guided 
  into the $R$-localized region. This is supported by large-scale 
  numerical evidence, while a fully rigorous justification remains 
  future work.

  \item \textbf{Emergence of negative amplitudes.} 
  We demonstrate the emergence of negative amplitudes using the effective 
  Hamiltonian $\mb{H}_{\ms{eff}}^{(\mc{M}\mc{D})}$, which is sufficient 
  for our analysis. A fully rigorous treatment would require mathematically 
  justifying this construction as the correct effective reduction of the 
  full system.
\end{itemize}

\subsection*{Quantumness and Absence of Classical Analogues}
The emergence of negative amplitudes---produced by sign-generating interference due to the Proper Non-Stoquastic Hamiltonian---serves as a witness to the quantum nature of the speedup.
Classical simulation algorithms, such as quantum Monte Carlo methods~(see \cite{Hen2021} and references therein), rely on the non-negativity of wavefunctions in the computational basis and break down in regimes where interference induces sign structure. This places the algorithm beyond the reach of eventually stoquastic annealing and efficient classical simulation.

From the perspective of classical solvers, the effect of sign-generating interference may be hard to reproduce classically.  
Algorithmically speaking, the \(\XX\)-driver with appropriately tuned coupling strength \(\Jxx\) enables a form of collective suppression of local minima, induced by edge couplings that effectively reduce vertex weights---while selectively preserving the global minimum, even in the presence of partial overlap.
To the best of our understanding, no classical algorithm is able to replicate this effect, which appears to be a genuinely quantum capability.

Together, these points position \DDD{} as a candidate for demonstrating quantum advantage.

\subsection*{Experimental Validation of the Quantum Advantage Mechanism}

A by-product of our analysis is that a bipartite substructure graph $\Gshare$ with $m_l n_c + m_r$ vertices can be effectively reduced to $2 m_l + m_r$ vertices, with $n_c$ appearing only as a parameter, see Figure~\ref{fig:V8-effective}. This reduction enables the construction of \textit{small-scale models} amenable to \textit{experimental verification} of quantum advantage. In particular, a minimal $3$-vertex model (with $m_l = m_r = 1$) demonstrates the emergence of negative amplitudes through sign-generating interference, thereby providing a direct coherence test for the negative-amplitude phase structure in the computational basis, while an $8$-vertex model (with $m_l = 3$, $m_r = 2$, $m_g = m_l + m_r = 5$) as shown in  Figure~\ref{fig:V8-effective}(b), not only exhibits negative amplitudes (see Figures~\ref{fig:neg-amp}) but also shows a measurable speedup compared to the stoquastic case with $\Jxx = 0$ (see Figures~\ref{fig:L2} and \ref{fig:L2-analysis}). Both models are within reach of current gate-model quantum computers via simulation~\cite{Lloyd1996,Berry2007}.
\begin{figure}[!htbp]
  \centering
  \begin{subfigure}[b]{0.25\textwidth}
    \centering
    \includegraphics[width=\textwidth]{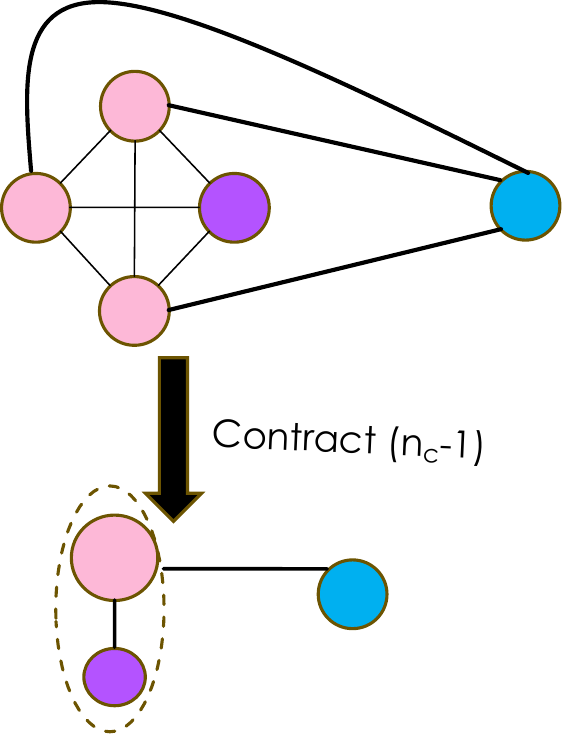}
    \caption{}
  \end{subfigure}
  \hspace{1.5cm} % manual spacing
  \begin{subfigure}[b]{0.17\textwidth}
    \centering
    \includegraphics[width=\textwidth]{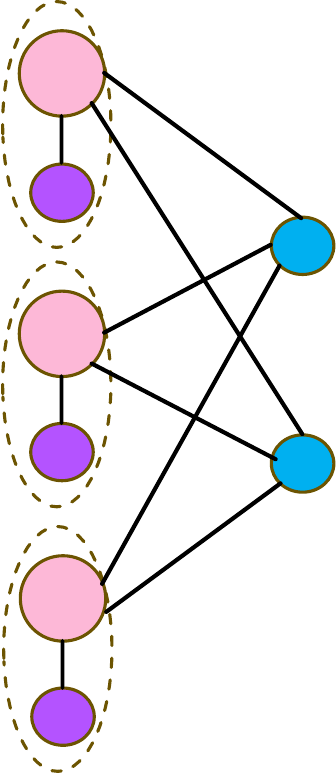}
    \caption{}
  \end{subfigure}

  \caption{
    (a) Schematic contraction. A clique consisting of $n_c-1$ pink vertices and one purple vertex. 
    The subclique of $n_c-1$ pink vertices is contracted into a single pink vertex. 
    Each pink vertex is $\ZZ$-coupled to both the purple and blue vertices, while the purple and blue vertices have no direct edge. 
    $\XX$-couplings occur between pink--pink and pink--purple pairs. 
    (b) Example of a bipartite substructure graph $\Gshare$ with $m_l n_c + m_r$ vertices reduced to $2 m_l + m_r$ vertices. 
    Shown is the case $m_l=3$, $m_r=2$, yielding an 8-vertex effective graph where each blue vertex is adjacent to each pink vertex. 
    The (global) maximum independent set has size $m_g = m_l + m_r$, consisting of all purple and blue vertices.
  }
  \label{fig:V8-effective}
\end{figure}

\subsection*{On the Relaxation of the Structured Input Assumption}
Our analysis is developed under a structured assumption on the problem graph of MIS, which we refer to as the \emph{\GIC{} assumption}: each critical degenerate local minima---corresponding to a set of maximal independent sets of fixed size---is formed by a set of \emph{independent cliques} (\MIC{}). This assumption underlies both the design of the \(\XX\)-driver and the block decomposition of the Hamiltonian. The independence of cliques is assumed primarily to allow efficient identification during Phase~I.
In principle, one can allow the cliques to be dependent (\MDC{}), meaning that some edges are permitted between cliques. In this case, the cliques may need to be identified heuristically rather than exactly. Under suitable conditions, the algorithm may remain robust even when applied to \MDC{} structures.
More generally, each critical degenerate local minima might consist of a set of disjoint \MDC{}s. A clique-based driver graph may still be constructed in such cases, but the generalization becomes more intricate and may require further structural insights.
Other future directions include the study of weighted MIS instances and adaptive strategies for selecting or optimizing \(\Jxx\) during the anneal.

In conclusion, the results presented here establish both a theoretical and numerical foundation for exponential quantum advantage via the \DDD{} algorithm.  
Beyond the algorithm itself, a key by-product of our analysis is the identification of \textit{effective small-scale models} that distill the essential dynamics into experimentally accessible settings.  
These models provide a concrete opportunity to verify the quantum advantage mechanism on currently available gate-model quantum devices through simulation.  
We hope this work not only offers a foundational framework for quantum optimization algorithm design and analysis, but also inspires the development of experimental devices equipped with the required $\XX$ non-stoquastic couplers, capable of directly implementing the presented algorithm.

\section*{Acknowledgment}
This work was written by the author with the help of ChatGPT (OpenAI),
which assisted in refining the presentation and in expressing the intended ideas with greater clarity and precision.  
The author thanks Jamie Kerman for introducing her to the angular-momentum basis 
and for early discussions that contributed to the companion work~\cite{Choi-Limitation}.
We also thank Siyuan Han for helpful comments.
We recognize that this manuscript may not cite all relevant literature.  
If you believe your work should be included in a future version, please contact the author with the appropriate references.

\appendix

\subsection*{Appendix A: Font Conventions for Notation}
We adopt the following conventions throughout:  
\begin{itemize}
\item Hilbert space / Subspace / Basis: calligraphic, e.g.\ \(\mc{V}, \mc{B}\).  
  \item Hamiltonian / Matrix: blackboard bold, e.g.\ \(\mb{H}, \mb{B}\).  
  \item Time-dependent quantity: typewriter, e.g.\ \(\mt{x} := \mt{x}(t), \mt{jxx}:=\mt{jxx}(t)\).  
  \item Named object / Abbreviation: capital typewriter, e.g.\ \LM{}, \GM{}, \MLIS{}.  
\end{itemize}

\end{document}